\documentclass[11pt,a4paper]{article}
\usepackage{amssymb}
\usepackage{graphicx}
\usepackage{amsmath}
\usepackage{makeidx}
\usepackage{indentfirst}
\usepackage[T1]{fontenc}
\usepackage[utf8]{inputenc}

\topmargin- 20mm
\newtheorem{assumption}{Assumption}[section]

\newcounter{assumptionnum}[section]
\newcounter{resultnum}[section]
\newtheorem{consequence}{Consequence}[section]

\newcounter{consequencenum}[section]
\newtheorem{conclusion}{Conclusion}[section]

\newcounter{conclusionnum}[section]
\newcounter{conditionnum}[section]
\newcounter{conjecturenum}[section]
\newtheorem{example}{Example}[section]

\newcounter{examplenum}[section]
\newcounter{exercisenum}[section]
\newtheorem{lemma}{Lemma}[section]

\newcounter{lemmanum}[section]
\newcounter{notationnum}[section]
\newtheorem{theorem}{Theorem}[section]

\newcounter{theoremnum}[section]
\newtheorem{definition}{Definition}[section]

\newcounter{definitionnum}[section]
\newtheorem{corollary}{Corollary}[section]

\newcounter{corollarynum}[section]
\newtheorem{remark}{Remark}[section]

\newcounter{remarknum}[section]
\newtheorem{proposition}{Proposition}[section]

\newcounter{propositionnum}[section]
\newcounter{acknowledgementnum}[section]
\newcounter{algorithmnum}[section]
\newcounter{axiomnum}[section]
\newcounter{casenum}[section]
\newcounter{claimnum}[section]
\newtheorem{conseq}{Consequence}[section]

\newcounter{conseqnum}[section]
\newtheorem{convention}{Convention}[section]

\newcounter{conventionnum}[section]
\newcounter{summarynum}[section]
\newtheorem{principle}{Principle}[section]

\newcounter{principlenum}[section]
\newcounter{problemnum}[section]
\newenvironment{proof}[1][]{\textbf{Proof.} }{}

\begin{document}

\title{On axiomatic formulation of gravity and matter field theories with
MDRs and Finsler-Lagrange-Hamilton geometry on (co)tangent Lorentz bundles}
\date{January 16, 2018}
\author{ \textbf{Sergiu I. Vacaru}\thanks{\textit{The address for post
correspondence:\ 67 Lloyd Street South, Manchester, M14 7LF, the UK;}
\newline
emails: sergiu.vacaru@gmail.com ; sergiuvacaru@mail.fresnostate.edu } \\
{\qquad } \\
{\small \textit{Physics Department, California State University at Fresno,
Fresno, CA 93740, USA}} \\
{\small and } {\small \textit{\ Project IDEI, University "Al. I. Cuza" Ia\c
si, Romania}}{\qquad } }
\maketitle

\begin{abstract}
We develop an axiomatic geometric approach and provide an unconventional review  of modified gravity theories, MGTs, with modified dispersion relations, MDRs, encoding Lorentz invariance violations, LIVs, classical and quantum random effects, anisotropies etc. There are studied Lorentz-Finsler like theories elaborated as extensions of general relativity, GR, and
 quantum gravity, QG, models and constructed on (co) tangent Lorentz bundles, i.e. (curved) phase spaces or locally anisotropic spacetimes. An indicator of MDRs is considered as a functional on various type functions depending on phase space coordinates and physical constants. It determines respective generating functions and fundamental physical objects (generalized metrics, connections and nonholonomic frame structures) for relativistic models of Finsler, Lagrange and/or Hamilton spaces. We show that there are canonical almost symplectic differential forms and adapted (non) linear connections which allow us to formulate equivalent almost K\"{a}hler-Lagrange / - Hamilton geometries. This way, it is possible to unify geometrically various classes of (non) commutative MGTs with locally anisotropic gravitational, scalar, non-Abelian gauge field, and Higgs interactions. We elaborate on theories with Lagrangian densities containing massive graviton terms and bi-connection and bi-metric modifications which can be modelled as Finsler-Lagrange-Hamilton geometries. An example of short-range locally anisotropic gravity on (co) tangent Lorentz bundles is analysed. We conclude that a large class of such MGTs admits a self-consistent causal axiomatic formulation which is similar to GR but involving generalized (non) linear connections, Finsler metrics and adapted frames on phase spaces. Such extensions of the standard model of particle physics and gravity offer a comprehensive guide to classical formulation of MGTs with MDRs, their quantization, applications in modern astrophysics and cosmology, and search for observable phenomena and experimental verifications. An appendix contains historical remarks on elaborating Finsler MGTs and a summary of author's results  in twenty directions of research on (non) commutative/ supersymmetric Finsler geometry and gravity; nonholonomic geometric flows, locally anisotropic superstrings and cosmology, etc. 

\vskip2pt

\end{abstract}

\tableofcontents

\bigskip


\section{Introduction}

Over the past two decades, the literature on classical and quantum gravity, QG, and accelerating cosmology has been increased substantially involving spacetime models with non-Riemannian geometries and modified gravity theories, MGTs. Various approaches to commutative and noncommutative theories, when Planck-scale features and deformed classical and quantum symmetries are described by modified dispersion relations, MDRs, and possible (local) Lorentz invariance violations, LIVs, have been developed in
\cite{v83,vstoch96,vjmp96,hooft96,vap97,vnp97,castro97,amelino98,vapny01,magueijo04,girelli06,castro07,kostelecky11,kostelecky12,
mavromatos11,mavromatos13b,vcqg11,vgrg12} and references therein. Such works have provided a series of important results on
QG and string phenomenology and physics of relativistic particles propagating and interacting in effective media with MDRs. There were elaborated models for doubly-special and/or deformed-special relativity; theories with LIVs, deformed curved phase spacetimes with (non) commutative and/or (non) associative variables etc. A subclass of MGTs was formulated as models of (generalized) Finsler geometry, see Refs. \cite{vplb10,vijmpd12,vijgmmp08,vcqg10,vmon98,vmon02,vmon06,laemmerzahl08,rosati15,basilakos13,
hohmann13,hohmann16,kouretsis10,pfeifer12} for reviews, critical remarks and applications.

The Einstein gravity theory (i.e. general relativity, GR) has to be modified into a theory on (co) tangent bundles to Lorentz manifolds if MDRs are stated, for instance, as small deformations of standard quadratic dispersion relations in special relativity theory, SRT, and (locally) in GR, by some indicator functions depending on velocity/ momentum variables. Such
constructions were elaborated for generalized spinor, gauge and string gravity models with Finsler--Lagrange-Hamilton configurations \cite{vjmp96,vap97,vnp97,vhsp98,svcqg13,svvijmpd14}. There were studied also models on higher order (co) tangent bundles with possible supersymmetric / noncommutative variables and almost symplectic structures. A series of results on MGTs with local anisotropy and Finsler like modifications were reviewed and presented in \cite{vmon98,vmon02,vmon06,vrev08,vplb10}. It was
shown that such constructions involve MDRs which can be derived for propagation of point mass classical particles and/or models with quantum variables and (non)commutative relations describing quantum fundamental field interactions
\cite{vplbnc01,vncs01,ingarden54,ingarden03,ingarden04,ingarden08,ingarden08a,vjmp05,castro07,vpla08,vjmp09,yang09,yang10}. Various classes of locally anisotropic (non) commutative/ associative classical and quantum field and spacetime theories can be described
naturally in terms of nonholonomic variables as some models of commutative and/or noncommutative geometries. For corresponding parameterizations and distributions of geometric/ physical objects, we can construct (effective) Lagrange--Finsler, (dual) Hamilton--Cartan, and/or other type phase space and spacetime models. We cite also some results on Hamilton geometry for
generalizations of standard physical theories and general relativity, GR, contained in Part I of monograph \cite{vmon02} and  article \cite{avjmp09}.\footnote{In our works, we use various terms like pseudo Lagrange / Finsler space/geometry and/or, equivalently, relativistic Lagrange geometry/ mechanics. There are also elaborated the concepts of "pseudo, i.e. relativistic" Finsler, Hamilton or other types non-Riemannian geometries. The term pseudo-Lagrange is used for a corresponding Hessian (Lagrange metric) like in pseudo-Euclidean (locally) and/or pseudo-Riemannian geometry.  Mathematicians in many cases write semi-Riemannian (instead of pseudo-Riemannian) when such a space is endowed, at least locally and/or effectively, with metric structures of pseudo-Euclidean signature, which in this work are labeled $(+,+,+,-)$, on Lorentz manifolds and $(+,+,+,-;+,+,+,-)$, on (co) tangent Lorentz bundle. This encodes a very important experimental fact that the physical interactions of fields and propagation of particles are described by a (maximal) constant speed of light and possible polarizations in certain effective media. (Pseudo) Finsler-Lagrange-Hamilton geometries (which will be defined and studied rigorously in next sections) are determined by nonlinear quadratic elements for which the concept of signature cannot be defined in a general form. Nevertheless, we assume and prove that for certain very general conditions and realistic physical models we can always introduce an associated (effective) quadratic form uniquely determined by a fundamental nonlinear quadratic form following certain physically motivated geometric principles. For such effective quadratic forms (metrics), we can consider the concept of signature, classify some small perturbations of physical fields and propagation of particles to be relativistic, or non-relativistic, ones in an effective media (aether, or vacuum spacetime model). In a number of recent work, the terms "pseudo" and "relativistic" are used for MGTs with local properties described by an associated quadratic form are like in Lorentz geometry. In a more general context, there are considered nonlinear and linear connection structures, nonholonomic constraints and generalized symmetries, which are different from those in (pseudo) Riemannian geometry.} Those works contain a number of original results on relativistic models on (co) tangent bundles and/or on nonholonomic manifolds, with higher order extensions and, in particular, for constructing new classes exact solutions in Einstein gravity by introducing nonholonomic 2+2 (co) fibred structures, see \cite{vexsol98,vjhep01,vijgmmp07,vijtp10,vijtp10a,
vijgmmp11,vjpcs13,vepjc14,vvyijgmmp14a,gvvepjc14}. Recently, a model of Hamilton geometry with MDRs was studied in \cite{barcaroli15}. 

Planck-scale modifications of dispersion relations are studied in phenomenological particle models, cosmology, and astrophysics. There are some indications that Finsler like generalizations allows us to account for the Planckian structures of relativistic classical and quantum field interactions by extending theories on conventional configuration space. Various ideas and
geometric and physical models were analyzed in \cite{amelino96,amelino98,amelino02,amelino02a,amelino14,castro08,castro08a,
castro14,girelli06,gvvcqg15,kostelecky11,kostelecky12,kouretsis10,kouretsis14,kouretsis14a,mavromatos10,mavromatos10a,mavromatos13b,
russell15,schreck15,schreck15a,schreck16,svcqg13, svvijmpd14,vjpcs11,vcqg11,vgrg12}. In such approaches, the standard-model extensions originate from the idea of spontaneous breaking of the Lorentz and CPT symmetries in the string theory and involved vacuum expectation values of tensor fields with spacetime indices, all resulting in certain models with MDRs and LIVs.

The goal of this work is to formulate a self-consistent geometric approach to MGTs constructed on phase spaces encoding MDRs and LIV when such theories extend GR on (co) tangent Lorentz bundles. We show that such models can be elaborated equivalently as Finsler-Lagrange-Hamilton gravity theories. Applying N-adapted variational and nonholonomic geometric methods, there are proved/ formulated generalized Einstein equations. There are derived also dynamical equations and (effective) matter sources for locally anisotropic scalar fields, non--Abelian gauge and Higgs fields and distortions of (non) linear connection structures. This article provides a general theoretical background for a series of future partner papers on exact solutions with Lagrange-Hamilton variables, quantization of such theories, and applications in modern cosmology and astrophysics. Such works will develop the anholonomic frame deformation method, AFDM, for constructing exact solutions in MGTs formulated in explicit Hamilton like variables. Various applications of AFDM and examples of Finsler like generic off-diagonal commutative and/or noncommutative, supersymmetric/ brane/string and other types of black ellipsoid/ hole, wormhole, solitonic and/or cosmological solutions have been elaborated and reviewed in Refs. \cite{vexsol98,vmon98,vjhep01,vpcqg01,vsbd01,vmon02,vtnpb02,vsjmp02,vijmpd03,vijmpd03a,vjmp05,
vmon06,avjgp09,vijtp10,vcqg10,vijtp10a,vijgmmp11,vcqg11,vepl11,vijtp13,vjpcs13, vport13,vepjc14,vepjc14a,vvyijgmmp14a,
gvvepjc14,gvvcqg15}. Alternative approaches on Finsler gravity theories and attempts to construct physical models and find approximate solutions can be studied in \cite{horvath50,takano68,takano74,tavakol09,tayebi11,tadmon12,chang08,chang12,
havaloyes13,laemmerzahl12,itin14,li14,li14a,lin14,li16,minguzzi15,minguzzi15a,minguzzi16,silagadze15,rahman15,rahman16,caponio16}. It is considered that readers are familiar with standard results on mathematical relativity, geometry of (non) linear connections in fiber bundles, spinor differential geometry, and exact classical solutions in GR (see, for instance, \cite{hawking73,kramer03,misner73,wald82}). On certain technical topics under consideration, we shall refer to specific monographs and reviews where details and ample biographic information can be found. This article is also oriented to open-minded scholars and young researchers on geometry and physics who are curious about application of advanced geometric methods for solving new extraordinary open problems in modern classical and quantum gravity, cosmology and astrophysics, and standard and nonstandard particle physics.

The paper is organized as follows:\ In section \ref{mgttmt}, we formulate a geometric approach to gravity theories with MDRs and LIVs modelled on relativistic Lagrange-Finsler and Hamilton-Cartan phase spaces. There are provided important motivations and main assumptions which are necessary for elaborating such geometric and physical models. There are defined canonical (non) linear connections, metrics and almost symplectic structures determined by MDRs encoding LIVs and possible contributions from QG, massive
gravity and bi-metric and bi-connection terms. Corresponding curvature, torsion and nonmetricity tensors (and related Ricci and Einstein tensors, scalar curvatures of various classes of Finsler like connections and their distortions) are constructed in N-adapted forms for (co) tangent Lorentz bundles.

Section \ref{saxiom} is devoted to a study of general principles for formulating MGTs with MDRs on (co)tangent Lorentz bundles. It is proven that such extensions of the Einstein gravity can be performed in canonical forms for Lagrange and Hamilton like variables. There are defined minimal actions and Lagrange densities for Einstein-Yang-Mills-Higgs systems with MDRs and analyzed possible contributions by massive gravitons and theories with bi-metric locally anisotropic structure. Actions and sources are considered for short-range gravity models with LIVs and anisotropic interactions.

Following geometric and N-adapted variational methods, we study (in section \ref{smodensteq}) possible generalizations of the Einstein equations for theories with MDRs on (co) tangent Lorentz bundles. There are derived gravitational and matter field equations for Lagrange-Hamilton spaces, EYMH systems, massive and bi-metric locally anisotropic MGTs and short-range phenomenological gravitational theories with LIVs resulting in locally anisotropic configurations. The constructions are provided both in coordinate-free and N-adapted coefficient forms. We speculate on axiomatic approaches to modeling and geometrizing MGTs on (co) tangent Lorentz bundles.

In section \ref{scd}, there are concluded the main results. We speculate on geometrization of MGTs with MDRs and the physical picture of Finsler-Lagrange-Hamilton gravity theories. In brief, we summarize what has been achieved with Finsler modifications of the GR and discuss what is missing. Open problems and further perspectives in study cosmological Einstein--Hamilton and locally anisotropic Einstein-Dirac configurations determined by MDRs, elaborating quantum gravity models on generalized (non) commutative and (non) associative phase spaces are analyzed.

Proofs of theorems are sketched in some forms being accessible both for researchers on mathematical physics and phenomenology of particle physics. We prefer the so-called abstract geometric method for proofs and emphasize the possibility of alternative N-adapted variational approaches. Certain technical details are given in Appendix \ref{aspt}. Throughout the main part
of this article and Appendix \ref{ashr}, there are also provided brief historical comments on "relativistic" Finsler like theories (with effective metrics of local Lorentzian signature) and generalizations of the Einstein gravity with MDRs on Lagrange-Hamilton spaces.\footnote{\label{fn01a} For reviews on relativistic developments and MGTs, we cite \cite{vmon98,vmon02,vmon06,avjmp09,vrev08,vijgmmp08,vplb10,vijmpd12} and references therein. Here we note that for local Euclidean signatures, the main ideas on Lagrange-Hamilton geometries formulated as generalized Finsler spaces on (co) tangent bundles were proposed and studied in a series of works due to J. Kern, M. Matsumoto, S. Ikeda, see
\cite{kern74,matsumoto66,matsumoto86,ikeda76,ikeda77,ikeda78,ikeda79,ikeda81,ishikawa80,ishikawa81}. There were elaborated  on supersymmetric and higher order generalizations, almost K\"{a}hler models and applications in particle physics and gravity. Later, similar constructions and some applications in non-relativistic geometric mechanics classical field models were developed in a series of monographs \cite{miron94,miron00,miron03}, see critical remarks in Appendix \ref{ssectromania}. Those books were based on some hundreds of papers containing as the main author the name Radu Miron (a member of the Romanian Academy of Sciences), who falsified and omitted to stipulate a series of original and important former contributions of a number of researchers from Japan, former USSR, Hungary, Germany, USA, China etc. and his former co-authors like G. Atanasiu, A. Bejancu, V. Oproiu, and tenths others. During last 60 years, more than 250 articles were published (the first papers contained co-authors who "disappeared" in furhter developments) in "not-accessible" Romanian journals and local preprints.  Later those results were re-published by R. Miron (as an individual or first author, and with various modifications) in more than 30 monographs. This can be checked, for instance, tracking the names in MathSciNet. Contributions of Western authors and even from the former USSR were censored in Romania. A series of monographs and collections of works were published by Kluwer, Hadronic Press, Romanian Academy etc. under the names R. Miron and a few other politically and ideologically selected (by communists and Ceau\c{s}escu's secret service) co-authors who sworn as doctor-docents an "absolute devote" to Romanian dictator's wife, Elena Ceau\c{s}escu. Those works were on developments and mechanical applications of (higher order) Lagrange-Finsler-Hamilton-Cartan geometries. No matter where those articles and books were published (in Romania, Japan, or in some Western Countries), the content was  elaborated  in a style of "higher order hidden plagiarism", with falsification of results and politically / ideologically screened teams of authors. Unfortunately, none ideological and communist political lustration of former communist and secrete service leaders, and administrators in science, was in Romania (like it was, for instance, in Eastern Germany or other former "socialist" Countries). There is not a complete access even to files on the history of physics and mathematics when "ideology and secret service" were deeply involved.} Such a survey of author's research (together with his former students and co-authors) performed beginning 1982 in (former) USSR, the R. Moldova, and (later) in Romania and Western Countries, and supported by a number of NATO, DAAD, UNESCO and CERN fellowships and grants, is important because it provides evidences that those results anticipated a series of recent publications on modified Finsler gravity theories and applications in modern cosmology.

\section{Phase Spaces with MDRs and Finsler--Lagrange--Hamilton Geometry}

\label{mgttmt} This section contains an introduction to the geometry of nonholonomic (co) tangent Lorentz bundles which is used for elaborating relativistic Finsler-Lagrange and Hamilton-Cartan gravity theories. There are developed a series of concepts and constructions considered in Part I of monograph \cite{vmon02}, on relativistic and modified Hamilton spaces, and  for MGTs defined by MDRs and LIVs generalizations of GR as generalized Finsler theories. For reviews and critical remarks on Finsler-Lagrange-Hamilton geometry and gravity theories with nonholonomic Lorentzian manifolds and tangent bundles, and applications in modern cosmology and astrophysics, we cite \cite{vplb10,vijmpd12,vijgmmp08,avjmp09,vcqg10,vmon98,vmon02,vmon06} and references therein.

\subsection{Motivations and assumptions}

\label{ssassumpts}

In semi-classical commutative and/or noncommutative models of MGTs and/or QG theories, a  MDR can be written locally in a general form
\begin{equation}
c^{2}\overrightarrow{\mathbf{p}}^{2}-E^{2}+c^{4}m^{2}=\varpi (E,%
\overrightarrow{\mathbf{p}},m;\ell _{P}).  \label{mdrg}
\end{equation}%
An \textbf{indicator} of deformations/ modifications $\varpi (...)$ encodes in a functional form possible contributions of modified geometric and physical theories and LIVs. Such MDRs can be extended to dependencies on spacetime coordinates $x^{i}=(x^{1},x^{2},x^{3},x^{4}=ct)$ (one can be introduced extra dimensions) on a standard Lorentz manifold $V$ and/or
generalizations with metric-affine structures. In explicit form, certain classes of
$\varpi (x^{i},E,\overrightarrow{\mathbf{p}},m;\ell _{P})$ are chosen and studied following theoretical/ phenomenological arguments, determined experimentally, and/or computed in a generalized classical/quantum theory of gravity and matter filed interactions. The formula (\ref{mdrg}) transforms into a standard quadratic dispersion relation for a relativistic point particle with mass $m, $ energy $E,$ and momentum $p_{\acute{\imath}}$ (for $\acute{\imath}=1,2,3 $) propagating in a four dimensional, 4-d, flat Minkowski spacetime of signature $(+,+,+,-)$ if $\varpi =0.$ It is supposed that certain modifications of the special relativity theory, SRT, and GR could be consequences of some (deformed) modified symmetries, for instance, with (non) commutative deformed Poincar\'{e} transforms, quantum groups and interactions, Lorentz invariance violation etc. Such values involve redefinitions of the physical energy-momentum
$p_{a}=(p_{\acute{\imath}},p_{4}=E),\overrightarrow{\mathbf{p}}=\{p_{\acute{\imath}}\},$ (for $a=1,2,3,4$), at the Planck scale $\ell_{p}:=\sqrt{\hbar G/c^{3}}\sim 10^{-33}cm.$ In this work, the light velocity is fixed $c=1$ for a respective system of physical units. The type of modification $\varpi $ is considered differently in various approaches to QG and (non) commutative MGTs, supergravity and (super) string models etc. Modifications of quantum mechanics and certain QG theories have been studied
also for models when $\ell _{p}$ is replaced by $\ell {s}$ (a string length which is found in the analysis of high-energy string scattering), for soliton-like structures known as Dirichlet p-brane and "D-particles" (Dirichlet 0-branes) could probe the structure of spacetime down to scales higher than $\ell _{s}$. Here we note that locally anistropic MGTs with MDRs are studied also as candidates for explaining acceleration cosmology and dark energy, DE, and dark matter, DM, physics, see
\cite{vijgmmp08,vijtp10a,vjpcs11,kouretsis10,vijmpd12,basilakos13,vepjc14a,svcqg13,stavr99,stavr04,stavr05}
and references therein.

\begin{example}
\label{exmdrenergy}\textsf{[MDRs in QG and cosmology] } In QG and various cosmological scenarios for which the Hamiltonian equations of motion $\dot{x}^{\acute{\imath}}=\partial H/\partial p^{\acute{\imath}}$ (we can consider similarly some Lagrange equations with possible nonholonomic constraints) are still valid at least approximately, quantum effects are modeled by a
deformed dispersion relation for photons, $c^{2}\mathbf{p}^{2}=E^{2}[1+f(E/\ ^{qg}E)],$ where $\ ^{qg}E$ is an effective quantum gravity scale and the function $f$ is model--dependent. For $E\ll \ ^{qg}E,$ one expects a series extension of MDRs, $c^{2}\mathbf{p}^{2}=E^{2}[1+\xi E/\ ^{qg}E+ O [E/\ ^{qg}E]^{2})],$ where $\xi =\pm 1$ should be fixed in a dynamical
framework for a corresponding theory. Such extensions correspond to energy-dependent velocities
$v=\frac{\partial E}{\partial p}\sim c(1-\xi E/\ ^{qg}E).$ These formulas are analogous to those for a conventional
relativistic medium, such as a gravy-electromagnetic plasma, modeling a gravitational aether which is believed to contain microscopic quantum fluctuations which may occur for Planck values $\ell _{p},t_{p}\sim 1/E_{p},$ where $E_{p}\sim 10^{19}Gev.$
\end{example}
The vacuum in above example is viewed as a non-trivial medium containing "foamy" quantum-gravity fluctuations which may include pair creations of virtual black holes, wormholes etc. One considers that $\xi =1$ and
$c^{2}p^{2}= (\ ^{qg}E)^{2}[1-e^{E/\ ^{qg}E}]^{2}$ for theories with $\kappa $-deformation of Poincar\'{e} symmetries. MDRs can be also found in models with quantization of point particles in a discrete spacetime \cite{hooft96}. The modifications emerging at the level of the $\kappa $-Poincar\'{e} (Hopf) algebra \cite{lukierski95}, and other types noncommutative black holes
\cite{vjmp05,vcqg10}, require physical interpretations on generalized Lorentz manifolds.

Our approach to geometrizing classical and quantum theories on (co) tangent Lorentz bundles (in this work, there are considered 3-d and, in the bulk, 4-d pseudo-Riemannian base manifolds) will be elaborated following three assumptions:

\begin{assumption}
\label{assumptqelorentz} \textsf{[Background quadratic elements on total spaces of (co) tangent bundles] } The are standard gravity and particle physics theories based on the special relativity and Einstein gravity principles and axioms. In such theories, the spacetime geometry is described by a four dimensional, 4-d, Lorentz spacetime manifold $V$ and respective (co) tangent bundles, $TV$ and/or $T^{\ast }V$, enabled with corresponding quadratic elements determined by total phase space metrics with signature $(+++-;+++-)$,
\begin{eqnarray}
ds^{2} &=&g_{\alpha \beta }(x^{k})du^{\alpha }du^{\beta
}=g_{ij}(x^{k})dx^{i}dx^{j}+\eta _{ab}dy^{a}dy^{b}, \mbox{ for }y^{a}\sim
dx^{a}/d\tau ;\mbox{ and/ or }  \label{lqe} \\
d\ ^{\shortmid }s^{2} &=&\ ^{\shortmid }g_{\alpha \beta }(x^{k})d\
^{\shortmid }u^{\alpha }d\ ^{\shortmid }u^{\beta
}=g_{ij}(x^{k})dx^{i}dx^{j}+\eta ^{ab}dp_{a}dp_{b}, \mbox{ for }p_{a}\sim
dx_{a}/d\tau .  \label{lqed}
\end{eqnarray}
\end{assumption}

In formulas (\ref{lqe}) and (\ref{lqed}), the local frame and dual frame (co-frame) coordinates are labeled in the forms: $u^{\alpha }=(x^{i},y^{a}),$ (or in brief, $u=(x,y)),$ on the tangent bundle $TV;$ and $\ ^{\shortmid }u^{\alpha }=(x^{i},p_{a}),$
(or in brief, $\ ^{\shortmid }u=(x,p)),$ on the cotangent bundle $T^{\ast}V. $ The pseudo--Riemannian spacetime metric $g=\{g_{ij}(x)\}$ can be a solution of the Einstein equations for the Levi-Civita connection $\nabla$. In diagonal form, the vertical metric $\eta _{ab}$ and its dual $\eta ^{ab}$ are standard Minkowski metrics, $\eta _{ab}=diag[1,1,1,-1]$ used for
computations in typical fibers for certain bundle geometric and physical models. The geometric and physical models are elaborated for general frame/ coordinate transforms in total spaces when the metric structures can be parameterized equivalently by the same h-components of $g_{\alpha \beta}(x^{k})$ and $\ ^{\shortmid }g_{\alpha \beta }(x^{k})=g_{\alpha \beta}(x^{k})$ in quadratic elements (\ref{lqe}) and (\ref{lqed}). Curves $x^{a}(\tau )$ on $V$ are parameterized by a positive parameter $\tau .$%
\footnote{In this work, there are such conventions for indices:\ the "horizontal" indices, h--indices, run values $i,j,k,...=1,2,3,4;$ the vertical indices, v-vertical, run values $a,b,c...=5,6,7,8$;  respectively, the v-indices can be identified/ contracted  with h-indices $1,2,3,4$ for lifts on total (co) tangent Lorentz bundles, when $\alpha =(i,a),\beta =(j,b),\gamma =(k,c),...=1,2,3,...8.$ We shall use letters labelled by an abstract left up/low symbol "$\ ^{\shortmid }$" (for
instance, $\ ^{\shortmid }u^{\alpha }$ and $\ ^{\shortmid }g_{\alpha \beta }) $ in order to emphasize that certain geometric/ physical objects are defined on $T^{\ast }V.$}

\begin{assumption}
\label{assumptnonlinq} \textsf{[nonlinear quadratic elements for modeling Finsler-Lagrange-Hamilton geometries on (co) tangent bundles ] } MGTs and quasi-classical limits of QG are characterized by MDRs (\ref{mdrg}) with possible small values of indicator $\varpi $ are described by basic Lorentzian and non-Riemannian total phase space geometries determined by nonlinear quadratic line elements
\begin{eqnarray}
ds_{L}^{2} &=&L(x,y),\mbox{ for models on  }TV;  \label{nqe} \\
d\ ^{\shortmid }s_{H}^{2} &=&H(x,p),\mbox{ for models on  }T^{\ast }V.
\label{nqed}
\end{eqnarray}
\end{assumption}
This assumption involves geometric/physical theories with an effective phase spacetime modelled by generalized frame, metric and linear and nonlinear connection structures defined on (co) tangent Lorentz bundles. For localized zero indicators in (\ref{mdrg}), $\varpi =0,$ the nonlinear quadratic line elements (\ref{nqe}) and (\ref{nqed}) transform correspondingly into linear quadratic elements (\ref{lqe}) and (\ref{lqed}).

\vskip5pt In modern literature on geometric mechanics, kinetics and statistical mechanics of locally anisotropic processes, classical MGTs, and QG, (see \cite{vjmp96,vstoch96,vap97,vnp97,vapny01,vsym13,vjmp13,vijgmmp14}), there were studied a series of important examples involving such concepts:

\begin{example}
\textsf{[relativistic models of Hamilton geometry and phase spaces ] } A 4-d relativistic model of Hamilton space
$H^{3,1}=(T^{\ast }V,H(x,p))$ is determined by a fundamental function (equivalently, generating Hamilton function) on a Lorentz manifold $V$, constructed as $T^{\ast }V\ni(x,p)\rightarrow H(x,p)\in \mathbb{R}$, which defines a real valued function being differentiable on $\widetilde{T^{\ast }V}:=T^{\ast }V/\{0^{\ast }\},$ for $\{0^{\ast }\}$ being the null section of $T^{\ast }V,$ and continuous on the null section of $\pi ^{\ast }:\ T^{\ast }V\rightarrow V.$ Such a model is regular if the Hessian (cv-metric)
\begin{equation}
\ \ ^{\shortmid }\widetilde{g}^{ab}(x,p):=\frac{1}{2}\frac{\partial ^{2}H}{\partial p_{a}\partial p_{b}}  \label{hesshs}
\end{equation}%
is non-degenerate, i.e. $\det |\ ^{\shortmid }\widetilde{g}^{ab}|\neq 0,$ and of constant signature.
\end{example}

For elaborating physical and mechanical models, one follows different geometric and physical principles in order to define a (non) relativistic Hamiltonian, i.e generating function, $H(x,p).$ Such a function may describe propagation of test particles, or perturbations of scalar fields in an effective phase space, various noncommutative generalizations, quantum fluctuations etc.

\begin{remark}
\textsf{[from MDR-indicators to Hamilton spaces ] } For any MDR of type (\ref{mdrg}), we can construct a Hamilton space $H^{3,1}$ if the effective Hamilton function is defined
\begin{equation}
H(p):=E=\pm (c^{2}\overrightarrow{\mathbf{p}}^{2}+c^{4}m^{2}-\varpi (E,\overrightarrow{\mathbf{p}},m;\ell _{P}))^{1/2}.  \label{hamfp}
\end{equation}%
This describes a Hamilton like geometry of relativistic point particles propagating in a cv--space with MDRs. In general, such a propagation is in an effective phase space endowed with local coordinates $(x^{i},p_{a})$ and generalized indicator
$\varpi (x^{i},E,\overrightarrow{\mathbf{p}},m;\ell_{P})$ (\ref{mdrg}).
\end{remark}

Changing the system of frames/coordinates on total space, we obtain generating functions of type $H(x,p)$. We can use for geometric modeling certain general (for simplicity, regular) generating functions $H(x,p)$ on $T^{\ast }V$.

\begin{example}
\textsf{[Relativistic Lagrange spaces ] } A relativistic 4-d model of Lagrange space $L^{3,1}=(TV,L(x,y))$ is defined by a fundamental function (equivalently, generating function) $TV\ni (x,y)\rightarrow L(x,y)\in \mathbb{R},$ i.e. a real valued function which is differentiable on $\widetilde{TV}:=TV/\{0\},$ for $\{0\}$ being the null section of $TV,$ and continuous on the null section of $\pi :TV\rightarrow V. $ Such a model is regular if the Hessian (v-metric)
\begin{equation}
\widetilde{g}_{ab}(x,y):=\frac{1}{2}\frac{\partial ^{2}L}{\partial y^{a}\partial y^{b}}  \label{hessls}
\end{equation}%
is non-degenerate, i.e. $\det |\widetilde{g}_{ab}|\neq 0,$ and of constant signature.
\end{example}

The values $\widetilde{g}_{ab}$ and $\ ^{\shortmid }\widetilde{g}^{ab}$ are labeled by tilde "\symbol{126}" in order to emphasize that such conventional v--metrics are defined canonically by respective Lagrange and Hamilton generating functions, which may encode various types of MDRs and LVs terms. Considering general frame/ coordinate transforms on $TV $ and/or $T^{\ast}V, $ we can express any "tilde" Hessian in a general quadratic form, respectively as a vertical metric (v-metric), $g_{ab}(x,y),$ and/or
co-vertical metric (cv-metric), $\ ^{\shortmid }g^{ab}(x,p).$ Inversely, prescribing any v-metric (cv-metric), we can introduce respective (co) frame /coordinate systems, when such values can transformed into certain canonical ones, with "tilde" values. In general, a $g_{ab}$ is different from the inverse of $\ ^{\shortmid }g^{ab},$ i.e. from $\ ^{\shortmid }g_{ab}$. We shall omit tildes on geometrical/ physical objects if certain formulas hold in general (not only canonical) forms and/or that will not result in ambiguities.

\begin{remark}
\textsf{[Legendre transforms and $\mathcal{L}$-duality of Lagrange and Hamilton spaces ] } There are Legendre transforms $L\rightarrow H,$ with $H(x,p):=p_{a}y^{a}-L(x,y)$ and $y^{a}$ determining solutions of the equations
$p_{a}=\partial L(x,y)/\partial y^{a}.$ In a similar manner, the inverse Legendre transforms can be introduced, $H\rightarrow L,$ for
\begin{equation}
L(x,y):=p_{a}y^{a}-H(x,p)  \label{invlegendre}
\end{equation}%
and $p_{a}$ determining solutions of the equations $y^{a}=\partial H(x,p)/\partial p_{a}.$
\end{remark}

For further details on relativistic Lagrange and Hamilton spaces, with reviews of non-relativistic original results, readers may consult \cite{vmon02,vjmp07,vjgp10,avjmp09}, where theories with degenerate Hessians are also considered (we omit such constructions in this work). Here we cite a very important paper \cite{kern74} where an original geometrization of Lagrange mechanics was proposed. The main idea was to drop the homogeneity condition for Finsler generating functions and apply in
studies of properties of mechanical systems certain methods of Finselr and almost K\"{a}hler geometry  \cite{matsumoto66,matsumoto86}. The geometry of Hamilton spaces consists a natural dualization on $T^{\ast }V$ of the
Lagrange geometry on $TV.$ Nevertheless, such constructions are not completely dual because of Legendre transforms and "symplectomophisms", see details in the mentioned references. In general, the Hamilton mechanics is not equivalent to the Lagrange mechanics and different almost symplectic realizations of such theories can be elaborated. Using Lagrange or Hamilton geometries, we can model different types of MGTs on (co) tangent bundles. Only for certain well-defined conditions such theories and respective
classes of solutions of motion/ evolution equations are equivalent.

\begin{example}
\label{rfgenf}\textsf{[Finsler-Cartan geometries as particular cases of Lagrange-Hamilton spaces] } A relativistic 4-d model of Finsler space is an example of Lagrange space when a regular $L=F^{2}$ is defined by a fundamental (generating) Finsler function subjected to the conditions: 1) $F$ is a real positive valued function which is differential on $\widetilde{TV}$ and continuous on the null section of the projection $\pi :TV\rightarrow V;$ 2) it is satisfied the homogeneity condition $F(x,\lambda y)=|\lambda |$ $F(x,y),$ for a nonzero real value $\lambda ;$ and 3) the Hessian (\ref{hessls}) is defined by $F^{2}$ in such a form that in any point $(x_{(0)},y_{(0)})$ the v-metric is of signature $(+++-).$ In a similar form, there are defined relativistic 4-d Cartan spaces $C^{3,1}=(V,C(x,p)),$ when $H=C^{2}(x,p)$ is 1-homogeneous on co-fiber coordinates $p_{a}.$
\end{example}

For simplicity, the bulk of geometric constructions in this work will be performed for (effective and/or generalized) Lagrange and Hamilton spaces considering that via corresponding frame and Legendre transforms, or homogeneity conditions, we can generate necessary type Lagrange/ Finsler/ Cartan configurations. Nevertheless, a series of important formulas and proofs will be presented both on tangent and cotangent bundles in order to emphasize the generality of our geometric methods which can be applied to a
large class of MGRs with different types of geometrization of classical and quantum physical theories. Another argument to dub the formulas will be that in many cases the (non) linear connections and/or almost symplectic structures are constructed differently on tangent and cotangent bundles. This results in different geometric and physical models (the (non) linear symmetries and related conservation laws are also different) even being written in abstract geometric forms the formulas are very similar. For certain well-defined geometric/ physical conditions, it is possible to establish certain equivalence and/or duality of constructions but this is an issue of additional assumptions and a more rigorous analysis for geometric structures and fundamental geometric/physical equations.

\begin{definition}
\textsf{[Nonlinear connections and nonholonomic h-v and/or h-cv splitting ] } \label{defnc}A nonlinear connection, N--connection, structure for $TV,$ or $T^{\ast }V,$ is defined as a Whitney sum of conventional $h$ and $v$--distributions, or $h$ and $cv$--distributions,
\begin{equation}
\mathbf{N}:TTV=hTV\oplus vTV\mbox{ or }\ \ ^{\shortmid }\mathbf{N}:TT^{\ast }V=hT^{\ast }V\oplus vT^{\ast }V.  \label{ncon}
\end{equation}
\end{definition}

There were formulated different equivalent definitions of N-connections on (co) tangent bundles and fibred manifolds, studied in \cite{vmon02,vjmp07,vjgp10,avjmp09}.%
\footnote{\label{fnanhc} The concept of N--connection is equivalent to that of Ehresmann connection in differential geometry \cite{ehresmann55}. In coordinate form, N-connections are used in the first monograph on Finsler geometry \cite{cartan35}. The N--connection and Finsler geometries were studied originally in \cite{kawaguchi37,kawaguchi52}. The formalism of N-connections was introduced and developed in the Einstein / string / brane / gauge theories \cite{vap97,vnp97,vexsol98,vjhep01,dvgrg03,vdgrg03,vjmp05,vijgmmp07,vijgmmp08}. On 4-d pseudo--Riemannian manifolds, a N--connection can be defined as a conventional nonholonomic 2+2 splitting by considering certain local (N--adapted) bases $\mathbf{e}_{\mu }=(\mathbf{e}_{i},e_{a})$ and their duals $\mathbf{e}^{\nu }=(e^{j},\mathbf{e}^{b}),$
\begin{equation*}
\mathbf{e}_{i}=\frac{\partial }{\partial x^{i}}-N_{i}^{c}\frac{\partial }{\partial y^{c}},e_{a}=\partial _{a}=\frac{\partial }{\partial y^{a}}\mbox{and }e^{j}=dx^{j},\mathbf{e}^{b}=dy^{b}+N_{k}^{b}dx^{k}.
\end{equation*}%
In this footnote, indices run values $i,j,k...=1,2$; $a,b,c,...=3,4$ and $\alpha ,\beta ,...=1,2,3,4$ for $u^{4}=t$ being a time like coordinate. Such frames are called nonholonomic satisfying, in general, the relations $[\mathbf{e}_{\alpha },\mathbf{e}_{\beta }]=\mathbf{e}_{\alpha }\mathbf{e}_{\beta }-\mathbf{e}_{\beta }\mathbf{e}_{\alpha }=W_{\alpha \beta }^{\gamma} \mathbf{e}_{\beta }$. If the anholonomy coefficients $W_{ia}^{b}=\partial _{a}N_{i}^{b},W_{ji}^{b}=\Omega {ij}^{b}=\mathbf{e}_{j}(N_{i}^{b})-\mathbf{e}_{i}(N_{j}^{b})$ are zero, we get holonomic bases which allows to consider some
coordinate transforms when $\mathbf{e}_{\alpha }\rightarrow \mathbf{\partial }%
_{\alpha }$ and $\mathbf{e}^{\beta }\rightarrow du^{\beta }.$ On 8-d (co) tangent Lorentz bundles, the N-connections and respective N-adapted frames are defined in the forms (\ref{ncon}) and (\ref{nadapb}) and (\ref{cnadap}). With respect to N--adapted bases, one say that a vector, a tensor and other geometric objects are represented by N-adapted coefficients and called correspondingly as a distinguished vector (d--vector), a distinguished tensor (d--tensor) and distinguished objects (d--object).
\par
The geometry of N--connections is related to the geometry of nonholonomic manifolds in the sense considered by G. Vr\v{a}nceanu \cite{vranc31,vranc31}, see a review of former results in \cite{bejancu03}. In brief, a nonholonomic manifold is a usual one endowed with a nonholonomic distribution, for instance, defined by a nonholonomic frame structure. N-connection and nonholonomic geometric methods were used for elaborating the AFDM. Coordinate free and global approaches to Finsler geometry and various relativistic/ string / brane / gauge generalizations were developed in \cite{vmon98,vmon02,vmon06,youssef09}.} In local form (parameterizing the corresponding N-connections by coefficients $\mathbf{N}=\{N_{i}^{a}\}$ and
$\ ^{\shortmid }\mathbf{N}=\{\ ^{\shortmid }N_{ia}\}$ with respect to coordinate (dual) bases), one prove by explicit constructions:

\begin{lemma}
\textsf{[N-adapted (co) frames ] } \label{lemadap}A N--connection (\ref{ncon}) defines respective systems of N--adapted bases%
\begin{eqnarray}
\mathbf{e}_{\alpha } &=&(\mathbf{e}_{i}=\frac{\partial }{\partial x^{i}}-N_{i}^{a}(x,y)\frac{\partial }{\partial y^{a}},e_{b}=\frac{\partial }{\partial y^{b}}),  \label{nadapb} \\
\mathbf{e}^{\alpha } &=&(e^{i}=dx^{i},\mathbf{e}^{a}=dy^{a}+N_{i}^{a}(x,y)dx^{i}),  \notag
\end{eqnarray}%
\begin{eqnarray}
\mbox{ and/ or }\ ^{\shortmid }\mathbf{e}_{\alpha } &=&(\ ^{\shortmid }\mathbf{e}_{i}=\frac{\partial }{\partial x^{i}}-\ ^{\shortmid }N_{ia}(x,p)\frac{\partial }{\partial p_{a}},\ ^{\shortmid }e^{b}=\frac{\partial }{\partial p_{b}}),  \label{cnadap} \\
\ \ ^{\shortmid }\mathbf{e}^{\alpha } &=&(\ ^{\shortmid }e^{i}=dx^{i},\ ^{\shortmid }\mathbf{e}_{a}=dp_{a}+\ ^{\shortmid }N_{ia}(x,p)dx^{i}).  \notag
\end{eqnarray}
\end{lemma}

In our works, boldface symbols are used in order to emphasize that certain geometric/ physical objects are considered in N--adapted form for certain spaces enabled with N--connection structure and when the coefficients of tensors, spinors, and fundamental geometric objects can be computed with respect to N-elongated bases of type (\ref{nadapb}) and/or (\ref{cnadap}). A splitting (\ref{ncon}), and respective N--adapted bases, defines corresponding non-integrable (equivalently, nonholonomic, or anholonomic)
structures on (co) tangent Lorentz bundles and transforms such spaces into nonholonomic manifolds, see details in Refs. \cite%
{vranc31,vranc57,bejancu03,vijgmmp07,vijgmmp08,vijtp10a}.

The sets of N--connection coefficients and necessary types of (co) frame/ coordinate transforms can be used for constructing lifts of metric structures $(V,g)$ to respective nonholonomic (co)tangent bundles, $(\mathbf{TV,N,g})$ and $(\mathbf{T}^{\ast }\mathbf{V,\ ^{\shortmid }N,\ ^{\shortmid }g})$.

\begin{assumption}
\textsf{[d-metrics on (co) tangent Lorentz bundles ] } \label{assumpt3} The total spaces of tangent, $\mathbf{TV},$ and cotangent, $\mathbf{T}^{\ast }\mathbf{V},$ Lorentz bundles used for elaborating physical theories with MDRs and LVs generalizations of the Einstein gravity are enabled, respectively, with pseudo-Riemannian metric, $\mathbf{g},$ and $\ ^{\shortmid }\mathbf{g,}$ structures. Such metrics can be parameterized by frame transforms in N--adapted form, i.e. as distinguished metrics (d-metrics)
\begin{eqnarray}
\mathbf{g} &=&\mathbf{g}_{\alpha \beta }(x,y)\mathbf{\mathbf{e}}^{\alpha }\mathbf{\otimes \mathbf{e}}^{\beta}=g_{ij}(x)e^{i}\otimes e^{j}+\mathbf{g}_{ab}(x,y)\mathbf{e}^{a}\otimes \mathbf{e}^{a}\mbox{ and/or }  \label{dmt} \\
\ ^{\shortmid }\mathbf{g} &=&\ ^{\shortmid }\mathbf{g}_{\alpha \beta }(x,p)\ ^{\shortmid }\mathbf{\mathbf{e}}^{\alpha }\mathbf{\otimes \ ^{\shortmid }\mathbf{e}}^{\beta }=g_{ij}(x)e^{i}\otimes e^{j}+\ ^{\shortmid }\mathbf{g}^{ab}(x,p)\ ^{\shortmid}\mathbf{e}_{a}\otimes \ ^{\shortmid }\mathbf{e}_{b}.
\label{dmct}
\end{eqnarray}
\end{assumption}

In this paper, we shall work with metrics on 8-d manifolds of signature (+,+,+,-,+,+,+,-)). A pseudo--Riemannian metric $g_{ij}(x)$ can be subjected to the condition that it defines a solution of the standard Einstein equations in GR and a corresponding Lorentz manifold $\mathbf{V}.$ Such constructions for the base spacetime manifold $V$ are possible for the Levi--Civita connection completely determined by $g_{ij}(x)$ by imposing the metric compatibility and zero torsion conditions. Working with more general
classed of geometric/ physical models with general MDRs and nonholonomic (co) frame structures elaborated on (co) tangent Lorentz bundles, we have to introduce into consideration non-Riemannian geometries with more general metric and connection structures. There are necessary additional geometrically and physically motivated assumptions on how nonlinear quadratic elements of type (\ref{nqe}), or (\ref{nqed}), and/or (\ref{dmt}), or (\ref{dmct}), encode MDRs.

\subsection{Canonical N--connections, metrics and almost symplectic structures}

We shall follow a coordinate-free formalism for nonholonomic manifolds and bundles enabled with N--connection structure. Certain important formulas and results will be formulated in coefficient forms, with respect to N--adapted frames, which is important for constructing in explicit form exact and parametric solutions following the AFDM.

\subsubsection{MDR, Hamilton--Lagrange generating functions and N--connections}

Let us consider a spacetime Lorentzian manifold $\mathbf{V}$ modeled as a pseudo--Riemannian manifold endowed with a metric $hg=\{g_{ij}(x)\}$ of signature $(3,1)$. Such metrics can be deformed by off-diagonal / nonholonomic transforms to metrics depending on velocity/ momentum coordinates.

\begin{proposition}
\label{prophs}A MDR (\ref{mdrg}) defines naturally canonical data for a v-metric (Hessian) $\widetilde{g}_{ab}(x,p)$ (\ref{hessls}) and a nonlinear quadratic element (\ref{nqed}) determining a relativistic model of Hamilton space $\widetilde{H}^{3,1}=(T^{\ast}V,\widetilde{H}(x,p))$.
\end{proposition}
\begin{proof}
Let us fix a point $\ x_{[0]}\in \mathbf{V}$, with a fiber space of co-vectors in this point $\ ^{\shortmid }u[x_{[0]}]=(x_{[0]},p)\in
T_{x_{[0]}}^{\ast }V$, and consider an effective Hamilton function $%
H(x_{[0]},p)$ (\ref{hamfp}) determined by a nontrivial MDR (\ref{mdrg}) with
nontrivial indicator $\varpi .$ The generating Hamilton function is defined
as a union of all points $x\in U_{x_{[0]}}\subset \mathbf{V}$, $\widetilde{H}%
(x,p):=\bigcup\limits_{x,U}H(x_{[0]},p)$ for an atlas of unions of carts $%
U_{x}$ covering $\mathbf{V}$. Tilde "$\widetilde{...}$" on symbols will be
used in this work in order to emphasize that certain spaces and
geometric/physical values are written in a "canonical" adapted form
determined by certain MDR and associated Hamilton, $\widetilde{H},$ and/or
Lagrange, $\widetilde{L},$ structures. In a constructive approach, one
prescribes a necessary smooth class (for instance, analytic or of some
finite class of differentiability) fundamental function $H(x,p)$ on $T^{\ast
}V$ determining for any fixed $x=x_{[0]}$ a MDR. Such a relativistic
Hamilton space is regular a point $\ ^{\shortmid }u=(x,p)$ if $\ \widetilde{g%
}_{ab}$ is not degenerate in this point. For $\varpi =0$ (\ref{mdrg}), the
cotangent bundle admits frame/coordinate transforms when the nonlinear
quadratic element (\ref{nqed}) can be parameterized as a linear quadratic
element (\ref{lqed}). Considering arbitrary frame/coordinate transform on $%
\mathbf{V}$ and/or $T^{\ast }V,$ the geometric objects transforms into
general ones which can be parameterized by symbols without tilde.

$\square $ (end proof)
\end{proof}

\vskip5pt

\begin{conseq}
\textsf{[Canonical Lagrange-Hamilton spaces determined by MDRs] } A MDR (\ref{mdrg}) defines $L$--dual, i.e. related via Legendre transforms, canonical relativistic models of Hamilton space $\widetilde{H}^{3,1}=(T^{\ast }V,\widetilde{H}(x,p))$ and Lagrange space $\widetilde{L}^{3,1}=(TV,\widetilde{L}(x,y))$.
\end{conseq}

\begin{proof}
It follows from a inverse Legendre transform (\ref{invlegendre}), when $\widetilde{L}(x,y):=p_{a}y^{a}-\widetilde{H}(x,p)$ and $p_{a}$ taken as a solution of the equations $y^{a}=\partial \widetilde{H}(x,p)/\partial p_{a}.$

$\square $
\end{proof}

\vskip5pt

Let us consider a regular curve $c(\tau )$ defined $c:\tau \in \lbrack
0,1]\rightarrow x^{i}(\tau )\subset U\subset V,$ for a real parameter $\tau
. $ Such a curve can be lifted to $\pi ^{-1}(U)\subset \widetilde{TV}$
defining a curve in the total space, when $\widetilde{c}(\tau ):\tau \in
\lbrack 0,1]\rightarrow \left( x^{i}(\tau ),y^{i}(\tau )=dx^{i}/d\tau
\right) $ with a nonvanishing v-vector field $dx^{i}/d\tau .$

There are on $T^{\ast }V$ a canonical symplectic structure $\theta
:=dp_{i}\wedge dx^{i}$ and a unique vector filed
\begin{equation*}
\widetilde{X}_{H}:=\frac{\partial \widetilde{H}}{\partial p_{i}}\frac{%
\partial }{\partial x^{i}}-\frac{\partial \widetilde{H}}{\partial x^{i}}%
\frac{\partial }{\partial p_{i}}
\end{equation*}
defined by $\widetilde{H},$ following the equation $i_{\widetilde{X}%
_{H}}\theta =-d\widetilde{H}.$ In above formulas $\wedge $ is the
antisymmetric product and $i_{\widetilde{X}_{H}}$ denotes the interior
produce defined by $\widetilde{X}_{H}.$ In result, we can formulate and
prove using an explicit calculus for any functions $\ ^{1}f(x,p)$ and $\
^{2}f(x,p)$ on $T^{\ast }V:$

\begin{conclusion}
A MDR (\ref{mdrg}) determines a canonical Poisson structure $\{\ ^{1}f,\
^{2}f\}:=\theta (\widetilde{X}_{^{1}f},\widetilde{X}_{^{2}f}).$
\end{conclusion}

Motion of probing point particles in a phase space modeled by an effective $%
\widetilde{H}^{3,1}$ are described by an effective relativistic Hamilton
mechanics induced by MDRs as follows from

\begin{corollary}
For any effective Hamilton phase space model on $T^{\ast }V, $ one holds the
canonical Hamilton-Jacobi equations
\begin{equation*}
\frac{dx^{i}}{d\tau }=\{\widetilde{H},x^{i}\}\mbox{ and }\frac{dp_{a}}{d\tau
}=\{\widetilde{H},p_{a}\}.
\end{equation*}
\end{corollary}

\begin{proof}
It follows from the previous Conclusion.

$\square $
\end{proof}

\vskip5pt

Following a standard variational calculus (see similar details, for
instance, in Ref. \cite{avjmp09}), one obtains the proof of

\begin{theorem}
\textsf{[Semi-sprays induced by MDRs and canonical Hamilton-Jacoby and
Euler-Lagrange equations] } The dynamics of a probing point particle in $L$%
-dual effective phase spaces $\widetilde{H}^{3,1}$ and $\widetilde{L}^{3,1}$
is described by fundamental generating functions $\widetilde{H}$ and $%
\widetilde{L}$ determined canonically by MDRs (\ref{mdrg}) and satisfy the
Hamilton-Jacobi equations written equivalently as
\begin{equation*}
\frac{dx^{i}}{d\tau }=\frac{\partial \widetilde{H}}{\partial p_{i}}%
\mbox{
and }\frac{dp_{i}}{d\tau }=-\frac{\partial \widetilde{H}}{\partial x^{i}},
\end{equation*}%
or as Euler-Lagrange equations,
\begin{equation*}
\frac{d}{d\tau }\frac{\partial \widetilde{L}}{\partial y^{i}}-\frac{\partial
\widetilde{L}}{\partial x^{i}}=0,
\end{equation*}%
which, in their turn, are equivalent to the \textbf{nonlinear geodesic
(semi-spray) equations}
\begin{equation}
\frac{d^{2}x^{i}}{d\tau ^{2}}+2\widetilde{G}^{i}(x,y)=0,  \label{ngeqf}
\end{equation}%
for $\widetilde{G}^{i}=\frac{1}{2}\widetilde{g}^{ij}(\frac{\partial ^{2}%
\widetilde{L}}{\partial y^{i}}y^{k}-\frac{\partial \widetilde{L}}{\partial
x^{i}}),\,\ $ with $\widetilde{g}^{ij}$ being inverse to $\widetilde{g}_{ij}$
(\ref{hessls}).
\end{theorem}

It should be noted that the equations (\ref{ngeqf}) emphasize that point
like probing particles move not along usual geodesics as on Lorentz
manifolds but follow some nonlinear geodesic equations determined by MDRs
and/or LIVs.

Using above Theorem and by construction on open sets covering $V,TV$ and $%
T^{\ast }V$ (see Definition \ref{defnc}), one proves

\begin{theorem}
\textsf{[existence of canonical N-connections determined by MDRs and pseudo
Hamilton and/or Lagrange generating functions] } \label{thcnc}There are
canonical N--connections determined by MDRs in $L$--dual form following
formulas
\begin{equation*}
\ \ \ \ ^{\shortmid }\widetilde{\mathbf{N}} = \left\{ \ ^{\shortmid }%
\widetilde{N}_{ij}:=\frac{1}{2}\left[ \{\ \ ^{\shortmid }\widetilde{g}_{ij},%
\widetilde{H}\}-\frac{\partial ^{2}\widetilde{H}}{\partial p_{k}\partial
x^{i}}\ ^{\shortmid }\widetilde{g}_{jk}-\frac{\partial ^{2}\widetilde{H}}{%
\partial p_{k}\partial x^{j}}\ ^{\shortmid }\widetilde{g}_{ik}\right]
\right\} \mbox{ and } \widetilde{\mathbf{N}} = \left\{ \widetilde{N}%
_{i}^{a}:=\frac{\partial \widetilde{G}}{\partial y^{i}}\right\},
\end{equation*}%
where $\ \ ^{\shortmid }\widetilde{g}_{ij}$ is inverse to $\ \ ^{\shortmid }%
\widetilde{g}^{ab}$ (\ref{hesshs}).
\end{theorem}

Hereafter we shall consider that using necessary type frame/coordiante
transforms we can always establish (if necessary) certain nonholonomic
frames and geometric variables determined by MDRs and respective $L$-dual
relations between geometric/ physical values on tangent and cotangent
Lorentz bundles.

Introducing the canonical N--connection coefficients defined by this
Theorem, respectively, into formulas (\ref{nadapb}) and (\ref{cnadap}) \
introduced by Lemma \ref{lemadap}, we prove

\begin{proposition}
\textsf{[Canonical N-adapted frames determined by MDRs and/or
Lagrange-Hamilton generating functions ]} \label{prcnadapb}The N--connection
structures $\widetilde{\mathbf{N}}$ and $\ \ ^{\shortmid }\widetilde{\mathbf{%
N}}$ define respective systems of N--adapted (co) frames
\begin{eqnarray}
\widetilde{\mathbf{e}}_{\alpha } &=&(\widetilde{\mathbf{e}}_{i}=\frac{%
\partial }{\partial x^{i}}-\widetilde{N}_{i}^{a}(x,y)\frac{\partial }{%
\partial y^{a}},e_{b}=\frac{\partial }{\partial y^{b}}),\mbox{ on }TV;
\label{cnddapb} \\
\widetilde{\mathbf{e}}^{\alpha } &=&(\widetilde{e}^{i}=dx^{i},\widetilde{%
\mathbf{e}}^{a}=dy^{a}+\widetilde{N}_{i}^{a}(x,y)dx^{i}),\mbox{ on }%
(TV)^{\ast };  \notag \\
\mbox{and \ } \ ^{\shortmid }\widetilde{\mathbf{e}}_{\alpha } &=&(\
^{\shortmid }\widetilde{\mathbf{e}}_{i}=\frac{\partial }{\partial x^{i}}-\
^{\shortmid }\widetilde{N}_{ia}(x,p)\frac{\partial }{\partial p_{a}},\
^{\shortmid }e^{b}=\frac{\partial }{\partial p_{b}}),\mbox{ on }T^{\ast }V;
\label{ccnadap} \\
\ \ ^{\shortmid }\widetilde{\mathbf{e}}^{\alpha } &=&(\ ^{\shortmid
}e^{i}=dx^{i},\ ^{\shortmid }\mathbf{e}_{a}=dp_{a}+\ ^{\shortmid }\widetilde{%
N}_{ia}(x,p)dx^{i})\mbox{ on }(T^{\ast }V)^{\ast }.  \notag
\end{eqnarray}
\end{proposition}

We note that we can introduce canonical N-splitting
\begin{equation*}
\widetilde{\mathbf{N}}:TTV=hTV\oplus vTV \mbox{ and/or }\ \ ^{\shortmid }%
\widetilde{\mathbf{N}}:TT^{\ast }V=hT^{\ast }V\oplus vT^{\ast }V,
\end{equation*}
respectively, on any tangent Lorentz bundle and cotangent Lorentz bundle by
prescribing a system of so-called nonholonomic variables with N-adapted
frames of type (\ref{cnddapb}) and (\ref{ccnadap}). In such cases, probing
point like particles are described by some effective models of Lagrange
and/or Hamilton mechanics for the total phase spaces. If arbitrary frame and
coordinate transforms are considered on such spacetime manifolds and phase
(co) tangent bundle, the values with tilde transform correspondingly into
arbitrary ones, i.e. into (\ref{ncon}), (\ref{nadapb}) and (\ref{cnadap}).

\begin{conclusion}
\textsf{[Canonical geometric data for Lagrange-Hamilton spaces] } The
nonholonomic structure of a Lorentz manifolds and respective (co) tangent
bundles can be described in equivalent forms using canonical data $(%
\widetilde{L},\ \widetilde{\mathbf{N}};\widetilde{\mathbf{e}}_{\alpha},%
\widetilde{\mathbf{e}}^{\alpha }),$ with effective Largange density $%
\widetilde{L}$ (correspondingly, $(\widetilde{H},\ ^{\shortmid }\widetilde{%
\mathbf{N}};\ ^{\shortmid }\widetilde{\mathbf{e}}_{\alpha },\ ^{\shortmid }%
\widetilde{\mathbf{e}}^{\alpha }),$ with effective Hamilton density $%
\widetilde{H}$ ) or by a general N-splitting without effective Lagrangians
(Hamiltonians), i.e. in terms of geometric data $(\mathbf{N};\mathbf{e}%
_{\alpha },\mathbf{e}^{\alpha })$ (correspondingly $(\ ^{\shortmid }\mathbf{N%
};\ ^{\shortmid }\mathbf{e}_{\alpha },\ ^{\shortmid }\mathbf{e}^{\alpha })$%
). Such structures are determined by a MDR (\ref{mdrg}) if, for instance, $%
\widetilde{H}$ is chosen following the procedure stated by Proposition \ref%
{prophs}. Such an analogous Lagrange (Hamilton) like interpretation of the
geometry of phase spaces is "hidden" into the structure of nonholonomic
manifolds/ bundles endowed with N-connection splitting if general
frame/coordinate transforms are considered on respective relativistic
spacetimes and phase spaces.
\end{conclusion}

The geometric constructions can be performed in equivalent forms for various
nonholonomic data. For instance, we can consider arbitrary frame/coordinate
transforms on (co) bundle spaces and write formulas without "tilde", i.e. in
general (non canonical) forms, emphasizing certain N-splitting, or in
coordinate forms. Nevertheless, some classes nonholonomic variables can be
more convenient for decoupling physically important systems of nonlinear
PDEs, generating new classes of exact/ parametric solutions, and other
classes of nonholonomic variables can be more convenient for elaborating and
effective/analogous geometric mechanics interpretation, or (for instance)
for deformation quantization of MGTs and the Einstein gravity, see Refs.
\cite{vpla08,vjgp10,avjmp09}, and next sections.

Vector fields on nonholonomic (co) tangent bundles are called d--vectors if
they are written in a form adapted to a prescribed N--connection structure.
For instance, we decompose
\begin{eqnarray*}
\mathbf{X} &=&\widetilde{\mathbf{X}}^{\alpha }\widetilde{\mathbf{e}}_{\alpha
}=\widetilde{\mathbf{X}}^{i}\widetilde{\mathbf{e}}_{i}+X^{b}e_{b}=\mathbf{X}%
^{\alpha }\mathbf{e}_{\alpha }=\mathbf{X}^{i}\mathbf{e}_{i}+X^{b}e_{b}\in T%
\mathbf{TV}, \\
\ ^{\shortmid }\mathbf{X} &=&\ ^{\shortmid }\widetilde{\mathbf{X}}^{\alpha }%
\widetilde{\mathbf{e}}_{\alpha }=\ ^{\shortmid }\widetilde{\mathbf{X}}^{i}\
^{\shortmid }\widetilde{\mathbf{e}}_{i}+\ ^{\shortmid }X_{b}\ ^{\shortmid
}e^{b}=\ ^{\shortmid }\mathbf{X}^{\alpha }\ ^{\shortmid }\mathbf{e}_{\alpha
}=\ ^{\shortmid }\mathbf{X}^{i}\ ^{\shortmid }\mathbf{e}_{i}+\ ^{\shortmid
}X_{b}\ ^{\shortmid }e^{b}\in T\mathbf{T}^{\ast }\mathbf{V,}
\end{eqnarray*}%
for decompositions with respect to canonical, or arbitrary, N-adapted bases.
In brief, one considers such h-v and/or h-cv decompositions, $\mathbf{X}%
^{\alpha }=\widetilde{\mathbf{X}}^{\alpha }=(\widetilde{\mathbf{X}}%
^{i},X^{b})=(\mathbf{X}^{i},X^{b}),\ ^{\shortmid }\mathbf{X}^{\alpha }=\
^{\shortmid }\widetilde{\mathbf{X}}^{\alpha }=(\ ^{\shortmid }\widetilde{%
\mathbf{X}}^{i},\ ^{\shortmid }X_{b})=(\ ^{\shortmid }\mathbf{X}^{i},\
^{\shortmid }X_{b})$.

We can write $\mathbf{X}$ and $\ ^{\shortmid }\mathbf{X}$ as 1-forms and
N-adapted canonical, or arbitrary, coefficients,
\begin{eqnarray*}
\mathbf{X} &=&\widetilde{\mathbf{X}}_{\alpha }\ \mathbf{e}^{\alpha }=X_{i}\
e^{i}+\widetilde{\mathbf{X}}^{a}\widetilde{\mathbf{e}}_{a}=\widetilde{%
\mathbf{X}}_{\alpha }\mathbf{e}^{\alpha }=X_{i}e^{i}+\mathbf{X}^{a}\mathbf{e}%
_{a}\ \in T^{\ast }\mathbf{TV} \\
\ ^{\shortmid }\mathbf{X} &=&\ ^{\shortmid }\widetilde{\mathbf{X}}_{\alpha
}\ ^{\shortmid }\mathbf{e}^{\alpha }=\ ^{\shortmid }X_{i}\ ^{\shortmid
}e^{i}+\ ^{\shortmid }\widetilde{\mathbf{X}}^{a}\ ^{\shortmid }\widetilde{%
\mathbf{e}}_{a}=\ ^{\shortmid }\widetilde{\mathbf{X}}_{\alpha }\ ^{\shortmid
}\mathbf{e}^{\alpha }=\ ^{\shortmid }X_{i}\ ^{\shortmid }e^{i}+\ ^{\shortmid
}\mathbf{X}^{a}\ ^{\shortmid }\mathbf{e}_{a}\ \in T^{\ast }\mathbf{T}^{\ast }%
\mathbf{V,}
\end{eqnarray*}%
or, in brief, $\mathbf{X}_{\alpha }=\widetilde{\mathbf{X}}_{\alpha }=(X_{i},%
\widetilde{\mathbf{X}}^{a})=(X_{i},\mathbf{X}^{a}),\ ^{\shortmid }\mathbf{%
X_{\alpha }=}\ ^{\shortmid }\widetilde{\mathbf{X}}_{\alpha }=(\ ^{\shortmid
}X_{i},\ ^{\shortmid} \widetilde{\mathbf{X}}^{a})=(\ ^{\shortmid }X_{i},\
^{\shortmid}\mathbf{X}^{a})$.

Considering tensor products of N-adapted (co) frames, we can parameterized in N-adapted forms (canonical or general ones) arbitrary tensors fields, called as d-tensors.

\subsubsection{Locally anisotropic metrics on (co) tangent Lorentz bundles}

We can introduce canonical metric structures determined by MDRs and respective Lagrange and/or Hamilton effective functions using the so-called Sasaki lifts \cite{yano73} from a base manifold to the total (co) tangent bundles. Using formulas (\ref{hessls}) and (\ref{hesshs}), Theorem \ref{thcnc} (on existence of canonical N-connections ) and Proposition \ref{prcnadapb} (on canonical geometric data for Lagrange-Hamilton spaces) and by construction, we prove the

\begin{theorem}
\textsf{[existence of canonical d-metrics completely defined by MDR and generating Lagrange-Hamilton functions]} There are canonical d-metric structures $\widetilde{\mathbf{g}}$ and $\ ^{\shortmid }\widetilde{\mathbf{g}}$ completely determined by a MDR (\ref{mdrg}), respectively for data $(\widetilde{L},\ \ \widetilde{\mathbf{N}};\widetilde{\mathbf{e}}_{\alpha },%
\widetilde{\mathbf{e}}^{\alpha };\widetilde{g}_{jk},\widetilde{g}^{jk})$ and/or $(\widetilde{H},\ ^{\shortmid }\widetilde{\mathbf{N}};\ ^{\shortmid }\widetilde{\mathbf{e}}_{\alpha },\ ^{\shortmid }\widetilde{\mathbf{e}}^{\alpha };\ \ ^{\shortmid }\widetilde{g}^{ab},\ \ ^{\shortmid }\widetilde{g}_{ab}),$ metrics (d-metrics)%
\begin{eqnarray}
\widetilde{\mathbf{g}} &=&\widetilde{\mathbf{g}}_{\alpha \beta }(x,y)%
\widetilde{\mathbf{e}}^{\alpha }\mathbf{\otimes }\widetilde{\mathbf{e}}%
^{\beta }=\widetilde{g}_{ij}(x,y)e^{i}\otimes e^{j}+\widetilde{g}_{ab}(x,y)%
\widetilde{\mathbf{e}}^{a}\otimes \widetilde{\mathbf{e}}^{a}\mbox{
and/or }  \label{cdms} \\
\ ^{\shortmid }\widetilde{\mathbf{g}} &=&\ ^{\shortmid }\widetilde{\mathbf{g}%
}_{\alpha \beta }(x,p)\ ^{\shortmid }\widetilde{\mathbf{e}}^{\alpha }\mathbf{%
\otimes \ ^{\shortmid }}\widetilde{\mathbf{e}}^{\beta }=\ \ ^{\shortmid }%
\widetilde{g}_{ij}(x,p)e^{i}\otimes e^{j}+\ ^{\shortmid }\widetilde{g}%
^{ab}(x,p)\ ^{\shortmid }\widetilde{\mathbf{e}}_{a}\otimes \ ^{\shortmid }%
\widetilde{\mathbf{e}}_{b}.  \label{cdmds}
\end{eqnarray}
\end{theorem}

One follows:
\begin{corollary}
By frame transforms, the canonical d-metric structures (\ref{cdms}) and (\ref%
{cdmds}) [with tildes] can be written, respectively, in general d-metric
forms (\ref{dmt}) and (\ref{dmct}) [without tildes].
\end{corollary}

This Corollary motivates the Assumption \ref{assumpt3} on metric properties
of (co) tangent Lorentz bundles endowed with MDRs.

There are also possible general vierbein transforms $e_{\alpha }=e_{\ \alpha
}^{\underline{\alpha }}(u)\partial /\partial u^{\underline{\alpha }}$ and $%
e^{\beta }=e_{\ \underline{\beta }}^{\beta }(u)du^{\underline{\beta }}$,
where the local coordinate indices are underlined in order to distinguish
them from arbitrary abstract ones. In such formulas, the matrix $e_{\
\underline{\beta }}^{\beta }$ is inverse to $e_{\ \alpha }^{\underline{%
\alpha }}$ for orthonormalized bases. \ For Hamilton like configurations,
one writes $\ ^{\shortmid }e_{\alpha }=\ ^{\shortmid }e_{\ \alpha }^{%
\underline{\alpha }}(\ ^{\shortmid }u)\partial /\partial \ ^{\shortmid }u^{%
\underline{\alpha }}$ and $\ ^{\shortmid }e^{\beta }=\ ^{\shortmid }e_{\
\underline{\beta }}^{\beta }(\ ^{\shortmid }u)d\ ^{\shortmid }u^{\underline{%
\beta }}.$ It should be noted that there are not used boldface symbols for
such transforms because an arbitrary decomposition (for instance, one can be
considered as particular cases certain diadic 2+2+2+2 splitting) is not
adapted to a N--connection structure. Introducing formulas of type (\ref%
{cnddapb}) and (\ref{ccnadap}), respectively, into (\ref{dmt}) and (\ref%
{dmct}) and regrouping with respect to local coordinate bases, one proves

\begin{corollary}
\textsf{[equivalent re-writing of d-metrics as off-diagonal metrics]} With
respect to local coordinate frames, any d--metric structures on $\mathbf{TV}$
and/or $\mathbf{T}^{\ast }\mathbf{V,}$%
\begin{eqnarray*}
\mathbf{g} &=&\mathbf{g}_{\alpha \beta }(x,y)\mathbf{e}^{\alpha }\mathbf{%
\otimes e}^{\beta }=g_{\underline{\alpha }\underline{\beta }}(x,y)du^{%
\underline{\alpha }}\mathbf{\otimes }du^{\underline{\beta }}\mbox{
and/or } \\
\ ^{\shortmid }\mathbf{g} &=&\ ^{\shortmid }\mathbf{g}_{\alpha \beta }(x,p)\
^{\shortmid }\mathbf{e}^{\alpha }\mathbf{\otimes \ ^{\shortmid }e}^{\beta
}=\ ^{\shortmid }g_{\underline{\alpha }\underline{\beta }}(x,p)d\
^{\shortmid }u^{\underline{\alpha }}\mathbf{\otimes }d\ ^{\shortmid }u^{%
\underline{\beta }},
\end{eqnarray*}%
can be parameterized via frame transforms, $\mathbf{g}_{\alpha \beta }=e_{\
\alpha }^{\underline{\alpha }}e_{\ \beta }^{\underline{\beta }}g_{\underline{%
\alpha }\underline{\beta }}$ and $\ ^{\shortmid }\mathbf{g}_{\alpha \beta
}=\ ^{\shortmid }e_{\ \alpha }^{\underline{\alpha }}\ ^{\shortmid }e_{\
\beta }^{\underline{\beta }}\ ^{\shortmid }g_{\underline{\alpha }\underline{%
\beta }},$ in respective off-diagonal forms:
\begin{eqnarray}
g_{\underline{\alpha }\underline{\beta }} &=&\left[
\begin{array}{cc}
g_{ij}(x)+g_{ab}(x,y)N_{i}^{a}(x,y)N_{j}^{b}(x,y) & g_{ae}(x,y)N_{j}^{e}(x,y)
\\
g_{be}(x,y)N_{i}^{e}(x,y) & g_{ab}(x,y)%
\end{array}%
\right] \mbox{
and/or }  \notag \\
\ ^{\shortmid }g_{\underline{\alpha }\underline{\beta }} &=&\left[
\begin{array}{cc}
\ ^{\shortmid }g_{ij}(x)+\ ^{\shortmid }g^{ab}(x,p)\ ^{\shortmid
}N_{ia}(x,p)\ ^{\shortmid }N_{jb}(x,p) & \ ^{\shortmid }g^{ae}\ ^{\shortmid
}N_{je}(x,p) \\
\ ^{\shortmid }g^{be}\ ^{\shortmid }N_{ie}(x,p) & \ ^{\shortmid
}g^{ab}(x,p)\
\end{array}%
\right] .  \label{offd}
\end{eqnarray}
\end{corollary}

Parameterizations of type (\ref{offd}) are considered, for instance, in the
Kaluza--Klein theory. Such metrics are generic off-diagonal if the
corresponding N-adapted structure is not integrable (see footnote \ref%
{fnanhc}). For MDR-generalizations of the Einstein gravity, we can consider
that the h-metrics $g_{ij}(x)=\ ^{\shortmid }g_{ij}(x)$ are determined by a
solution of standard Einstein equations but the terms with $N $%
--coefficients are determined by solutions of certain generalized
gravitational field equations on noholonomic phase spaces. In general, such
solutions are not compactified on velocity/ momentum like coordinates, $%
y^{a} $ / $p_{a}$ like in standard Kaluza-Klein models. For decompositions
with respect to coordinate bases of canonical d-metrics (\ref{cdms}) and (%
\ref{cdmds}), we obtain coefficients of type (\ref{offd}) when (following a
formal notation procedure) the geometric objects are labeled with tilde and $%
\widetilde{g}_{ij}(x,y)\neq \ ^{\shortmid }\widetilde{g}_{ij}(x,p).$ In
result, we obtain

\begin{conclusion}
\textsf{[definition by MDRs of canonical frames and noncompactified
Kaluza-Klein metrics for phase spacetims with equivalent Lagrange-Hamilton
interpretation] } A MDR-structure determines on (co) tangent Lorentz bundles
nonholonomic generalized frame structures and equivalent d-metric and
off-diagonal metric structures with dependencies on velocity/ momentum type
coordinates. For respective sets of nonholonomic variables, such locally
anisotropic gravitational models admit analogous Lagrange and/or Hamilton
mechanics interpretation, or as a generalized (nonholonomic) Kaluza-Klein
theory without compactification on extra dimension (phase) coordinates.
\end{conclusion}

\begin{remark}
\textsf{[two general classes of phase spaces and MGTs encoding MDTs and
LIVs] } There are two general classes of MGTs (on tangent and/or cotangent
bundles) constructed with total metric structures determined by MDRs. In
this work, we shall study with priority theories on cotangent Lorentz
bundles $\mathbf{T}^{\ast }\mathbf{V}$ but also provide and compare
important formulas for tangent Lorentz bundles (on $\mathbf{TV}$ and
conventional fibered nonholonomic manifolds, such theories were studied in
Refs. \cite{vepjc14,gvvepjc14} and \cite%
{castro05,castro11,castro12,castro14,castro16}.
\end{remark}

The type of MGT on a nonholonomic (co) tangent bundle, or nonholonomic
manifold, depend on the type of metric and nonlinear and linear connection
structures (see below) are involved for such constructions.

\begin{remark}
\textsf{[frame transforms, nonholonomic frame transforms and equivalence of
canonical and noncanonical d-metric structures] } If we fix a metric
structure of type $\ \ ^{\shortmid }\widetilde{\mathbf{g}}$ (\ref{cdmds}),
we can elaborate equivalent models with $\ ^{\shortmid }\mathbf{g}$ (\ref%
{dmct}) determined by certain classes of frame transforms. Inversely,
prescribing a d-metric $\ ^{\shortmid }\mathbf{g,}$ we can define
nonholonomic variables when this metric structure can be represented as a $\
^{\shortmid }\widetilde{\mathbf{g}}.$ In such a model, $\ ^{\shortmid }%
\mathbf{g=}\ ^{\shortmid }\widetilde{\mathbf{g}}.$ Nevertheless, we can
elaborate on bi-metric (and even multi-metric theories) if we consider that $%
\ ^{\shortmid }\widetilde{\mathbf{g}}$ and $\ ^{\shortmid }\mathbf{g}$ are
related via certain generalized transforms considered, for instance, in
bi-metric, bi-gravity, and/or massive gravity \cite%
{vijgmmp14,vepjc14,vepjc14a,gvvepjc14,gvvcqg15}.
\end{remark}

Instead of (canonical) metric structures, we can define on nonholonomic (co)
tangent bundles equivalent (canonical) almost symplectic structures (in both
cases, the constructions are determined by a MDR (\ref{mdrg})) as in next
subsection.

\subsection{Almost K\"{a}hler Lagrange--Hamilton structures and MDRs}

Spacetime models encoding MDRs and LIVs and formulated as almost K\"{a}hler geometries for \textsf{relativistic} Lagrange-Hamilton configurations were studied in \cite{vmon02,avjmp09}. Further developments were elaborated for almost symplectic (algebroid, commutative and noncommutative) models of deformation/ geometric quantization and/or geometric flow theories
\cite{vjmp07,vpla08,vijgmmp09,vjgp10,vjmp13,vmjm15,vch2416,tayebi10} and references therein. Fundamental ideas on almost K\"{a}hler realisation of Finsler and Lagrange geometry were proposed in \cite{matsumoto66,matsumoto86,kern74}. In nonrelativistic form, K. Matsumoto and J. Kern results were applied in geometric mechanics by a series of Romanian authors whose works were summarized in monographs \cite{miron94,miron00} (see footnote \ref{fn01a} on ethical and political issues related to
publication of those books).

\subsubsection{Canonical almost complex structures and Neijenhuis fields}

Fundamental (generating) Lagrange-Hamilton functions and MDRs determine canonical models of almost K\"{a}hler geometry. Such nonholonomic variables can be introduced in classical and quantum MGTs on (co) tangent bundles. One holds:

\begin{proposition}
\textsf{[existence of canonical almost complex structures for Lagrange-Hamilton spaces] } MDRs (\ref{mdrg}) determining canonical
N--connections $\widetilde{\mathbf{N}}$ and $\ ^{\shortmid }\widetilde{\mathbf{N}}$ following conditions of  Theorem \ref{thcnc} define respectively canonical almost complex structures $\widetilde{\mathbf{J}}\mathbf{,}$ on $\mathbf{TV},$ and
$\ ^{\shortmid }\widetilde{\mathbf{J}},$ on $\mathbf{T}^{\ast }\mathbf{V}.$
\end{proposition}

\begin{proof}
Let us introduce the linear operator $\widetilde{\mathbf{J}}$ acting on $%
\widetilde{\mathbf{e}}_{\alpha }=(\widetilde{\mathbf{e}}_{i},e_{b})$ (\ref%
{cnddapb}) a as follows: $\widetilde{\mathbf{J}}(\mathbf{e}_{i})=-\widetilde{%
\mathbf{e}}_{n+i}$ and $\widetilde{\mathbf{J}}(e_{n+i})=\widetilde{\mathbf{e}%
}_{i}$. This operator defines globally an almost complex structure ($%
\widetilde{\mathbf{J}}\mathbf{\circ \widetilde{\mathbf{J}}=-I}$ for $\mathbf{%
I}$ being the unity matrix) on $\mathbf{TV}$ completely determined for
Lagrange spaces by a $\widetilde{L}(x,y).$

On $\mathbf{T}^{\ast }\mathbf{V,}$ we can consider a linear operator $\
^{\shortmid }\widetilde{\mathbf{J}}$ acting on$\ ^{\shortmid }\mathbf{e}%
_{\alpha }=(\ ^{\shortmid }\mathbf{e}_{i},\ ^{\shortmid }e^{b})$ (\ref%
{ccnadap}) following formulas $\ ^{\shortmid }\widetilde{\mathbf{J}}(\
^{\shortmid }\mathbf{e}_{i})=-\ ^{\shortmid }e^{n+i}$ and $\ ^{\shortmid }%
\widetilde{\mathbf{J}}(\ ^{\shortmid }e^{n+i})=\ ^{\shortmid }\mathbf{e}_{i}$%
. Then $\ \ ^{\shortmid }\widetilde{\mathbf{J}}$ defines globally an almost
complex structure ( $\ \ ^{\shortmid }\widetilde{\mathbf{J}}\mathbf{\circ \ }%
\ ^{\shortmid }\widetilde{\mathbf{J}}=$ $-\mathbf{\ I}$ for $\mathbf{I}$
being the unity matrix) on $\mathbf{T}^{\ast }\mathbf{V}$ completely
determined for Hamilton spaces by a $\widetilde{H}(x,p).$
\end{proof}

$\square $ \vskip5pt

We note that $\widetilde{\mathbf{J}}$ and $\ ^{\shortmid }\widetilde{\mathbf{J}}$ are standard almost complex structures only for the Euclidean signatures, respectively, on $\mathbf{TV}$ and $\mathbf{T}^{\ast }\mathbf{V}$. Considering arbitrary frame/coordinate transforms, we can omit tildes
and write $\mathbf{J}$ and $\ ^{\shortmid }\mathbf{J.}$ Inversely, we can state certain almost complex nonholonomic variables prescribing a MDR inducing respective canonical structures.

\begin{definition}
\textsf{[curvatures of canonical N-connections and almost complex structures]%
} The canonical Neijenhuis tensor fields determined by MDRs for respective
canonical almost complex structures $\widetilde{\mathbf{J}}$ on $\mathbf{TV}$
and/or $\ ^{\shortmid }\widetilde{\mathbf{J}}$ on $\mathbf{T}^{\ast }\mathbf{%
V},$ are considered as curvatures of respective N--connections
\begin{eqnarray}
\widetilde{\mathbf{\Omega }}\mathbf{(}\widetilde{\mathbf{X}}\mathbf{,}%
\widetilde{\mathbf{Y}})&:=&\mathbf{\ -[\widetilde{\mathbf{X}}\mathbf{,}%
\widetilde{\mathbf{Y}}]+[\widetilde{\mathbf{J}}\widetilde{\mathbf{X}},%
\widetilde{\mathbf{J}}\widetilde{\mathbf{Y}}]-\widetilde{\mathbf{J}}[%
\widetilde{\mathbf{J}}\widetilde{\mathbf{X}},\widetilde{\mathbf{Y}}]-%
\widetilde{\mathbf{J}}[\widetilde{\mathbf{X}},\widetilde{\mathbf{J}}%
\widetilde{\mathbf{Y}}]}\mbox{ and/or }  \notag \\
\ ^{\shortmid }\widetilde{\mathbf{\Omega }}(\ ^{\shortmid }\widetilde{%
\mathbf{X}}\mathbf{,}\ ^{\shortmid }\widetilde{\mathbf{Y}}\mathbf{)} &:=&%
\mathbf{\ -[\ ^{\shortmid }\widetilde{\mathbf{X}}\mathbf{,}\ ^{\shortmid }%
\widetilde{\mathbf{Y}}]+[\ ^{\shortmid }\widetilde{\mathbf{J}}\ ^{\shortmid }%
\widetilde{\mathbf{X}},\ ^{\shortmid }\widetilde{\mathbf{J}}\ ^{\shortmid }%
\widetilde{\mathbf{Y}}]-\ ^{\shortmid }\widetilde{\mathbf{J}}[\ ^{\shortmid }%
\widetilde{\mathbf{J}}\ ^{\shortmid }\widetilde{\mathbf{X}},\ ^{\shortmid }%
\widetilde{\mathbf{Y}}]-\ ^{\shortmid }\widetilde{\mathbf{J}}[\ ^{\shortmid }%
\widetilde{\mathbf{X}},\ ^{\shortmid }\widetilde{\mathbf{J}}\ ^{\shortmid }%
\widetilde{\mathbf{Y}}],}  \label{neijt}
\end{eqnarray}%
for any d--vectors $\mathbf{X,}$ $\mathbf{Y}$ and $\ ^{\shortmid }\mathbf{%
X,\ ^{\shortmid }Y.}$
\end{definition}

For arbitrary N--connection splitting, the formulas (\ref{neijt}) can be
written in general form without tilde values. Applying the left label "$%
\mathbf{\ ^{\shortmid }",}$ we can rewrite geometric formulas on $\mathbf{TV%
}$ into respective ones on $\mathbf{T}^{\ast }\mathbf{V}$ (if necessary, for
$\mathcal{L}$--dual values and with, or not "tildes"). One follows:

\begin{corollary}
In local coordinate form, a N--connection on $\mathbf{TV,}$ or $\mathbf{T}%
^{\ast }\mathbf{V,}$ is characterized by such coefficients of Neijenhuis
tensors (\ref{neijt}), i.e. N--connection curvature,
\begin{equation}
\Omega _{ij}^{a} =\frac{\partial N_{i}^{a}}{\partial x^{j}}-\frac{\partial
N_{j}^{a}}{\partial x^{i}}+N_{i}^{b}\frac{\partial N_{j}^{a}}{\partial y^{b}}%
-N_{j}^{b}\frac{\partial N_{i}^{a}}{\partial y^{b}},\mbox{\ or\ } \ \
\mathbf{\ ^{\shortmid }}\Omega _{ija} = \frac{\partial \mathbf{\ ^{\shortmid
}}N_{ia}}{\partial x^{j}}-\frac{\partial \mathbf{\ ^{\shortmid }}N_{ja}}{%
\partial x^{i}}+\ \mathbf{^{\shortmid }}N_{ib}\frac{\partial \mathbf{\
^{\shortmid }}N_{ja}}{\partial p_{b}}-\mathbf{\ ^{\shortmid }}N_{jb}\frac{%
\partial \mathbf{\ ^{\shortmid }}N_{ia}}{\partial p_{b}}.  \label{neijtc}
\end{equation}
\end{corollary}

Some almost complex structures $\mathbf{J}$ and $\ ^{\shortmid }\mathbf{J}$
transform into standard complex structures for Euclidean signatures if $%
\mathbf{\Omega }=0$ and/or $\ ^{\shortmid }\mathbf{\Omega }=0.$

\begin{remark}
For almost complex structures determined by MDRs, formulas (\ref{neijtc})
can be written (using frame transforms) in respective canonical forms with
"tilde" values determined by $\widetilde{\mathbf{N}}=\{\widetilde{N}%
_{j}^{b}\}$ and $\ ^{\shortmid }\widetilde{\mathbf{N}}=\{\mathbf{\
^{\shortmid }}\widetilde{N}_{ia}\}.$
\end{remark}

Using Proposition \ref{prcnadapb} (see also the footnote \ref{fnanhc} on
anholonomic frames and N--connections), by straightforward N-adapted
calculus using formulas $\widetilde{\mathbf{e}}_{\alpha }=(\widetilde{%
\mathbf{e}}_{i},e_{b})$ (\ref{cnddapb}), $\ ^{\shortmid }\widetilde{\mathbf{e%
}}_{\alpha }=(\ ^{\shortmid }\widetilde{\mathbf{e}}_{i},\ ^{\shortmid
}e^{b}) $ (\ref{ccnadap}) and (\ref{neijtc}), we prove

\begin{consequence}
\textsf{[existence of MDR-induced canonic anholonomic frame structures ]} %
\label{anhr}MDRs and respective canonical N-connections induce canonical
nonholonomic frame structures on $\mathbf{TV}$ and/or $\mathbf{T}^{\ast }%
\mathbf{V}$ characterized by corresponding anholonomy relations
\begin{equation}
\lbrack \widetilde{\mathbf{e}}_{\alpha },\widetilde{\mathbf{e}}_{\beta }]=%
\widetilde{\mathbf{e}}_{\alpha }\widetilde{\mathbf{e}}_{\beta }-\widetilde{%
\mathbf{e}}_{\beta }\widetilde{\mathbf{e}}_{\alpha }=\widetilde{W}_{\alpha
\beta }^{\gamma }\widetilde{\mathbf{e}}_{\gamma }  \label{anhrelc}
\end{equation}%
with (antisymmetric) anholonomy coefficients $\widetilde{W}%
_{ia}^{b}=\partial _{a}\widetilde{N}_{i}^{b}$ and $\widetilde{W}_{ji}^{a}=%
\widetilde{\Omega }_{ij}^{a},$ and
\begin{equation}
\lbrack \ ^{\shortmid }\widetilde{\mathbf{e}}_{\alpha },\ ^{\shortmid }%
\widetilde{\mathbf{e}}_{\beta }]=\ ^{\shortmid }\widetilde{\mathbf{e}}%
_{\alpha }\ ^{\shortmid }\widetilde{\mathbf{e}}_{\beta }-\ ^{\shortmid }%
\widetilde{\mathbf{e}}_{\beta }\ ^{\shortmid }\widetilde{\mathbf{e}}_{\alpha
}=\ ^{\shortmid }\widetilde{W}_{\alpha \beta }^{\gamma }\ ^{\shortmid }%
\widetilde{\mathbf{e}}_{\gamma }  \label{anhrelcd}
\end{equation}%
with anholonomy coefficients $\ ^{\shortmid }\widetilde{W}_{ib}^{a}=\partial
\ ^{\shortmid }\widetilde{N}_{ib}/\partial p_{a}$ and $\ ^{\shortmid }%
\widetilde{W}_{jia}=\ \mathbf{\ ^{\shortmid }}\widetilde{\Omega }_{ija}.$
\end{consequence}

We note that we obtain holonomic (integrable) frame configurations if
respective anholonomy coefficients (\ref{anhrelc}) and/or (\ref{anhrelcd})
are zero.

\begin{remark}
\textsf{[generic off-diagonal metric structures canonically induced by MDRs ]%
} Canonical d-metric structures $\ \widetilde{\mathbf{g}}$ (\ref{cdms}) and $%
\ ^{\shortmid }\widetilde{\mathbf{g}}$ (\ref{cdmds}) are described by
generic off--diagonal metrics (\ref{offd}) if respective anholonomy
coefficients (\ref{anhrelc}) and (\ref{anhrelcd}) are not trivial. This
means that MDRs (\ref{mdrg}) generate off-diagonal metric structures on $%
\mathbf{TV}$ and $\mathbf{T}^{\ast }\mathbf{V}$ if certain special
conditions for integrability of respective frame structures and
diagonalization (on some finite, or infinite phase space regions) are not
imposed. It is necessary to elaborate more advanced and sophisticate
geometric and numeric methods for constructing exact, parametric and
approximate solutions with generic off-diagonal metrics and generalized
connections (for instance, BH, or cosmological type) in MGTs with MDRs.
\end{remark}

\subsubsection{Canonical almost symplectic structures determined by MDRs}

The geometry of N-connections and d-metric structures determined by MDRs and LIVs on nonholonomic (co) tangent Lorentz bundles  can be described equivalently in terms of canonical almost symplectic variables. For relativistic generalizations and models of geometric flows and deformation and A-brane complexified quantization of Einstein-Lagrange-Hamilton and generalized Finsler spaces, such constructions were elaborated in a series of our works
\cite{vjmp07,vpla08,vijgmmp09,vjgp10,bvnd11,vjmp13,vmjm15,vch2416,vmon02,avjmp09}; on preliminary geometric ideas see
\cite{matsumoto66,matsumoto86,kern74,etayo05}.

\begin{definition}
Almost symplectic structures on $\mathbf{TV}$ and $\mathbf{T}^{\ast }\mathbf{%
V}$ are defined by respective nondegenerate N-adapted 2--forms
\begin{equation*}
\ \theta =\frac{1}{2}\ \theta _{\alpha \beta }(u)\ \mathbf{e}^{\alpha
}\wedge \mathbf{e}^{\beta }\mbox{ and }\ \ ^{\shortmid }\theta =\frac{1}{2}\
\ ^{\shortmid }\theta _{\alpha \beta }(u)\ \ ^{\shortmid }\mathbf{e}^{\alpha
}\wedge \ ^{\shortmid }\mathbf{e}^{\beta }.
\end{equation*}
\end{definition}

One holds the following
\begin{proposition}
\label{pr01}For any $\mathbf{\theta }$ and $\ ^{\shortmid }\mathbf{\theta }$
on respective tangent and cotangent Lorentz bundle, there are unique
N--connections $\mathbf{N}=\{N_{j}^{b}\}$ and $\ ^{\shortmid }\mathbf{N}=\{%
\mathbf{\ ^{\shortmid }}N_{ia}\}$ satisfying the conditions:%
\begin{eqnarray}
\ \theta &=&(h\ \mathbf{X},v\mathbf{Y})=0,\mbox{ when }\ \theta =h\theta
+v\theta, \mbox{ and }  \notag \\
\ \ ^{\shortmid }\theta &=&(h\ ^{\shortmid }\mathbf{X},cv\ ^{\shortmid }%
\mathbf{Y})=0,\mbox{ when }\ ^{\shortmid }\theta \doteq h\ ^{\shortmid
}\theta +cv\ ^{\shortmid }\theta ,  \label{aux02v}
\end{eqnarray}%
for any $\mathbf{X}=h\mathbf{X}+v\mathbf{X,Y}=h\mathbf{Y}+v\mathbf{Y}$ and $%
\ ^{\shortmid }\mathbf{X}=h\ ^{\shortmid }\mathbf{X}+cv\ ^{\shortmid }%
\mathbf{X,}\ \ ^{\shortmid }\mathbf{Y}=h\ ^{\shortmid }\mathbf{Y}+cv\
^{\shortmid }\mathbf{Y,}$ where
\begin{eqnarray*}
h\theta (\mathbf{X,Y}) &:=&\theta (h\mathbf{X,}h\mathbf{Y}),v\theta (\mathbf{%
X,Y}):=\theta (v\mathbf{X,}v\mathbf{Y});\mbox{ and }\  \\
h\ ^{\shortmid }\theta (\ ^{\shortmid }\mathbf{X,\ ^{\shortmid }Y}) &:=&\
^{\shortmid }\theta (h\ ^{\shortmid }\mathbf{X,}h\ ^{\shortmid }\mathbf{Y}%
),\ cv\ ^{\shortmid }\theta (\ ^{\shortmid }\mathbf{X,\ ^{\shortmid }Y}):=\
^{\shortmid }\theta (cv\ ^{\shortmid }\mathbf{X,}cv\ ^{\shortmid }\mathbf{Y}%
).
\end{eqnarray*}
\end{proposition}

\begin{proof}
Let us sketch the proof on $\mathbf{T}^{\ast }\mathbf{V}$ for $\ ^{\shortmid
}\theta $ (the constructions are similar on $\mathbf{TV}).$

For $\ ^{\shortmid }\mathbf{X=\ ^{\shortmid }\mathbf{e}_{\alpha }=}(\
^{\shortmid }\mathbf{e}_{i},\ ^{\shortmid }e^{a})$ and $\ \ ^{\shortmid }%
\mathbf{Y=\ ^{\shortmid }e}_{\beta }=(\ ^{\shortmid }\mathbf{e}_{j},\
^{\shortmid }e^{b}),$ as in (\ref{ccnadap}), where $\ \mathbf{\ ^{\shortmid }%
\mathbf{e}_{\alpha }}$ is a N--adapted basis\ of type (\ref{cnadap}), we
write the first equation in (\ref{aux02v}) in the form%
\begin{equation*}
\ ^{\shortmid }\theta =\ ^{\shortmid }\theta (\ ^{\shortmid }\mathbf{e}%
_{i},\ ^{\shortmid }e^{a})=\ ^{\shortmid }\theta (\frac{\partial }{\partial
x^{i}},\frac{\partial }{\partial p_{a}})-\ ^{\shortmid }N_{ib}\mathbf{\ }\
^{\shortmid }\theta (\frac{\partial }{\partial p_{b}},\frac{\partial }{%
\partial p_{a}})=0.
\end{equation*}%
These conditions uniquely define $\ ^{\shortmid }N_{ib}$ if $\ ^{\shortmid
}\theta $ is non--degenerate, i.e. $rank|\ ^{\shortmid }\theta (\frac{%
\partial }{\partial p_{b}},\frac{\partial }{\partial p_{a}})|=4.$ Setting
locally
\begin{equation}
\mathbf{\ }\ ^{\shortmid }\theta =\frac{1}{2}\mathbf{\ }\ ^{\shortmid
}\theta _{ij}(u)e^{i}\wedge e^{j}+\frac{1}{2}\mathbf{\ }\ ^{\shortmid
}\theta ^{ab}(u)\mathbf{\ ^{\shortmid }e}_{a}\wedge \mathbf{\ }\ ^{\shortmid
}\mathbf{e}_{b},  \label{aux03}
\end{equation}%
where the first term is for $h\ ^{\shortmid }\theta $ and the second term is
$cv\ ^{\shortmid }\theta ,$ we get the second formula in (\ref{aux02v}).
\end{proof}

$\square $\vskip5pt

In above Proposition, the constructed N--connections are not canonical ones
as in Theorem \ref{thcnc}. There are necessary additional nonholonomic frame
deformations/ transforms in order to generate almost symplectic structures
in canonical N-adapted forms:\ A N--connection $\ ^{\shortmid }\mathbf{N}$
defines a unique decomposition of a d--vector $\ ^{\shortmid }\mathbf{X=\ }%
X^{h}+\ ^{\shortmid }X^{cv}$ on $T^{\ast }\mathbf{V},$ for $\mathbf{\ }%
X^{h}=h\ ^{\shortmid }\mathbf{X}$ and $\mathbf{\ }^{\shortmid }X^{cv}=cv\ \
^{\shortmid }\mathbf{X},$ where the projectors $h$ and $cv$ defines
respectively the dual distribution $\ ^{\shortmid }\mathbf{N}$ on $\mathbf{V}
$. They have the properties $h+cv=\mathbf{I},h^{2}=h,(cv)^{2}=cv,h\circ
cv=cv\circ h=0.$ This allows us to construct on $T^{\ast }\mathbf{V}$ the
almost product operator $\ ^{\shortmid }\mathbf{P:=}I-2cv=2h-I$ acting on $\
^{\shortmid }\mathbf{e}_{\alpha }=(\ ^{\shortmid }\mathbf{e}_{i},\
^{\shortmid }e^{b})$ following formulas
\begin{equation*}
\ \mathbf{\ }^{\shortmid }\mathbf{P}(\ ^{\shortmid }\mathbf{e}_{i})=\
^{\shortmid }\mathbf{e}_{i}\mbox{\ and  } \ ^{\shortmid }\mathbf{P}(\
^{\shortmid }e^{b})=-\ \ ^{\shortmid }e^{b}.
\end{equation*}

In a similar form, a N--connection $\ \mathbf{N}$ induces an almost product
structure $\mathbf{P}$ on $T\mathbf{V}.$

In geometric and physical models, there are used also the almost tangent
(co) operators
\begin{eqnarray*}
\mathbb{J(}\mathbf{e}_{i}\mathbb{)} &=&e_{4+i}\mbox{\ and
\ }\ \mathbb{J}\left( e_{a}\right) =0,\mbox{ \ or \ }\mathbb{J=}\frac{%
\partial }{\partial y^{i}}\otimes dx^{i}; \\
\mathbf{\ }^{\shortmid }\mathbb{J}(\mathbf{\ }^{\shortmid }\mathbf{e}_{i}%
\mathbb{)} &=&\mathbf{\ }^{\shortmid }g_{ib}\mathbf{\ }^{\shortmid }e^{b}%
\mbox{\ and
\ }\ \mathbf{\ }^{\shortmid }\mathbb{J}\left( \mathbf{\ }^{\shortmid
}e^{b}\right) =0,\mbox{ \ or \ }\mathbf{\ }^{\shortmid }\mathbb{J=}\mathbf{\
}^{\shortmid }g_{ia}\frac{\partial }{\partial p_{a}}\otimes dx^{i}.
\end{eqnarray*}%
The operators $\ ^{\shortmid }\mathbf{P,}\ \mathbf{\ }^{\shortmid}\mathbf{J}$
and $\ ^{\shortmid }\mathbb{J}$ are respectively $\mathcal{L}$--dual to $\
\mathbf{P,}\ \mathbf{J}$ and $\ \mathbb{J}$ if and only if $\ \mathbf{\ }%
^{\shortmid }\mathbf{N}$ and $\ \mathbf{N}$ are $\mathcal{L}$--dual and
there are constructed respective (co) frame transforms to canonical values $[%
\mathbf{\ }^{\shortmid }\widetilde{\mathbf{P}}\mathbf{,}\ \mathbf{\ }%
^{\shortmid }\widetilde{\mathbf{J}}\mathbf{,\ }^{\shortmid }\widetilde{%
\mathbb{J}}]$ and $[\widetilde{\mathbf{P}}\mathbf{,}\widetilde{\mathbf{J}}%
\mathbf{,}\widetilde{\mathbb{J}}].$

For the above--introduced almost complex and almost product operators, we
can check by straightforward computations the properties:

\begin{proposition}
\label{propctf}Let $\left( \mathbf{N,\ \mathbf{\ }^{\shortmid }N}\right) $
be a pair of $\mathcal{L}$--dual N--connections. Then, we can construct
canonical d--tensor fields (defined respectively by $L(x,y)$ and $H(x,p)$
related by Legendre transforms, see Remark on (\ref{invlegendre})):%
\begin{equation*}
\mathbf{J}=-\delta _{i}^{a}e_{a}\otimes e^{i}+\delta _{a}^{i}\mathbf{e}%
_{i}\otimes \mathbf{e}^{a},\ \ \mathbf{\ }^{\shortmid }\mathbf{J}=-\
^{\shortmid }g_{ia}\ ^{\shortmid }e^{a}\otimes \ ^{\shortmid }e^{i}+\
^{\shortmid}g^{ia} \ ^{\shortmid }\mathbf{e}_{i}\otimes \ ^{\shortmid }%
\mathbf{e}_{a}
\end{equation*}%
corresponding to the $\mathcal{L}$--dual pair of almost complex structures $%
\left( \mathbf{J},\ ^{\shortmid }\mathbf{J}\right) ,$
\begin{equation*}
\mathbf{P}=\mathbf{e}_{i}\otimes e^{i}-e_{a}\otimes \mathbf{e}^{a},\ \
^{\shortmid }\mathbf{P}= \ ^{\shortmid } \mathbf{e}_{i}\otimes \ ^{\shortmid
}e^{i}-\ ^{\shortmid }e^{a}\otimes \ ^{\shortmid }\mathbf{e}_{a}
\end{equation*}%
corresponding to the $\mathcal{L}$--dual pair of almost product structures $%
\left( \mathbf{P},\ \ ^{\shortmid }\mathbf{P}\right)$, and almost symplectic
structures
\begin{equation}
\theta =g_{aj}(x,y)\mathbf{e}^{a}\wedge e^{i}\mbox{ and } \ ^{\shortmid
}\theta =\delta _{i}^{a}\ ^{\shortmid }\mathbf{e}_{a}\wedge \ ^{\shortmid
}e^{i}  \label{sympf}
\end{equation}
\end{proposition}

The formulas in this Proposition can be re-written in canonical form by considering canonical N-adapted bases with tilde. For instance, one re-writes (using frame transforms) (\ref{sympf}) as
$\widetilde{\theta }=\widetilde{g}_{aj}(x,y)\widetilde{\mathbf{e}}^{a}\wedge e^{i}$ and
$\ ^{\shortmid }\widetilde{\theta }=\delta _{i}^{a} \ ^{\shortmid } \widetilde{\mathbf{e}}_{a}\wedge \ ^{\shortmid }e^{i}.$

For modeling of (co) tangent bundle N-connection and almost symplectic geometries on (co) tangent bundles with total dimension 8, we introduce the

\begin{definition}
\label{defaks}An almost Hermitian model of a tangent Lorentz bundle $T%
\mathbf{V}$\ (or a cotangent Lorentz bundle $T^{\ast }\mathbf{V}$) equipped
with a N--connection structure $\mathbf{N}$\ (or $\ ^{\shortmid }\mathbf{N})$
is defined by a triple $\mathbf{H}^{8}=(T\mathbf{V},\theta ,\mathbf{J}),$
where $\theta \mathbf{(X,Y)}:=\mathbf{g}\left( \mathbf{JX,Y}\right) $ (or by
a triple $\ ^{\shortmid }\mathbf{H}^{8}=(T^{\ast }\mathbf{V},\ ^{\shortmid
}\theta ,\ ^{\shortmid }\mathbf{J}),$ where $\ ^{\shortmid }\theta (\
^{\shortmid }\mathbf{X},\ ^{\shortmid }\mathbf{Y}):=\ ^{\shortmid }\mathbf{g}%
\left( \ ^{\shortmid }\mathbf{J}\ ^{\shortmid }\mathbf{X},\ ^{\shortmid }%
\mathbf{Y}\right) ).$ A space $\mathbf{H}^{8}$ (or $\ ^{\shortmid }\mathbf{H}%
^{8})$ is almost K\"{a}hler and denoted $\mathbf{K}^{8}$ if $d\ \theta =0$
(or $\ ^{\shortmid }\mathbf{K}^{8}$ if $d\ ^{\shortmid }\theta =0).$
\end{definition}

The following theorem holds:

\begin{theorem}
\textsf{[MDR-induced canonical almost K\"{a}hler-Lagrange/ - Hamilton spaces] } \label{thaklh}The Lagrange and Hamilton spaces (including those determined by MDRs (\ref{mdrg})) can be represented respectively as canonical almost K\"{a}hler spaces (called almost K\"{a}hler-Lagrange and almost K\"{a} hler-Hamilton) on $\mathbf{TV}$ and $\mathbf{T}^{\ast }\mathbf{V.}$
\end{theorem}
\begin{proof}
It follows from the existence on on $\mathbf{TV}$ and $\mathbf{T}^{\ast }\mathbf{V}$ of canonical N-connections under conditions of Theorem (\ref{thcnc}) and 1--forms, respectively, defined by a regular Lagrangian $L$ and Hamiltonian $H$ (related by a Legendre transform), $\widetilde{\omega }=\frac{\partial L}{\partial y^{i}}e^{i}$ and $\ ^{\shortmid }\widetilde{\omega }=p_{i}dx^{i},$ for which $\widetilde{\theta }=d\widetilde{\omega }$ and $\ ^{\shortmid }\widetilde{\theta }=d\ ^{\shortmid }\widetilde{\omega }$.
As a result, we get that $d\theta =0$ and $d\ ^{\shortmid }\widetilde{\theta }=0$ corresponding to the Definition \ref{defaks}.
\end{proof}

$\square $\vskip5pt

There are almost K\"{a}hler models on (co) tangent bundles defined canonically by (pseudo) Riemannian metrics on base manifolds with nonholonomic deformations to MGTs modeled as (effective phase) regular Lagrange and/or Hamilton spaces and respective almost Hermitian/ K\"{a}hler spaces. In particular, we can consider almost K\"{a}hler-Finsler spaces  originally elaborated in \cite{matsumoto66,matsumoto86}.

\subsection{Lagrange-Hamilton connections and curvatures}

The geometry of a (pseudo) Riemannian spacetime $(V,\{g_{ij}(x)\})$ is
completely determined by its metric structure $\{g_{is}\}$. There is a
uniquely defined by $\{g_{is}\}$ linear connection called the Levi-Civita
(LC) connection, $\nabla $, which by definition is torsionless and metric
compatible. In GR, geodesic equations and auto parallel curves are described
by equivalent systems on PDEs. For nontrivial MDRs and/or LIVs, point like
probing particles and small perturbations of wave equations and scalar field
equations do not move along usual geodesics (as in GR, on Lorentz manifolds)
but follow certain curves described by nonlinear geodesic equations (\ref%
{ngeqf}). In non-relativistic form, there were developed certain approaches
related to Finsler geometry and semi-spray configurations \cite%
{shen01,shen01a,bao00}, where the priority was given to the Chern connection
for Finsler spaces. Such a connection is not compatible with the metric
structure on the total bundle. This creates a number of ambiguities related
to elaborating metric noncompatible Finsler gravity theories (including
definition of spinors, definition of compatible motion equations and
conservation laws), see explicit results, critics and discussions in Refs.
\cite{vplb10,vmon06,vrev08,vijgmmp12}.

The Theorem \ref{thcnc}, Proposition \ref{prcnadapb} and Consequence \ref%
{anhr} state explicitly that MDRs and LIVs define on (co) tangent Lorentz
bundle certain generic off-diagonal metric structures (\ref{offd}) and
related N--connection (\ref{ncon}) and anholonomic frame structures
characterized by respective Neijenhuis tensors (\ref{neijt}). Elaborating on
different types of Hamilton-Lagrange-Finsler models encoding (\ref{mdrg}),
we are not able to perform the constructions in N-adapted anholonomic form
if we work only with generalized (Finsler like) metrics determined by
nonlinear quadratic forms $L(x,y)$ (\ref{nqe}) and/or $H(x,p)$ (\ref{nqed}).
The geometry of nonholonomic (co) tangent Lorentz manifolds encodes new
types of fundamental geometric/ physical objects and more rich geometric
structures. Physically viable nonholonomic deformations of GR on $\mathbf{TV}
$ and $\mathbf{T}^{\ast }\mathbf{V}$ can be elaborated for locally
anisotropic gravity and matter field theories if there are used certain well
defined geometrically and physically motivated linear connection structures.

The goal of this subsection is to analyze which classes of linear
connections and respective covariant derivative operators can be generated
canonically by fundamental nonlinear quadratic forms and applied for
formulating classical and quantum gravity and matter filed theories.

\subsubsection{Distinguished connections, N-adapted distortions and curvatures}

Let $D$ be a linear connection on $\mathbf{TV}$ when a $\mathcal{L}$
--duality between the tangent and corresponding cotangent bundles can be
defined by pull--back and push--forward maps (we omit geometric details on
constructing such maps from/to base space to total space, considered, for
instance, in Ref \cite{avjmp09}). One defines a linear connection $\
^{\shortmid }D$ on $\mathbf{T}^{\ast }\mathbf{V}$ as follows: $\
^{\shortmid}D_{\ ^{\shortmid }\mathbf{X}}\ ^{\shortmid }\mathbf{Y}:=(D_{%
\mathbf{X}}\mathbf{Y})^{\ast }= \ ^{\shortmid}(D_{\mathbf{X}}\mathbf{Y}),$
for any vector fields $\ ^{\shortmid }\mathbf{X}$ and $\ ^{\shortmid }%
\mathbf{Y}$ on $\mathbf{T}^{\ast }\mathbf{V}.$ Inversely, considering a
linear connection $\ \ ^{\shortmid }D$ on $\mathbf{T}^{\ast }\mathbf{V},$ we
construct a linear connection $\ ^{\circ }D$ on $\mathbf{TV},$ following the
rule $\ ^{\circ}D_{\mathbf{X}}\mathbf{Y}:=(\ ^{\shortmid }D_{\ ^{\shortmid }%
\mathbf{X}}\ ^{\shortmid }\mathbf{Y})^{\circ },$ for any vector fields $%
\mathbf{X}$ and $\mathbf{Y}$ on $\mathbf{TV}$.

On (co) tangent bundles, we can elaborate a geometry of affine (linear)
connections and respective covariant derivatives in certain forms being (or
not) adapted to a N--connection structure.

\begin{definition}
\textsf{[d-connections as N-adapted linear connections] } A distinguished connection (d--connection) is a linear connection $\mathbf{D} $ on $\mathbf{TV}$ (or $\ ^{\shortmid }\mathbf{D}$ on $\mathbf{T}^{\ast }\mathbf{V})$ which is compatible with the almost product structure $\mathbf{DP}=0$ (or $\ ^{\shortmid }\mathbf{D\ ^{\shortmid }P}=0)$. In equivalent form, such a
d--connection is defined to preserve under parallelism a respective N--connection splitting (\ref{ncon}).
\end{definition}

For a MDR (\ref{mdrg}), it is possible to construct $\mathcal{L}$--dual
Lagrange and Hamilton spaces when $\mathbf{DP}=0$ induces $\ ^{\shortmid }%
\mathbf{D\ ^{\shortmid }P}=0,$ and inversely. In general, we can define and
study respective independent N- and/or P-structures on $\mathbf{TV},$ or $%
\mathbf{T}^{\ast }\mathbf{V}$.

The coefficients of d--connections can be defined in corresponding N-adapted
forms with respect to N--adapted frames (\ref{nadapb}) and (\ref{cnadap}),
\begin{equation*}
\mathbf{D}_{\mathbf{e}_{\beta }}\mathbf{e}_{\gamma }:= \mathbf{\Gamma }_{\
\beta \gamma }^{\alpha }\mathbf{e}_{\alpha }\mbox{ and }\mathbf{\
^{\shortmid }D}_{\mathbf{\ ^{\shortmid }e}_{\beta }}\ \mathbf{^{\shortmid }e}%
_{\gamma }:=\mathbf{\ ^{\shortmid }\Gamma }_{\ \beta \gamma }^{\alpha }%
\mathbf{\ ^{\shortmid }e}_{\alpha },
\end{equation*}%
where (for a h-v splitting)
\begin{equation*}
\mathbf{D}_{\mathbf{e}_{k}}\mathbf{e}_{j}:=L_{\ jk}^{i}\mathbf{e}_{i},%
\mathbf{D}_{\mathbf{e}_{k}}e_{b}:=\acute{L}_{\ bk}^{a}e_{a},\mathbf{D}%
_{e_{c}}\mathbf{e}_{j}:=\acute{C}_{\ jc}^{i}\mathbf{e}_{i},\mathbf{D}%
_{e_{c}}e_{b}:=C_{\ bc}^{a}e_{a}
\end{equation*}%
and (for a h-cv splitting)
\begin{equation*}
\ \ ^{\shortmid }\mathbf{D}_{\ ^{\shortmid }\mathbf{e}_{k}}\ ^{\shortmid }%
\mathbf{e}_{j}:=\ ^{\shortmid }L_{\ jk}^{i}\ ^{\shortmid }\mathbf{e}_{i},\
^{\shortmid }\mathbf{D}_{\mathbf{e}_{k}}\ ^{\shortmid }e^{b}:=-\ ^{\shortmid
}\acute{L}_{a\ k}^{\ b}\ ^{\shortmid }e^{a},\ ^{\shortmid }\mathbf{D}_{\
^{\shortmid }e^{c}}\ ^{\shortmid }\mathbf{e}_{j}:=\ ^{\shortmid }\acute{C}%
_{\ j}^{i\ c}\ ^{\shortmid }\mathbf{e}_{i},\ ^{\shortmid }\mathbf{D}_{\
^{\shortmid }e^{c}}\ ^{\shortmid }e^{b}:=-\ ^{\shortmid }C_{a}^{\ bc}\
^{\shortmid }e^{a}.
\end{equation*}%
So, the N-adapted coefficients of d-connections on (co) tangent Lorentz
bundles are respectively parameterized
\begin{equation*}
\mathbf{\Gamma }_{\ \beta \gamma }^{\alpha }=\{L_{\ jk}^{i},\acute{L}_{\
bk}^{a},\acute{C}_{\ jc}^{i},C_{\ bc}^{a}\}\mbox{ and }\ ^{\shortmid }%
\mathbf{\Gamma }_{\ \beta \gamma }^{\alpha }=\{\ ^{\shortmid }L_{\ jk}^{i},\
^{\shortmid }\acute{L}_{a\ k}^{\ b},\ ^{\shortmid }\acute{C}_{\ j}^{i\ c},\
^{\shortmid }C_{a}^{\ bc}\}.
\end{equation*}%
Such values can be used for computations of h-- and/or v--splitting,
cv-splitting, of covariant derivatives
\begin{equation*}
\mathbf{D=}\left( _{h}\mathbf{D,\ }_{v}\mathbf{D}\right) \mbox{
and/or }\ ^{\shortmid }\mathbf{D}=\left( \mathbf{\ }_{h}^{\shortmid }\mathbf{%
D,\ }_{v}^{\shortmid }\mathbf{D}\right) ,
\end{equation*}%
where $\ _{h}\mathbf{D}=\{L_{\ jk}^{i},\acute{L}_{\ bk}^{a}\},\ _{v}\mathbf{D%
}=\{\acute{C}_{\ jc}^{i},C_{\ bc}^{a}\}$ and $\ _{h}^{\shortmid }\mathbf{D}%
=\{\ ^{\shortmid }L_{\ jk}^{i},\ ^{\shortmid }\acute{L}_{a\ k}^{\ b}\},$ $\
_{v}^{\shortmid }\mathbf{D}=\{\ ^{\shortmid }\acute{C}_{\ j}^{i\ c},\
^{\shortmid }C_{a}^{\ bc}\}.$

\begin{lemma}
\textsf{[distortion of linear connection structures ]} \label{ldist}Let us
consider on $\mathbf{TV}$ a linear connection $\underline{D}$ (which is not
obligatory a d-connection) and a d-connection $\mathbf{D}$. Such values on $%
\mathbf{T}^{\ast }\mathbf{V,}$ are respectively denoted $\ ^{\shortmid }%
\underline{D}$ and $\ ^{\shortmid }\mathbf{D}$ and can be related by
corresponding distortion d-tensors%
\begin{equation}
\mathbf{Z:=D-}\underline{D}\mbox{ and/or }\ ^{\shortmid }\mathbf{Z:=\
^{\shortmid }D-}\ ^{\shortmid }\underline{D}.  \label{dist}
\end{equation}
\end{lemma}

\begin{proof}
Fixing respective N-adapted frames, such distortion d-tensors
\begin{eqnarray*}
\mathbf{Z}_{\ \beta \gamma }^{\alpha } &=&\mathbf{\Gamma }_{\ \beta \gamma
}^{\alpha }-\underline{\Gamma }_{\ \beta \gamma }^{\alpha } \\
&=&\{Z_{\ jk}^{i}=L_{\ jk}^{i}-\underline{L}_{\ jk}^{i},\acute{Z}_{\ bk}^{a}=%
\acute{L}_{\ bk}^{a}-\underline{\acute{L}}_{\ bk}^{a},\acute{Z}_{\ jc}^{i}=%
\acute{C}_{\ jc}^{i}-\underline{\acute{C}}_{\ jc}^{i},Z_{\ bc}^{a}=C_{\
bc}^{a}-\underline{C}_{\ bc}^{a}\}\mbox{ and } \\
\ ^{\shortmid }\mathbf{Z}_{\ \beta \gamma }^{\alpha } &=&\ ^{\shortmid }%
\mathbf{\Gamma }_{\ \beta \gamma }^{\alpha }-\ ^{\shortmid }\underline{%
\mathbf{\Gamma }}_{\ \beta \gamma }^{\alpha } \\
&=&\{\ ^{\shortmid }Z_{\ jk}^{i}=\ ^{\shortmid }L_{\ jk}^{i}-\ ^{\shortmid }%
\underline{L}_{\ jk}^{i},\ ^{\shortmid }\acute{Z}_{a\ k}^{\ b}=\ ^{\shortmid
}\acute{L}_{a\ k}^{\ b}-\ ^{\shortmid }\underline{\acute{L}}_{a\ k}^{\ b},\
^{\shortmid }\acute{Z}_{\ j}^{i\ c}=\ ^{\shortmid }\acute{C}_{\ j}^{i\ c}-\
^{\shortmid }\underline{\acute{C}}_{\ j}^{i\ c},\ ^{\shortmid }Z_{a}^{\
bc}=\ ^{\shortmid }C_{a}^{\ bc}-\ ^{\shortmid }\underline{C}_{a}^{\ bc}\},
\end{eqnarray*}%
can be constructed in explicit form by considering corresponding differences
of N-adapted coefficients for linear and d--connections.
\end{proof}

$\square $\vskip5pt

Using similar definitions and theorems as for linear connections, we can
prove for d-connections:

\begin{definition}
\textbf{-Theorem}\footnote{%
In mathematical physics, there are used terms like Definition-Theorem / -
Lemma / - Corollary etc. for such definitions (new ideas, concepts, or
conventions) which motivated by certain explicit geometric constructions
and/or requesting formulation of some theorems and respective mathematical
proofs.}\ \textsf{[curvature, torsion and nonmetricity of d-connections ] }
Any d--connection $\mathbf{D,}$ or $\ ^{\shortmid }\mathbf{D,}$ is
characterized by respective curvature $(\mathcal{R},$ or $\ ^{\shortmid }%
\mathcal{R}),$ torsion $(\mathcal{T},$ or $\ ^{\shortmid }\mathcal{T}),$ and
nonmetricity, $(\mathcal{Q},$ or $\ ^{\shortmid }\mathcal{Q})$, d-tensors
defined and computed in standard forms:
\begin{eqnarray}
\mathcal{R}(\mathbf{X,Y}):= &&\mathbf{D}_{\mathbf{X}}\mathbf{D}_{\mathbf{Y}}-%
\mathbf{D}_{\mathbf{Y}}\mathbf{D}_{\mathbf{X}}-\mathbf{D}_{\mathbf{[X,Y]}},
\label{dcurvabstr} \\
\mathcal{T}(\mathbf{X,Y}):= &&\mathbf{D}_{\mathbf{X}}\mathbf{Y}-\mathbf{D}_{%
\mathbf{Y}}\mathbf{X}-[\mathbf{X,Y}]\mbox{ and }\mathcal{Q}(\mathbf{X}):=%
\mathbf{D}_{\mathbf{X}}\mathbf{g},  \notag \\
\mbox{ or }\ ^{\shortmid }\mathcal{R}(\ ^{\shortmid }\mathbf{X,\ ^{\shortmid
}Y}):= &&\ ^{\shortmid }\mathbf{D}_{\ ^{\shortmid }\mathbf{X}}\ ^{\shortmid }%
\mathbf{D}_{\ ^{\shortmid }\mathbf{Y}}-\ ^{\shortmid }\mathbf{D}_{\
^{\shortmid }\mathbf{Y}}\ ^{\shortmid }\mathbf{D}_{\ ^{\shortmid }\mathbf{X}%
}-\ ^{\shortmid }\mathbf{D}_{\mathbf{[\ ^{\shortmid }X,\ ^{\shortmid }Y]}},
\notag \\
\ ^{\shortmid }\mathcal{T}(\ ^{\shortmid }\mathbf{X,\ ^{\shortmid }Y}):= &&\
^{\shortmid }\mathbf{D}_{\ ^{\shortmid }\mathbf{X}}\ ^{\shortmid }\mathbf{Y}%
-\ ^{\shortmid }\mathbf{D}_{\ ^{\shortmid }\mathbf{Y}}\ ^{\shortmid }\mathbf{%
X}-[\ ^{\shortmid }\mathbf{X,\ ^{\shortmid }Y}]\mbox{ and }\ ^{\shortmid }%
\mathcal{Q}(\ ^{\shortmid }\mathbf{X}):=\ ^{\shortmid }\mathbf{D}_{\
^{\shortmid }\mathbf{X}}\ ^{\shortmid }\mathbf{g}.  \notag
\end{eqnarray}
\end{definition}

The N--adapted coefficients for the curvature, torsion and nonmetricity
d-tensors (\ref{dcurvabstr}) are provided in Appendix, see Corollary \ref%
{acorolcurv}.

\subsubsection{The Ricci and Einstein d--tensors}

The Ricci tensor for a d--connection on a (co) tangent bundle can be
constructed in standard form by contracting, for instance, the first and
forth indices of respective curvature d-tensors $\mathcal{R}$ and/or $\
^{\shortmid }\mathcal{R}$ (\ref{dcurvabstr}).

\begin{definition}
\textbf{-Theorem\ } \textsf{[Ricci tensors for d--connections ] } The Ricci
d--tensors are defined and computed as $Ric=\{\mathbf{R}_{\alpha \beta }:=%
\mathbf{R}_{\ \alpha \beta \tau }^{\tau }\},$ for a d-connection $\mathbf{D}$%
, and $\ ^{\shortmid }Ric=\{\ ^{\shortmid }\mathbf{R}_{\alpha \beta }:=\
^{\shortmid }\mathbf{R}_{\ \alpha \beta \tau }^{\tau }\},$ for a
d-connection $\ ^{\shortmid }\mathbf{D}$.
\end{definition}

In N-adapted form using formulas (\ref{dcurv}), one proves

\begin{corollary}
\textsf{[computation of Ricci d-tensors ] } The N-adapted coefficients of
the Ricci d--tensors of a d-connection in a (co) tangent Lorentz bundle are
parameterized in $h$- and/or $v$-, or $cv$-form, by formulas
\begin{eqnarray}
\mathbf{R}_{\alpha \beta } &=&\{R_{hj}:=R_{\ hji}^{i},\ \ R_{ja}:=-P_{\
jia}^{i},\ R_{bk}:=P_{\ bka}^{a},R_{\ bc}=S_{\ bca}^{a}\},\mbox{ or }
\label{dricci} \\
\ ^{\shortmid }\mathbf{R}_{\alpha \beta } &=&\{\ ^{\shortmid }R_{hj}:=\
^{\shortmid }R_{\ hji}^{i}, \ ^{\shortmid }R_{j}^{\ a}:=-\ ^{\shortmid }P_{\
ji}^{i\ \ \ a},\ \ ^{\shortmid }R_{\ k}^{b}:=\ ^{\shortmid }P_{a\ k}^{\ b\ \
a},\ ^{\shortmid }R_{\ }^{bc}=\ ^{\shortmid }S_{a\ }^{\ bca}\}.
\label{driccid}
\end{eqnarray}
\end{corollary}

If a (co) tangent bundle is enabled both with a d-connection, $\mathbf{D}$
(or$\ ^{\shortmid }\mathbf{D),}$ and d-metric, $\mathbf{g}$ (\ref{dmt}) (or$%
\ ^{\shortmid }\mathbf{g}$ (\ref{dmct})), [in particular, we can consider
canonical d-metrics $\widetilde{\mathbf{g}}$ (\ref{cdms}) and/or $\
^{\shortmid }\widetilde{\mathbf{g}}$ (\ref{cdmds}) encoding MDRs], we can
introduce nonholonomic Ricci scalars:

\begin{definition}
\textbf{-Theorem\ } \textsf{[scalar curvature of d-connection] } The scalar
curvature of a d-connection $\mathbf{D,}$ or $\ ^{\shortmid }\mathbf{D,}$
can be defined and computed for the inverse d-metric $\mathbf{g}^{\alpha
\beta },$ or $\ ^{\shortmid }\mathbf{g}^{\alpha \beta },$
\begin{equation*}
\ _{s}R:=\mathbf{g}^{\alpha \beta }\mathbf{R}_{\alpha
\beta}=g^{ij}R_{ij}+g^{ab}R_{ab}=R+S,\ \mbox{ or } \ _{s}^{\shortmid }R :=\
^{\shortmid }\mathbf{g}^{\alpha \beta }\ ^{\shortmid }\mathbf{R}_{\alpha
\beta }= \ ^{\shortmid }g^{ij}\ ^{\shortmid }R_{ij}+\ ^{\shortmid }g^{ab}\
^{\shortmid }R_{ab}=\ ^{\shortmid }R+\ ^{\shortmid }S,
\end{equation*}%
with respective h-- and v--components $R=g^{ij}R_{ij}, S=g^{ab}S_{ab},$ or $%
\ ^{\shortmid }R=\ ^{\shortmid }g^{ij}\ ^{\shortmid }R_{ij}, \ ^{\shortmid
}S=\ ^{\shortmid }g_{ab}\ ^{\shortmid }S^{ab}.$
\end{definition}

Using above Definitions-Theorems and Corollaries, we formulate
\begin{definition}
\textbf{-Theorem\ }\label{dteinstdt} \textsf{[the Einstein tensors for
d-connections ] } By constructions, the Einstein d-tensors on $\mathbf{TV}$
and/or $\mathbf{T}^{\ast }\mathbf{V}$ are defined:
\begin{equation*}
En=\{\mathbf{E}_{\alpha \beta }:=\mathbf{R}_{\alpha \beta }-\frac{1}{2}%
\mathbf{g}_{\alpha \beta }\ _{s}R\}\mbox{ and/or }\ ^{\shortmid }En=\{\
^{\shortmid }\mathbf{E}_{\alpha \beta }:=\ ^{\shortmid }\mathbf{R}_{\alpha
\beta }-\frac{1}{2}\ ^{\shortmid }\mathbf{g}_{\alpha \beta }\
_{s}^{\shortmid }R\}.
\end{equation*}
\end{definition}
\begin{proof}
Such proofs follow from explicit constructions on regions of some atlases
covering respectively $\mathbf{TV}$ and/or $\mathbf{T}^{\ast }\mathbf{V}$
using N-adapted coefficients (\ref{dcurv}) and (\ref{dricci}) and/or (\ref%
{driccid}).
\end{proof}

$\square $\vskip5pt

It should be noted that nonholonomic Einstein-Cartan models on (co) tangent
Lorentz bundles can be elaborated for respective data $(\mathbf{g,D})$
and/or $(\ ^{\shortmid }\mathbf{g,\ ^{\shortmid }D})$ considering, for
instance, generalized spinning liquid models for effective matter as sources
of d-torsions $\mathcal{T}$ and $\ ^{\shortmid }\mathcal{T},$ see (\ref%
{dtors}) and imposing the metricity conditions. How to formulate such
nonholonomic Riemann-Cartan and metric-affine and generalized
Lagrange-Hamilton-Finsler theories with nontrivial nonmetricities $\mathcal{Q%
} $ and/or $\ ^{\shortmid }\mathcal{Q}$, see formulas (\ref{dnonm}), was
studied in monograph \cite{vmon06}\footnote{%
see Chapter 1 and references therein related to Lagrange-Hamilton-Finsler
generalizations of the metric-affine gravity elaborated by F. W. Hehl's group%
}. For $\mathbf{g}=\widetilde{\mathbf{g}}$ and/or $\ ^{\shortmid }\mathbf{g}%
=\ ^{\shortmid }\widetilde{\mathbf{g}}$ determined by MDRs, we can model
generalized Lagrange-Finsler and/or Hamilton-Cartan phase spaces with
nonmetric backgrounds extending the class of (co) Finsler geometries with
the Chern and/or Berwald type connections \cite%
{berwald26,berwald41,bao00,shen01,shen01a}. Such geometric models result in
a number of ambiguities for constructing physically self-consistent theories
of locally anisotropic interactions for gravitational, gauge, scalar and
spinor fields (and further (super) string / noncommutative, quantum etc.
generalizations) as it is concluded in Refs. \cite%
{vplb10,vijgmmp08,vijgmmp12,vijmpd12,vmon06}.

An axiomatic approach to Finsler like generalizations of GR on (co) tangent
Lorentz bundles \cite{vjpcs11,vijmpd12} can be formulated for theories when
a triple consisting from a N-connection, a d-metric (or almost symplectic)
and a d-connection structures are canonically and metric compatible
determined (following certain minimal geometric and physical principles) by
a fundamental generating Lagrange and/or Hamilton function. In next
subsections, we construct and study some canonical and physically important
d-connection structures which allows us to elaborate physically viable classical
and quantum models of locally anisotropic gravitational and matter field
interactions. Such theories encode respective MDRs of type (\ref{mdrg}),
with a fixed point on a base spacetime manifold, or for a generalized phase
spacetime modelled on a (co) tangent Lorentz bundle.

\subsubsection{Physically important (Filsler like) d-connections on (co)
tangent bundles}

For elaborating MGTs and geometric mechanics models, one considers more
specials classes of d--connections on which can be defined completely by a
d-metric/ almost symplectic structure determined by a respective
Lagrange-Finsler and/or Hamilton-Cartan fundamental form.

\begin{definition}
\textbf{-Theorem\ } \label{phidc} \textsf{[physically important
d-connections] } The Almost K\"{a}hler-Lagrange and/or almost K\"{a}%
hler-Hamilton phase spaces (determined, or not, by respective MDRs (\ref%
{mdrg}) and a possible $\mathcal{L}$--duality) are characterized
respectively by such geometric and physically important linear connections
and canonical/ almost symplectic connections:
\begin{eqnarray}
\lbrack \mathbf{g,N]} &\mathbf{\simeq }&\mathbf{[}\widetilde{\mathbf{g}},%
\widetilde{\mathbf{N}}]\mathbf{\simeq \lbrack }\widetilde{\theta }:=%
\widetilde{\mathbf{g}}(\widetilde{\mathbf{J}}\cdot ,\cdot ),\widetilde{%
\mathbf{P}}\mathbf{,}\widetilde{\mathbf{J}}\mathbf{,}\widetilde{\mathbb{J}}]
\label{canondcl} \\
&\Longrightarrow &\left\{
\begin{array}{ccccc}
\nabla : &  & \nabla \mathbf{g}=0;\ \mathbf{T[\nabla ]}=0, &  & %
\mbox{Lagrange LC--connection}; \\
\widehat{\mathbf{D}}: &  & \widehat{\mathbf{D}}\ \mathbf{g}=0;\ h\widehat{%
\mathbf{T}}=0,\ v\widehat{\mathbf{T}}=0. &  &
\mbox{canonical Lagrange
d-connection}; \\
\widetilde{\mathbf{D}}: &  & \widetilde{\mathbf{D}}\widetilde{\theta }=0,%
\widetilde{\mathbf{D}}\widetilde{\theta }=0 &  &
\mbox{almost symplectic
Lagrange d-connection.};%
\end{array}%
\right.  \notag
\end{eqnarray}%
and/or
\begin{eqnarray}
\lbrack \ ^{\shortmid }\mathbf{g,\ ^{\shortmid }N]} &\mathbf{\simeq }&%
\mathbf{[}\ ^{\shortmid }\widetilde{\mathbf{g}},\ ^{\shortmid }\widetilde{%
\mathbf{N}}]\mathbf{\simeq \lbrack }\ ^{\shortmid }\widetilde{\theta }:=\
^{\shortmid }\widetilde{\mathbf{g}}(\ ^{\shortmid }\widetilde{\mathbf{J}}%
\cdot ,\cdot ),\ ^{\shortmid }\widetilde{\mathbf{P}},\ ^{\shortmid }%
\widetilde{\mathbf{J}},\ ^{\shortmid }\widetilde{\mathbb{J}}]
\label{canondch} \\
&\Longrightarrow &\left\{
\begin{array}{ccccc}
\ ^{\shortmid }\nabla : &  & \ ^{\shortmid }\nabla \ ^{\shortmid }\mathbf{g}%
=0;\ \ ^{\shortmid }\mathbf{T[\ ^{\shortmid }\nabla ]}=0, &  & %
\mbox{Hamilton LC-connection}; \\
\ ^{\shortmid }\widehat{\mathbf{D}}: &  & \ ^{\shortmid }\widehat{\mathbf{D}}%
\ \mathbf{g}=0;\ h\ ^{\shortmid }\widehat{\mathbf{T}}=0,\ cv\ ^{\shortmid }%
\widehat{\mathbf{T}}=0. &  & \mbox{canonical Hamilton d-connection}; \\
\ ^{\shortmid }\widetilde{\mathbf{D}}: &  & \ ^{\shortmid }\widetilde{%
\mathbf{D}}\ ^{\shortmid }\widetilde{\theta }=0,\ ^{\shortmid }\widetilde{%
\mathbf{D}}\ ^{\shortmid }\widetilde{\theta }=0 &  &
\mbox{almost symplectic
Hamilton d-connection.}%
\end{array}%
\right.  \notag
\end{eqnarray}
\end{definition}

\begin{proof}
It is sketched step by step by a respective linear connection / d-connection
determined by the same fundamental geometric objects up to frame transforms:

\begin{itemize}
\item Both variants of LC-connections on (co) tangent bundles, $\nabla $
and/or $\ ^{\shortmid }\nabla ,$ are defined and constructed in standard
abstract, coordinate, or N-adapted forms using (respectively) $\mathbf{g}$ (%
\ref{dmt}) and$\ ^{\shortmid }\mathbf{g}$ (\ref{dmct}). We have to consider $%
\widetilde{\mathbf{g}}$ (\ref{cdms}) and $\ ^{\shortmid }\widetilde{\mathbf{g%
}}$ (\ref{cdmds}) if we work in not N-adapted form with generic
off--diagonal metrics of type (\ref{offd}). Here, we note that
LC--connections can be defined without N--connection structures, i.e. such
linear connections are not d--connections. Nevertheless, such values may
encode Lagrange and/or Hamilton structures if they are computed for metrics/
d-metrics encoding, for instance, MDRs and respective Hessians (modelling
Lagrange and/or Hamilton spaces, see (\ref{hessls}) and/or (\ref{hesshs})).

\item The definition and proofs of existence of $\widehat{\mathbf{D}}$ and $%
\ ^{\shortmid }\widehat{\mathbf{D}}$ are provided by globalizing the
constructions from Corollary \ref{acoroldand} and respective N-adapted
coefficients (\ref{canlc}) and/or (\ref{canlc}). Such canonical
d-connections play a crucial role for decoupling and solving in general
off-diagonal forms for various types of Finsler like d-connections of
gravitational and matter filed like equations in MGTs, i.e. for elaborating
the AFDM.

\item The fundamental Lagrange-Finsler generating functions (which may be,
or not, induced by MDRs and LIVs) can be used for modeling almost K\"{a}hler
models of generalized Lagrange and/or Hamilton spaces. In relativistic form,
such theories were elaborated and applied to deformation quantization of
gravity theories and possible noncommutative generalizations in Refs. \cite%
{vjmp07,vpla08,vijgmmp09,vjgp10,bvnd11,vjmp13,vmjm15,vch2416,vmon02,avjmp09}%
, where there are provided necessary proofs and N-adapted coefficients for
the almost symplectic connections $\widetilde{\mathbf{D}}$ and $\
^{\shortmid }\widetilde{\mathbf{D}}.$ There are not considered such models
and solutions in this work (even $\widetilde{\mathbf{D}}$ was used for
finding commutative and noncommutative Finsler black hole solutions \cite%
{vcqg10,vcqg11,vijtp13}). It is important to analyze formulas with $%
\widetilde{\mathbf{D}}$ and $\ ^{\shortmid }\widetilde{\mathbf{D}}$ because
such d--connections transform into standard Finsler-Cartan ones for
respective nonholonomic parameterizations of the generating functions,
Sasaki type induced d-metrics and canonical N-connection structures.
\end{itemize}
\end{proof}

$\square $\vskip5pt

It is well-known Chern's definition \cite{chern48,bao07} that Finsler
geometry is an example of geometry when the assumption on quadratic linear
elements is dropped and (we emphasize additionally) there are elaborated new
geometric constructions determined by nonlinear quadratic line elements.
Various mathematical approaches were developed for respective Finsler
generalized norms and \textbf{metric structures}. We need additional
assumptions and have to performed more sophisticate geometric constructions
involving \textbf{N-connection and d-connection } structures in order to
formulate self-consistent and physically viable Finsler like generalizations
of Einstein gravity for theories with MDRs and LIVs.

\begin{remark}
\textsf{[complete and self-consistent models of Finsler geometry] } The
first self-consistent model of Finsler geometry (with local geometric
constructions with generalized metric, N-connection and d-connection
structures, and associated N-frames) was elaborated before 1935 by E. Cartan
\ \cite{cartan35} and citations therein. In those works, there were used
coordinate transforms of nonlinear and linear connections and developed the
original constructions with nonlinear quadratic elements due to the famous
habilitation thesis of B. Riemann defended in 1854, see \cite{riem1854}; and
introduced the term of Finsler geometry using the original work \cite%
{finsler18}. Conventionally, that model of Finsler-Cartan geometry, which is
metric compatible, can be described on tangent bundles (or on manifolds with
fibred structure) by a triple of fundamental geometric structures $(F:%
\widetilde{\mathbf{g}},\widetilde{\mathbf{N}},\widetilde{\mathbf{D}})$ all
determined by a so-called Finsler metric (generating function) $F$ which can
be associated to a class of MDRs and LIVs subjected to certain homogeneity
conditions.
\end{remark}

One of the most important nontrivial characteristics of a Finsler geometry
model is that it posses, in general, a nontrivial nonholonomic structure
determined by a N-connection structure. Even for certain additional
geometric/ physical assumptions, all geometric objects on a Finsler space
can be determined by a $F(x,y)$ on a tangent bundles, such a theory is with
a triple of fundamental geometric objects. The geometric and physical models
related to (generalized) Finsler theories are very different from the
(pseudo) Riemann geometry which is completely determined by the metric
structure.

\subsubsection{Distortion tensors for connections and curvature and Ricci
tensors}

MGTs with MDRs on (co) tangent bundles are characterized by multi-connection
structures which, in principle, can be derived by a metric structure
(induced by a fundamental Lagrange-Hamilton function) as we proved in
Definition-Theorem \ref{phidc}. One fixes a (non) linear connection
structure following certain physical/geometric principles resulting in a
self-consistent and experimentally verified physical theory. In order to
construct exact solutions following the AFDM, it is most convenient to work
with the canonical d-connections $\widehat{\mathbf{D}}$ and $^{\shortmid }%
\widehat{\mathbf{D}}.$ Following the conditions of Lemma \ref{ldist}, we
prove

\begin{theorem}
\textsf{[existence of unique and physically important distortions of
connections ] } \label{thdistr} \newline
There are unique distortions relations
\begin{eqnarray}
\widehat{\mathbf{D}} &=&\nabla +\widehat{\mathbf{Z}},\widetilde{\mathbf{D}}%
=\nabla +\widetilde{\mathbf{Z}},\mbox{ and }\widehat{\mathbf{D}}=\widetilde{%
\mathbf{D}}+\mathbf{Z,}\mbox{  determined by }(\mathbf{g,N)};
\label{candistr} \\
\ ^{\shortmid }\widehat{\mathbf{D}} &=&\ ^{\shortmid }\nabla +\ ^{\shortmid }%
\widehat{\mathbf{Z}},\ ^{\shortmid }\widetilde{\mathbf{D}}=\ ^{\shortmid
}\nabla +\ ^{\shortmid }\widetilde{\mathbf{Z}},\mbox{ and }\ ^{\shortmid }%
\widehat{\mathbf{D}}=\ ^{\shortmid }\widetilde{\mathbf{D}}+\ ^{\shortmid }%
\mathbf{Z,}\mbox{ determined by }(\ ^{\shortmid }\mathbf{g,\ ^{\shortmid }N)}%
;  \notag
\end{eqnarray}%
for distortion d-tensors $\widehat{\mathbf{Z}},\widetilde{\mathbf{Z}},$ and $%
\mathbf{Z,}$ on $T\mathbf{TV,}$ and $\ ^{\shortmid }\widehat{\mathbf{Z}},\
^{\shortmid }\widetilde{\mathbf{Z}},$ and $\ ^{\shortmid }\mathbf{Z,}$ on $T%
\mathbf{T}^{\ast }\mathbf{V.}$
\end{theorem}

It should be noted that the d--tensor $\widehat{\mathbf{Z}}$ in above
formulas is an algebraic combination of coefficients $\widehat{\mathbf{T}}%
_{\ \alpha \beta }^{\gamma }[\mathbf{g,N}]$ computed by introducing formulas
(\ref{canlc}) into (\ref{dtors}) (similarly for cotangent bundles). Such
values can be induced by corresponding nonholonomic structures (\ref{anhrelc}%
) and (\ref{anhrelcd}) determined by MDRs. This proves:

\begin{consequence}
MDRs (\ref{mdrg}) are characterized by respective canonical and/or almost
symplectic distortion d-tensors $\widehat{\mathbf{Z}}[\widetilde{\mathbf{g}},%
\widetilde{\mathbf{N}}],\widetilde{\mathbf{Z}} [\widetilde{\mathbf{g}},%
\widetilde{\mathbf{N}}],$ and $\mathbf{Z}[\widetilde{\mathbf{g}}, \widetilde{%
\mathbf{N}}],$ for (almost symplectic) Lagrange models, and $\ ^{\shortmid }%
\widehat{\mathbf{Z}}[\ ^{\shortmid }\widetilde{\mathbf{g}},\ ^{\shortmid }%
\widetilde{\mathbf{N}}],\ ^{\shortmid }\widetilde{\mathbf{Z}}[\ ^{\shortmid }%
\widetilde{\mathbf{g}},\ ^{\shortmid }\widetilde{\mathbf{N}}],$ and $\
^{\shortmid }\mathbf{Z}[\ ^{\shortmid }\widetilde{\mathbf{g}},\ ^{\shortmid }%
\widetilde{\mathbf{N}}],$ for (almost symplectic) Hamilton models.
\end{consequence}

Using nonholonomic frame transforms and distortions, we can elaborate
equivalent models of phase spaces and MGTs formulated in terms of
"preferred" for certain purposes geometric data. In the geometric "language"
of "tilde" objects, we obtain relativistic mechanical like formulations for
the geometry of phase spaces (locally anisotropic ether) with
straightforward procedures for performing deformation quantization. For
"hat" objects, we obtain many possibilities for decoupling and integrating
physically important systems of nonlinear PDEs.

\begin{conclusion}
\textsf{[equivalent canonical geometric data for modeling phase spaces with
MDRs] } \newline
The phase space geometry can be described in equivalent forms (up to
respective nonholonomic deformations of the linear connection structures and
nonholonomic frame transforms) by such data{\small
\begin{equation}
\begin{array}{ccccc}
\mbox{MDRs} & \nearrow & (\mathbf{g,N,}\widehat{\mathbf{D}})\leftrightarrows
(L:\widetilde{\mathbf{g}}\mathbf{,}\widetilde{\mathbf{N}},\widetilde{\mathbf{%
D}}) & \leftrightarrow (\widetilde{\theta },\widetilde{\mathbf{P}},%
\widetilde{\mathbf{J}},\widetilde{\mathbb{J}},\widetilde{\mathbf{D}}) &
\leftrightarrow \lbrack (\mathbf{g[}N],\nabla )],\mbox{ on }\mathbf{TV} \\
\mbox{indicator }\varpi &  & \updownarrow \mbox{ possible }\mathcal{L}%
\mbox{-duality }\& & \mbox{symplectomorphisms \cite{avjmp09}} & \updownarrow %
\mbox{ not N-adapted } \\
\mbox{ see (\ref{mdrg})} & \searrow & (\ ^{\shortmid }\mathbf{g,\
^{\shortmid }N,}\ ^{\shortmid }\widehat{\mathbf{D}})\leftrightarrows (H:\
^{\shortmid }\widetilde{\mathbf{g}},\ ^{\shortmid }\widetilde{\mathbf{N}},\
^{\shortmid }\widetilde{\mathbf{D}}) & \leftrightarrow (\ ^{\shortmid }%
\widetilde{\theta },\ ^{\shortmid }\widetilde{\mathbf{P}},\ ^{\shortmid }%
\widetilde{\mathbf{J}},\ ^{\shortmid }\widetilde{\mathbb{J}},\ ^{\shortmid }%
\widetilde{\mathbf{D}}) & \leftrightarrow \lbrack (\ ^{\shortmid }\mathbf{g}%
[\ ^{\shortmid }N],\ ^{\shortmid }\nabla )],\mbox{on}\mathbf{T}^{\ast }%
\mathbf{V}.%
\end{array}
\label{phspgd}
\end{equation}%
}
\end{conclusion}

In brief, we shall say that certain geometric constructions are canonical
(i.e. formulated in canonical nonholonomic variables) if they are performed
for "hat", or "tilde", d-connections and related geometric objects uniquely
derived for certain Lagrange-Hamilton fundamental generating functions (in
particular, for a Finsler metric $F$).

\begin{convention}
\textsf{[existence of a preferred d-connection for decoupling (modified)
Einstein equations and generating off-diagonal solutions] } We can work with
canonical d-connection structures on (co) tangent bundles, $\widehat{\mathbf{%
D}}$ and/or $\ ^{\shortmid }\widehat{\mathbf{D}}$ which allows us to
decouple and integrate in most general exact and parametric forms (with
generic off-diagonal metrics and generalized connections and effective
matter sources depending, in principle, on all spacetime and phase space
coordinates) the gravitational and matter field equations in MGTs and GR,
see details and proofs for above presented references for the AFDM (in the
proof of Definition-Theorem \ref{phidc}) and next sections.
\end{convention}

Lagrange-Finsler variables can be introduced on 4-d, and higher dimension,
(pseudo) Riemann spaces and in GR, see details and a number of examples in
Refs. \cite%
{vap97,vnp97,vhsp98,vmon98,vjhep01,vmon02,vmon06,vijgmmp07,vrev08,vijtp10a,vijgmmp11,gvvepjc14,gheorghiuap16}%
.

\begin{remark}
\textsf{[Lagrange-Hamilton variables in Einstein gravity, GR]} \newline
Prescribing a generating function $L(x^{1},x^{2},y^{3},y^{4}=t),$ [in this
remark, $\alpha =(i,a),$ for $i=1,2$ and $a=3,4]$ on a Lorentz manifold $%
\left( \ ^{4}V,\mathbf{g}_{\alpha \beta }\right) ,\dim (\ ^{4}V),$ endowed
with a metric of local signature $(+++-)$ and conventional splitting, we can
construct a canonical N-connection $\ ^{2+2}\widetilde{\mathbf{N}}:T(\
^{4}V)=h(\ ^{4}V)\oplus v(\ ^{4}V),$ see Theorem \ref{thcnc}. Respectively,
it is defined an N-adapted frame structure with 2+2 splitting of type $%
\widetilde{\mathbf{e}}_{\alpha }=(\widetilde{\mathbf{e}}_{1},\widetilde{%
\mathbf{e}}_{2},e_{3},e_{4}),$ see (\ref{cnddapb}). Via frame/coordinate
transforms, we can express the pseudo-Riemannian metric as a canonical
d-metric, $\mathbf{g}_{\alpha \beta }\simeq \widetilde{\mathbf{g}}_{\alpha
\beta }$(\ref{cdms}), and/or in off-diagonal form (\ref{offd}). Following
conditions (\ref{canondcl}), we can construct in unique form three types of
linear connections $\nabla \lbrack \mathbf{g}\simeq \widetilde{\mathbf{g}}],%
\widehat{\mathbf{D}}[\mathbf{g}\simeq \widetilde{\mathbf{g}}]$ and/o $%
\widetilde{\mathbf{D}}[\mathbf{g}\simeq \widetilde{\mathbf{g}}].$ This way,
the pseudo-Riemannian geometry can be formulated equivalently in standard
form with data $(\mathbf{g},\nabla ),$ and/or in Lagrange-Finsler like
variables $(\widetilde{\mathbf{g}},\widehat{\mathbf{N}},\widehat{\mathbf{D}}%
) $ and/or $(\widetilde{\mathbf{g}},\widetilde{\mathbf{N}},\widetilde{%
\mathbf{D}}),$ with respective distortion relations, $\widehat{\mathbf{D}}%
=\nabla +\widehat{\mathbf{Z}}$ and $\widetilde{\mathbf{D}}=\nabla +%
\widetilde{\mathbf{Z}}.$ Such geometric models (we can say "toy"
Lagrange-Finsler models on manifolds with conventional nonholonomic 2+2
splitting) are elaborated for (pseudo) Riemannian manifolds with prescribed
nonholonomic fibered structure. Fibered splitting is of type $h(\
^{4}V)\oplus v(\ ^{4}V)$ but the "standard" Lagrange-Finsler geometrise are
constructed for $TTV=hTV\oplus vTV.$
\end{remark}

An important example is that when imposing certain (in general,
nonholonomic) constraints of type $\widehat{\mathbf{Z}}=0,$ we obtain $%
\widehat{\mathbf{D}}_{|\widehat{\mathbf{Z}}=0}\simeq \nabla $ even $\widehat{%
\mathbf{D}}\neq \nabla .$\footnote{%
The frame/coordinate transformation laws of nonlinear and distinguished /
linear connections are different from that of tensors. It is possible to
define such a frame structure when different connections may be determined
by the same set of coefficients with respect to such a special frame and by
different sets in other systems of reference.} If such conditions are
satisfied, we can extract (pseudo) Riemannian LC-configurations from more
(general) nonholonmic metric-affine structures.

\begin{corollary}
\textsf{[extracting LC-configurations by additional (non) holonomic
constraints]} \newline
One extracts LC-configurations from $\widehat{\mathbf{D}}$ and/or $\
^{\shortmid }\widehat{\mathbf{D}}$ \ for respective zero distortions, $%
\widehat{\mathbf{Z}}$ and/or $\ ^{\shortmid }\widehat{\mathbf{Z}},$ if there
are imposed zero torsion conditions for $\widehat{\mathcal{T}}$ $=\{\widehat{%
\mathbf{T}}_{\ \alpha \beta }^{\gamma }\}=0$ and/or $\ \ ^{\shortmid }%
\widehat{\mathcal{T}}$ $=\{\ ^{\shortmid }\widehat{\mathbf{T}}_{\ \alpha
\beta }^{\gamma }\}=0,$ see $\ $(\ref{dtors}). Such conditions are satisfied
if
\begin{eqnarray}
\widehat{C}_{jb}^{i} &=&0,\Omega _{\ ji}^{a}=0\mbox{ and }\widehat{L}%
_{aj}^{c}=e_{a}(N_{j}^{c});  \label{lccondl} \\
\ ^{\shortmid }\widehat{C}_{j}^{i\ b} &=&0,\ ^{\shortmid }\Omega _{\ aji}=0%
\mbox{ and }\ ^{\shortmid }\widehat{L}_{c\ j}^{\ a}=\ ^{\shortmid }e^{a}(\
^{\shortmid }N_{cj}).  \label{lccondh}
\end{eqnarray}
\end{corollary}

\begin{proof}
Let us sketch such a proof on $\mathbf{TV}$ (the constructions on $\mathbf{T}%
^{\ast }\mathbf{V}$ are similar). Introducing (\ref{canondcl}) in (\ref%
{dtors}), we can check that one obtains zero values for
\begin{equation*}
\widehat{T}_{\ jk}^{i}=\widehat{L}_{jk}^{i}-\widehat{L}_{kj}^{i},\widehat{T}%
_{\ ja}^{i}=\widehat{C}_{jb}^{i},\widehat{T}_{\ ji}^{a}=-\Omega _{\
ji}^{a},\ \widehat{T}_{aj}^{c}=\widehat{L}_{aj}^{c}-e_{a}(N_{j}^{c}),%
\widehat{T}_{\ bc}^{a}=\ \widehat{C}_{bc}^{a}-\ \widehat{C}_{cb}^{a}.
\end{equation*}%
if the conditions (\ref{lccondl}) are satisfied.
\end{proof}

$\square $\vskip5pt

In a series of works \cite%
{vijgmmp08,vijtp10a,vjpcs11,kouretsis10,vijmpd12,basilakos13,vepjc14a,svcqg13}%
, we proved that the equations (\ref{lccondl}) can be solved in explicit
form for manifolds/ bundle spaces of dimensions 4-10. Similarly, the
equations (\ref{lccondh}) can be integrated on $\mathbf{T}^{\ast }\mathbf{V.}
$

Introducing distortions from Theorem \ref{thdistr} into formulas (\ref%
{dcurvabstr}), we can prove in abstract and N-adapted forms:

\begin{theorem}
\textsf{[existence of canonical distortions of Riemannian and Ricci
d-tensors determined by MDRs] } \label{thcandist}There are canonical
distortion relations encoding MDRs for respective Lagrange-Finsler
nonholonomic variables:

\begin{itemize}
\item For the curvature d-tensors,%
\begin{eqnarray*}
\widehat{\mathcal{R}}[\mathbf{g},\widehat{\mathbf{D}} &=&\nabla +\widehat{%
\mathbf{Z}}]=\mathcal{R}[\mathbf{g},\nabla ]+\widehat{\mathcal{Z}}[\mathbf{g}%
,\widehat{\mathbf{Z}}], \\
\ ^{\shortmid }\widehat{\mathcal{R}}[\ ^{\shortmid }\mathbf{g},\ ^{\shortmid
}\widehat{\mathbf{D}} &=&\ ^{\shortmid }\nabla +\ ^{\shortmid }\widehat{%
\mathbf{Z}}]=\ ^{\shortmid }\mathcal{R}[\ ^{\shortmid }\mathbf{g},\
^{\shortmid }\nabla ]+\ ^{\shortmid }\widehat{\mathcal{Z}}[\ ^{\shortmid }%
\mathbf{g},\ ^{\shortmid }\widehat{\mathbf{Z}}],
\end{eqnarray*}
with respective distortion d-tensors $\ \widehat{\mathcal{Z}},$ on $\mathbf{%
TV,}$ and $\ ^{\shortmid }\widehat{\mathcal{Z}},$ on $\mathbf{T}^{\ast }%
\mathbf{V};$

\item For the Ricci d-tensors,%
\begin{eqnarray*}
\widehat{R}ic[\mathbf{g},\widehat{\mathbf{D}} &=&\nabla +\widehat{\mathbf{Z}}%
]=Ric[\mathbf{g},\nabla ]+\widehat{Z}ic[\mathbf{g},\widehat{\mathbf{Z}}], \\
\ ^{\shortmid }\widehat{R}ic[\ ^{\shortmid }\mathbf{g},\ ^{\shortmid }%
\widehat{\mathbf{D}} &=&\ ^{\shortmid }\nabla +\ ^{\shortmid }\widehat{%
\mathbf{Z}}]=\ ^{\shortmid }Ric[\ ^{\shortmid }\mathbf{g},\ ^{\shortmid
}\nabla ]+\ ^{\shortmid }\widehat{Z}ic[\ ^{\shortmid }\mathbf{g},\
^{\shortmid }\widehat{\mathbf{Z}}],
\end{eqnarray*}%
with respective distortion d-tensors $\ \widehat{Z}ic,$ on $\mathbf{TV,}$
and $\ \ ^{\shortmid }\widehat{Z}ic,$ on $\mathbf{T}^{\ast }\mathbf{V};$

\item For the scalar curvature of canonical d-connection $\widehat{\mathbf{D}%
}\mathbf{,}$ or $\ ^{\shortmid }\widehat{\mathbf{D}}\mathbf{,}$%
\begin{eqnarray*}
\ _{s}^{\shortmid }\widehat{R}[\mathbf{g},\widehat{\mathbf{D}} &=&\nabla +%
\widehat{\mathbf{Z}}]=\mathcal{R}[\mathbf{g},\nabla ]+\ _{s}\widehat{Z}[%
\mathbf{g},\widehat{\mathbf{Z}}], \\
\ _{s}^{\shortmid }\widehat{R}[\ ^{\shortmid }\mathbf{g},\ ^{\shortmid }%
\widehat{\mathbf{D}} &=&\ ^{\shortmid }\nabla +\ ^{\shortmid }\widehat{%
\mathbf{Z}}]=\ _{s}^{\shortmid }R[\ ^{\shortmid }\mathbf{g},\ ^{\shortmid
}\nabla ]+\ _{s}^{\shortmid }\widehat{Z}[\ ^{\shortmid }\mathbf{g},\
^{\shortmid }\widehat{\mathbf{Z}}],
\end{eqnarray*}%
with respective distortion scalar functionals $\ \ _{s}\widehat{Z},$ on $%
\mathbf{TV,}$ and $\ _{s}^{\shortmid }\widehat{Z},$ on $\mathbf{T}^{\ast }%
\mathbf{V.}$
\end{itemize}
\end{theorem}

\begin{proof}
We omit such tedious abstract and/or frame computations in this work. In
\cite{vijmpd03,vijmpd03a,vjmp05,vmon06,avjgp09,vijtp10,vcqg10,
vijtp10a,vijgmmp11,vcqg11,vepl11,vijtp13,vjpcs13,vport13,vepjc14,vepjc14a,vvyijgmmp14a,gvvepjc14,gvvcqg15}%
, there are provided details of similar proofs for Lagrange-Finsler spaces
and on (pseudo) Riemannian manifolds with nonholonomic fibered structure. In
this work, we shall construct exact solutions for MGTs encoding MDRs working
with the canonical d-connection $\widehat{\mathbf{D}}\mathbf{,}$ or $\
^{\shortmid }\widehat{\mathbf{D}},$ which can be restricted to
LC-configurations by solving, respectively, the equations (\ref{lccondl}),
or (\ref{lccondh}).
\end{proof}

$\square $\vskip5pt

\begin{remark}
\textsf{[MDR-distortions of canonical almost symplectic Lagrange-Hamilton
structures] } \newline
The conditions of above Theorem can be reformulated for distortions of the
almost symplectic Lagrange, or Finsler, d-connections, for instance,
considering
\begin{eqnarray*}
\widetilde{\mathcal{R}}[\widetilde{\mathbf{g}} &\simeq &\widetilde{\theta },%
\widetilde{\mathbf{D}}=\nabla +\widetilde{\mathbf{Z}}]=\mathcal{R}[%
\widetilde{\mathbf{g}}\simeq \widetilde{\theta },\nabla ]+\widetilde{%
\mathcal{Z}}[\widetilde{\mathbf{g}}\simeq \widetilde{\theta },\widetilde{%
\mathbf{Z}}], \\
\ ^{\shortmid }\widetilde{\mathcal{R}}[\ ^{\shortmid }\widetilde{\mathbf{g}}
&\simeq &\ ^{\shortmid }\widetilde{\theta },\ ^{\shortmid }\widetilde{%
\mathbf{D}}=\ ^{\shortmid }\nabla +\ ^{\shortmid }\widetilde{\mathbf{Z}}]=\
^{\shortmid }\mathcal{R}[\ ^{\shortmid }\widetilde{\mathbf{g}}\simeq \
^{\shortmid }\widetilde{\theta },\ ^{\shortmid }\nabla ]+\ ^{\shortmid }%
\widetilde{\mathcal{Z}}[\ ^{\shortmid }\mathbf{g}\simeq \ ^{\shortmid }%
\widetilde{\theta },\ ^{\shortmid }\widetilde{\mathbf{Z}}],
\end{eqnarray*}%
and any similar geometric objects with "tilde" symbols. Similar distortions
can be defined and computed, for instance, for the Chern d-connection,
Berwald d-connection and any d-connection structure which can be constructed
in unique forms by the same d-metric and N-connection structures \cite%
{vijgmmp12}. Working with "hat" distortions stated in Theorem \ref{thcandist}%
, we can prove the decoupling property of modified/ generalized Einstein
equations on $\mathbf{TV}$ and $\mathbf{T}^{\ast }\mathbf{V.}$
\end{remark}

\subsubsection{On Akbar-Zadeh definition of the Ricci tensor for Finsler
like spaces}

In a series of geometric models, it is applied an alternative concept of Ricci
tensor for Finsler spaces which by definition does not involve the concept
of N-connection and/or d--connection, see \cite{akbar88,akbar95}. Following
Akbar-Zadeh constructions for a base manifold $M,$ it in used as the
curvature for Finsler geometry a value
\begin{equation*}
\breve{\mathbf{R}}=\breve{R}_{k}^{i}\ dx^{k}\otimes \frac{\partial }{%
\partial x^{i}}|_{x}:T_{x}M\rightarrow T_{x}M,
\end{equation*}%
defined for any point $x\in M.$ This type of "curvature" is not with an
associated 2-form for a linear connection like in (\ref{dcurvabstr}), but
constructed directly from the data for nonlinear geodesic equations (\ref%
{ngeqf}) and respective semi-spray $\tilde{G}^{k},$ when
\begin{equation*}
\breve{R}_{k}^{i}:=2\frac{\partial \tilde{G}^{i}}{\partial x^{k}}-y^{j}\frac{%
\partial ^{2}\tilde{G}^{i}}{\partial x^{j}\partial y^{k}}+2\tilde{G}^{j}%
\frac{\partial ^{2}\tilde{G}^{i}}{\partial y^{j}\partial y^{k}}-\frac{%
\partial \tilde{G}^{i}}{\partial y^{j}}\frac{\partial \tilde{G}^{j}}{%
\partial y^{k}}.
\end{equation*}%
Contracting indices in this formula, we can define and compute a scalar
function $\breve{R}(x,y):=F^{-2}\breve{R}_{i}^{i}.$ By definition, the
scalar $\breve{R}$ is positive homogeneous of degree $0$ in $v$--variables $%
y^{a}.$ Such values are convenient for constructing and studies of geometric
objects in $T_{x}M$ for a point $x\in $ $M.$

The Ricci tensor "a la Akbar-Zadeh" is introduced by
\begin{equation}
\breve{R}ic_{jk}:=F^{-2}\frac{\partial ^{2}\breve{R}}{\partial y^{j}\partial
y^{k}},  \label{ricciaz}
\end{equation}%
which is different from the (\ref{dricci}) and other types of Ricci
d-tensors constructed in explicit form, for instance, using the
Finsler-Cartan d-connection. This geometric object is induced by the Finsler
metric $F$ via inverse Hessian ${\tilde{g}}^{ij}$ and $\tilde{G}^{k}.$ In $%
\breve{R}ic_{jk}$ (\ref{ricciaz}), there are not involved definitions for a
N--connection structure, lifts of metrics on total space tangent bundle, and
d--connections. For similar construction on a Lorentz manifold, $%
M\rightarrow \mathbf{V},$ \ such values may characterize certain MDRs (\ref%
{mdrg}) for ${\tilde{g}}_{ij}$ encoding such modifications. It is possible
also to introduce a variant of Einstein like equation, $\breve{R}%
ic_{jk}=\lambda (x){\tilde{g}}_{jk}$, i.e. when the scalar function $\breve{R%
}(x,y)=\lambda (x)$ is a function only on $h$--variables $x^{k}.$ Another
important property of $\breve{R}ic_{jk}$ (\ref{ricciaz}) is that it is
always symmetric (by definition) which provides a "simplified" model to
study Ricci fields and/or evolution dynamics in any point $T_{x}M,$ see \cite%
{bao07}.

Nevertheless, we consider that is not possible to define a complete,
self-consistent and physically viable Finlser like MGT on $\mathbf{TV}$ and/
or $\mathbf{T}^{\ast }\mathbf{V}$ working only with geometric objects generated by fundamental metric structures and do not
involving additional constructions with N--connection and d--connection
structures. In order to introduce matter fields (tensor and spinor ones) and
define geometric flows (with Perelman's functionals), we need covariant and
local derivatives which are positively defined by some (non) linear
connection objects, see critics in Refs. \cite{vplb10,vijgmmp12,vijmpd12}.
We can distinguish in physical theories certain configurations with $\breve{R%
}ic_{jk}$ for the h-components of a distortion relation like in Theorem \ref%
{thcandist} computing distortion relations, for instance,
\begin{equation*}
\widehat{R}ic_{ik}[\mathbf{\tilde{g}},\widehat{\mathbf{D}}]=\breve{R}ic_{jk}[%
{\tilde{g}}_{ij}]+\breve{Z}ic_{ij}[\mathbf{\tilde{g}},\widehat{\mathbf{Z}}].
\end{equation*}%
Such constructions can not be complete without additional assumptions on $v$%
- and/or $cv$-components and it is not possible to integrate in certain
general forms modified Einstein equations with $\breve{R}ic_{jk}$ and/or $%
\breve{Z}ic$. In our works, we shall work with Ricci tensors for
d--connections, of type (\ref{dricci}), which allow certain nonholonomic
transforms and constraints on $h$-subspaces in order to reproduce geometric
structures related to $\breve{R}ic_{jk}$.

\section{Geometric and Physical Principles for Gravity Theories on (Co) Tangent Bundles}

\label{saxiom} In this section, we formulate and analyse a set of geometric and physical principles which are necessary for self-consistent causal formulations of MGTs defined by (generalized) Finsler multi-connection and/or bi-metric structures. For metric compatible d-connections, the geometric constructions are similar to the GR theory but extended on nonholonomic phase space manifolds or (co)tangent bundles. Physically important Lagrange densities are postulated as in the Einstein gravity theory but in terms of respective data with triples of fundamental geometric objects (i.e. some metric compatible d-metric, d-connection, and N-connection structures) on base Lorentz manifolds and their total (co) tangent bundles. Applying formal geometric and/or
N-adapted variational methods, we derive formulas for (effective) sources of matter fields on curve phase spacetimes. There are constructed and studied explicit models with MDRs and locally anisotropic gauge and Higgs field interactions, modified massive and bi-metric theories, short-range models with LIVs.

\subsection{Principles for extending GR to Finsler--Lagrange-Hamilton gravity theories}

\label{ssmpext}The concept of flat Minkowski spacetime (with pseudo--Euclidean signature) and the postulates which are necessary for formulating the special relativity theory, SRT, allow us to unify in a relativistic manner the classical mechanics and the Maxwell electromagnetic field theory. The approach was formulated in agreement with various types of Michelson--Morley experiments proving the existence of a constant maximal speed of light. Such logical and explicit experiments prove and involve certain
fundamental properties of local spacetime local isotropy and homogeneity under the assumption that the concept of ether is not necessary for describing vacuum configurations. The most important symmetries in SRT are those of Lorentz (pseudo-rotation) and Poincar\'{e} (with additional translations) invariance with respect to linear group transforms. Such groups of automorphism of the Minkowski spacetime determine the conservation laws for energy and momentum (rotation momentum) values.

The GR theory was formulated in a standard (pseudo) Riemannian form using geometric data $(g,\nabla )$ for a Lorentz manifold $\mathbf{V}$ with causality structure. That approach preserves locally the symmetries of SRT following certain fundamental principles and axioms which can be extended for a large class of generalized theories with MDRs and LIVs. In this subsection, we speculate on modifications of GR for metric compatible Finsler like connections on $\mathbf{TV}$ and $\mathbf{T}^{\ast }\mathbf{V}$ of the principle of equivalence, with a partial realisation of the Mach principle; and of the general covariance principle. It is discussed the relation of equations of motion and conservation laws to Bianchi relations and nonholonomic deformations of d-connections. We elaborate on equivalent geometric and variational formulations in N-adapted variables of generalized Einstein
equations for MGTs with MDRs.

\subsubsection{Modified equivalence principles}

The experimental data show that the Newtonian gravitational force on a body is proportional to its inertial mass. This supports a fundamental idea that all bodies and fields are influenced in a similar "manner" by gravity and, indeed, point masses, light and small perturbations of scalar fields fall precisely in the same way in gravitational fields. Considering that motion of probing bodies and linearized interactions of classical fields are independent of the nature of the bodies, the paths of freely falling bodies define a preferred set of lines on a curved spacetime which locally is just as in SRT.

The world lines of freely falling bodies and small perturbations of electromagnetic/scalar fields in a gravitational field in GR are simply described by the geodesics of the (curved) spacetime metric. This suggests the possibility of ascribing the properties of the gravitational field to the structure of spacetime itself. Because MDRs on a Minkowski spacetime can be associated with metrics of type $g_{ij}(y)$ and/or $g_{ij}(p),$ when GR is defined by metrics of type $g_{ij}(x),$ we suppose that one can be formulated a generalized equivalence principle on some generalized Finsler spacetimes with metrics of type $\ ^{F}g_{ij}(x,y)$ and/or $\ ^{F}g_{ij}(x,p) $. For general coordinate transforms, we can omit the left label F and consider structures on certain (co) tangent spaces. Such locally anisotropic metrics can be related to Hessians (\ref{hessls}) and/or (\ref{hesshs}) of respective nonlinear quadratic elements (\ref{nqe}) and/or (\ref{nqed}) and corresponding N--connection structures. For small values of an
indicator of MDRs (\ref{mdrg}), we suppose that it is possible to preserve the ideas of Universality of Free Fall (and similar Universality of the Gravitational Redshift, or propagation of small perturbations of light/scalar fields etc.) in a Finsler-Lagrange-Hamilton type generalized spacetime modelled by data $\left( \mathbf{N,g,D}\right) $ and/or
$(\ ^{\shortmid }\mathbf{N,\ ^{\shortmid }g,\ ^{\shortmid }D)}$. These geometric data can be stated on nonholonomic manifolds with fibred structure and/or on nonholonomic (co) tangent bundles. In such phase spaces (i.e. locally anisotropic spacetime models), the paths of freely falling bodies, and propagation of small perturbations of scalar fields, are not usual geodesics but certain nonlinear (semi--spray) curves which are different from auto--parallels of $\mathbf{D}$ and/or $\ ^{\shortmid }\mathbf{D}$. For
models of generalized Finsler spacetimes, it is important to study the geometry of semi--spray configurations (\ref{ngeqf}) as N--connection generalizations of auto--parallel and geodesic curves. In certain sense, semi--sprays characterize via N--connections and respective adapted nonholonomic variables certain MDRs effects on "physical paths" of test particles.

Working with metric compatible d--connections $\mathbf{D}$ and/or $\ ^{\shortmid }\mathbf{D}$ completely determined by some metric and N--connection structures on respective (co) tangent bundles, we can establish a 1--1 correspondence between one type of preferred curves (semi--sprays) and respective auto--parallels. This way, we can encode equivalently the experimental (curvature deviation) data with respect to both types of congruences. To derive important physical equations for a Finsler gravitational and matter fields, covariant derivatives determined by a d-connection $\mathbf{D}$, or $\ ^{\shortmid }\mathbf{D}$, must be
considered. For canonical constructions, we can work equivalently $\widehat{\mathbf{D}},$ or $\widetilde{\mathbf{D}},$ and respective generalizations for locally anisotropic spinors, see details in Refs.  \cite{vjmp96,vhsp98,vmon98,vpcqg01,vmon02,vtnpb02,vvicol04,vmon06,vjmp06,vjmp09}. Such a d--connection can be used for defining generalized Dirac, d' Alambert and other physically important operators which allows us to compute the light and particle propagation and study classical and quantum field interactions in a Finsler spacetime and/or Lagrange-Hamilton type modifications.

\begin{principle}
\label{pgpeq} \textbf{\ [modified equivalence principle]}:\ In a MGT with an indicator
$\varpi (x^{i},E,\overrightarrow{\mathbf{p}},m;\ell _{P})$ (\ref{mdrg}) modeling MDRs and LIVs on (co) tangent bundles, there are
nonholonomic variables when point masses, light and small perturbations of scalar fields fall precisely in the same way in an effective phase spacetime along semi-spray configurations (\ref{ngeqf}) and associated auto-parallel equations on $\mathbf{TV}$ or $\mathbf{T}^{\ast}\mathbf{V}$. Using canonical nonholonomic variables, we can describe such a phase space by $\mathcal{L}$%
-dual Lagrange and Hamilton geometries (\ref{phspgd}) with respective canonical d-connections ($\widehat{\mathbf{D}}$ and $\widetilde{\mathbf{D}},$ or $\ ^{\shortmid} \widehat{\mathbf{D}}$ and $\ ^{\shortmid}\widetilde{\mathbf{D}})$. Such canonical and/or almost symplectic d-connections are related by canonical distortion relations (\ref{candistr}) and can be used for definition of N-adapted covariant derivative operators, Dirac operators for locally anisotropic spinors and/or almost symplectic models.
\end{principle}

This Principle can be formulated in more general forms for nonholonomic configurations modeling noncommutative and/or supersymmetric MGTs for various types of metric, frame and (non) linear connection structures derived for certain fundamental geometric objects on generalized / analogous manifolds and (co) tangent bundles. Geometrically, such possibilities are motivated by the fact that complex / noncommutative / supersymmetric configurations can be modelled by certain nonholonomic distributions on real manifolds and/or bundle spaces. In results, we suppose that it can be provided always  a self-consistent physical interpretation for theories constructed as nonholonomic deformations of GR and standard particle physics models. Such geometric
models can be elaborated for some general data of type $\left( \mathbf{N,g,D}\right)$ which in explicit and non-explicit form involve certain MDRs (\ref{mdrg}) or in other nonlinear and non-quadratic forms.

\subsubsection{Generalized Mach principles}

The Einstein gravity theory was formulated using a second much less precise set of ideas which goes under the name of Mach`s principle. There were involved various philosophical speculations on properties of space and time (in unified forms, spacetime) aether and associated models of continuum mechanics for such an aether media. For elaborating pre--relativity notions of spacetime and classical field theories and in STR, the geometric structure of spacetime is given once and for all and is considered to be unaffected by the material bodies or certain field interactions that may be present. In particular, it was stated that the
properties of "inertial motion" and "non--rotating" space are not influenced by matter and fields in the universe.

Mach supposed that all matter in the universe (in a modern fashion, we can include nonlinear effective/ observational / dark matter field interactions and evolution of certain material and fundamental fields configurations) should contribute to the local definition of "non--acceleration" and "non--rotating" of the fundamental space (in a modern approach, spacetime) structure. Einstein accepted this idea and was strongly motivated to formulate a theory where, unlike SRT, the spacetime geometry is
influenced by the presence of matter. Such purposes were achieved only partially in GR: the influence of matter was encoded in the right part of the Einstein equations via energy-momentum tensor. In a more general context, we can consider certain effective sources determined cumulatively by generalized connections and/or generic off-diagonal configurations and nonholonomic gravitational distributions and non-minimal coupling to matter fields, for stochastic / kinetic processes, geometric evolution processes etc. Such generalized and nonholonomic spacetime configurations encode more richer geometric vacuum and nonvacuum structures with nontrivial topology, (non) linear symmetries and modified conservation laws. For such off-diagonal gravitational solutions and related spacetime models, the Max principle is extended to a more general set of ideas (still less precise) defining a gravitational spacetime ether with rich structure encoding nonlinear processes (relativistic evolutions and nonlinear interactions) and nonholonomic configurations.

For generalized Finsler gravity theories on (co) tangent bundles derived from classical and/or quantum MDRs, we have to formulate a generalized Mach principle stating that the quantum energy and (non) linear fluctuations, motion of (effective) particles and fields, background metrics and generalized hidden / (non) linear symmetries should all contribute to the spacetime structure modifying it into a nonholonomic phase space. We argue that generalizations to effective phase spaces extending the concept of Lorentz manifolds is motivated by the presence of the quantum world, by kinetic and diffusion processes, by nonlinear off-diagonal interactions and modified symmetries. The influence of such (non) linear effects is encoded both into the nonholonomic structure for data $\left( \mathbf{N,g,D}\right)$ and/or $(\ ^{\shortmid }\mathbf{N,\ ^{\shortmid }g,\ ^{\shortmid }D)}$ and into energy-momentum tensors for matter fields embedded self-consistently in a spacetime aether characterized additionally by velocity, and/or momentum, coordinates $y^{a}$ and/or $p_{a}$.

\begin{principle}
\label{pgmp}\textbf{\ [generalized Mach principle]}:\ An effective phase space with MDRs on (co) tangent bundles (i.e. a locally anisotropic spacetime modelled as a generalized Finsler-Lagrange-Hamilton geometry) encodes contributions of nonlinear and nonholonomic distortion structures and various types of evolution effects, nonlinear interactions, kinetic and diffusion processes,
various nonlinear models with complex structure, nonlinear information, fractional derivative and pattern structure process etc. This consists in a less precise concept of generalized Finsler-Mach principles which can be stated by additional assumptions for explicit models with respective MDRs; certain canonical geometric data and distortions of $\left( \mathbf{N,g,D}\right) $ and/or $(\ ^{\shortmid }\mathbf{N,\ ^{\shortmid }g,\ ^{\shortmid}D)}$; and effective sources of matter.
\end{principle}

Such Finsler-Lagrange-Hamilton generalized Mach principles can be re-formulated, restricted and/or generalized in various forms encoding not only MDRs but also characterizing information processes, complex structures, evolutions of ecological and biophysical systems. The Principle \ref{pgmp} can be correlated to more richer geometric structures (commutative and noncommutative ones, algebroids and gerbes, supersymmetric generalizations) in classical and quantum theories, for models of geometric flows, off-diagonal solutions quasi-periodic and/or pattern forming structures etc. Various examples of such constructions determined by  exact and parametric solutions were provided and studied in Refs. \cite{vjmp06,vjmp09,vcsf12,vsym13,vch2416,gvvcqg15,bubuianucqg17,avjmp09,avjgp09,bvnd11}.

\subsubsection{Principle of general covariance and equivalent geometrization of MGTs}

In GR and various classes of MGTs, the principle of general covariance is a natural consequence of geometrization of physical and spacetime models constructed on (pseudo) Riemannian manifolds and metric-affine generalizations. Such theories can be constructed for different Lagrange densities for gravitational and matter fields. This reflects the idea that the geometric and physical constructions should not depend on the type of frames of reference (observers) and coordinate transforms (local parameterizations on some finite, or infinite, regions consisting an atlas covering a spacetime manifold/ bundle spaces).

In the definitions of Finsler-Lagrange-Hamilton geometries, the concept of manifold is also involved. For certain classes of MGTs extending the GR to models with MDRs, such manifolds are stated as Lorentz (co) tangent bundle spaces. So, the principles of general covariance has to be extended on some nonholonomic $\mathbf{V}$, $T\mathbf{V} $ and/or $T^{\ast}\mathbf{V}$. The N-connection structures for such spaces can be defined in coordinate-free forms encoding MDRs. We can introduce certain "preferred" N-adapted systems of reference and respective coordinate transforms when a fixed $h$--$v$--decomposition is prescribed. This does not prohibit us to state respective principles of general covariance for total phase spaces endowed with geometric structures
$\left(\mathbf{N,g,D}\right) $ and/or $(\ ^{\shortmid }\mathbf{N,\ ^{\shortmid }g,\ ^{\shortmid }D)}$. For instance, it is possible to define certain nonholonomic canonical variables modeling a pseudo Hamilton space,
$(H:\ ^{\shortmid }\widetilde{ \mathbf{g}},\ ^{\shortmid }\widetilde{\mathbf{N}},\ ^{\shortmid }\widetilde{\mathbf{D}}),$ but such a model can be re-defined by general frame/coordinate transforms and nonholonomic deformations into some general data
$(\ ^{\shortmid }\mathbf{N,\ ^{\shortmid }g,\ ^{\shortmid }D),}$ or certain special ones with equivalent almost symplectic
variables or allowing a general integration of certain physical important systems of PDEs.

Using not N-adapted frame transforms, all canonical and noncanonical N-adapted constructions can be transformed into certain general frame/coordinate ones, with "hidden" N-connection structures (encoding MDRs) on $V, TV$ and/or $T^{\ast }V.$ In such cases, we can work without boldface symbols and use generic off-diagonal metrics parameterized, respectively, in the forms (\ref{offd}).

The \textsf{principle of general covariance in GR} and MGTs with MDRs can be generalized to a \textsf{principle of equivalent geometrization of gravitational theories in terms of canonical geometric/physical objects} (\ref{phspgd}) and respective distortion relations (\ref{candistr}). Distortions of all geometric and physical objects can be computed following geometric methods (see, for instance, Theorem \ref{thcandist}). Such distortions of will modify, for instance, the definition of (effective) sources of matter fields and corresponding gravitational an matter filed equations (as we shall prove in next section). Nevertheless, a "dictionary and instructions" for an equivalent geometric formulation of different physical models can be always formulated if the distortions are uniquely determined by the N-connection and metric structures. For instance, performing general nonholonomic frame transforms and deformations, we can geometrize physical theories in general form in terms of geometric objects without "tilde", "hats" etc and work with N-adapted, or not N-adapted, (non) holonomic variables. Such theories can be related to GR and certain classes of MGTs induced on base manifolds via nonholonomic constraints on distortions, when $\mathbf{Z}=0$ and/or $\ ^{\shortmid }\mathbf{Z}=0,$
respectively, for $\mathbf{D}_{|\mathbf{Z}=0}\simeq \nabla $ and/or
$\ ^{\shortmid }\mathbf{D}_{|\ ^{\shortmid }\mathbf{Z}=0}\simeq \ ^{\shortmid}\nabla ,$ even $\mathbf{D}\neq \nabla $ and/or
$\ ^{\shortmid }\mathbf{D}\neq \ ^{\shortmid }\nabla $. For canonical nonholonomic distortions, we have to impose conditions of type (\ref{lccondl}) and/or (\ref{lccondh}).

\begin{principle}
\label{pgcov}\textbf{\ [general covariance and equivalent nonholonomic geometrizations of gravitational theories]}: On (co) tangent Lorentz bundles and nonholonomic Lorentz manifolds endowed with N-connection structures (encoding MDRs), it is possible to model geometrically MGTs following principles of general covariance on total and base space manifolds. Using (canonical) distortions such geometric models can be described equivalently for different classes of nonholonomic variables and (nonlinear) connection structures. Imposing additional nonholonomic constraints, we can define/extract Levi-Civita configurations with generic off-diagonal metrics.
\end{principle}

Above formulated Principles \ref{pgpeq}, \ref{pgmp}, and \ref{pgcov} may involve in a less explicit form various possible multi-connection structures of MGTs with MDRs. In a more general context, considering generic off-diagonal metrics and d-metrics encoding certain background configurations (for instance, like in bi-metric theories, with massive gravitons etc.), we can construct MGTs with multi-metric structures.

\subsubsection{Principles of analogy of N-adapted operators defined by d-metrics and d-connections}

Minimal extensions of GR to (co) tangent Lorentz bundles for models embedding MDRs can be performed by analogy with constructions in pseudo-Riemann geometry (in a more general case, we can use metric-affine and vierbein/spinor geometry) but applying nonholonomic geometric methods and performing respective N-adapted variational, differential, and integral calculi. In abstract forms, there are considered respective Lagrange densities
\begin{eqnarray}
\mathcal{L}(\mathbf{N}(u)\mathbf{,g}(u)\mathbf{,D}(u)\mathbf{;\ ^{A}}\phi
(u)) &=&\ ^{g}\mathcal{L}(\mathbf{N,g,D})+\ ^{m}\mathcal{L}(\mathbf{N,g,D;}\
^{A}\phi )\mbox{ on }T\mathbf{V;}  \label{generlag} \\
\ \mathbf{\ _{\shortmid }}\mathcal{L}(\ ^{\shortmid }\mathbf{N}(\
^{\shortmid }u)\mathbf{,\ ^{\shortmid }g}(\ ^{\shortmid }u)\mathbf{,\
^{\shortmid }D}(\ ^{\shortmid }u);\ _{\shortmid }^{A}\phi (\ ^{\shortmid
}u)) &=&\ _{\shortmid }^{g}\mathcal{L}(\ ^{\shortmid }\mathbf{N,\
^{\shortmid }g,\ ^{\shortmid }D})+\ \ _{\shortmid }^{m}\mathcal{L}(\
^{\shortmid }\mathbf{N,\ ^{\shortmid }g,\ ^{\shortmid }D;}\ _{\shortmid
}^{A}\phi )\mbox{ on }T^{\ast }\mathbf{V.}  \notag
\end{eqnarray}%
In these formulas, $\ ^{g}\mathcal{L}$ and $\ _{\shortmid }^{g}\mathcal{L}$ are corresponding (on total spaces of tangent and contangent bundles) Lagrange densities for modified gravitational fields without matter fields. One can be constructed certain effective matter sources and modeled matter field interactions via distortions of d-connections and/or by introduced bi-/ multi-metric structures. Such effective matter terms can be included as additional terms to standard matter Lagrange densities. The respective matter fields $\ ^{A}\phi $ and $\ {\shortmid }^{A}\phi $ are labeled by an abstract index $A$. Such fields
can be scalar, gauge, fermion/spinor fields etc. The explicit constructions of $\ ^{m}\mathcal{L}$ and/or
$\ _{\shortmid }^{m}\mathcal{L}$ depend on the type of theories and phenomenological models we have to elaborate and should be supporte by a set of experimental and/or observational data which are considered for verifying respective MGTs. In general, (effective) matter Lagrangians are functionals of certain background geometric data $\left( \mathbf{N,g,D}\right) $ and/or
$(\ ^{\shortmid }\mathbf{N,\ ^{\shortmid }g,\ ^{\shortmid}D})$ and encode certain fundamental constants and MDRs, prescribe (non) linear symmetries, assumptions on topological configurations etc. We shall analyze some examples of $\mathcal{L}$ and
$\mathbf{\ _{\shortmid }}\mathcal{L} $ in section \ref{sslagrd}.

Formulas (\ref{generlag}) can be written in canonical variables for respective Finsler/ - Lagrange / - Hamilton spaces, i.e. with "tilde" values like $\ ^{m}\widetilde{\mathcal{L}}(\widetilde{\mathbf{N}},\widetilde{\mathbf{g}},\widetilde{\mathbf{D}}; \ ^{A}\widetilde{\phi })$ and necessary type tildes on other Lagrange densities and geometric objects. Such variables allow and effective Lagrange-Hamilton interpretation of the total phase space/ spacetime model, which is convenient for elaborating quantum models. Using general frame and/or coordinate transforms (i.e. Principle \ref{pgcov}), "tilde" Lagrange densities can be transformed into "hat" Lagrange densities (in such variables, it is possible to decouple and integrate physically important PDEs in general forms). For additional nonholonomic zero distortion conditions, it is possible to extract LC-configurations and project on base Lorentz manifolds in order to find new classes of exact solutions in GR encoding certain information, for instance, about MDRs.

\begin{principle}
\label{panalogy}\textbf{\ [principle of analogy of  Lagrange densities]}: Gravitational and matter field Lagrange densities considered in GR and MGTs for metric-affine manifolds admit generalizations on (co) tangent Lorentz bundles using nonholonomic variables and deformations of fundamental geometric/ physical objects which can be transformed into canonical ones modeling pseudo Finsler/ - Lagrange / - Hamilton geometries.
\end{principle}

We note that Principles \ref{pgpeq}--\ref{panalogy} can be formulated in general form for any d-connection $\mathbf{D}$ and/ $\mathbf{\ ^{\shortmid }D}$ (it can be metric compatible or not) related in a unique form via distortion relations to any prescribed LC-connection, canonical d--connection, or almost symplectic connection structures. There are a number of conceptual and technical difficulties in elaborating physical theories with metric noncompatible connections. For instance, there are problems with definition of spinors and Dirac operators, conservation laws, geometric flows etc., see details in Refs. \cite{vplb10,vijgmmp12,vijmpd12} and references therein. Nevertheless, self-consistent physical models can be constructed for any well-defined geometrical data generating metric compatible configurations and then distorting in certain unique forms the fundamental physical objects and equations. For more general geometric models and MGTs, we can consider theories with multi-/ metric/connection/measure structures. Certain realistic geometric data (\ref{phspgd}), distortions (\ref{candistr}), and Lagrange densities (\ref{generlag}) should be fixed in such forms which can be verified by experimental/observational/ theoretical simulation data. Such theories for nonholonomic manifolds and tangent bundles (for instance, $\mathbf{T}^{\ast }\mathbf{V}$) were studied in details, see some series of geometric and mathematical physics works with applications in modern cosmology and
astrophysics \cite{v94,vog94,vmon98,vmon02,vmon06,vijgmmp08,vlqg09,vjmp09,avjgp09,bvcejp11,gvvepjc14, gheorghiuap16,ruchinepjc17,svvijmpd14,vijgmmp10a,vplb10,vijgmmp12,vijmpd12,vport13,vijgmmp14, vepjc14,vepjc14a,vacaruplb16}. The constructions be can elaborated (generalized, re-defined) on "dual" spaces like $\mathbf{T}^{\ast }\mathbf{V}$ following certain $\mathcal{L}$-dual, almost symplectic and other type principles.

\subsection{Lagrange densities and energy-momentum tensors on (co) tangent Lorentz bundles}

\label{sslagrd}In this subsection, we state the conventions on Lagrange densities and derive respective energy-momentum tensors which will be used for formulating three classes of MGTs with MDRs. We shall consider
arbitrary metric compatible d-connections $\mathbf{D}$ and/or $\ ^{\shortmid}\mathbf{D}$.

\subsubsection{Scalar fields on (co) tangent Lorentz bundles}

Let us consider examples of Lagrange densities for matter fields with locally anisotropic interactions:

\begin{convention}
\label{convscfields}\textsf{[scalar fields with MDRs on (co) tangent bundles] } \newline
Scalar field interactions on phase spaces with MDRs can be modeled by Lagrange densities
\begin{equation}
\ ^{m}\mathcal{L}=\ ^{\phi }\mathcal{L}(\mathbf{g;}\phi )\mbox{ on }T\mathbf{%
V}\mbox{ and/or }\ _{\shortmid }^{m}\mathcal{L}=\ _{\shortmid }^{\phi }%
\mathcal{L}(\ ^{\shortmid }\mathbf{g};\ _{\shortmid }\phi )\mbox{ on }%
T^{\ast }\mathbf{V}  \label{lagscf}
\end{equation}%
depending only on d-metrics, $\mathbf{g}_{\mu \nu }$ and/or $\ ^{\shortmid }%
\mathbf{g}^{\alpha \beta }$ (this allows to construct exact off-diagonal
solutions in explicit form, see discussion and references for Direction 10, appendix \ref{sssdir10}), and on scalar fields $\ ^{A}\phi =\phi (u)$ and $\
_{\shortmid }^{A}\phi =\ _{\shortmid }\phi (\ ^{\shortmid }u)$.
\end{convention}

Performing a N-adapted variational calculus, we prove

\begin{consequence}
\textsf{[energy-momentum d-tensors for locally anisotropic interacting scalar fields ] } \newline
The symmetric energy-momentum d-tensors for scalar fields on (co) tangent bundles derived for respective Lagrange densities (\ref{lagscf}) are computed for possible h- and v-, or cv-splitting
\begin{eqnarray}
\ ^{\phi }\mathbf{T}_{\alpha \beta } &=&-\frac{2}{\sqrt{|\mathbf{g}_{\mu \nu
}|}}\frac{\delta (\sqrt{|\mathbf{g}_{\mu \nu }|}\ \ ^{\phi }\mathcal{L})}{%
\delta \mathbf{g}^{\alpha \beta }}=\ ^{\phi }\mathcal{L}\mathbf{g}_{\alpha
\beta }+2\frac{\delta (\ ^{\phi }\mathcal{L})}{\delta \mathbf{g}^{\alpha
\beta }}  \label{emscdt} \\
&=&\{\ ^{\phi }\mathbf{T}_{ij}=-\frac{2}{\sqrt{|\mathbf{g}_{\mu \nu }|}}%
\frac{\delta (\sqrt{|\mathbf{g}_{\mu \nu }|}\ \ ^{\phi }\mathcal{L})}{\delta
\mathbf{g}^{ij}}=\ldots ,\ ^{\phi }\mathbf{T}_{ab}=-\frac{2}{\sqrt{|\mathbf{g%
}_{\mu \nu }|}}\frac{\delta (\sqrt{|\mathbf{g}_{\mu \nu }|}\ \ ^{\phi }%
\mathcal{L})}{\delta \mathbf{g}^{ab}}=\ldots \};  \notag
\end{eqnarray}
\begin{eqnarray*}
\ _{\shortmid }^{\phi }\mathbf{T}_{\alpha \beta } &=&-\frac{2}{\sqrt{|\
^{\shortmid }\mathbf{g}_{\mu \nu }|}}\frac{\delta (\sqrt{|\ ^{\shortmid }%
\mathbf{g}_{\mu \nu }|}\ \ \ _{\shortmid }^{\phi }\mathcal{L})}{\delta \
^{\shortmid }\mathbf{g}^{\alpha \beta }}=\ \ _{\shortmid }^{\phi }\mathcal{L}%
\ ^{\shortmid }\mathbf{g}_{\alpha \beta }+2\frac{\delta (\ \ \ _{\shortmid
}^{\phi }\mathcal{L})}{\delta \ ^{\shortmid }\mathbf{g}^{\alpha \beta }} \\
&=&\{\ _{\shortmid }^{\phi }\mathbf{T}_{ij}=-\frac{2}{\sqrt{|\ ^{\shortmid }%
\mathbf{g}_{\mu \nu }|}}\frac{\delta (\sqrt{|\ ^{\shortmid }\mathbf{g}_{\mu
\nu }|}\ \ \ _{\shortmid }^{\phi }\mathcal{L})}{\delta \ ^{\shortmid }%
\mathbf{g}^{ij}}=\ldots ,\ _{\shortmid }^{\phi }\mathbf{T}^{ab}=-\frac{2}{%
\sqrt{|\ ^{\shortmid }\mathbf{g}_{\mu \nu }|}}\frac{\delta (\sqrt{|\
^{\shortmid }\mathbf{g}_{\mu \nu }|}\ \ \ _{\shortmid }^{\phi }\mathcal{L})}{%
\delta \ ^{\shortmid }\mathbf{g}_{ab}}=\ldots \}.
\end{eqnarray*}
\end{consequence}

Scalar field equations can be re-defined by an additional 3+1 splitting on
base and fiber spaces as certain moving equations for ideal fluids. For
instance, we can consider a velocity $\mathbf{v}_{\alpha }$ d-vector subjected
to the conditions $\mathbf{v}_{\alpha }\mathbf{v}^{\alpha }=1$ and $\mathbf{v%
}^{\alpha }\mathbf{D}_{\beta }\mathbf{v}_{\alpha }=0,$ for $\ \ ^{\phi }%
\mathcal{L}:=-p$ in a corresponding local N--adapted frame.
\begin{remark}
\textsf{[energy-momentum d-tensors for locally anisotropic perfect liquids] }
\newline
Conventionally, the sources (\ref{emscdt}) can be approximated as a perfect
liquid matter (we use a left abstract index "l" for liquid) with respective density and
pressure
\begin{equation}
\ ^{l}\mathbf{T}_{\alpha \beta }=(\rho +p)\mathbf{v}_{\alpha }\mathbf{v}%
_{\beta }-p\mathbf{g}_{\alpha \beta }\mbox{ and/or }\ _{\shortmid }^{l}%
\mathbf{T}_{\alpha \beta }=(\ ^{\shortmid }\rho +\ ^{\shortmid }p)\
^{\shortmid }\mathbf{v}_{\alpha }\ ^{\shortmid }\mathbf{v}_{\beta }-\
^{\shortmid }p\ ^{\shortmid }\mathbf{g}_{\alpha \beta }.  \label{lemd}
\end{equation}
\end{remark}

It should be noted that formulas (\ref{emscdt}) do not depend in explicit form on d--connections. Nevertheless, such dependencies should be considered if $\phi (u)$ and/or $\phi (\ ^{\shortmid }u)$ are subjected to the conditions to be solutions of some generalized scalar field equations (Klein-Gordon, or other types, for instance, certain hydrodynamical equations). Using
respective nonholonomic variables, we can re-write these formulas in terms of "tilde/hat" variables.

\subsubsection{Lagrange densities for Einstein-Yang-Mills-Higgs systems with MDRs}

Various models of locally anisotropic gravitational and gauge field interactions on modified Finsler spaces, nonholonomic Lorentz manifolds, and higher order Lagrange-Hamilton spaces were studied in a series of our works \cite{v83,v87,asanov85,bejancu90,v94,
vog94,vgon95,vmon98,vmon02,vmon06,vrev08,vijgmmp08, vepl11,vvyijgmmp14a,vacaruepjc17,bubuianucqg17}. That research was on twistor and spinor methods for generalized geometries, with nearly autoparallel maps and generalized connections, for formulating models of Finsler gravity as gauge like theories, for supersymmetric generalizations and, more recently, superstring and supergravity theories. Those directions were related also to the theory of spinors in Finsler-Lagrange-Hamilton spaces and modelling such
configurations, for instance, in brane gravity, with Clifford-algebroids, noncommutative Finsler geometry and locally anisotropic Dirac operators \cite{vjmp96,vjmp09,vpcqg01,vtnpb02}. We shall present a detailed study of Einstein-Dirac systems with MDRs resulting in Clifford-Lagrange-Hamilton geometries, and constructed exact solutions and quantum models, in other partner works (see also Directions 3 and 7, respectively, in appendices \ref{sssdir403} and \ref{sssdir07}). In this subsection, we define Lagrange densities and corresponding energy-momentum tensors for certain models of Yang-Mills-Higgs, YMH, interactions on (co) tangent Lorentz bundles.

Let us consider YMH systems modeled on nonholonomic Lorentz manifolds, or on their (co) tangent bundles, by respective geometric data $\left( \mathbf{N,g,D;A}^{\check{a}}(u),\phi ^{\check{a}}(u)\right) $ and/or $(\ ^{\shortmid}\mathbf{N},
\ ^{\shortmid }\mathbf{g},\ ^{\shortmid }\mathbf{D};\ ^{\shortmid} \mathbf{A}^{\check{a}}(u),
\ ^{\shortmid }\phi ^{\check{a}}(\ ^{\shortmid }u))$ for matter fields parameterized in the form $\ ^{A}\phi =[\mathbf{A}^{\check{a}}(u),\phi ^{\check{a}}(u)]$ and/or $\ _{\shortmid }^{A}\phi =[\_{\shortmid}A ^{\check{a}}(\ ^{\shortmid }u),\ _{\shortmid }\phi ^{\check{a}}(\ ^{\shortmid }u)].$ In these formulas, the non-Abelian gauge fields
$\mathbf{A}=\{\mathbf{A}^{\check{a}}(u)\}$ can be defined by a 1-form $\mathbf{A}=\mathbf{A}_{\delta }du^{\delta }$ when coefficients $\mathbf{A}_{\delta }^{\check{a}}$ take values in a Lie algebra $\mathcal{A}$ of so-called internal gauge symmetries. The index $\check{a}$ labels elements of such Lie aglebra and related groups. For $\mathcal{A}=su(2)$ and base Minkowski spaces, we obtain  standard Yang-Mills fields. The scalar-multiplet (complex) field $\phi ^{\check{a}}$ is called the Higgs field. The model is elaborated on (co) tangent bundles with respective 1-forms
\begin{equation*}
\mathbf{A}=\mathbf{A}_{\delta }du^{\delta },\mbox{ for }\mathbf{A}_{\delta }=%
\mathbf{A}_{\delta }^{\check{a}}\tau _{\check{a}},\ \mbox{ and }\
^{\shortmid }\mathbf{A}=\ ^{\shortmid }\mathbf{A}_{\delta }du^{\delta },%
\mbox{ for }\ \ ^{\shortmid }\mathbf{A}_{\delta }=\ ^{\shortmid }\mathbf{A}%
_{\delta }^{\check{a}}\tau _{\check{a}},
\end{equation*}%
where $\tau _{\check{a}}$ are generators of $\mathcal{A}$. We consider summation on repeating indices $\check{a}$ in all cases "up-low", "up-up", and "low-low". In this work, we consider gauge theories with the same gauge group both for tangent and cotangent bundles, extended to respective vector bundles with typical fibers and respective adjoint representations of Lie algebra of $\mathcal{A}$.

\begin{convention}
\label{convcovder}\textsf{[covariant derivatives of gauge fields in (co) tangent bundles ] } \newline
The non--Abelian gauge fields (on (co)tangent bundles enabled with respective metric compatible d--connections $\mathbf{D}$ and
$\ ^{\shortmid }\mathbf{D})$ define covariant derivatives in the associated vector bundles,%
\begin{equation*}
\mathcal{D}_{\delta }=\mathbf{D}_{\delta }+i\check{e}[\mathbf{A}_{\delta
},\cdot ]\mbox{ and/or  }\mathbf{\ ^{\shortmid }}\mathcal{D}_{\delta }=%
\mathbf{\ ^{\shortmid }D}_{\delta }+i\check{e}[\mathbf{\ ^{\shortmid }A}%
_{\delta },\cdot ],
\end{equation*}%
where $\check{e}$ is the coupling constant, $i^{2}=-1,$ and $[\cdot ,\cdot ]$ is used for the commutator.
\end{convention}

We note that such values should be written with "hat" or "tilde" symbols if the phase space d-connections are of respective type, for instance, $\widehat{\mathcal{D}}_{\delta }=\widehat{\mathbf{D}}_{\delta }+i\check{e}[\widehat{\mathbf{A}}_{\delta },\cdot ]$.  All formulas derived with "hat/ tilde" operators are labeled by "hat/ tilde" symbols.

By a standard differential form calculus, we prove
\begin{definition}
\textbf{-Theorem\ } \textsf{[strengths of gauge fields in curved phase spaces ] } \newline
The d-vector fields $\mathbf{A}_{\delta }$ and $\mathbf{^{\shortmid }A}_{\mu}$ are characterised by respective curvatures
\begin{equation}
\mathcal{F}_{\beta \mu }:=\mathbf{D}_{\beta }\mathbf{A}_{\mu }-\mathbf{D}%
_{\mu }\mathbf{A}_{\beta }+i\check{e}[\mathbf{A}_{\beta },\mathbf{A}_{\mu }];%
\mathbf{\ ^{\shortmid }}\mathcal{F}_{\beta \mu }:=\mathbf{\ ^{\shortmid }D}%
_{\beta }\mathbf{\ ^{\shortmid }A}_{\mu }-\mathbf{\ ^{\shortmid }D}_{\mu }%
\mathbf{A}_{\beta }+i\check{e}[\mathbf{\ ^{\shortmid }A}_{\beta },\mathbf{\
^{\shortmid }A}_{\mu }].  \label{gaugestr}
\end{equation}
\end{definition}

Following Principle \ref{panalogy}, we formulate:
\begin{convention}
\label{convldymh}\textsf{[Lagrange densities for YMH interactions on (co) tangent Lorentz bundles] } \newline
Non-Abelian gauge - Higgs, YMH, interactions on phase spaces with MDRs can be modeled by Lagrange densities
\begin{eqnarray}
\ ^{A}\mathcal{L} &=&-\frac{1}{4}\mathcal{F}^{\check{a}\mu \nu }\mathcal{F}%
_{\mu \nu }^{\check{a}}\mbox{ and }\ ^{H}\mathcal{L}=-\frac{1}{2}(\mathcal{D}%
_{\mu }\phi ^{\check{a}})(\mathcal{D}^{\mu }\phi ^{\check{a}})-\mathcal{V}%
(\phi ^{\check{a}});  \label{ymhld} \\
\ _{\shortmid }^{A}\mathcal{L} &=&-\frac{1}{4}\ ^{\shortmid }\mathcal{F}^{%
\check{a}\mu \nu }\mathbf{\ ^{\shortmid }}\mathcal{F}_{\mu \nu }^{\check{a}}%
\mbox{ and }\ _{\shortmid }^{H}\mathcal{L}=-\frac{1}{2}(\mathbf{\
^{\shortmid }}\mathcal{D}_{\mu }\mathbf{\ ^{\shortmid }}\phi ^{\check{a}})(%
\mathbf{\ ^{\shortmid }}\mathcal{D}^{\mu }\mathbf{\ ^{\shortmid }}\phi ^{%
\check{a}})-\mathbf{\ ^{\shortmid }}\mathcal{V}(\mathbf{\ ^{\shortmid }}\phi
^{\check{a}}),  \notag
\end{eqnarray}%
were $\mathcal{V}(\phi ^{\check{a}})=\frac{1}{4}\check{\lambda}(|\phi
_{\lbrack 0]}^{\check{a}}|^{2}-\overline{\phi }^{\check{a}}\phi ^{\check{a}%
})^{2}$ and $\mathbf{\ ^{\shortmid }}\mathcal{V}(\mathbf{\ ^{\shortmid }}%
\phi ^{\check{a}})=\frac{1}{4}\check{\lambda}(\overline{\phi }_{[0]}^{\check{%
a}}\phi _{\lbrack 0]}^{\check{a}}-\mathbf{\ ^{\shortmid }}\overline{\phi }^{%
\check{a}}\mathbf{\ ^{\shortmid }}\phi ^{\check{a}})^{2}$ are respective \
potentials for nonlinear interactions of Higgs fields with corresponding
vacuum expectations $|\phi _{\lbrack 0]}^{\check{a}}|^{2}=\overline{\phi }%
_{[0]}^{\check{a}}\phi _{\lbrack 0]}^{\check{a}}$ and $|\mathbf{\
^{\shortmid }}\phi _{\lbrack 0]}^{\check{a}}|^{2}=\mathbf{\ ^{\shortmid }}%
\overline{\phi }_{[0]}^{\check{a}}\mathbf{\ ^{\shortmid }}\phi _{\lbrack
0]}^{\check{a}}$ and self-integration constant $\check{\lambda}.$
\end{convention}

Above considered constants can be related to physical ones, for instance, in the Weinberg-Salam theory when vacuum expectation of the Higgs field which determines a masse$\ ^{H}M=\sqrt{\check{\lambda}}\eta ;$ the mass of W-bosons is $\ ^{W}M=ev.$

Performing a N-adapted variational calculus, we prove
\begin{consequence}
\textsf{[ energy-momentum d-tensors for YMH fields with locally anisotropic interactions ] } \newline
The symmetric energy-momentum d-tensors for scalar fieds on (co) tangent bundles derived for respective Lagrange densities (\ref{ymhld}) are computed for locally anisotropic gauge fields, {\small
\begin{equation}
\ ^{A}\mathbf{T}_{\beta \delta }=2(\mathbf{g}^{\mu \nu }\mathcal{F}_{\beta
\mu }^{\check{a}}\mathcal{F}_{\nu \delta }^{\check{a}}-\frac{1}{4}\mathbf{g}%
_{\beta \delta }\mathcal{F}^{\check{a}\mu \nu }\mathcal{F}_{\mu \nu }^{%
\check{a}})\mbox{ and }\ _{\shortmid }^{A}\mathbf{T}_{\beta \delta }=2(%
\mathbf{\ ^{\shortmid }g}^{\mu \nu }\mathbf{\ ^{\shortmid }}\mathcal{F}%
_{\beta \mu }^{\check{a}}\mathbf{\ ^{\shortmid }}\mathcal{F}_{\nu \delta }^{%
\check{a}}-\frac{1}{4}\mathbf{\ ^{\shortmid }g}_{\beta \delta }\mathbf{\
^{\shortmid }}\mathcal{F}^{\check{a}\mu \nu }\mathbf{\ ^{\shortmid }}%
\mathcal{F}_{\mu \nu }^{\check{a}}),\   \label{emdtym}
\end{equation}%
and for Higgs fields,}%
\begin{eqnarray}
\ ^{H}\mathbf{T}_{\beta \delta } &=&\frac{1}{4}\ (\mathcal{D}_{\delta }%
\overline{\phi }^{\check{a}}\ \mathcal{D}_{\beta }\phi ^{\check{a}}+\mathcal{%
D}_{\beta }\overline{\phi }^{\check{a}}\ \mathcal{D}_{\delta }\phi ^{\check{a%
}})-\frac{1}{4}\mathbf{g}_{\beta \delta }\mathcal{D}_{\alpha }\overline{\phi
}^{\check{a}}\ \mathcal{D}^{\alpha }\phi ^{\check{a}}-\mathbf{g}_{\beta
\delta }\mathcal{V}(\phi ^{\check{a}});  \label{emdthiggs} \\
\ _{\shortmid }^{H}\mathbf{T}_{\beta \delta } &=&\frac{1}{4}\ (\mathbf{\
^{\shortmid }}\mathcal{D}_{\delta }\mathbf{\ ^{\shortmid }}\overline{\phi }^{%
\check{a}}\ \mathbf{\ ^{\shortmid }}\mathcal{D}_{\beta }\mathbf{\
^{\shortmid }}\phi ^{\check{a}}+\mathbf{\ ^{\shortmid }}\mathcal{D}_{\beta }%
\mathbf{\ ^{\shortmid }}\overline{\phi }^{\check{a}}\ \mathbf{\ ^{\shortmid }%
}\mathcal{D}_{\delta }\mathbf{\ ^{\shortmid }}\phi ^{\check{a}})-\frac{1}{4}%
\mathbf{\ ^{\shortmid }g}_{\beta \delta }\mathbf{\ ^{\shortmid }}\mathcal{D}%
_{\alpha }\mathbf{\ ^{\shortmid }}\overline{\phi }^{\check{a}}\ \mathbf{\
^{\shortmid }}\mathcal{D}^{\alpha }\mathbf{\ ^{\shortmid }}\phi ^{\check{a}}-%
\mathbf{\ ^{\shortmid }g}_{\beta \delta }\mathbf{\ ^{\shortmid }}\mathcal{V}(%
\mathbf{\ ^{\shortmid }}\phi ^{\check{a}}).  \notag
\end{eqnarray}
\end{consequence}

Similar formulas can be derived for models of quantum chromodynamics,\ QCD, with MDRs generalized on (co) tangent Lorentz bundles. In this work, we omit geometric constructions with the gauge Lie algebra $\mathcal{A}=su(3)$ resulting in locally anisotropic interactions of gluonic and quark fields on (co)tangent Lorentz bundles.

\subsubsection{Actions for minimal MDR--extensions of GR and YMH theories}

In GR, the Lagrange density for gravitational fields is postulated in the
form%
\begin{equation}
\ ^{g}\mathcal{L}(g_{ij}, \nabla )=\frac{^{Pl}M^{2}}{2}R,  \label{einstlg}
\end{equation}%
where $R$ is the Ricci scalar of a Lorentz manifold $V$ and the Planck mass $%
^{Pl}M$ is determined by the Newton constant $^{New}G$. In the units $%
^{New}G=1/16\pi $ with $^{Pl}M=(8\pi ^{New}G)^{-1/2}=\sqrt{2}$, one states a
constant $\varkappa $ for the matter source $\Upsilon _{ij}:=\varkappa
(T_{ij}-\frac{1}{2}g_{ij}T)$, where $T:=g^{ij}T_{ij},$ for $\ T_{kl}:=-\frac{%
2}{\sqrt{|\mathbf{g}_{ij}|}}\frac{\delta (\sqrt{|\mathbf{g}_{ij}|}\ \ ^{m}%
\mathcal{L})}{\delta \mathbf{g}^{kl}}$ with Lagrange density of matter
fields $\ ^{m}\mathcal{L}$. The Einstein equations with the Ricci tensor for
$\nabla ,$
\begin{equation}
R_{ij}=\Upsilon _{ij},  \label{einsteq}
\end{equation}%
can be derived by a variational calculus on $V$ using action $\mathcal{S}=\
^{g}\mathcal{S}+ \ ^{m}\mathcal{S}=\frac{1}{16\pi }\int d^{4}x\sqrt{|g_{ij}|}%
(\ ^{g}\mathcal{L+}\ ^{m}\mathcal{L})$.

Let us consider on $T\mathbf{V}$ a metric compatible d--connection $\mathbf{D%
}$ completely defined by a d-metric $\mathbf{g}$ and with the property that
for respective nonholonomic constraints $\mathbf{D}_{\mid \mathcal{T}%
=0}=\nabla $ (such constraints may be non-integrable and or may not have a
smooth limit). Following Principle \ref{panalogy}, we consider that minimal
extensions of GR to tangent Lorentz bundles enabled with data $(\mathbf{%
N,g,D)}$ should be determined by nonholonomically deformed Einstein
equations of type
\begin{equation*}
\mathbf{R}_{\alpha \beta }[\mathbf{D}]=\Upsilon _{\ \beta \gamma}:=\varkappa
(\mathbf{T}_{\alpha \beta }-\frac{1}{2}\mathbf{g}_{\alpha\beta}T),
\end{equation*}%
where $T:=\mathbf{g}^{\alpha \beta }\mathbf{T}_{\alpha \beta }$ and the
constant $\varkappa $ may take different values on h- and v-subspaces (this
should be defined by experimental/ observational data). Similar geometric
formulas can be written on $T^{\ast }\mathbf{V,}$ for instance, as $\mathbf{%
\ ^{\shortmid }R}_{\alpha \beta }[\mathbf{\ ^{\shortmid }D}]=\ ^{\shortmid
}\Upsilon _{\ \beta \gamma }$ (when splitting into $h$- and $cv$ -components
can be considered with, or not, $L$-duality and with additional requests on
possible almost symplectic symmetries, see below necessary details).

Summarizing Lagrange densities introduced in previous subsections, we
formulate:

\begin{principle}
-\textbf{Convention} \label{pactminmodact}\textbf{\ [actions for minimally
MDR-modified EYMH systems]}:\newline
Locally anisotropic interactions of Einstein gauge and Higgs fields on
(co)tangent Lorentz bundles (endowed with metric compatible d-connections
uniquely determined by respective d-metric structures and admitting
nonholonomic variables for distinguishing Finsler-Lagrange-Hamilton phase
spaces) are described by respective actions%
\begin{eqnarray*}
\mathcal{S} &=&\ ^{\mathbf{g}}\mathcal{S+}\ ^{\phi }\mathcal{S+}\ \ ^{A}%
\mathcal{S+}\ ^{H}\mathcal{S=}\frac{1}{16\pi }\int \delta u\sqrt{|\mathbf{g}%
_{\alpha \beta }|}(\ ^{\mathbf{g}}\mathcal{L}+\ ^{\phi }\mathcal{L}+ \ ^{A}%
\mathcal{L}+\ ^{H}\mathcal{L}); \\
\mathbf{\ }\ _{\shortmid }\mathcal{S} &=&\ _{\shortmid }^{\mathbf{g}}%
\mathcal{S}+\ _{\shortmid }^{\phi }\mathcal{S}+\ \ _{\shortmid }^{A}\mathcal{%
S}+\ _{\shortmid }^{H}\mathcal{S=}\frac{1}{16\pi }\int \delta \mathbf{\
^{\shortmid }}u\sqrt{|\mathbf{\ ^{\shortmid }g}_{\alpha \beta }|}(\
_{\shortmid }^{\mathbf{g}}\mathcal{L+}\ _{\shortmid }^{\phi }\mathcal{L}+\ \
_{\shortmid }^{A}\mathcal{L+}\ _{\shortmid }^{H}\mathcal{L}),
\end{eqnarray*}%
where $\ ^{\mathbf{g}}\mathcal{L}:=\ ^{s}\mathbf{R}=\mathbf{g}^{\alpha \beta
}\mathbf{R}_{\alpha \beta }[\mathbf{D}],\ _{\ \shortmid }^{\mathbf{g}}%
\mathcal{L}:=\ _{\shortmid }^{s}\mathbf{R}=\mathbf{\ ^{\shortmid }g}^{\alpha
\beta }\mathbf{\ ^{\shortmid }R}_{\alpha \beta }[\mathbf{\ ^{\shortmid }D}]$
and the Lagrange densities for scalar-YMH fields are given by corresponding
formulas (\ref{lagscf}) and (\ref{ymhld}).
\end{principle}

Performing a N-adapted variational calculus and summarizing previous
Consequences on energy-momentum d-tensors, we prove

\begin{consequence}
\label{conseymhs}\textsf{[sources for locally anisotropic YMH fields ] }
\newline
On (co)tangent Lorentz bundles, minimal modifications of scalar and YMH
systems encoding MDRs are characterized by respective sources
\begin{eqnarray*}
\ ^{\phi }\Upsilon _{\ \beta \gamma }:= &&\varkappa (\ ^{\phi }\mathbf{T}%
_{\alpha \beta }-\frac{1}{2}\mathbf{g}_{\alpha \beta }\ ^{\phi }T),\
^{A}\Upsilon _{\ \beta \gamma }:=\varkappa (\ ^{A}\mathbf{T}_{\alpha \beta }-%
\frac{1}{2}\mathbf{g}_{\alpha \beta }\ ^{A}T),\ ^{H}\Upsilon _{\ \beta
\gamma }:=\varkappa (\ ^{H}\mathbf{T}_{\alpha \beta }-\frac{1}{2}\mathbf{g}%
_{\alpha \beta }\ ^{H}T); \\
\ _{\shortmid }^{\phi }\Upsilon _{\ \beta \gamma }:= &&\varkappa (\
_{\shortmid }^{\phi }\mathbf{T}_{\alpha \beta }-\frac{1}{2}\mathbf{\
^{\shortmid }g}_{\alpha \beta }\ _{\shortmid }^{\phi }T),\ _{\shortmid
}^{A}\Upsilon _{\ \beta \gamma }:=\varkappa (\ _{\shortmid }^{A}\mathbf{T}%
_{\alpha \beta }-\frac{1}{2}\mathbf{\ ^{\shortmid }g}_{\alpha \beta }\
_{\shortmid }^{A}T),\ _{\shortmid }^{H}\Upsilon _{\ \beta \gamma
}:=\varkappa (\ _{\shortmid }^{H}\mathbf{T}_{\alpha \beta }-\frac{1}{2}%
\mathbf{\ ^{\shortmid }g}_{\alpha \beta }\ _{\shortmid }^{H}T),
\end{eqnarray*}%
where the energy-momentum tensors are given respectively by formulas (\ref%
{emscdt}), (\ref{emdtym}) and (\ref{emdthiggs}).
\end{consequence}

In GR, the sources of matter fields are approximated (for instance, for
quantum fluctuations or certain summarized contributions) to (effective)
cosmological constants. Following Principle \ref{panalogy}, we consider

\begin{assumption}
\textsf{[effective cosmological constants for locally anisotropic YMH
sources ] } \label{assumpt4} \newline
Generalized sources for matter fields (including possible effective sources
defined by distortion tensors of d-connections) can be approximated by
respective cosmological constants $\Lambda $ and/or $\ _{\shortmid }\Lambda $%
,
\begin{equation*}
\ \Upsilon _{\ \gamma }^{\beta }=\Lambda \delta _{\ \gamma }^{\beta }%
\mbox{and/or }\ _{\shortmid }\Upsilon _{\ \gamma }^{\beta }=\mathbf{\ }\
_{\shortmid }\Lambda \delta _{\ \gamma }^{\beta },
\end{equation*}%
when, correspondingly, $\ \Upsilon _{\ \beta \gamma }=\ ^{\phi }\Upsilon _{\
\beta \gamma }+\ ^{A}\Upsilon _{\ \beta \gamma }+\ ^{H}\Upsilon _{\ \beta
\gamma }$ and/or $\ _{\shortmid }\Upsilon _{\ \beta \gamma }=\ _{\shortmid
}^{\phi }\Upsilon _{\ \beta \gamma }+\ _{\shortmid }^{A}\Upsilon _{\ \beta
\gamma }+\ _{\shortmid }^{H}\Upsilon _{\ \beta \gamma },$ considering that
such sources are subjected to relations of additivity,
\begin{equation}
\ \Lambda =\ ^{\phi }\Lambda +\ ^{A}\Lambda +\ ^{H}\Lambda \mbox{ and/or }\
_{\shortmid }\Lambda =\ _{\shortmid }^{\phi }\Lambda +\ _{\shortmid
}^{A}\Lambda +\ _{\shortmid }^{H}\Lambda .  \label{additcconst}
\end{equation}
\end{assumption}

We note that all such cosmological constants can be zero, positive, or
negative. They may compensate each other and result in (fictive) vacuum
configurations, or take different values in certain $h$-, $v$-, and/or $cv$%
-subspaces.

\begin{remark}
\textsf{[effective cosmological constants as nonlinear superpositions of associated cosmological constants for locally anisotropic YMH fields ]} In certain MGTs with quasi-classical corrections, renormalizations, nonlinear MDRs, nonlinear symmetries of generating functions and/or (effective) sources, nonlinear vacuum polarizations etc., one can be considered
nonlinear functionals of type $\Lambda =\Lambda \lbrack \ ^{\phi }\Lambda ,\ ^{A}\Lambda ,\ ^{H}\Lambda ]$ and/or
$_{\shortmid }\Lambda =\ _{\shortmid}\Lambda \lbrack \ _{\shortmid }^{\phi }\Lambda , \
_{\shortmid}^{A}\Lambda ,\ _{\shortmid }^{H}\Lambda ].$ We do not study such theories in this work.
\end{remark}

\subsubsection{Actions for MGTs with massive gravitons, bi-metric structures and MDRs}

\label{ssbms} A large class of MGTs are constructed with modifications of Lagrange density (\ref{einstlg}) to a general functional depending on Ricci scalar, torsion, energy-momentum and other fundamental geometric/ physical values. For instance, there are considered modifications of type $\ ^{g}\mathcal{L}(g_{ij},\nabla )\approx R$ $\rightarrow $ $f(R),$ where $R$ is a scalar curvature determined by the LC-connection or a generalized metric-affine connection. During last 20 years, it was elaborated a
paradigm of $f$-modified gravity theories attempting to explain the universe acceleration and solve the dark energy and dark matter problems. In a more general context, this paradigm involves theories with infra--red (IR) modifications of the GR theory and ultra-violet (UV) corrections expected to be of quantum origin. Various studies were on understanding possible physical implications of the massive spin--2 theory, MDRs and LIVs, generalized Finsler gravity theories, commutative and noncommutative /
nonassociative and/or supersymmetric generalizations etc. A series of recent our works
\cite{vjpcs13,vport13,vijgmmp14,vvyijgmmp14a,vepjc14,vepjc14a,gvvepjc14,gvvcqg15,vacaruplb16} is devoted to elaborating nonholonomic $f$-modified theories and study of cosmological implications of massive gravity and bi--metric gravity with
local anisotropy. Such constructions were inspired by papers on holonomic models of nonlinear massive gravitational theories including $f(R)$ modifications \cite{cai14,kluson13,nojiri13jcap}. Here we emphasize that locally (an)isotropic massive gravity theories contain the benefits of the so-called dRGT model and generalizations, see \cite{derham10,derham11,volkov12,strauss12}, being free of ghost modes \cite{boulware72,hassan12}. Advantages are those that by tuning the $f(R)$ and various MDR-deformed functionals (for holonomic configurations, see reviews \cite{nojiri07,capozziello10}), we can relate locally anisotropic black hole
solutions, stabilize cosmological backgrounds, elaborate various types of locally anisotropic cosmological evolution scenarios and provide for MGTs an unified description of inflation and late--time acceleration etc. In this paper, we geometrize specific models of massive $f(R)$ gravity with MDRs and (see subsections below) the corresponding systems of modified gravitational equations.

We extend the approaches elaborated in \cite{vijgmmp14,vvyijgmmp14a,vepjc14,vepjc14a,gvvepjc14,gvvcqg15,vacaruplb16} and
\cite{cai14,kluson13,nojiri13jcap} and consider theories on nonholonomic (co) tangent Lorentz bundles enabled with a common
N--connection structure $\mathbf{N}$ for two d--metrics where $\mathbf{g}=\{ \mathbf{g}_{\alpha \beta }\}$ is a dynamical d--metric and $\mathbf{q}=\{ \mathbf{q}_{\alpha \beta }\}$ is a non--dynamical reference metric. For Lagrange theories, we work with a metric compatible d-connection $\mathbf{D}$ (instead of the LC--connection $\nabla )$ and a corresponding Ricci scalar $\ {s}R$ both computed for $\mathbf{g.}$ On a base Lorentz manifold, the nonzero graviton mass is denoted by $\mu ,$ the constant $M_{P} $ is the
Planck mass and such constants are lifted with re-definition on fiber space (on total phase spaces, such values have to be defined experimentally).

For elaborating Hamilton massive and bi-metric gravity theories, there are considered corresponding values: a N--connection structure $\ _{\shortmid }\mathbf{N}$ for $\ _{\shortmid }\mathbf{g=\{\ _{\shortmid }g}_{\alpha \beta }\mathbf{\}}$ being the dynamical d--metric and $\ _{\shortmid }\mathbf{q}=\{\ _{\shortmid } \mathbf{q}_{\alpha \beta }\}$ being the
so--called non--dynamical reference (fiducial) d-metric, when respective $\mathbf{\ _{\shortmid }D}$ and $_{s}^{\shortmid }R$ are computed for $\ _{\shortmid }\mathbf{g.}$ Let us consider the d--tensor
\begin{eqnarray*}
(\sqrt{\mathbf{g}^{-1}\mathbf{q}})_{~\nu }^{\mu } &=&\left( (\sqrt{g^{-1}q})
_{~j}^{i},(\sqrt{g^{-1}q})_{~b}^{a}\right) ,\mbox{ and/ or }(\sqrt{\mathbf{\
_{\shortmid }g}^{-1}\mathbf{\ _{\shortmid }q}})_{~\nu }^{\mu }=\left( (\sqrt{%
\mathbf{\ _{\shortmid }}g^{-1}\mathbf{\ _{\shortmid }}q})_{~j}^{i},(\sqrt{%
\mathbf{\ _{\shortmid }}g^{-1}\mathbf{\ _{\shortmid }}q})_{a}^{\hspace{0in}%
~b}\right) , \\
\mbox{i.e. }(\sqrt{\mathbf{\ _{\shortmid }g}^{-1}\mathbf{\ _{\shortmid }q}})
&=&\left( \sqrt{h(\mathbf{\ _{\shortmid }}g^{-1}\mathbf{\ _{\shortmid }}q)},%
\sqrt{v(\mathbf{\ _{\shortmid }}g^{-1}\mathbf{\ _{\shortmid }}q)}\right) ,%
\mbox{ and/ or }(\sqrt{\mathbf{\ _{\shortmid }g}^{-1}\mathbf{\ _{\shortmid }q%
}})=\left( \sqrt{h(\mathbf{\ _{\shortmid }}g^{-1}\mathbf{\ _{\shortmid }}q)},%
\sqrt{cv(\mathbf{\ _{\shortmid }}g^{-1}\mathbf{\ _{\shortmid }}q)}\right) ,
\end{eqnarray*}%
where the square roots are computed respectively for $\mathbf{g}^{\mu \rho }\mathbf{q}_{\rho \nu }$ and/or $\ _{\shortmid }\mathbf{g}^{\mu \rho }\ _{\shortmid }\mathbf{q}_{\rho \nu },$ when
\begin{eqnarray*}
(\sqrt{g^{-1}q})_{~j}^{i}(\sqrt{g^{-1}q})_{~k}^{j} &=&g^{ij}q_{jk},%
\mbox{
and } \sum\limits_{\mathring{k}=0}^{4}~_{h}^{\mathring{k}}\beta ~_{h}e_{%
\mathring{k}}(\sqrt{h(g^{-1}q)})=3-tr\sqrt{h(g^{-1}q)}-\det \sqrt{h(g^{-1}q)}%
; \\
(\sqrt{g^{-1}q})_{~b}^{a}(\sqrt{g^{-1}q})_{~c}^{b} &=&g^{ab}q_{bc},%
\mbox{
and } \sum\limits_{\mathring{k}=0}^{4}~_{v}^{\mathring{k}}\beta ~_{v}e_{%
\mathring{k}}(\sqrt{v(g^{-1}q)})=3-tr\sqrt{v(g^{-1}q)}-\det \sqrt{v(g^{-1}q)}%
;
\end{eqnarray*}
\begin{eqnarray*}
\mbox{and/or }(\sqrt{\mathbf{\ _{\shortmid }}g^{-1}\mathbf{\ _{\shortmid }}q}%
)_{~j}^{i}(\sqrt{\mathbf{\ _{\shortmid }}g^{-1}\mathbf{\ _{\shortmid }}q}%
)_{~k}^{j} &=&\mathbf{\ _{\shortmid }}g^{ij}\mathbf{\ _{\shortmid }}q_{jk},%
\mbox{ and }\  \\
&&\sum\limits_{\mathring{k}=0}^{4}~_{h\shortmid }^{\mathring{k}}\beta
~_{h\shortmid }e_{\mathring{k}}(\sqrt{h(\mathbf{\ _{\shortmid }}g^{-1}%
\mathbf{\ _{\shortmid }}q)})=3-tr\sqrt{h(\mathbf{\ _{\shortmid }}g^{-1}%
\mathbf{\ _{\shortmid }}q)}-\det \sqrt{h(\mathbf{\ _{\shortmid }}g^{-1}%
\mathbf{\ _{\shortmid }}q)};
\end{eqnarray*}
\begin{eqnarray*}
(\sqrt{\mathbf{\ _{\shortmid }}g^{-1}\mathbf{\ _{\shortmid }}q})_{a}^{~b}(%
\sqrt{\mathbf{\ _{\shortmid }}g^{-1}\mathbf{\ _{\shortmid }}q})_{b}^{~c} &=&%
\mathbf{\ _{\shortmid }}g_{ab}\mathbf{\ _{\shortmid }}q^{bc},\mbox{ and } \\
&&\sum\limits_{\mathring{k}=0}^{4}~_{cv\shortmid }^{\mathring{k}}\beta
~_{cv\shortmid }e_{\mathring{k}}(\sqrt{cv(\mathbf{\ _{\shortmid }}g^{-1}%
\mathbf{\ _{\shortmid }}q)})=3-tr\sqrt{cv(\mathbf{\ _{\shortmid }}g^{-1}%
\mathbf{\ _{\shortmid }}q)}-\det \sqrt{cv(\mathbf{\ _{\shortmid }}g^{-1}%
\mathbf{\ _{\shortmid }}q)};
\end{eqnarray*}%
for respective coefficients $~_{h}^{\mathring{k}}\beta ,~_{v}^{\mathring{k}%
}\beta ,~_{h\shortmid }^{\mathring{k}}\beta $ and $~_{cv\shortmid }^{%
\mathring{k}}\beta .$ The values $\ _{h}e_{\mathring{k}},~_{v}e_{\mathring{k}%
}$ and/or $~_{h\shortmid }e_{\mathring{k}},~_{cv\shortmid }e_{\mathring{k}}$
can be computed correspondingly for any d--tensor $\mathbf{X}_{~\rho }^{\mu
}=(X_{~j}^{i},X_{~b}^{a})$ with trace
\begin{equation*}
X=\ \ ^{h}X+\ \ ^{v}X=[\ ^{h}X]+[\ ^{v}X]:=tr(\mathbf{X})=tr(\ ^{h}X)+tr(\
^{v}X)=\mathbf{X}_{~\mu }^{\mu }=X_{~i}^{i}+X_{~a}^{a}
\end{equation*}%
and/or any d--tensor $\mathbf{\ _{\shortmid }X}_{~\rho }^{\mu }=(\mathbf{\
_{\shortmid }}X_{~j}^{i},\mathbf{\ _{\shortmid }}X_{a}^{~b})$ with trace
\begin{equation*}
\mathbf{\ _{\shortmid }}X=\ \ _{\shortmid }^{h}X+\ \ _{\shortmid }^{cv}X=[\
_{\shortmid }^{h}X]+[\ _{\shortmid }^{cv}X]:=tr(\mathbf{\ _{\shortmid }X}%
)=tr(\ _{\shortmid }^{h}X)+tr(\ _{\shortmid }^{cv}X)=\mathbf{\ _{\shortmid }X%
}_{~\mu }^{\mu }=\mathbf{\ _{\shortmid }}X_{~i}^{i}+\mathbf{\ _{\shortmid }}%
X_{a}^{~a}.
\end{equation*}%
Such d-tensors are parameterized as "lifts" of certain $h$-objects to $v$%
-objects and/or $cv$-objects. The respective formulas are:
\begin{eqnarray*}
\mbox{ for }\ _{h}e_{\mathring{k}}(\mathbf{X}): && e_{0}(\ \
^{h}X)=1,e_{1}(\ \ ^{h}X)=\ \ ^{h}X, \\
2e_{2}(\ \ ^{h}X) &=&\ \ ^{h}X^{2}-[\ \ ^{h}X^{2}],\ 6e_{3}(\ \ ^{h}X)=\ \
^{h}X^{3}-3\ \ ^{h}X[\ \ ^{h}X^{2}]+2[\ \ ^{h}X^{3}], \\
24e_{4}(\ \ ^{h}X) &=&\ \ ^{h}X^{4}-6\ \ ^{h}X^{2}[\ \ ^{h}X^{2}]+3[\ \
^{h}X^{2}]^{2}+8\ \ ^{h}X[\ \ ^{h}X^{3}]-6[\ \ ^{h}X^{4}];\ e_{\mathring{k}%
}(\ \ ^{h}X)=0\mbox{ for }\mathring{k}>4;
\end{eqnarray*}%
\begin{eqnarray*}
\mbox{ for }\ _{v}e_{\mathring{k}}(\mathbf{X}): && e_{0}(\ \
^{v}X)=1,e_{1}(\ \ ^{v}X)=\ \ ^{v}X, \\
2e_{2}(\ \ ^{v}X) &=&\ \ ^{v}X^{2}-[\ \ ^{v}X^{2}],\ 6e_{3}(\ \ ^{v}X)=\ \
^{v}X^{3}-3\ \ ^{v}X[\ \ ^{v}X^{2}]+2[\ \ ^{v}X^{3}], \\
24e_{4}(\ \ ^{v}X) &=&\ \ ^{v}X^{4}-6\ \ ^{v}X^{2}[\ \ ^{v}X^{2}]+3[\ \
^{v}X^{2}]^{2}+8\ \ ^{v}X[\ \ ^{v}X^{3}]-6[\ \ ^{v}X^{4}];\ e_{\mathring{k}%
}(\ \ ^{v}X)=0\mbox{ for }\mathring{k}>4;
\end{eqnarray*}%
\begin{eqnarray*}
\mbox{ for }\ _{h\shortmid }e_{\mathring{k}}(\mathbf{\ _{\shortmid }X}): &&
e_{0}(\ \ _{\shortmid }^{h}X)=1,e_{1}(\ \ _{\shortmid }^{h}X)=\ \
_{\shortmid }^{h}X, \\
2e_{2}(\ \ _{\shortmid }^{h}X) &=&\ \ _{\shortmid }^{h}X^{2}-[\ \
_{\shortmid }^{h}X^{2}],\ 6e_{3}(\ \ _{\shortmid }^{h}X)=\ \ _{\shortmid
}^{h}X^{3}-3\ \ _{\shortmid }^{h}X[\ \ _{\shortmid }^{h}X^{2}]+2[\ \
_{\shortmid }^{h}X^{3}], \\
24e_{4}(\ \ _{\shortmid }^{h}X) &=&\ \ _{\shortmid }^{h}X^{4}-6\ \
_{\shortmid }^{h}X^{2}[\ \ _{\shortmid }^{h}X^{2}]+3[\ \ _{\shortmid
}^{h}X^{2}]^{2}+8\ \ _{\shortmid }^{h}X[\ \ _{\shortmid }^{h}X^{3}]-6[\ \
_{\shortmid }^{h}X^{4}];\ e_{\mathring{k}}(\ \ _{\shortmid }^{h}X)=0%
\mbox{
for }\mathring{k}>4;
\end{eqnarray*}%
\begin{eqnarray*}
\mbox{ for }\ _{cv\shortmid }e_{\mathring{k}}(\mathbf{\ _{\shortmid }X}): &&
e_{0}(\ \ _{\shortmid }^{cv}X)=1,e_{1}(\ \ _{\shortmid }^{cv}X)=\ \
_{\shortmid }^{cv}X, \\
2e_{2}(\ \ _{\shortmid }^{cv}X) &=&\ \ _{\shortmid }^{cv}X^{2}-[\ \
_{\shortmid }^{cv}X^{2}],\ 6e_{3}(\ \ _{\shortmid }^{cv}X)=\ \ _{\shortmid
}^{cv}X^{3}-3\ \ _{\shortmid }^{cv}X[\ \ _{\shortmid }^{cv}X^{2}]+2[\ \
_{\shortmid }^{cv}X^{3}], \\
24e_{4}(\ \ _{\shortmid }^{cv}X) &=&\ \ _{\shortmid }^{cv}X^{4}-6\ \
_{\shortmid }^{cv}X^{2}[\ \ _{\shortmid }^{cv}X^{2}]+3[\ \ _{\shortmid
}^{cv}X^{2}]^{2}+8\ \ _{\shortmid }^{cv}X[\ \ _{\shortmid }^{cv}X^{3}]-6[\ \
_{\shortmid }^{cv}X^{4}];\ e_{\mathring{k}}(\ \ _{\shortmid }^{cv}X)=0,%
\mathring{k}>4.
\end{eqnarray*}%
Following \cite%
{cai14,kluson13,nojiri13jcap,vijgmmp14,vvyijgmmp14a,vepjc14,vepjc14a,gvvepjc14,gvvcqg15,vacaruplb16}%
, we state for locally anisotropic gravity theories the

\begin{principle}
-\textbf{Convention} \label{pactmassive}\textbf{\ [actions for MDR-modified massive gravity theories]} \newline
We introduce the mass--deformed scalar curvatures%
\begin{eqnarray*}
\check{R} &:=&\ \ _{s}R+2~\ \ ^{h}\mu ^{2}(3-tr\sqrt{h(g^{-1}q)}-\det \sqrt{%
h(g^{-1}q)}) \\
&&+2~\ \ ^{v}\mu ^{2}(3-tr\sqrt{v(g^{-1}q)}-\det \sqrt{v(g^{-1}q)}) \\
\mbox{ and/or } && \\
\mathbf{\ _{\shortmid }}\check{R} &:=&\ \ _{s}^{\shortmid }R+2~\ \
_{\shortmid }^{h}\mu ^{2}(3-tr\sqrt{h(\mathbf{\ _{\shortmid }}g^{-1}\mathbf{%
\ _{\shortmid }}q)}-\det \sqrt{h(\mathbf{\ _{\shortmid }}g^{-1}\mathbf{\
_{\shortmid }}q)}) \\
&&+2~\ \ _{\shortmid }^{cv}\mu ^{2}(3-tr\sqrt{cv(\mathbf{\ _{\shortmid }}%
g^{-1}\mathbf{\ _{\shortmid }}q)}-\det \sqrt{cv(\mathbf{\ _{\shortmid }}%
g^{-1}\mathbf{\ _{\shortmid }}q)}),
\end{eqnarray*}%
with (MDR-induced, or other variants of N-connection structures) horizontal, $\ ^{h}\mu $ and/or $\ _{\shortmid }^{h}\mu ,$ and (co) vertical, $\ ^{v}\mu $ and $\ _{\shortmid }^{cv}\mu ,$ masses for locally anisotropic gravitons in respective phase spaces. Using such values, we construct respective $f$-modified Lagrange densities, $\ ^{\mathbf{g}\mu }\mathcal{L}=\mathbf{f}(\check{R})$ and/or $\ _{\shortmid }^{\mathbf{g}\mu }\mathcal{L}=\ _{\shortmid }\mathbf{f}(\ _{\shortmid}\check{R})$, and respective actions on
cotangent Lorentz bundles,%
\begin{eqnarray*}
\mathcal{S} &=&\ ^{\mathbf{g}\mu }\mathcal{S+}\ ^{m}\mathcal{S=}\frac{1}{%
16\pi }\int \delta u\sqrt{|\mathbf{g}_{\alpha \beta }|}(\ ^{\mathbf{g}\mu }%
\mathcal{L+}\ ^{m}\mathcal{L}\ ) \\
\mbox{ and/or } && \\
\mathbf{\ }\ _{\shortmid }\mathcal{S} &=&\ _{\shortmid }^{\mathbf{g}\mu }%
\mathcal{S+}\ _{\shortmid }^{m}\mathcal{S=}\frac{1}{16\pi }\int \delta
\mathbf{\ ^{\shortmid }}u\sqrt{|\mathbf{\ ^{\shortmid }g}_{\alpha \beta }|}%
(\ _{\shortmid }^{\mathbf{g}\mu }\mathcal{L+}\ _{\shortmid }^{m}\mathcal{L}),
\end{eqnarray*}%
where the Lagrange densities for matter fields, $\ ^{m}\mathcal{L}$\ and/or $\ _{\shortmid }^{m}\mathcal{L},$ can be taken in as for scalar fields (\ref{lagscf}), or generalized to YMH and/or fermion configurations.
\end{principle}

Performing a N-adapted variational calculus for above introduced Lagrange densities and actions for massive gravity models, with
$\ ^{1}\mathbf{f}(\check{R}):=d\mathbf{f}(\check{R})/d\check{R}$ and $\_{\shortmid }^{1}\mathbf{f}(\mathbf{\ {\shortmid}}\check{R}):=d\mathbf{\ {\shortmid }f}(\mathbf{\ _{\shortmid }}\check{R})/d\mathbf{\ _{\shortmid }}\check{R},$ we
prove:
\begin{consequence}
\textsf{[sources in massive gravity with MDRs on (co) tangent bundles ] }
\newline
On (co)tangent Lorentz bundles, massive gravity models encoding MDRs are
characterized by respective (effective) sources
\begin{equation}
\ ^{f\mu }\mathbf{\Upsilon }_{\mu \nu }=\ ^{g\mu }\mathbf{\Upsilon }_{\mu
\nu }~+~^{f}\mathbf{\Upsilon }_{\mu \nu }+\mathbf{\ }^{m}\mathbf{\Upsilon }%
_{\mu \nu },~^{g\mu }\mathbf{\Upsilon }_{\mu \nu }=(~^{g\mu }\mathbf{%
\Upsilon }_{ij},~^{g\mu }\mathbf{\Upsilon }_{ab}),\mbox{ where }
\label{smgb}
\end{equation}%
{\small
\begin{eqnarray*}
\ ^{g\mu }\mathbf{\Upsilon }_{ij} &=&-2\ \ ^{h}\mu ^{2}[3-tr\sqrt{h(g^{-1}q)}%
+\frac{1}{2}\det \sqrt{h(g^{-1}q)})]g_{ij}+\frac{\ \ ^{h}\mu ^{2}}{2}%
\{q_{ik}[(\sqrt{h(g^{-1}q)})^{-1}]_{~j}^{k}+q_{ik}[(\sqrt{h(g^{-1}q)}%
)^{-1}]_{~j}^{k}\}, \\
~^{g\mu }\mathbf{\Upsilon }_{ab} &=&-2\ \ ^{v}\mu ^{2}[3-tr\sqrt{v(g^{-1}q)}+%
\frac{1}{2}\det \sqrt{v(g^{-1}q)})]g_{ij}+\frac{\ \ ^{v}\mu ^{2}}{2}%
\{q_{ac}[(\sqrt{v(g^{-1}q)})^{-1}]_{~b}^{c}+q_{ac}[(\sqrt{v(g^{-1}q)}%
)^{-1}]_{~b}^{c}\},
\end{eqnarray*}%
}
\begin{equation*}
\ \ ^{f}\mathbf{\Upsilon }_{\mu \nu }=(\frac{\mathbf{f}}{2~^{1}\mathbf{f}}-%
\frac{\mathbf{D}^{2}\ ^{1}\mathbf{f}}{~^{1}\mathbf{f}})\mathbf{g}_{\mu \nu }+%
\frac{\mathbf{D}_{\mu }\mathbf{D}_{\nu }\ ^{1}\mathbf{f}}{~^{1}\mathbf{f}}%
;~^{m}\mathbf{\Upsilon }_{\mu \nu }=\frac{1}{2M_{P}^{2}}\ ^{m}\mathbf{T}%
_{\mu \nu };
\end{equation*}%
or (for sources) on cotangent bundles
\begin{equation}
\mathbf{\ _{\shortmid }^{~^{f\mu }}\Upsilon }_{\mu \nu }=\ _{\shortmid
}^{g\mu }\mathbf{\Upsilon }_{\mu \nu }~+~_{\shortmid }^{f}\mathbf{\Upsilon }%
_{\mu \nu }+\ _{\shortmid }^{m}\mathbf{\Upsilon }_{\mu \nu },~_{\shortmid
}^{g\mu }\mathbf{\Upsilon }_{\mu \nu }=(~_{\shortmid }^{g\mu }\mathbf{%
\Upsilon }_{ij},~_{\shortmid }^{g\mu }\mathbf{\Upsilon }^{ab}),\mbox{ where }
\label{smgcb}
\end{equation}%
{\small
\begin{eqnarray*}
\ _{\shortmid }^{g\mu }\mathbf{\Upsilon }_{ij} &=&-2\ \ _{\shortmid }^{h}\mu
^{2}[3-tr\sqrt{h(\mathbf{\ _{\shortmid }}g^{-1}\mathbf{\ _{\shortmid }}q)}+%
\frac{1}{2}\det \sqrt{h(\mathbf{\ _{\shortmid }}g^{-1}\mathbf{\ _{\shortmid }%
}q)})]g_{ij} \\
&&+\frac{\ \ _{\shortmid }^{h}\mu ^{2}}{2}\{\mathbf{\ _{\shortmid }}q_{ik}[(%
\sqrt{h(\mathbf{\ _{\shortmid }}g^{-1}\mathbf{\ _{\shortmid }}q)}%
)^{-1}]_{~j}^{k}+\mathbf{\ _{\shortmid }}q_{ik}[(\sqrt{h(\mathbf{\
_{\shortmid }}g^{-1}\mathbf{\ _{\shortmid }}q)})^{-1}]_{~j}^{k}\}, \\
~_{\shortmid }^{g\mu }\mathbf{\Upsilon }^{ab} &=&-2\ \ _{\shortmid }^{cv}\mu
^{2}[3-tr\sqrt{cv(\mathbf{\ _{\shortmid }}g^{-1}\mathbf{\ _{\shortmid }}q)}+%
\frac{1}{2}\det \sqrt{cv(\mathbf{\ _{\shortmid }}g^{-1}\mathbf{\ _{\shortmid
}}q)})]\mathbf{\ _{\shortmid }}g_{ij} \\
&&+\frac{\ \ _{\shortmid }^{cv}\mu ^{2}}{2}\{\mathbf{\ _{\shortmid }}q^{ac}[(%
\sqrt{cv(\mathbf{\ _{\shortmid }}g^{-1}\mathbf{\ _{\shortmid }}q)}%
)^{-1}]_{c}^{~b}+\mathbf{\ _{\shortmid }}q^{ac}[(\sqrt{cv(\mathbf{\
_{\shortmid }}g^{-1}\mathbf{\ _{\shortmid }}q)})^{-1}]_{c}^{~b}\},
\end{eqnarray*}%
}
\begin{equation*}
\ \ _{\shortmid }^{f}\mathbf{\Upsilon }_{\mu \nu }=(\frac{\mathbf{\
_{\shortmid }f}}{2~_{\shortmid }^{1}\mathbf{f}}-\frac{\mathbf{\ _{\shortmid
}D}^{2}\ _{\shortmid }^{1}\mathbf{f}}{~_{\shortmid }^{1}\mathbf{f}})\mathbf{%
\ _{\shortmid }g}_{\mu \nu }+\frac{\mathbf{\ _{\shortmid }D}_{\mu }\mathbf{\
_{\shortmid }D}_{\nu }\ _{\shortmid }^{1}\mathbf{f}}{~_{\shortmid }^{1}%
\mathbf{f}};~_{\shortmid }^{m}\mathbf{\Upsilon }_{\mu \nu }=\frac{1}{%
2M_{P}^{2}}\ _{\shortmid }^{m}\mathbf{T}_{\mu \nu }.
\end{equation*}
\end{consequence}

In 4-d massive MGTs \cite{cai14,kluson13,nojiri13jcap}, there are considered values of type
$S_{\ j}^{i}=g^{ik}\eta _{\overline{i}\overline{j}}\mathbf{e} _{k}s^{\overline{j}}\mathbf{e}_{j}s^{\overline{j}},$ with the Minkowski metric $\eta _{\overline{i}\overline{j}}=diag(1,1,1,-1),$ generated by four scalar St\"{u}kelberg fields $s^{\overline{j}},$ which is necessary for restoring the diffeomorphism invariance. Similar values can be considered
for fiducial fields $\mathbf{q}=\{_{\shortmid }\mathbf{q}_{\alpha \beta}\}$ and/or $\ _{\shortmid }\mathbf{q}=\{\ _{\shortmid }\mathbf{q}_{\alpha \beta}\}$ for massive and bi-metric locally anisotropic theories when the St\"{u}kelberg fields are introduced also on (co) fiber spaces, see details and examples in \cite{vijgmmp14,vvyijgmmp14a,vepjc14,vepjc14a,gvvepjc14,vacaruplb16}.

The sources of \ (effective) matter fields are approximated (for instance, for quantum fluctuations or certain summarized contributions) to (effective) cosmological constants. Following Principle \ref{panalogy}, we consider
\begin{assumption}
\textsf{[effective cosmological constants in massive gravity on (co) tangent bundles ] } \label{assumpt5} \newline
Generalized sources for matter fields (including possible effective sources defined by distortion tensors of d-connections) can be approximated by respective cosmological constants $\Lambda $ and/or $\ _{\shortmid }\Lambda $,
\begin{equation*}
\ ^{f\mu }\Upsilon _{\ \gamma }^{\beta }=\ ~^{f\mu }\Lambda \delta _{\
\gamma }^{\beta } \mbox{ and/or }\ \mathbf{\ }\ _{\shortmid }^{f\mu}\Upsilon
_{\ \gamma }^{\beta }=\mathbf{\ }_{\shortmid }^{f\mu }\Lambda \delta _{\
\gamma }^{\beta },
\end{equation*}%
when, correspondingly, $\ ^{f\mu }\mathbf{\Upsilon }_{\mu \nu }=\ ^{g\mu }\mathbf{\Upsilon }_{\mu \nu } + \ ^{f}\mathbf{\Upsilon }_{\mu \nu }+ \ ^{m}\mathbf{\Upsilon }_{\mu \nu }$ and/or $\ _{\shortmid }^{f\mu }\mathbf{\Upsilon }_{\mu\nu }= \ _{\shortmid }^{g\mu }\mathbf{\Upsilon }_{\mu \nu } +\ _{\shortmid }^{f}\mathbf{\Upsilon }_{\mu \nu }+\ _{\shortmid }^{m}\mathbf{%
\Upsilon } _{\mu \nu }$ considering that such sources are related to relations of additivity,
\begin{equation*}
\ \ ~^{f\mu }\Lambda =\ ^{g\mu }\Lambda +\ ^{f}\Lambda +\ ^{m}\Lambda \mbox{\ and/or }\ _{\shortmid }^{f\mu }\Lambda =\ _{\shortmid }^{g\mu }\Lambda +\ _{\shortmid }^{f}\Lambda +\ _{\shortmid }^{m}\Lambda .
\end{equation*}
\end{assumption}

Bi-metric and/or massive gravity LC--configurations can be extracted by imposing additional \ (non) holonomic constraints when $\mathbf{D}_{\mathcal{T}=0}\rightarrow \nabla .$

\subsubsection{Lagrange densities for short-range locally anisotropic gravity with MDRs and LIV}

In \cite{bailey15} (see also references therein), a systematic study of LIVs for dimensions $\dim >5$ (with explicit formulas for a class of effective MGTs of dimensions $\dim =4,5,6)$ was initiated. In this subsection, we elaborate on toys models of Lagrange-Hamilton gravity theories of dimension $\dim =3+3,$ with signatures of local metrics of type $(++-;++-),$ and
indices running values $i,j,...=1,2,3$ and $a,b,...=4,5,6,$ for $\alpha =(i,a),\beta =(j,b),...$

Following Principle \ref{panalogy}, we change the LC--connection into a metric compatible d-connection completely determined by the same metric structure and formulate:
\begin{principle}
-\textbf{Convention} \label{princsr}\textbf{\ [Lagrange densities for MDR-modified short-range gravity and LIV]} The Lagrange densities of underlying actions for effective gravity theory with short-range, sr, locally anisotropic interactions with spontaneous LIVs are determined by such sums:%
\begin{eqnarray}
\ ^{\mathbf{sr}}\mathcal{L} &=&\ ^{\mathbf{g}}\mathcal{L}+\ ^{m}\mathcal{L}%
+\ ^{k}\mathcal{L}+\ ^{lv}\mathcal{L},\mbox{ where }  \label{lagrdshortord}
\\
\ ^{\mathbf{g}}\mathcal{L} &=&\ ^{s}\mathbf{R}-2\Lambda ,\ ^{m}\mathcal{L=}\
^{m}\mathcal{L}\mbox{ see  }(\ref{lagscf}),  \notag \\
\ ^{k}\mathcal{L} &\mathcal{=}&%
\mbox{ is for the dynamics of fiels
triggering sponanteous LIVs}  \notag \\
\ ^{lv}\mathcal{L} &=&\ ^{lv}\mathcal{L}_{[4]}+\ ^{lv}\mathcal{L}_{[5]}+\
^{lv}\mathcal{L}_{[6]}+...,\mbox{ for }  \notag \\
\ ^{lv}\mathcal{L}_{[4]} &=&\kappa _{\alpha \beta \gamma \delta }^{[4]}(x,y)%
\mathbf{R}^{\alpha \beta \gamma \delta }[\mathbf{D}],\ ^{lv}\mathcal{L}%
_{[5]}=\kappa _{\alpha \beta \gamma \delta \tau }^{[5]}(x,y)\mathbf{D}^{\tau
}\mathbf{R}^{\alpha \beta \gamma \delta },  \notag \\
\ ^{lv}\mathcal{L}_{[6]} &=&\frac{1}{2}\ ^{1}\kappa _{\alpha \beta \gamma
\delta \tau \nu }^{[6]}(x,y)\left( \mathbf{D}^{\tau }\mathbf{D}^{\nu }+%
\mathbf{D}^{\nu }\mathbf{D}^{\tau }\right) \mathbf{R}^{\alpha \beta \gamma
\delta }+\ ^{2}\kappa _{\alpha \beta \gamma \delta \tau \nu \lambda \theta
}^{[6]}(x,y)\mathbf{R}^{\tau \nu \lambda \theta }\mathbf{R}^{\alpha \beta
\gamma \delta };  \notag
\end{eqnarray}
and/ or
\begin{eqnarray*}
\ _{\shortmid }^{\mathbf{sr}}\mathcal{L} &=&\ _{\shortmid }^{\mathbf{g}}%
\mathcal{L+}\ _{\shortmid }^{lv}\mathcal{L}+\ \ _{\shortmid }^{k}\mathcal{L+}%
\ _{\shortmid }^{m}\mathcal{L} \\
\ _{\shortmid }^{\mathbf{g}}\mathcal{L} &=&\ _{\shortmid }^{s}\mathbf{R}-2\
_{\shortmid }\Lambda ,\ _{\shortmid }^{m}\mathcal{L=}\ _{\shortmid }^{m}%
\mathcal{L}\mbox{ see  }(\ref{lagscf}), \\
\ _{\shortmid }^{k}\mathcal{L} &\mathcal{=}&%
\mbox{ is for the dynamics of
fiels triggering sponanteous LIVs} \\
\ _{\shortmid }^{lv}\mathcal{L} &=&\ _{\shortmid }^{lv}\mathcal{L}_{[4]}+\
_{\shortmid }^{lv}\mathcal{L}_{[5]}+\ _{\shortmid }^{lv}\mathcal{L}%
_{[6]}+...,\mbox{ for } \\
\ _{\shortmid }^{lv}\mathcal{L}_{[4]} &=&\ \ _{\shortmid }\kappa _{\alpha
\beta \gamma \delta }^{[4]}(x,p)\ \ _{\shortmid }\mathbf{R}^{\alpha \beta
\gamma \delta }[\ \ _{\shortmid }\mathbf{D}],\ _{\shortmid }^{lv}\mathcal{L}%
_{[5]}=\ \ _{\shortmid }\kappa _{\alpha \beta \gamma \delta \tau
}^{[5]}(x,p)\ \ _{\shortmid }\mathbf{D}^{\tau }\ \ _{\shortmid }\mathbf{R}%
^{\alpha \beta \gamma \delta }, \\
\ _{\shortmid }^{lv}\mathcal{L}_{[6]} &=&\frac{1}{2}\ _{\shortmid
}^{1}\kappa _{\alpha \beta \gamma \delta \tau \nu }^{[6]}(x,p)\left( \
_{\shortmid }\mathbf{D}^{\tau }\ _{\shortmid }\mathbf{D}^{\nu }+\
_{\shortmid }\mathbf{D}^{\nu }\ _{\shortmid }\mathbf{D}^{\tau }\right) \
_{\shortmid }\mathbf{R}^{\alpha \beta \gamma \delta }+\ _{\shortmid
}^{2}\kappa _{\alpha \beta \gamma \delta \tau \nu \lambda \theta
}^{[6]}(x,p)\ _{\shortmid }\mathbf{R}^{\tau \nu \lambda \theta }\
_{\shortmid }\mathbf{R}^{\alpha \beta \gamma \delta },
\end{eqnarray*}%
where $\ \kappa $-coefficients are taken as in \cite{bailey15} (see also references therein) but with $h$-$v$- and/or $h$-$cv$-decompositions.
\end{principle}

The $h$-components of above formulas transform into holonomic ones for $\mathbf{D}_{\mathcal{T}=0}\rightarrow \nabla .$

Performing respective N-adapted variational calculuses for actions
\begin{eqnarray*}
\ \ ^{\mathbf{sr}}\mathcal{S} &=&\ ^{\mathbf{g}}\mathcal{S+}\ ^{m}\mathcal{S+%
}\ \ ^{k}\mathcal{S+}\ ^{lv}\mathcal{S=}\frac{1}{16\pi G_{\dim }}\int \delta
u\sqrt{|\mathbf{g}_{\alpha \beta }|}(\ ^{\mathbf{g}}\mathcal{L}+\ ^{m}%
\mathcal{L}+\ ^{k}\mathcal{L}+\ ^{lv}\mathcal{L})\mbox{ and/or } \\
\mathbf{\ }\ _{\shortmid }^{\mathbf{sr}}\mathcal{S} &=&\ _{\shortmid }^{%
\mathbf{g}}\mathcal{S}+\ _{\shortmid }^{m}\mathcal{S}+\ \ _{\shortmid }^{k}%
\mathcal{S}+\ _{\shortmid }^{lv}\mathcal{S=}\frac{1}{16\pi \mathbf{\
^{\shortmid }}G_{\dim }}\int \delta \mathbf{\ ^{\shortmid }}u\sqrt{|\mathbf{%
\ ^{\shortmid }g}_{\alpha \beta }|}(\ _{\shortmid }^{\mathbf{g}}\mathcal{L+}%
\ _{\shortmid }^{m}\mathcal{L}+\ \ _{\shortmid }^{k}\mathcal{L+}\
_{\shortmid }^{lv}\mathcal{L}),
\end{eqnarray*}%
with conventional Newton constants $G_{\dim }$ and $\mathbf{\ ^{\shortmid }}G_{\dim }$ one proves:
\begin{consequence}
\textsf{[effective matter energy-momentum d-tensors in short-range gravity on (co) tangent Lorentz bundles ] } The symmetric energy-momentum d-tensors for matter fields and sources induced by LIV on (co) tangent bundles derived for respective Lagrange densities (\ref{lagrdshortord}) are computed
\begin{eqnarray}
\ ^{sr}\mathbf{T}_{\beta \delta } &=&\ ^{m}\mathbf{T}_{\beta \delta }+\ ^{lv}%
\mathbf{T}_{\beta \delta },\mbox{ for }\ ^{m}\mathbf{T}_{\beta \delta }%
\mbox{ as in }(\ref{emscdt})\mbox{ and }  \label{emdthsr} \\
\ ^{lv}\mathbf{T}_{\beta \delta } &=&\frac{1}{4\pi G_{\dim }}\widehat{\bar{s}%
}^{\alpha \beta }\mathbf{E}_{\alpha (\mu \nu )\beta }-\frac{1}{16\pi G_{\dim
}}\widehat{\bar{u}}\mathbf{E}_{\mu \nu }+  \notag \\
&&\frac{a}{8\pi G_{\dim }}\ ^{1}\overline{\kappa }_{\alpha (\mu \nu )\beta
\gamma \delta }^{[6]}\mathbf{e}^{\alpha }\mathbf{e}^{\beta }\mathbf{R}%
^{\gamma \delta }+\frac{1}{2\pi G_{\dim }}\ ^{2}\overline{\kappa }_{\alpha
(\mu \nu )\beta \gamma \delta \epsilon \zeta }^{[6]}\mathbf{e}^{\alpha }%
\mathbf{e}^{\beta }\mathbf{R}^{\gamma \delta \epsilon \zeta },  \notag \\
\mbox{ and/or } &&  \notag \\
\ _{\shortmid }^{sr}\mathbf{T}_{\beta \delta } &=&\ _{\shortmid }^{m}\mathbf{%
T}_{\beta \delta }+\ _{\shortmid }^{lv}\mathbf{T}_{\beta \delta },%
\mbox{ for
}\ _{\shortmid }^{m}\mathbf{T}_{\beta \delta }\mbox{ as in }(\ref{emscdt})%
\mbox{ and }  \notag \\
\ _{\shortmid }^{lv}\mathbf{T}_{\beta \delta } &=&\frac{1}{4\pi \mathbf{\
^{\shortmid }}G_{\dim }}\mathbf{\ ^{\shortmid }}\widehat{\bar{s}}^{\alpha
\beta }\mathbf{\ ^{\shortmid }E}_{\alpha (\mu \nu )\beta }-\frac{1}{16\pi
G_{\dim }}\mathbf{\ ^{\shortmid }}\widehat{\bar{u}}\mathbf{\ ^{\shortmid }E}%
_{\mu \nu }+  \notag \\
&&\frac{\mathbf{\ ^{\shortmid }}a}{8\pi G_{\dim }}\ _{\shortmid }^{1}%
\overline{\kappa }_{\alpha (\mu \nu )\beta \gamma \delta }^{[6]}\mathbf{\
^{\shortmid }e}^{\alpha }\mathbf{\ ^{\shortmid }e}^{\beta }\mathbf{\
^{\shortmid }R}^{\gamma \delta }+\frac{1}{2\pi \mathbf{\ ^{\shortmid }}%
G_{\dim }}\ _{\shortmid }^{2}\overline{\kappa }_{\alpha (\mu \nu )\beta
\gamma \delta \epsilon \zeta }^{[6]}\mathbf{\ ^{\shortmid }e}^{\alpha }%
\mathbf{\ ^{\shortmid }e}^{\beta }\mathbf{\ ^{\shortmid }R}^{\gamma \delta
\epsilon \zeta }.  \notag
\end{eqnarray}
\end{consequence}

The coefficients and operators in (\ref{emdthsr}) are computed in certain forms which are similar to formula (6) in \cite{bailey15}\footnote{readers may see that paper and references therein for details on symmetries of LIV coefficients, definition of operators and sources etc.} when $\partial ^{\alpha }\rightarrow $ $\mathbf{e}^{\alpha }$ and $\nabla
\rightarrow \mathbf{D}$ on $T\mathbf{V}.$Here we note that the double dual of the Riemannian d-tensor $\mathbf{E}_{\alpha \beta \gamma \delta }:=\frac{1}{4}\epsilon _{\alpha \beta \kappa \lambda }\epsilon _{\gamma \delta \mu \nu
}\mathbf{R}^{\kappa \lambda \mu \nu },$ the symmetrization of indices is denoted by "$(\alpha \beta )$" and $a$ is an integration constant. Similar values for $T^{\ast }\mathbf{V}$ are labeled by "$\mathbf{\ ^{\shortmid }}$".

The sources of (effective) matter fields are approximated (for instance, for quantum fluctuations or certain summarized contributions) to (effective) cosmological constants. Following Principle \ref{panalogy}, we consider
\begin{assumption}
\textsf{[effective cosmological constants for short-range locally anisotropic interactions ] } \label{assumpt6} \newline
Generalized (effective) sources for matter and fields (including possible effective sources defined by distortion tensors of d-connections) in short-range locally anisotropic gravity can be approximated by respective cosmological constants $\ ^{sr}\Lambda $ and/or $\ _{\shortmid }^{sr}\Lambda $, when effective sources
\begin{equation*}
\ ^{sr}\Upsilon _{\ \gamma }^{\beta }=\ ^{sr}\Lambda \delta _{\ \gamma
}^{\beta } \mbox{ and/or } \ _{\shortmid }^{sr}\Upsilon _{\ \gamma }^{\beta
}=\ _{\shortmid }^{sr}\Lambda \delta _{\ \gamma }^{\beta }
\end{equation*}
are taken respectively $\ ~^{sr}\mathbf{\Upsilon }_{\mu \nu }=\ ^{m}\mathbf{%
\Upsilon }_{\mu \nu }+~^{lv}\mathbf{\Upsilon }_{\mu \nu }$ and/or $\
_{\shortmid }^{sr}\mathbf{\Upsilon }_{\mu \nu }=~_{\shortmid }^{m}\mathbf{%
\Upsilon }_{\mu \nu }~+ \ _{\shortmid }^{lv}\mathbf{\Upsilon }_{\mu \nu }$,
for relations of additivity
\begin{equation*}
\ ^{sr}\Lambda =\ ^{m}\Lambda +\ ^{lv}\Lambda \mbox{ and } \ _{\shortmid
}^{sr}\Lambda =\ _{\shortmid }^{m}\Lambda + \ _{\shortmid }^{lv}\Lambda .
\end{equation*}
\end{assumption}

Short-range LC--configurations (in general, parameterized by off-diagonal metrics) can be extracted by imposing additional (non) holonomic constraints when $\mathbf{D}_{\mathcal{T}=0}\rightarrow \nabla $. The constructions can be performed in similar forms, for instance, if respective energy-momentum tensors for YMH fields are considered instead of $\ ^{m}\mathbf{T}_{\beta \delta },$ or there are studied massive gravity effects and bi-metric structures.

\subsubsection{The equations of motion and nonholonomic conservation laws for (effective) sources}

\label{ssecnhcons}The conservation laws
\begin{equation}
\nabla ^{i}(R_{ij}-\frac{1}{2}g_{ij}R)=0\mbox{ and }\nabla _{i}T^{ij}=0
\label{divemt}
\end{equation}%
in GR are consequences of the Bianchi relations. Such laws involve the idea that the Einstein's equations actually imply the geodesic hypothesis when the world lines of test bodies are geodesics of the spacetime metric. However, it should be noted that bodies which are "large" enough to feel the tidal forces of the gravitational field will deviate from geodesic motion. Such deviations may be caused by internal symmetries and (non) linear interactions and certain nonholonomic constraints on the dynamics of gravitational and matter fields. Nevertheless, the equations of motion of "large and/or complex structure" bodies in GR also can be found from the condition $\nabla _{a}T^{ab}=0.$

For a Finsler d--connection $\mathbf{D}$ (even it can be chosen to metric compatible,
$\mathbf{D}_{\alpha }\mathbf{g}^{\beta \gamma }=0)$, we have $\mathbf{D}_{\alpha }\mathbf{T}^{\alpha \beta }\neq 0,$ which is a consequence of non--symmetry of the Ricci and Einstein d--tensors;\ see explanations for formulas (\ref{dricci}) and (\ref{driccid}) and Definition-Theorems \ref{dteinstdt} and generalized Bianchi identities. Such a property is also related to nonholonomic constraints on the dynamics of Finsler gravitational fields. It is not surprising that the "covariant divergence" of (effective) matter sources does not vanish even for canonical d--connections $\ \widehat{\mathbf{D}}$ and/or
$\widetilde{\mathbf{D}}$ (on nonholonomic \ manifolds and/or (co) tangent bundles) see formulas (\ref{canondcl}) and (\ref{canondch}). In such cases, the conservation law became more sophisticate because of nonholonomic constraints and differences between the autoparallels of $\mathbf{D}$ and/or $\mathbf{\ ^{\shortmid }D}$ and nonlinear geodesic (semi-spray) equations (\ref{ngeqf}). Nevertheless it is possible to compute effective nonholonomic tidal forces of locally anisotropic gravitational fields using distorting relations of type $\mathbf{D}=\nabla -\mathbf{Z}$ \ (\ref{dist}), see Lemma \ref{ldist}  and, for canonical distortions, formulas (\ref{candistr}).

\begin{assumption}
\textsf{[energy-momentum d-tensors for (non) massive gravity theories on (co) tangent Lorentz bundles] } \label{assumpt7}The (effective) energy momentum d-tensors on (co) tangent Lorentz bundles are defined by all possible components (\ref{emscdt}), (\ref{emdtym}), (\ref{emdthiggs}), (\ref{smgb}), (\ref{smgcb}) and (\ref{emdthsr}) constructed by respective geometric principles and N-adapted variational calculuses,
\begin{eqnarray}
\mathbf{T}_{\alpha \beta } &=&\ ^{\phi }\mathbf{T}_{\alpha \beta }+\ ^{A}%
\mathbf{T}_{\alpha \beta }+\ ^{H}\mathbf{T}_{\alpha \beta }+\ ^{g\mu }%
\mathbf{T}_{\alpha \beta }~+~^{f}\mathbf{T}_{\alpha \beta }+\ ^{lv}\mathbf{T}%
_{\alpha \beta }+\ldots  \label{generemdt} \\
\ _{\shortmid }\mathbf{T}_{\alpha \beta } &=&\ _{\shortmid }^{\phi }\mathbf{T%
}_{\alpha \beta }+\ _{\shortmid }^{A}\mathbf{T}_{\alpha \beta }+\
_{\shortmid }^{H}\mathbf{T}_{\alpha \beta }+\ _{\shortmid }^{g\mu }\mathbf{T}%
_{\alpha \beta }~+~_{\shortmid }^{f}\mathbf{T}_{\alpha \beta }+\ _{\shortmid
}^{lv}\mathbf{T}_{\alpha \beta }+\ldots .  \notag
\end{eqnarray}
\end{assumption}
Dots in these formulas contain possible other types contributions of locally anisotropic spinor fields, from (super) string gravity, noncommutative deformations, fractional, diffusion and kinetic processes etc., see examples in Refs.
\cite{vacaruplb16,vacaruepjc17,vch2416,vgrg12,vcqg11,vjmp09,vrev08,vjmp06,vmon06,vvicol04,vmon02,vmon98,vhsp98,vnp97,vap97,vjmp96,
bvcejp11,gvvepjc14,gvvcqg15,gheorghiuap16,bubuianucqg17}. The sources in such works were considered in the bulk for nonholonomic
manifolds and tangent bundles. But analogy, the constructions can re-defined for cotangent Lorentz bundles. In this work, we omit cumbersome formulas for the Bianchi identities and conservation laws with nonholonomic constraints written in variables $(\mathbf{g,N,D})$ and/or $(\mathbf{\ ^{\shortmid }g,\ ^{\shortmid }N,\ ^{\shortmid }D}).$

We note that the Bianchi identities for some data $(\mathbf{g,N,}\widehat{\mathbf{D}})$ can be computed by introducing nonholonomic deformations $\nabla =\widehat{\mathbf{D}}-\widehat{\mathbf{Z}}$ into the standard relations
$\nabla ^{\alpha }(R_{\alpha \beta }-\frac{1}{2}g_{\alpha \beta }R)=0$ and $\nabla ^{\alpha }T_{\alpha \beta }=0.$ Even, in general, $\widehat{\mathbf{D} }^{\alpha }\mathbf{T}_{\alpha \beta }=\mathbf{Q}_{\beta }\neq 0,$ such a
$\mathbf{Q}_{\beta }[\mathbf{g,N}]$ is completely defined by the d--metric and chosen N--connection structure. This is a consequence of the nonholonomic structure. A similar "problem" exists in Lagrange mechanics with non--integrable constraints when the standard conservation laws do not hold true. A new class of effective variables can be introduced using Lagrange multiples.

\begin{principle}
\textsf{[nonholonomic deformations by MDRs of conservation laws on (co) tangent Lorentz bundles]} \label{prinndefclaw} Postulating for respective LC-connections $\nabla $ and/or $\ ^{\shortmid}\nabla $ on (co)tangent Lorentz bundles conservation laws of type%
\begin{equation*}
\nabla ^{\alpha }(R_{\alpha \beta }-\frac{1}{2}g_{\alpha \beta} R) =0
\mbox{
and } \nabla ^{\alpha } T_{\alpha \beta }=0; \mbox{ and/or } \ ^{\shortmid
}\nabla ^{\alpha }(\ ^{\shortmid}R_{\alpha \beta }- \frac{1}{2}\ ^{\shortmid
}g_{\alpha \beta } \ ^{\shortmid }R) = 0 \mbox{ and } \ ^{\shortmid }\nabla
^{\alpha } \ ^{\shortmid }T_{\alpha \beta }=0,
\end{equation*}%
we define unique nonholonomic conservation laws
\begin{equation}
\mathbf{D}^{\alpha }\mathbf{T}_{\alpha \beta }=\mathbf{Q}_{\beta }\neq 0%
\mbox{ and/or }\mathbf{\ ^{\shortmid }D}^{\alpha }\mathbf{\ ^{\shortmid }T}%
_{\alpha \beta }=\mathbf{\ ^{\shortmid }Q}_{\beta }\neq 0,
\label{nonhclsourc}
\end{equation}%
with d-tensors $\mathbf{Q}_{\beta }$ and/or $\mathbf{\ ^{\shortmid }Q}_{\beta }$ uniquely determined by respective data with unique distortion relations $(\mathbf{g,N,D}=\nabla -\mathbf{Z})$ and/or $(\mathbf{\ ^{\shortmid }g,\ ^{\shortmid }N,\ ^{\shortmid }D}=\ ^{\shortmid }\nabla - \ ^{\shortmid }\mathbf{Z}).$
\end{principle}

Nonholonomic deformations of conservation laws (\ref{nonhclsourc}) are similar to modifications of conservation laws in nonholonomic Lagrange and/or Hamilton mechanics. The motion of probing point particles in phase spaces with MDRs are modeled by canonical data for Finsler-Lagrange-Hamilton spaces with respective d--connections (\ref{canondcl}) and (\ref{canondch})
and canonical distortions (\ref{candistr}). Nonholonomic canonical deformations
$\widetilde{\mathbf{D}}^{\alpha }\widetilde{\mathbf{T}}_{\alpha \beta }=\widetilde{\mathbf{Q}}_{\beta }\neq 0$ and
$\ ^{\shortmid }\widetilde{\mathbf{D}}^{\alpha }\ ^{\shortmid }\widetilde{\mathbf{T}}_{\alpha \beta }=\ ^{\shortmid }\widetilde{\mathbf{Q}}_{\beta }\neq 0$ are considered also in almost symplectic classical and (deformation) quantum
models, commutative and noncommutative, of gravity. Finally, we note that Principle \ref{prinndefclaw} has to be re-formulated and stated in model dependent forms for theories with multi-metric structure, in metric-affine spaces with nontrivial nonmetricity, nonlocal interactions etc.

\section{Modified Field Equations with MDRs and Lagrange-Hamilton Gravity}

\label{smodensteq} The goal of this section is to formulate the gravitational and matter field equations in MGTs with MDRs and LIVs. Such systems of nonlinear PDEs can be derived in  abstract and/or N-adapted (in general, coordinate free) forms following Principles \ref{pgcov} and \ref{panalogy} and applying geometric and variational methods using Lagrange densities on (co) tangent Lorentz bundles postulated in previous section.

\subsection{Generalized Einstein equations on nonholonomic (co) tangent Lorentz bundles}

\label{ssmeeq} We shall analyse modified gravitational field equations for d-metrics (\ref{dmt}) and (\ref{dmct}) (see Assumption \ref{assumpt3} and equivalent to off-diagonal metrics (\ref{offd})) and metric compatible d-connections. The (effective) matter sources $\Upsilon _{\ \beta \gamma}:=\varkappa (\mathbf{T}_{\alpha \beta}-\frac{1}{2}\mathbf{g}_{\alpha \beta }\mathbf{T})$ and/or $\ ^{\shortmid }\Upsilon _{\ \beta \gamma} :=\varkappa (\ ^{\shortmid}\mathbf{T}_{\alpha \beta }- \frac{1}{2}
\ ^{\shortmid}\mathbf{g}_{\alpha \beta }\ ^{\shortmid }\mathbf{T})$ are determined by energy-momentum d-tensors of type (\ref{generemdt}), see Assumption \ref{assumpt7}. Such sources are defined both by distortion d-tensors of respective d--connections and generalized energy-momentum d--tensors.

\subsubsection{Nonholonomic Einstein equations with general (effective) sources}

Following a N-adapted variational calculus for data $(T\mathbf{V},\mathbf{N,g,D})$ and source $\Upsilon _{\beta \gamma },$ we motivate and prove:

\begin{principle}
\textbf{-Theorem} \textsf{[nonholonomic modifications of Einstein equations on tangent Lorentz bundles] } \label{princtheinsttb} \newline
The modified Einstein equations for the Ricci d-tensor (\ref{dricci}) corresponding to some metric compatible data $(\mathbf{g,Dg}=0)$ and an effective source $\Upsilon _{\ \beta \gamma}$ computed for (\ref{generemdt}) on a tangent Lorentz bundle $\mathbf{TV}$ are
\begin{equation}
\mathbf{R}_{\alpha \beta }[\mathbf{D}]=\Upsilon _{\alpha \beta }.
\label{meinsteqtb}
\end{equation}
\end{principle}

Considering distortion d-tensors, $\mathbf{D\rightarrow \check{D}}=\mathbf{D+Z}$ determined in unique forms by data $(\mathbf{N,g)}$ following certain geometric principles, the equations (\ref{meinsteqtb}) can be re-defined for
metric non-compatible d-connections, $\mathbf{\check{D}g\neq 0}$. Such theories determined by different types of (effective) sources $\Upsilon _{\alpha \beta }$, and examples of exact solutions, are studied in details in monograph \cite{vmon02}. Using frame transforms $\mathbf{g}_{\alpha \beta}=e_{\ \alpha }^{\alpha ^{\prime }}e_{\ \beta}^{\beta ^{\prime }}g_{\alpha
^{\prime }\beta ^{\prime }}$ and distortion relations of type (\ref{candistr}), we prove

\begin{corollary}
\textsf{[canonical form of locally anisotropic Einstein equations on tangent bundles] } \newline
The modified Einstein equations (\ref{meinsteqtb}) on $\mathbf{TV}$ can be written equivalently for a canonical d-connection $\widehat{\mathbf{D}}$
\begin{equation}
\widehat{\mathbf{R}}_{\alpha \beta }[\widehat{\mathbf{D}}]=\widehat{\Upsilon}_{\alpha \beta }.  \label{meinsteqtbcan}
\end{equation}%
Imposing additional (non)holonomic constraints for zero torsion, $\widehat{\mathbf{D}}_{\mid \widehat{\mathcal{T}}=0}=\nabla ,$ when
\begin{equation}
\widehat{\mathcal{T}}=0.  \label{lcconddt}
\end{equation}
\end{corollary}

The equations (\ref{lcconddt}) can be solved in exact form for configurations subjected to the conditions (\ref{lccondl}).

\begin{remark}
\textsf{[ nonholonomic Lagrange-Finsler variables for locally anisotropic Einstein equations ] } \newline
The system of nonlinear PDEs (\ref{meinsteqtb}) and (\ref{lcconddt}) can be considered on (pseudo) Riemannian manifolds with nonholonomic fibered $h$-and $v$-structures. Prescribing N--connection structures, we can reformulate the GR theory and the Einstein equations using data $(\mathbf{g},\nabla )$ and/or in terms of nonholonomic variables with data $(\mathbf{g,}%
\widehat{\mathbf{D}}).$
\end{remark}

Variables $(\mathbf{g,}\widehat{\mathbf{D}})$ with corresponding parameterizations for 2+2+2+... splitting of spacetimes and phase spaces play a preferred role in elaborating the AFDM for decoupling and integrating (modified) Einstein equations, see Refs.
\cite{vexsol98,vmon98,vjhep01,vpcqg01,vsbd01,vmon02,vtnpb02,vsjmp02,vijmpd03,vijmpd03a,vjmp05,vmon06,avjgp09,vijgmmp11,vcqg11,vepl11,
vijtp13,vjpcs13, vport13,vepjc14, vepjc14a,vvyijgmmp14a,gvvepjc14,gvvcqg15}. Certain classes of generic off-diagonal solutions can be constructed in explicit form for the canonical data $(\widetilde{\mathbf{g}},\widetilde{\mathbf{D}})$ \cite{vijtp10,vcqg10,vijtp10a} and (in almost K\"{a}hler variables) for performing deformation quantization of gravity theories \cite{vch2416,vjmp13,vjgp10,vijgmmp09,avjmp09,vpla08,vjmp07}.

\subsubsection{Gravitational equations on Lorentz cotangent bundles}

For geometric data $(\mathbf{T}^{\ast }\mathbf{V,\ ^{\shortmid }N,\ ^{\shortmid }g,\ ^{\shortmid }D})$ and prescribed source
$\ ^{\shortmid}\Upsilon _{\ \beta \gamma },$ we derive
\begin{principle}
\textbf{-Theorem} \textsf{[nonholonomic modifications of Einstein equations on cotangent Lorentz bundles ] } \label{princtheinstctb} \newline
On $\mathbf{T}^{\ast }\mathbf{V,}$ the modified Einstein equations, when
$(\mathbf{\ ^{\shortmid }g,\ ^{\shortmid }D\ ^{\shortmid }g}=0)$, an effective source $\ ^{\shortmid }\Upsilon _{\alpha \beta }$ is computed for d-tensors (\ref{generemdt}) and Ricci d-tensor (\ref{driccid})  are
\begin{equation}
\mathbf{\ ^{\shortmid }R}_{\alpha \beta }[\mathbf{\ ^{\shortmid }D}]=\ ^{\shortmid }\Upsilon _{\alpha \beta }.  \label{meinsteqctb}
\end{equation}
\end{principle}

Considering distortion d-tensors, $\mathbf{\ ^{\shortmid }D\rightarrow \ ^{\shortmid }\check{D}}=\mathbf{\ ^{\shortmid }D+\ ^{\shortmid }Z}$ and nonholonomic deformations, we prove

\begin{corollary}
\textsf{[canonical form of locally anisotropic Einstein equations on cotangent bundles ] } \newline
The analog of modified Einstein equations (\ref{meinsteqtb}) on $T\mathbf{V}$ can be derived on $T^{\ast }\mathbf{V}$ for a canonical d-connection $\mathbf{\ ^{\shortmid }}\widehat{\mathbf{D}},$%
\begin{equation}
\ ^{\shortmid }\widehat{\mathbf{R}}_{\alpha \beta }[\ ^{\shortmid}\widehat{\mathbf{D}}]= \ ^{\shortmid }\widehat{\Upsilon }_{\alpha \beta }.
\label{meinsteqtbcand}
\end{equation}
\end{corollary}

We extract LC-configurations from solutions of such equations if there are imposed additional zero-torsion constraints (\ref{lccondh}).

\begin{remark}
\textsf{[ nonholonomic Hamilton-Cartan variables for locally anisotropic Einstein equations ] } \newline
The system of nonlinear PDEs (\ref{meinsteqctb}) can be re-written equivalently in terms of nonholonomic canonical variables with data $(\ ^{\shortmid }\mathbf{g},\ ^{\shortmid }\widehat{\mathbf{D}}),$ or
$(\ ^{\shortmid }\widetilde{\mathbf{g}},\ ^{\shortmid }\widetilde{\mathbf{D}}).$
Such equations can be used for constructing generic off-diagonal solutions, or for performing deformation quantization.
\end{remark}

In general, theories with gravitational field equations (\ref{meinsteqtb}) are different from theories determined by systems of type (\ref{meinsteqctb}) because the corresponding phase spaces and d-connection and d-metric structures are different. Some subclasses of solutions can be transformed mutually in equivalent forms for different theories if canonical variables with $L$--duality are introduced.

\subsection{Modified Einstein equations for pseudo Lagrange-Hamilton spaces}

Using geometric and N-adapted variational methods, we can derive gravitational field equations for different models of Lagrange and Hamilton gravity theory respectively constructed  on (co) tangent Lorentz bundles. If the geometric objects and fundamental field equations are written in canonical variables, there are obtained systems of nonlinear PDEs containing forth and higher order partial derivatives of generating functions. This makes technically impossible to find physically important exact solutions using standard analytic methods for some reduced nonlinear systems on ODEs. Nevertheless, we can construct very general classes of exact solutions in Einstein-Finsler/-Lagrange/-Hamilton theories considering nonholonomic frame transforms and distortion of connections to certain configurations with decoupling and general integrability properties. In this section, the gravitational field equations
are formulated in terms of canonical, i.e. "tilde", variables because such equations can be re-defined in straightforward forms in almost symplectic variables. Such almost K\"{a}hler models are important for geometric/deformation quantization and in order to study locally anisotropic kinetic and stochastic processes on relativistic curved spacetimes and phase space generalizations
\cite{vjmp07,vpla08,vrmp09,vjgp10,vsym13,vjmp13,vstoch96,vmon98,vapl00,vapny01,vijgmmp09}.

\subsubsection{Einstein-Finsler-Lagrange gravity}

\label{ssseflg}We consider canonical d-metric, $\widetilde{\mathbf{g}}$ (\ref{cdms}), N-connection, $\widetilde{\mathbf{N}},$ see Theorem (\ref{thcnc}), and d-connection, $\widetilde{\mathbf{D}}$ (\ref{canondcl}), structures (all defined by a Lagrange generating function $\widetilde{L},$ which can be determined by a MDR (\ref{mdrg})). Possible sources
$\widetilde{\Upsilon }_{\beta \gamma }$ are prescribed by energy-momentum d-tensors of type (\ref{generemdt}), see Assumption \ref{assumpt7}, being compatible and constructed by analogy to GR but for such canonical data. In such nonholonomic variables, the Principle-Theorem \ref{princtheinsttb} transforms into

\begin{principle}
\textbf{-Corollary} \textsf{[generalized canonical Einstein equations for Lagrange gravity] } \label{princcoroleinsttb} \newline
For the canonical data $(\widetilde{L},\widetilde{\mathbf{N}},\widetilde{\mathbf{g}},\widetilde{\mathbf{D}}),$ the gravitational field equations for the Einstein-Lagrange phase spaces are
\begin{equation}
\widetilde{\mathbf{R}}_{\alpha \beta }[\widetilde{\mathbf{D}}]=\widetilde{\Upsilon }_{\alpha \beta }.  \label{meinsteqlg}
\end{equation}
\end{principle}

Considering distortion d-tensors of type $\widetilde{\mathbf{D}}\mathbf{\rightarrow \check{D}=\widetilde{\mathbf{D}}+Z}$ determined in unique forms by $(\widetilde{\mathbf{N}},\widetilde{\mathbf{g}}),$ the equation (\ref{meinsteqlg}) can be re-defined for metric non-compatible d-connections, $\mathbf{\check{D}\widetilde{\mathbf{g}}\neq 0.}$ For instance, we can consider the Chern \cite{chern48,bao00} or Berwald \cite{berwald26,berwald41} d-connections. Such constructions result in systems of PDEs which can not be integrated in general forms and with ambiguities in introducing spinor fields on pseudo Lagrange spaces, see critics \cite{vplb10,vijmpd12}. It is more efficient to consider frame transforms correlated to distortions to the canonical d-connection, $\widetilde{\mathbf{D}}\mathbf{\rightarrow \widehat{\mathbf{D}}=\widetilde{\mathbf{D}}+ \widehat{\mathbf{Z}}}$. In result, we
prove

\begin{corollary}
\textsf{[a general integrable canonical form of locally anisotropic Einstein equations on tangent bundles ] } The modified Einstein equations (\ref{meinsteqlg}) for Einstein-Lagrange phase spaces can be written equivalently using the canonical d-connection $\widehat{\mathbf{D}},$ see equations (\ref{meinsteqtbcan}), with re-defined sources,
$\widetilde{\Upsilon }_{\beta\gamma }\rightarrow \widehat{\Upsilon }_{\ \beta \gamma }.$
\end{corollary}

We conclude that the gravitational field equations for the Einstein-Lagrange phase space geometry can be transformed into certain systems of nonlinear PDEs with decoupling properties. In particular, such transforms can be considered for pseudo Finsler spaces with $L=F^{2}$ as in Example \ref{rfgenf}.

\subsubsection{Einstein-Hamilton gravity}

The constructions for Einstein-Lagrange phase spaces can reproduced on cotangent Lorentz bundles taking canonical data for Hamilton geometries. In result, the Principle--Theorem \ref{princtheinsttb} transforms into

\begin{principle}
\textbf{-Corollary} \textsf{[generalized canonical Einstein equations for Hamilton gravity] } \label{princcoroleinstctb} For canonical data
$(\widetilde{H},\ ^{\shortmid }\widetilde{\mathbf{N}},\ ^{\shortmid }\widetilde{\mathbf{g}},\ ^{\shortmid}\widetilde{\mathbf{D}}),$ the gravitational field equations on the Einstein-Hamilton phase spaces are written
\begin{equation}
\mathbf{\ ^{\shortmid }}\widetilde{\mathbf{R}}_{\alpha \beta }[\ ^{\shortmid}%
\widetilde{\mathbf{D}}]=\ ^{\shortmid }\widetilde{\Upsilon }_{\alpha \beta}.
\label{meinsteqhg}
\end{equation}
\end{principle}

Considering frame transforms distortions to the canonical d-connection, $\ ^{\shortmid }\widetilde{\mathbf{D}}\mathbf{\rightarrow \ ^{\shortmid }\widehat{\mathbf{D}}=\ ^{\shortmid }\widetilde{\mathbf{D}}+\ ^{\shortmid }\widehat{\mathbf{Z}},}$ we prove

\begin{corollary}
\textsf{[a general integrable canonical form of locally anisotropic Einstein equations on cotangent bundles] } \newline
The modified Einstein equations (\ref{meinsteqhg}) for Einstein-Hamilton phase spaces can be written equivalently in terms of the canonical d-connection $\mathbf{\ ^{\shortmid }}\widehat{\mathbf{D}},$ see equations (\ref{meinsteqtbcand}), with re-defined sources, $\ ^{\shortmid }\widetilde{\Upsilon }_{\ \beta \gamma }\rightarrow \ ^{\shortmid }\widehat{\Upsilon }_{\ \beta \gamma }.$
\end{corollary}

We conclude that the gravitational field equations for the Einstein-Hamilton phase space geometry can be transformed into certain systems of nonlinear PDEs with decoupling properties. Examples of exact solutions are discussed in Directions 10-12, 18, 19 (subsections \ref{sssdir10}, \ref{sssdir11}, \ref{sssdir12}, \ref{sssdir18}, \ref{sssdir19}). Such solutions  have been constructed by applying the AFDM for the case Einstein-Lagrange spaces and various modifications on nonholonomic manifolds. In dual form, the method can be develped for (modified) Einstein-Hamilton spaces.

\begin{remark}
\textsf{[ (non) equivalence of Einstein-Lagrange and Einstein-Hamilton theories ] } \newline
The models of locally anisotropic gravity defined by equations (\ref{meinsteqlg}) and (\ref{meinsteqhg}) are different because they are derived on different phase spaces and for different types of d-metric and d-connection structures. Nevertheless, the geometric and physical data can be transformed equivalently from a tangent bundle to a cotangent bundle, and inversely, if a well-defined $L$-duality map (\ref{invlegendre}) is considered.
\end{remark}

Finally, we note that the locally anisotropic gravitational field equations (\ref{meinsteqhg}) (with d-metric coefficients transformed into almost symplectic structures) can be studied in deformation quantization theories as in Refs. \cite{vmon02,avjmp09}.

\subsection{ MDRs and the Einstein-Yang-Mills-Higgs equations}

We can define such equations in abstract geometric form by extending the analogy Principle \ref{panalogy} both to Lagrange densities and gravitational and matter field equations on (co) tangent Lorentz bundles. All formulas can be derived alternatively applying the Principle-Convention \ref{pactminmodact} and Consequence \ref{conseymhs}.

\subsubsection{EYMH systems on pseudo Lagrange spaces}

Modeling locally anisotropic gravitational and matter fields interactions on $\mathbf{TV}$ determined by geometric and physical data $\left( \mathbf{N,g,D};\mathbf{A}^{\check{a}}(u),\phi ^{\check{a}}(u)\right) $ with a metric compatible d-connection $\mathbf{D},$ we prove following geometric and/or N-adapted variational methods:

\begin{theorem}
\textsf{[MDR-modified EYMH equations on pseudo Lagrange phase spaces ]} \label{thseymhlagsp}\newline
Nonholonomic gravitational -- scalar - gauge fields - Higgs field interactions on tangent Lorentz bundles can be described by such a system of nonlinear PDEs:
\begin{eqnarray}
\mathbf{R}_{\alpha \beta }[\mathbf{D}] &=&\ ^{\phi }\Upsilon _{\ \alpha
\beta }+\ ^{A}\Upsilon _{\alpha \beta }+\ ^{H}\Upsilon _{\alpha \beta },
\label{eymhtb} \\
\ \mathcal{D}_{\mu }(\sqrt{|\mathbf{g}|}\mathcal{F}^{\check{a}\mu \nu })&=&%
\frac{1}{2}i\check{e} \sqrt{|\mathbf{g}|}[\overline{\phi }^{\check{a}},
\mathcal{D}^{\nu }\phi ^{\check{a}}],  \notag \\
\ \mathcal{D}^{\mu }\mathcal{D}_{\mu }(\sqrt{|\mathbf{g}|}\phi ^{\check{a}%
})&=& \check{\lambda}\sqrt{|\mathbf{g}|}|(|\phi _{\lbrack 0]}^{\check{a}%
}|^{2}- \overline{\phi }^{\check{a}}\phi ^{\check{a}})\phi ^{\check{a}},
\notag
\end{eqnarray}%
with sources determined by the energy-momentum tensors (\ref{emscdt}), (\ref{emdtym}) and (\ref{emdthiggs}).
\end{theorem}

Equations (\ref{eymhtb}) can be integrated in general off-diagonal forms for configurations with $\mathbf{D}=\widehat{\mathbf{D}}$ and, if necessary, for further restrictions to $\widehat{\mathbf{D}}_{\mid \widehat{\mathcal{T}}=0}=\nabla .$ Exact classical and quantum solutions with generalized connections and for nonholonomic configurations have been studied in a series of works
\cite{v87,vdgrg03,v94,vgon95,vhsp98,vncs01,vplbnc01,vmon98,vmon02,vmon06,vijgmmp10,vijgmmp10a}. Recently, such solutions were found for locally anisotropic interactions of gravitational fields and effective matter fields, EYMH systems in GR and string gravity, for nonholonomic quantization of gauge models of gravity etc., see \cite{vepl11,vjpcs13,vvyijgmmp14a,vacaruepjc17,bubuianucqg17} and references therein. We note also that having constructed a class of solutions for $\widehat{\mathbf{D}}$ it is possible to define nonholonomic configurations with $\widetilde{\mathbf{D}}$ generating new classes of solutions via nonholonomic transforms and distorting nonlinear and linear connections. In such case, contributions of MDRs can be computed in explicit form.

\subsubsection{EYMH systems on pseudo Hamilton spaces}

On cotangent Lorentz bundles, we consider the data $(\ ^{\shortmid }\mathbf{N},\ ^{\shortmid }\mathbf{g},\ ^{\shortmid}\mathbf{D};
\ ^{\shortmid }\mathbf{A}^{\check{a}}(u),\ ^{\shortmid }\phi ^{\check{a}}(\ ^{\shortmid }u))$. Applying the analogy Principle \ref{panalogy}, we formulate and prove
\begin{theorem}
\textsf{[MDR-modified EYMH equations on pseudo Hamilton phase spaces ]} %
\label{thseymhhamp} \newline
Nonholonomic gravitational -- scalar - gauge fields - Higgs field interactions on cotangent Lorentz bundles are defined by a system of nonlinear PDEs with "dual" geometric data:
\begin{eqnarray}
\mathbf{\ ^{\shortmid }R}_{\alpha \beta }[\mathbf{\ ^{\shortmid }D}] &=&\
_{\shortmid }^{\phi }\Upsilon _{\ \alpha \beta }+\ _{\shortmid }^{A}\Upsilon
_{\alpha \beta }+\ _{\shortmid }^{H}\Upsilon _{\alpha \beta },
\label{eymhctb} \\
\ \ ^{\shortmid } \mathcal{D}_{\mu }(\sqrt{|\mathbf{\ ^{\shortmid }g}|}%
\mathbf{\ ^{\shortmid }}\mathcal{F}^{\check{a}\mu \nu }) &=&\frac{1}{2}i%
\check{e}\sqrt{|\mathbf{\ ^{\shortmid }g}|}[\mathbf{\ ^{\shortmid }}%
\overline{\phi }^{\check{a}},\mathbf{\ ^{\shortmid }}\mathcal{D}^{\nu }%
\mathbf{\ ^{\shortmid }}\phi ^{\check{a}}],  \notag \\
\ \mathbf{\ ^{\shortmid }}\mathcal{D}^{\mu }\mathbf{\ ^{\shortmid }}\mathcal{%
D}_{\mu }(\sqrt{|\mathbf{\ ^{\shortmid }g}|}\mathbf{\ ^{\shortmid }}\phi ^{%
\check{a}}) &=&\check{\lambda}\sqrt{|\mathbf{\ ^{\shortmid }g}|}|(|\mathbf{\
^{\shortmid }}\phi _{\lbrack 0]}^{\check{a}}|^{2}-\mathbf{\ ^{\shortmid }}%
\overline{\phi }^{\check{a}}\mathbf{\ ^{\shortmid }}\phi ^{\check{a}})%
\mathbf{\ ^{\shortmid }}\phi ^{\check{a}},  \notag
\end{eqnarray}%
with sources determined by the "dual" formulas for energy-momentum tensors (\ref{emscdt}), (\ref{emdtym}) and (\ref{emdthiggs}). For simplicity, we can consider the same interaction constants for the h-, v-, and cv-subspaces.
\end{theorem}

The solutions for EYMH systems with MDRs on tangent Lorentz bundles and nonholonomic manifolds in GR and extra dimensions can be re-defined in cotangent bundle variables for $L$-dual (\ref{invlegendre}) configurations when
$\widehat{\mathbf{D}}\rightarrow \ ^{\shortmid }\widehat{\mathbf{D}}$ are correlated with respective nonholonomic transforms and distortions for $\widetilde{\mathbf{D}}\rightarrow \ ^{\shortmid }\widetilde{\mathbf{D}}.$ The symmetries of systems (\ref{eymhctb}) encode more general simplectomporphysms of d-connection structures (for Hamilton systems) which usually are not
considered for Lagrange type phase space models.

\subsection{Massive and bi-metric MGTs as Lagrange-Hamilton geometries}

We speculated on geometric principles for formulating MGTs with $f$-deformations, massive gravitons, bi-metric structures and MDRs in section \ref{ssbms}. In this subsection, we provide main Theorems on such modifications of the Einstein equations on nonholonomic (co) tangent Lorentz bundles.

\subsubsection{Generalized Einstein equations with massive graviton on tangent Lorentz bundles}

Following a N-adapted variational calculus for the action in Principle-Convention \ref{pactmassive}, we prove

\begin{theorem}
\textsf{[MDR-modified massive bi-metric gravity on tangent Lorentz bundles ]} \label{thmgtb}\newline
Nonholonomic $f$-modified and/or massive gravitational with bi-metric configurations on tangent Lorentz bundles are defined by such a system of nonlinear PDEs:
\begin{equation}
\mathbf{R}_{\alpha \beta }[\mathbf{D}]=\ ^{f\mu } \Upsilon _{\mu \nu }= \
^{g\mu }\Upsilon _{\mu \nu } + \ ^{f} \Upsilon _{\mu \nu}+\ ^{m}\Upsilon
_{\mu \nu },  \label{mgmassivetb}
\end{equation}%
with sources determined by (\ref{smgb}).
\end{theorem}

The equations (\ref{mgmassivetb}) can be integrated following the AFDM for configurations with $\mathbf{D}=\widehat{\mathbf{D}}$ and nonholonomic constraints for extracting LC-configurations, $\widehat{\mathbf{D}}_{\mid \widehat{\mathcal{T}}=0}=\nabla .$ A series of locally anisotropic wormhole and cosmological solutions were constructed and analyzed in Refs.
\cite{vijgmmp14,vepjc14,vepjc14a,vacaruplb16}. The right side of (\ref{mgmassivetb}) can be completed with additional sources EYMH systems and/or modifications from string gravity, for nonholonomic quantization of gauge models of gravity etc. like we generated cosmological and black hole like solutions in \cite{vepl11,vjpcs13,vvyijgmmp14a,vacaruepjc17,bubuianucqg17} and references therein. Such generalizations can be defined following Assumption \ref{assumpt4} when the contributions of gauge and Higgs fields can be approximated with respect to N-adapted frames to certain effective cosmological constants. In general frame/coordinate systems, such solutions are generic off-diagonal and for generalized connections with coefficients depending, in principle, on all phase space coordinated. In order to compute contributions of MDRs, we have to re-define the physical equations and solutions for $\widetilde{\mathbf{D}}$--configurations.

\subsubsection{Massive gravitational equations on cotangent Lorentz bundles}

On cotangent Lorentz bundles, the main result for $f$--modified massive/ bi-metric gravity is stated by
\begin{theorem}
\textsf{[MDR-modified massive bi-metric gravity on cotangent Lorentz bundles] } \label{thmgtbd}\newline
Nonholonomic $f$-modified and/or massive gravitational with bi-metric configurations on tangent Lorentz bundles are defined by such a system of nonlinear PDEs:
\begin{equation}
\mathbf{\ ^{\shortmid }R}_{\alpha \beta }[\mathbf{\ ^{\shortmid }D}]=\
_{\shortmid }^{f\mu }\mathbf{\Upsilon }_{\mu \nu }=\ _{\shortmid }^{g\mu }
\mathbf{\Upsilon }_{\mu \nu }~+~_{\shortmid }^{f}\mathbf{\Upsilon }_{\mu \nu
}+\mathbf{\ }_{\shortmid }^{m}\mathbf{\Upsilon }_{\mu \nu },
\label{mgmassivetbd}
\end{equation}%
with sources determined by (\ref{smgcb}).
\end{theorem}

Examples of exact solutions for the system (\ref{mgmassivetbd}) are discussed in Directions 10, 18, 19 (subsections \ref{sssdir10},
\ref{sssdir18}, \ref{sssdir19}). Such solutions exist at least as $L$-dual ones which have been constructed for (\ref{mgmassivetb}).

\subsection{Short-range locally anisotropic gravity}

In this section, we formulate the gravitational filed equations for short-range gravitational toy models determined by MDRs with LIVs on (co) tangent Lorentz bundles with total phase spaces of 3+3 dimension. Such models are constructed following the Principle-Convention \ref{princsr}.

\subsubsection{Field equations for short-range gravity with MDRs}

Following a N-adapted variational calculus of actions (\ref{lagrdshortord}) on tangent bundle, we prove

\begin{theorem}
\textsf{[Gravitational field equations for short-range gravity with MDR and LIV on tangent bundles ] } \label{thshrgtb}\newline
Nonholonomic short-range gravitational equations on tangent bundles determined by MDRs and LIVs can be written in the form
\begin{equation}
\mathbf{R}_{\beta \delta }[\mathbf{D}]=\ ^{sr}\mathbf{\Upsilon }_{\beta
\delta }=\ ^{m}\mathbf{\Upsilon }_{\beta \delta }+\ ^{lv}\mathbf{\Upsilon }
_{\beta \delta },  \label{mgshrtb}
\end{equation}%
with sources determined by energy-momentum d-tensors (\ref{emdthsr}).
\end{theorem}

Examples of exact solutions for the system (\ref{mgshrtb}) for $\mathbf{D}=\widehat{\mathbf{D}}$ will be provided in our future partner works. Such solutions consist some particular 6-d examples (with short range effective source and associated cosmological constants, see Assumption \ref{assumpt6}) of 8-d and 10-d solutions constructed in Refs.  \cite{gvvepjc14,vacaruepjc17,bubuianucqg17}.

\subsubsection{Field equations for short-range co-gravity with MDRs}

Following geometric methods (considering the dual of the short range gravity equations from tangent to cotangent bundles) and N-adapted variational calculus of actions (\ref{lagrdshortord}) on cotangent bundle, we prove
\begin{theorem}
\textsf{[Gravitational field equations for short-range gravity with MDR and LIV on cotangent bundles ] } \label{thshrgtbd}\newline
Nonholonomic short-range gravitational equations on cotangent bundles determined by MDRs and LIVs can be written in the form
\begin{equation}
\mathbf{\ ^{\shortmid }R}_{\alpha \beta }[\mathbf{\ ^{\shortmid }D}]=\
_{\shortmid }^{f\mu }\mathbf{\Upsilon }_{\mu \nu }=\ _{\shortmid }^{g\mu }%
\mathbf{\Upsilon }_{\mu \nu }~+~_{\shortmid }^{f}\mathbf{\Upsilon }_{\mu \nu
}+\mathbf{\ }_{\shortmid }^{m}\mathbf{\Upsilon }_{\mu \nu },
\label{mgshrtbd}
\end{equation}%
with sources determined by energy-momentum d-tensors (\ref{emdthsr}) with labels "$\ ^{\shortmid }$".
\end{theorem}

The systems (\ref{mgshrtb}) and/or (\ref{mgshrtbd}) can be extended to 8-d and 10-d phase spaces. Such black hole and cosmological type solutions (in general, with nontrivial torsion) in Directions 18 and 19 (subsections \ref{sssdir18} and \ref{sssdir19}).  Using nonholonomic frame transforms and distortions of d-connections, the short-range gravitational equations and respective exact/ parametric solutions can be re-defined equivalently as Finsler-Lagrange-Hamilton systems of higher order. Geometrically, such
generic off-diagonal equations are equivalent, or can be imbedded, into certain classes of physically important solutions constructed in \cite{vepl11,vjpcs13,vvyijgmmp14a,vacaruepjc17,bubuianucqg17} and references therein.

\subsection{Towards axiomatic formulation of MGTs with MDRs on (co) tangent bundles}

A constructive--axiomatic approach to GR was proposed in 1964 by J. Ehlers, F. A. E. Pirani and A. Schild \cite{ehlers72} (the system of so-called EPS axioms). It was completed with a series of conventions and results on physically important exact solutions and fundamental singularity and topological censorship theorems in gravity theories (for reviews, see monographs \cite{hawking73,misner73,wald82,kramer03}). The concept of EPS spacetime, as a physically motivated model of spacetime geometry and \textsf{locally isotropic gravitational interactions}, was studied in a series of publications in the early 1970's, see original results and references in \cite{pirani73,woodhouse73,perlick87,meister90}. That EPS axiomatic approach led to an "orthodox" belief that the underlying spacetime geometry can be only pseudo--Riemannian. It resulted in the paradigmatic concept of "Lorentzian 4--manifold" described by a metric structure $g$ with pseudo-Euclidean signature $\left( +++-\right) $ and a corresponding unique
metric compatible and torsionless (Levi Civita) LC-connection $\nabla$. It was postulated that the geometric data $(g,\nabla)$ are determined by a well-defined physical solution of the Einstein equations. For gravitational and matter field interactions, it was imposed a local causality condition that there are preserved locally all postulates of the special relativity theory, SRT.

All constructions in GR (on finite, or infinite, spacetime regions) are based on a well-defined physical principle of propagation of light along null geodesics. This way, a causal structure for massive particles and fields can be established if the speed of interactions is smaller than the maximal speed of light. The EPS axiomatic formalism is not valid, and has to be revised, for MGTs developed with the aim to explain acceleration and dark matter and dark energy in cosmology. Here we note that the ESP
axiomatic approach has to be also extended in order to include in the scheme QG theories with MDRs and LIVs effects.

In this subsection, we conclude on how the axiomatic EPS scheme can be modified and generalized for Finsler like gravity theories. Such an approach was proposed in Refs. \cite{vjpcs11,vijmpd12} for locally anisotropic spacetimes modeled on nonholonomic manifolds and vector/tangent bundles and can be performed in similar forms for MGTs with MDRs constructed on cotangent Lorentz bundles. In previous two sections, we have shown that the Lorentzian manifold geometric formulation of GR can be extended on $\mathbf{TV}$ and $\mathbf{T}^{\ast }\mathbf{V}$ for geometric data (\ref{phspgd}) defining models of generalized Einstein-Finsler/-Lagrange /-Hamilton gravity on phase spaces with MDRs and LIVs. We proved that using geometric and physical data for a large class of indicator functionals and respective canonical nonholonomic variables, $\varpi (x^{i},E,\overrightarrow  {\mathbf{p}},m;
 \ell _{P})\rightarrow (\ ^{\shortmid }\mathbf{g,\ ^{\shortmid }N},\ ^{\shortmid} \widehat{\mathbf{D}}) \leftrightarrows
 (H:\ ^{\shortmid }\widetilde{\mathbf{g}}, \ ^{\shortmid} \widetilde{\mathbf{N}}, \ ^{\shortmid}\widetilde{\mathbf{D}})  \leftrightarrow \lbrack (\ ^{\shortmid}\mathbf{g}[\ ^{\shortmid }N], \ ^{\shortmid }\nabla)] $ is possible to elaborate well-defined physical models in certain forms which are very similar to GR. In a more general context, the approach allows us to include in such a scheme various classes of MGTs with MDRs. There were considered models with multi-connection/-metric structure and motion equations formulated as effective Einstein-Yang-Mills-Higgs equations determined by equivalent geometric data $(\mathbf{g,N},\widehat{\mathbf{D}})$ and/or $(\ ^{\shortmid}\mathbf{g,\ ^{\shortmid}N},\ ^{\shortmid} \widehat{\mathbf{D}})$. We argue that such gravitational and matter field equations possess certain generalized nonlinear symmetries with respect to redefinition of (effective) sources and distortion of distinguished (Finsler like) connections and fundamental geometric data. 

Perhaps, the first attempt to formulate an axiomatic approach to Finsler gravity theories was considered in the early 80ths in the former URSS by I. Pimenov \cite{pimenov87} in his habilitation thesis. In parallel, it was proposed also a minimal set of axioms for Finsler geometry due to M. Matsumoto, see a summary of work in \cite{matsumoto86}. Those works were not devoted to theories with multi-connection and multi-metric structures and had not studied how a corresponding Finsler axiomatic formalism should be
related to the set of EPS axioms for GR. Recently, we proved \cite{vjpcs11,vijmpd12} that the EPS approach can be generalized for a class of so-called Einstein-Finsler/ -Lagrange gravity theories. That geometric scheme was elaborated for nonholonomic deformations of fundamental geometric objects, $(g,\nabla)\rightarrow (\mathbf{g,N}, \widehat{\mathbf{D}}) \leftrightarrows
(L: \widetilde{\mathbf{g}},\widetilde{\mathbf{N}},\widehat{\mathbf{D}}),$ working on nonholonomic manifolds, or tangent bundles, with N-adapted splitting of total spacetime dimensions, $n+m=2+1,2+2,2+2+1,...$. Such locally anisotropic spacetimes (phase space gravity theories) are endowed with metrics of local pseudo-Euclidean signature and generalized Finsler connections, when probing point mass particles and light follow along semi-spray equations stating a causal structure at least for small nonholonomic deformations. Corresponding geometric methods and physical models were elaborated and studied in our monographs and articles on locally anisotropic black ellipsoid, wormhole, cosmological and other type solutions, see \cite
{vexsol98,vjhep01,vijgmmp07,vijtp10,vijtp10a,vijgmmp11,vjpcs13,vepjc14,vvyijgmmp14a,gvvepjc14,vmon98,vmon02,vmon06,vrev08,vplb10}
and references therein. Various issues on geometrization and axiomatic formulation of locally anisotropic relativistic/ supersymmetric/ noncommutative/ quantum deformed/ string like/ gauge like and other type theories (including corresponding types of MDRs) were studied in \cite{vjmp96,vap97,vnp97,vhsp98,svcqg13,svvijmpd14}, see recent results
\cite{vacaruplb16,gheorghiuap16,vacaruepjc17,ruchinepjc17,bubuianucqg17}.

A set of assumptions, conventions and principles which could be formalized into a complete system of axioms should be formulated for a generalized Finsler like theory (as a generalized Lagrange and/or Hamilton type) must be verified by experimental and observational data. This is a task for future research and theoretical constructions. In next subsection, we summarize key points and speculate only on structural blocks which have to be elaborated in detail for respective geometric models on tangent and/or cotangent Lorentz bundles: 

\subsubsection{Assumptions on metrics, nonlinear quadratic line elements and d-metrics}

Let us discuss the assumptions considered in subsection \ref{ssassumpts}:

\vskip5pt

\textbf{Assumption} \ref{assumptqelorentz} \textsf{[Background quadratic elements on total spaces of (co) tangent bundles] } is important for theories with MDRs and LIVs which have a limit to a base Lorentz manifold with a typical Minkowski fiber for a corresponding tangent bundle.\footnote{It is not obligatory defined in a smooth form and one could be considered nonholonomic constraints and various singular structures.} Using nonholonomic frame transforms and deformations of geometric structures, such
background quadratic elements are transformed into general ones on (co) tangent Lorentz bundles. Intuitively, this assumption allows us to elaborate a causality structure on total phase space (at least for small nonholonomic deformations).

\vskip4pt

\textbf{Assumption} \ref{assumptnonlinq} \textsf{[nonlinear quadratic elements for modeling Finsler-Lagrange-Hamilton geometries on (co) tangent bundles ] } states a natural relation between the indicators of MDRs and respective generalized Finsler generating functions. For MDRs on cotangent Lorentz bundles, we can parameterize the construction in a form when the effective Hamiltonian can be considered as a generating function. Considering necessary type nonholonomic variables, we construct Legendre transforms to effective Lagrange generating functions. Relativistic and non-relativistic Finsler metrics can be considered for some particular cases when certain homogeneity conditions are imposed. It should be emphasized that such nonlinear quadratic line elements can be transformed into (pseudo) Riemannian ones only for very special "quadratic" configurations (as in the B. Riemann habilitation thesis \cite{riem1854}). We can consider some "simplified" examples of Randers type geometries, with applications in gravity and cosmology, see for instance, \cite{ingarden04,ingarden08a,basilakos13,silva14,silva15,silva16}. Nevertheless, a general assumption relating indicators with generating functions and nonlinear quadratic line elements is important for formulating an axiomatic
approach to theories with MDRs.

\vskip4pt

\textbf{Assumption} \ref{assumpt3} \textsf{[d-metrics on (co) tangent Lorentz bundles ] } is motivated for physical theories geometrized with linear quadratic elements when nonlinear contributions are considered for certain small parameters, nonholonomic constraints and certain linearizable symmetries. We are not able to construct self-consistent and viable physically theories working only with nonlinear generating functions. Sasaki lifts of geometric objects from the base to total spaces (see details in
monograph \cite{yano73}) allows us to define a total metric structure completely defined by a generating (for instance, Finsler) function. More than that, we can perform such a construction in N-adapted form and consider equivalent variants for nonholonomic Klauza-Klein theories. In result, the concept of d-metric allows us to encode geometrically contributions for MDRs into generating functions and respective quadratic elements in frame/coordinate free forms.

\vskip4pt

Above assumptions should be completed with additional ones for Clifford bundles with N-connections and spinor metric structures if MDRs are considered for spinor fields and generalized Dirac operators \cite{vjmp96,vhsp98,vmon98,vmon02,vvicol04,vjmp06,vjmp09}.

\subsubsection{Geometric and physical principles for Finsler-Lagrange-Hamilton gravity extensions}

Principles \ref{pgpeq} - \ref{pgcov} formulated in subsection \ref{ssmpext} consist a set of main principles for extending GR to
Finsler-Lagrange-Hamilton theories. We should add the Principle \ref{panalogy} if it is supposed that the locally anisotropic interactions of matter fields are modelled by Lagrange densities for a gravity theory with MDRs when the vacuum Finsler like gravitational fields are supposed to be described in a variational form (with minimal action).

\vskip4pt

\textbf{Principle }\ref{pgpeq} \textbf{\ [modified equivalence principle] } states that we can extend the fundamental equivalence principle for GR in self-consistent forms for certain classes of generalized Finsler theories. The constructions on (co) tangent Lorentz bundles are natural ones and allow "simplest causal" realizations of theories with MDRs. In such locally anisotropic spacetimes, free propagating particles and small perturbations of fields do not only follow certain (nonlinear) geodesic equations. They are described additionally by certain autoparallel equations which should be analyzed together with experiments and observational data for determining which linear connection is more appropriate for covariant derivatives. One can be constructed MGTs with multi-metric and multi-connection structures when a linear connection can be considered for one type of gravitational-matter filed interactions and another type of linear connections taken for other types of matter field equations. Some metric structures can be considered as "hidden / un-physical" ones but one of them must be defined as the physical one. This allows us to define a Levi-Civita connection which should be completed with distorting tensors for respective linear connections and nonholonomic frame effects. A modified equivalence principle can be formulated and verified in a "minimal form" for metric
compatible linear connections which are adapted to certain physically motivated N-connection structures. In such a case, a causality structure can be established in a self-consistent form. It is not clear how such a principle could be formulated for general metric noncompatible connections. Nevertheless, if such a noncompatible connection is described by a unique distortion tensor from, for instance, the Levi-Civita connection, a modified equivalence principle can be reconsidered for certain congruences of special classes of curves determined by distortion relations.

\vskip4pt

\textbf{Principle }\ref{pgmp}\textbf{\ [generalized Mach principle] } can be realized in more complex forms on generalized phase spaces modelled on (co) tangent Lorentz bundles. The geometric structure of such nonholonomic spaces is more richer and encodes more sophisticate correlations of spacetime properties and general locally anisotropic interactions with various noncommutative, stochastic, fractional configurations etc. In principle, we elaborate on certain models of generalized spacetime and phase space aether for different types of MGTs. The approach with nonholonomic bundles and manifolds allows us to formulate theories with MDRs and LIVs. For such theories, a generalized Mach principle can be verified by local and global observational and experimental data.

\vskip4pt

\textbf{Principle}\ref{pgcov}\textbf{\ [general covariance and equivalent nonholonomic geometrizations of gravitational theories]}  states explicitly that the formulation of geometric models and physical theories should be performed in general frame and coordinate free forms. The nonholonomic and multi-connection structures of such locally anisotropic spacetime and phase spaces involve more sophisticated geometric constructions. If there are distortion relations for linear connections uniquely defined by the metric and N-connection structure, all physically viable theories must admit an equivalent reformulation for any admissible and well-defined linear connection structure. Certain experimental/ observational data may allow us choosing a physical connection for certain classes of matter filed interactions. Nevertheless, we shall be always able to distort in an equivalent form the physical data if a unique distortion connection structure can be defined. This is a consequence of the fact that theories with MDRs and LIVs are naturally geometrized on spaces with three fundamental geometric objects (N-connection, d-metric and d-connection) which is very different from (pseudo) Riemannian geometry determined completely by the metric structure.

\vskip4pt

\textbf{Principle } \ref{panalogy} \textbf{\ [analogy principle for Lagrange densities] } has to be considered additionally to the previous ones if we attempt to construct variational physical models. For instance, we can substitute the LC-connection on a Lorentzian manifold by a metric compatible d-connection in the total space for all respective Lagrangians in GR. This way the constructions are naturally extended on (co) tangent Lorentz bundles and admit a self-consistent causal formulation.

\vskip4pt

Principle \ref{panalogy} has to re-formulated if spinor fields are introduced into consideration, see discussions for Direction 3 (i.e. subsection \ref{sssdir403}).

\subsubsection{Conventions on Lagrangians and energy-momentum tensors on (co) tangent bundles}

The conventions, principles and assumptions formulated in subsection \ref{sslagrd} are necessary for constructing models with locally anisotropic gravitational and matter field interactions and MDRs in N-adapted variational form. An explicit formulation of the main principles depends on the type of gravity and matter field theory we elaborate. For certain classes of canonical d-connections, such models can be elaborated to be  exactly integrable in some general forms and/or quantized following methods of geometric/ deformation quantization.

\paragraph{\textit{\qquad On scalar field interactions and MDRs: }}

{\ \newline
}\textbf{Convention }\ref{convscfields} \textsf{[scalar fields with MDRs on (co) tangent bundles]} is for models of gravitational and matter field interactions with minimal coupling but extended on generalized phase spaces. In a similar form, we can postulate Lagrange densities with non-minimal coupling, complex scalar fields etc. when corresponding Lagrangians for MGTs on some manifolds are lifted by using Sasaki lifts with N--connections and d-metrics on total spaces of (co) tangent bundles.

\paragraph{\textit{\qquad EYMH systems with MDRs on cotangent bundles:}}

{\ \newline
} \textbf{Convention } \ref{convcovder} \textsf{[covariant derivatives of gauge fields in (co) tangent bundles] } states that a metric compatible d--connection on a phase space can be elongated by a gauge field potential for describing interactions with gauge fields in locally anisotropic spacetimes. In a simplest approach, we can consider canonical d-connections on (co) tangent bundles which allows us to integrate in some general forms respective gravitational - gauge filed equations with MDRs.

\vskip4pt

\textbf{Convention } \ref{convldymh} \textsf{[Lagrange densities for YMH interactions on (co) tangent Lorentz bundles]} is necessary for formulating Yang-Mills and Higgs theories with MDRs. The constructions are similar to those in GR but for respective metric compatible d-connection structures nonholonomically extended on total (co) bundle spaces.

\vskip4pt

\textbf{Principle}-\textbf{Convention} \ref{pactminmodact} \textbf{[actions for minimally MDR-modified EYMH systems]} defines a model of locally anisotropic EYMH theory with minimal coupling on generalized Finsler-Lagrange-Hamilton spaces. MDRs and LIVs are encoded in respective generating functions via corresponding indicators. In a similar way, we can construct such theories to be  with non-minimal coupling, for instance, in certain limits of (super) string theory.

\vskip4pt

\textbf{Assumption} \ref{assumpt4} \textsf{[effective cosmological constants for locally anisotropic YMH sources]} concerns both the possibility to construct self-dual gauge like configurations and to encode such configurations into an effective cosmological constant. For generalized Finsler spacetimes, this assumption allows us to consider a new class of nonlinear symmetries for generating functions, effective matter sources and effective cosmological constants.

\paragraph{\textit{\qquad Massive gravity theories with MDRs:}}

{\ \newline
} \textbf{Principle}-\textbf{Convention} \ref{pactmassive}\textbf{\ [actions for MDR-modified massive gravity theories ]} is necessary for formulating massive gravity theories with MDRs and LIVs on (co) tangent Lorentz bundles. In particular, we can model a new class of locally anisotropic theories which similar GR on base spacetime but extended to effective massive gravity on phase space.

\vskip4pt

\textbf{Assumption} \ref{assumpt5} \textsf{[effective cosmological constants in massive gravity on (co) tangent bundles ]} reflects the existence of a new class of nonlinear symmetries for generating functions and effective sources and cosmological constants.

\paragraph{\textit{\qquad MDR-modified short-range gravity: }}

\textit{\qquad }{\ \newline
} \textbf{Principle}-\textbf{Convention} \ref{princsr} \textbf{\ [Lagrange densities for MDR-modified short-range gravity and LIV]} is important for formulating a self-consistent extension of such short-range theories on (co) tangent bundles. This allows us to construct exact solutions in corresponding canonical variables by applying the AFDM.

\vskip4pt

\textbf{Assumption} \ref{assumpt6} \textsf{[effective cosmological constants for short-range locally anisotropic interactions ]} is based on nonlinear symmetries of generating functions and effective matter fields and cosmological constants which can introduces for a large class of theories.

\paragraph{\textit{\qquad Conservation laws (effective) sources and MDRs: }}

{\ \newline}  The equations of motion and nonholonomic conservation laws for (effective) sources are discussed in section \ref{ssecnhcons}. Additionally, we conclude:

\vskip4pt

\textbf{Assumption} \ref{assumpt7} \textsf{[energy-momentum d-tensors for (non) massive gravity theories on (co) tangent Lorentz bundles]} allows us to define such values as in GR and MGTs on metric--affine manifolds but redefining the covariant derivatives into respective ones on (co) tangent bundles.

\vskip5pt

\textbf{Principle} \ref{prinndefclaw} \textsf{[nonholonomic deformations by MDRs of conservation laws on (co) tangent Lorentz bundles]} reflects the properties that nonholonomic deformations and distortions result in theories for which the covariant derivative of the energy-momentum tensor is not zero. This is typical for nonholonomic systems, for instance, in nonholonomic mechanics when the variational problem is solved by using additional Lagrange multiples. In gravity theories constructed on
nonholonomic manifolds and/or (co) tangent bundles, this issue is encoded into generalized Bianchi identities when distortions of covariant derivatives result in unique effective sources for nonzero covariant derivatives of energy momentum tensor. We can define distortions to LC-configurations resulting in effective sources which allows us to define conservation laws with covariant derivatives as in GR but on (co) tangent Lorentz bundles.

\subsubsection{Principles and main theorems for modified field equations with MDRs}

A series of principles stated in section \ref{smodensteq} can be formulated and derived as some main theorems which can be proven by respective N-adapted variational calculus for respective actions and Lagrange densities outlined in previous subsection. Alternatively, such gravity theories can be formulated following only geometric principles like in the original A. Einstein works, see details monograph \cite{misner73}. For "pure" geometric constructions, the Principles-Theorems / -Corollaries and certain important Theorems discussed below can be considered as postulates for respective modified gravitational equations.

\vskip4pt

\textbf{Principle-Theorem} \ref{princtheinsttb} \textsf{[nonholonomic modifications of Einstein equations on tangent Lorentz bundles]} and \textbf{Principle-Theorem} \ref{princtheinstctb} \textsf{[nonholonomic modifications of Einstein equations on cotangent Lorentz bundles]} can be formulated for the gravitational field equations on corresponding (co) tangent bundles. In
general, such theories can be constructed in terms of triples (N-connection, d-connection and d-metric) without any assumptions on the properties of indicators of MDRs and/or related generating functions. The monographs \cite{vmon98,vmon02,vmon06} and review \cite{vrev08} summarize such constructions for generalized Finsler and metric-affine spaces endowed with N-connection structure.

\vskip4pt

\textbf{Principle-Corollary} \ref{princcoroleinsttb} \textsf{[generalized canonical Einstein equations for Lagrange gravity]} and \textbf{Principle-Corollary } \ref{princcoroleinstctb} \textsf{[generalized canonical Einstein equations for Hamilton gravity]} are considered for canonical Lagrange and/or Hamilton variables which allow almost symplectic formulations. Such systems of nonlinear PDEs are equivalent if nonsingular Legendre transforms are defined. In general, such theories are not equivalent because they are described by different types of symmetries, different almost symplectic forms and symplectomorphysms, and non-equivalent
d-connections. They have different classes of solutions and/or noncommutative/nonassociative and supersymmetric generalizations.

\vskip4pt

\textbf{Theorem} \ref{thseymhlagsp} \textsf{[MDR--modified EYMH equations on pseudo Lagrange phase spaces]} and \textbf{Theorem} \ref{thseymhhamp} \textsf{[MDR--modified EYMH equations on pseudo Hamilton phase spaces]} show how we can formulate generalizations of EYMH theories with MDRs on pseudo Lagrange and/or Hamilton phase spaces. The constructions can be performed in canonical nonholonomic variables admitting decoupling and integration of fundamental field equations; using almost symplectic forms which is more suitable for deformation quantization; or in locally anisotropic spinor variables which allow twistor formulations and twistor quantization \cite{v87,vmon98,vmon02,vmon06}.

\vskip4pt

\textbf{Theorem} \ref{thmgtb} \textsf{[MDR--modified massive bi-metric gravity on tangent Lorentz bundles]} and \textbf{Theorem} \ref{thmgtbd} \textsf{[MDR--modified massive bi-metric gravity on cotangent Lorentz bundles]} derive generalized Einstein equations for massive bi-metric gravity theories with MDRs on phase spaces. Such theories and exact solutions have been studied recently in a series of works \cite{gvvepjc14,strauss12,vijgmmp14,vepjc14,vepjc14a,cai14,kluson13,derham10,derham11,sebastiani16,vacaruplb16,volkov12}. In this paper, the gravitational field equations are proven originally for bi-metric and connection structures depending on momentum type coordinates.

\vskip4pt

\textbf{Theorem} \ref{thshrgtb} \textsf{[Gravitational field equations for short-range gravity with MDR and LIV on tangent bundles]} and \textbf{Theorem} \ref{thshrgtbd} \textsf{[Gravitational field equations for short-range gravity with MDR and LIV on cotangent bundles]} provide proofs how short-range gravity theories with MDRs \cite{bailey15} can be geometrized on (co) tangent Lorentz bundles.

\vskip4pt

Principles analyzed in this subsection have to be extended for new classes of nonholonomic distributions with N- and d-connections if spinor fields and twistors are introduced into consideration \cite{vaaca15,v87,vmon02}. It is possible to elaborate a N-adapted Palatini formalism \cite{misner73}, quasi-classical extensions, noncommutative and nonassociative generalizations, deformation quantization, gauge like gravity theories etc., see Refs. \cite{castro07,castro08,castro08a,castro14,castro16,
vplbnc01,vjmp96,vhsp98,vpcqg01,vtnpb02,vvicol04,vjmp05, vjmp06,vjmp09,vch2416,dvgrg03,vgon95,vncs01,vepl11}.

\section{Conclusions, Achievements and Perspectives}

\label{scd} In this work, we have presented a unified geometric perspective on modified gravity theories, MGTs, with modified dispersion relations, MDRs, encoding Lorentz invariance violations, LIVs. We omitted technical aspects but attempted to formulate an axiomatic background and elaborate on new geometric methods, and discuss solutions to conceptual and fundamental problems in such ways we understand them. As this article is not intended to be a complete  review of generalized Finsler geometries and applications we were able to mention only the sources which are most relevant to our research interests. We discussed the most important fundamental contributions of various international schools and individual researchers, provided main bibliographic sources but we should note that  citations on possible applications are not comprehensive. This work also contains original results on geometric, variational,  and axiomatic formulations of a series of relativistic models of Hamilton-Lagrange geometry and commutative MGTs elaborated on (co) tangent Lorentz bundles. We conclude with relevant historical remarks which are important for understanding the logic and motivations for developing such directions in modern geometry and physics (see details in Appendix \ref{ashr} and references therein), a summary of the main achievements and results, and a discussion of open problems and
perspectives.

\subsection{On geometric models of MGTs with MDRs and generalized relativistic Finsler spaces}

As we have shown in section \ref{mgttmt}, a large class of MGTs can be geometrized as models of relativistic Finsler-Lagrange-Hamilton spacetimes elaborated on nonholonomic (co) tangent Lorentz bundles/manifolds. The fundamental geometric objects (defined as d-metrics, N-connections, and d-connections) in such generalized phase spacetimes are with generic dependence on "velocity or momentum" type coordinates and determined by respective generating functions. The first examples of such nonlinear quadratic linear elements (i.e. Finsler metrics which can be modelled as some particular  examples with homogeneous Lagrange generating functions (\ref{nqe}), see also Example \ref{rfgenf}) was considered in the famous habilitation thesis (B. Riemann, 1854) \cite{riem1854}. That author emphasized that for simplicity his research was restricted only to geometries with linear quadratic elements which allowed an axiomatic formulation for Riemann geometry and, later, for pseudo-Riemannian (Lorentz spacetime) geometry.

The term \textsf{Finsler geometry} was introduced due to (E. Cartan, 1935), see \cite{cartan35} and references therein. That monograph was devoted to further developments of locally anisotropic geometric models proposed in (P. Finsler, 1918) \cite{finsler18}. E. Cartan contributions in elaborating the fundamentals of Finsler geometry were crucial. He understood that such a geometry is determined not only by a (nonlinear) quadratic line element with certain homogeneity conditions but that there are also two other important fundamental geometric objects. It was used in a coefficient form a new geometric object later called the N-connection (\ref{ncon}) and introduced in Finsler geometry a metric compatible and N-adapted distinguished connection (the so-called Cartan d-connection). A series of fundamental books \cite{cartan35,cartan38,cartan63} with definition of bundle spaces,
spinors, torsion fields and Riemann-Cartan geometry, Pfaff forms for systems of partial differential equations, PDEs, etc. were translated in many languages and had a substantial influence for thousands of researchers on developments of Finsler geometry by research schools in France, Germany, the USA, Czech Republic, Romania, Hungary, Poland, Russia, Japan, China, India, Greece, Iran, Egypt etc.

We mention here certain very important contributions to Finsler geometry due to (L. Berwald, 1926, 1941) \cite{berwald26,berwald41}. He introduced new types of N- and d-connection structures completely determined by a semi-spray generating functions. His d-connection is not metric compatible but also admits a semi-spray geometric formulation in terms of certain nonlinear geodesic congruences. Here, we note that the concept of N-connection is equivalent to that of (C. Ehresmann, 1955) \cite{ehresmann55}. A rigorous and original study of N-connections in Finsler geometry was performed due to (A. Kawaguchi, 1937, 1952) \cite{kawaguchi37,kawaguchi52}.

A very important concept for geometrizing theories with locally anisotropic interactions and/or MDRs is that of \textsf{nonholonomic manifold}. Such a space is defined as a real or complex (super) manifold with prescribed nonintegrable distributions, in particular, with N-connections, which are very important in Finsler geometry and generalizations. In certain sense, any geometric and/or physical model involve nonholonomic structures/constraints/frames and correspondingly adapted geometric objects and formulas. The approach was elaborated by (G. Vr\v{a}nceanu 1931, 1957) \cite{vranc31,vranc57}, see references therein and \cite{bejancu03} on relevant papers beginning 1926)\footnote{\label{fvranceanu} Historically, it is interesting to note that G. Vr\v{a}nceanu was taught mathematics at the University of Ia\c{s}i (the first University in Romania) and studied at G\"{o}ttingen under D. Hilbert; then at Rome under Levi-Civita; and at Paris, where he worked with E. Cartan. His PhD was validated by an examining board consisted of 11 professors headed by V. Volterra. He was awarded a Rockefeller scholarship to study in France and the USA. In 1929, G. Vr\v{a}nceanu was appointed as a professor at the Cernau\c{t}i University which was one of the most important (at that time)  in Romania, where he became one of the leading geometers (at that time) in the World.  At present, his name is almost not known at the Chernivtsy State University, Ukraine. During his career as an university professor and leading member of the Romanian Academy of Sciences, Vr\v{a}nceanu published more than 300 articles in all branches of modern geometry (till the end 1990).}. We cite also an independent research and similar results due to (Z. Horak, 1927) \cite{horac27}. The geometry of nonholonomic manifolds may encode and/or model via respective classes of distributions various complex / noncommutative / supersymmetric Finsler, Lagrange, Hamilton and other type nonholonomic spaces. Certain methods of the geometry of nonholonomic manifold (generalizing also various constructions in nonholonomic mechanics and classical and quantum field theories with nonintergrable constraints) were very important for elaborating the AFDM \cite{vmon98,vmon02,vmon06,vrev08,vplb10} for generating exact and parametric solutions in MGTs.

Then, a very important contribution in Finsler geometry was due to (S. -S. Chern, 1948) \cite{chern48}. He introduced a new type of d-connection which is N-adapted, with zero torsion, and coincides with the Levi-Civita connection on the base manifold. But the Chern d-connection is metric non-compatible on the total tangent bundle. Here we note that there is also a Chern connection for complex manifolds which is different from that in Finsler geometry. In certain sense, Finsler-Chern spaces consist examples of the Weyl geometry \cite{weyl29} which is defined, in general, with metric-affine noncompatible metric and linear connection structures. In a detailed form with various applications and noncommutative modifications, Finsler-/ Lagrange- / Hamilton - affine geometric and physical theories were studied both on nonholonomic manifolds and (co) tangent bundles endowed with generalized Finsler like metric and N-connection structure, see monograph \cite{vmon06}. Chern's d-connection was re-discovered by (H. Rund, 1959) \cite{rund59}. That monograph played a very important international role in the education of a new generation of researchers on Finsler geometry and applications. Later (after a number of Finsler books appeared in the former USSR, Romania and R. Moldova, due to G. Asanov, A. Bejancu, R. Miron and co-authors, S. Vacaru), it was published a series of important geometric books developing the Finsler-Chern geometry, see (D. Bao, S.-S. Chern, Z. Zhen, 2000) \cite{bao00} and (Z. Shen, 2001) \cite{shen01} and references therein. Those three monographs, and mentioned above E. Cartan and H. Rund books, are most cited and known books on Finsler geometry. Unfortunately, the style, methods and results of such works written by mathematicians, and respective geometric methods, are quite sophisticated for applications in modern physics and elaborating MGTs admitting limits to standard particle theories and GR. We had to elaborate our "metric compatible" approach to relativistic generalized Finsler theories, locally anisotropic gravity and cosmology as it was performed in  \cite{vmon98,vmon02,vmon06,vrev08,vplb10,vjmp07,vpla08,vijgmmp09,vjgp10,bvnd11,vjmp13,vmjm15,vch2416,vmon02,avjmp09}.

Experts in geometry and physics know Chern's approach to Finsler geometry when the quadratic conditions are not imposed on linear quadratic elements and metric noncompatible d-connection structures are derived for nonlinear quadratic elements. Some authors consider that the Chern's d-connection is "the most appropriate one to study such locally anisotropic geometries and possible physical implications".\footnote{For applications in physics and mechanics, it is more constructive to elaborate (generalized) Finsler geometry models defined in a self-consistent physical and axiomatic mathematical form for certain fundamental metric, nonlinear and/or linear connection structures, vierbein and/or spinor like variables, etc. The N-adapted coefficients of such geometric objects are with generic dependence both on some basic spacetime coordinates and fiber (ed) type variables. We conclude that a Finsler space can be constructed as a nonholonomic geometric model with a nonlinear quadric line element, inducing certain fundamental  metric and/or almost symplectic structures and generalized connections depending on velocity/ momentum type coordinates on a (co) vector bundle or a nonholonomic fibered manifold. In explicit form, such Finsler-Lagrange-Hamilton theories can be elaborated for certain additional geometric and/or physical assumptions.} We emphasize that it is a very difficult technical task and less physically motivated to elaborate well-defined Finsler-Chern/ -Berwald (and other not metric compatible) modifications of standard physical theories. In brief, the main problems with applications in physics of such "noncompatible" geometries are those that the nonmetricity of a d-connection does not allow a self-consistent and generally accepted definition of spinors and Dirac operators and that formulation of some physically motivated conservative laws seems to be not possible. This creates a number of conceptual and technical difficulties for constructing physically viable, for instance, Finsler-Chern generalizations of GR; to formulate locally anisotropic (super) string and brane theories, noncommutative gravity and nonholonomic flow theories, EYMH and Einstein-Dirac system with MDRs; to study of locally anisotropic kinematic, stochastic and thermodynamic processes etc. Readers may consider critics of metric noncompatible physical models and alternative developments on Finsler MGTs in a series of works due to S. Vacaru and co-authors (almost all such works have been dubbed in arXiv.org and inspirehep.net or reviewed in MathSciNet and Zentralblatt\footnote{During 1987-1995, there were a series of original publications in the former USSR, the R. Moldova, and Romania. Here, we cite also some further original and review articles and books published by Western Editors beginning 1995 \cite{vplb10,vijgmmp12,vijmpd12,vmon98,vmon02,vmon06,vrev08}, see references therein and Appendix \ref{ashr}.}). The main conclusions of those works were that there are natural extensions and/or equivalent realizations of GR and standard theories of particle physics, and various quantum models like loop quantum gravity/ deformation quantization/ noncommutative geometry etc. In a more general context, the constructions can be extended on (co) tangent Lorentz bundles enabled with d-metric N-connection and Cartan d-connection (or other types metric compatible) structures determined by MDRs and respective fundamental generating Finsler-Lagrange-Hamilton functions. The fundamental field and/or evolution field equations in such theories can be transformed into systems of nonlinear PDEs with general decoupling property.

Some very influent schools on Finsler geometry, during 40 years after the second World War, were created in Japan. Such a research on pure and applied mathematics is related to two prominent geometers (A. Kawaguchi, 1937, 1952) \cite{kawaguchi37,kawaguchi52} and (M. Matsumoto, 1966, 1986) \cite{matsumoto66,matsumoto86,bao07} and their co-authors. In this work, we emphasized the importance of a class of Finsler metric compatible geometric models which can be described in terms of almost K\"{a}hler variables, see Theorem \ref{thaklh} (they are important for deformation and geometric quantization \cite{vjmp07,vpla08,vijgmmp09,vjgp10,bvnd11,vjmp13,vmjm15,vch2416,vmon02,avjmp09}).

It should be noted that in Finsler geometry a symmetric variant of Ricci tensor can be constructed (H. Akbar Zadeh, 1988) \cite{akbar88,akbar95}). The idea was to derive such a value directly from the fundamental generating function (Finsler metric) and semispray equations and do not enter in the ambiguities related to multi-connection character of such geometries. For applications in physics, such an approach remains incomplete because we can not elaborate viable theories without "covariant derivatives", i.e. without introducing a Finsler d-connection, on total bundle spaces (see discussions related to formula (\ref{ricciaz})).

\subsection{The physical picture of Finsler-Lagrange-Hamilton gravity theories}

The special relativity, SR, and general relativity, GR, theories have modified our understanding of causal and space-temporal structure of classical reality in a way admitting an axiomatic formulation but whose consequences have not been fully explored yet. Various directions of research in QG and, for instance, modern accelerating cosmology with DE and DM problem result in fundamental problems on formulating theories with MDRs and possible LIVs. Our exploration of foundational issues for geometrizing MGTs with MDRs give rise to the conclusion that classical and quantum physical models can be described in a self-consistent axiomatic form on (co) tangent Lorentz bundles. This way we have to revise a series of fundamental and conventional physical ideas and principles for locally anisotropic phase spaces and work with more rich geometric structures and advanced mathematical methods. We keep a lot from a credible physical intuition which is similar to that in standard theories of physics and do not need any drastic conceptual elaborations, for instance, on absolutely new and different Finsler causality and axiomatic approaches which cannot be elaborated, in principle, for general generating functions and indicators of MDRs. Nevertheless, various possible nonlinear symmetries/classifications/conservation laws of locally anisotropic and inhomogeneous interactions and modifications of standard theories determined by a (generalized) Finsler metric (nonlinear quadratic element) can be encoded into respective data for effective d-metrics, N-connections and d-connections. We also emphasize that our approach allows to analyze and compute explicit classical and quantum MDRs and LIVs effects, to find exact solutions and quantize in both perturbative and non-perturbative forms such MGTs.

The SR theory is based on the idea that the concepts of Newton space and independent time can be unified partially for physical models on a flat Minkowski spacetime as a consequence of experimentally verified constant speed of light and homogeneous and isotropic properties of light propagation. This allows a self-consistent formulation of relativistic mechanics together with Maxwell electrodynamics and such a unification can be realized with a respective Poincar\'{e}/ Lorentz invariance both the fundamental flat spacetime, main physical equations and respective conservation laws. The physical meaning of GR is that the concept of curved spacetime, modelled as a Lorentz manifold, and the gravitational field are the same entity when the causality structure is determined in any point (for instance, along with a light geodesic) as in SR. If any MDRs of type (\ref{mdrg}) is considered (as a modification from QG, a generalized commutative or noncommutative theory etc.), the geometric constructions are naturally extended to a (co) tangent Lorentz bundle. For elaborating such approaches,  it  is always invoked the concept of a generalized spacetime aether (phase space, or locally anisotropic spacetime) with rich geometric structure depending both on space, time and velocity/momentum type coordinates.

In a locally anisotropic spacetime, the fundamental concepts of space and time of Newton, the Minkowski spacetime in SR with electromagnetic but without gravitational interactions, and the GR curved spacetime (the gravitational field, in our approach with zero MDR indicator) can be reinterpreted for respective  configurations of locally anisotropic gravitational and matter field in a conventional phase space. Such a locally anisotropic picture implies also physical entities - particles, propagating dynamically and under nonholonomic constraints, and perturbations of locally anisotropic interacting real and effective fields - are not all immersed in a space moving in time or falling along curved pseudo-Riemannian geodesics and/or causal curves. All they do not live only in a (generalized) spacetime. The dynamics and evolution of particles and field are mutually correlated (let say, live on one another) via a locally anisotropic aehter for a nonholonomic phase space. It is a matter of convention which type of (generalized) Finsler-Lagrange-Hamilton, or their almost symplectic models, we use for modelling such geometries and physical theories. This approach presents a natural generalization of curved spacetime geometry for any MDRs and LIVs encoded in a nonlinear quadratic element $L(x,y)$ (\ref{nqe}) and/or $H(x,p)$ (\ref{nqed}). Such fundamental generating functions can be of any nature and with very general nonlinear symmetries. It is not possible to formulate a self-consistent and viable causal-axiomatic approach in terms of general nonlinear quadratic elements (for instance, considering that they may describe certain nonlinear geodesics with space, time, or null-like interpretation). We can elaborate axiomatic approaches as in GR but for Finsler generalized metrics/ connections / nonholonomic frames if we make one natural assumption: for a nontrivial indicator $\varpi $ of MDRs, the dynamics and evolution of particles and field is geometrized by respective nonholonomic structures on (co) tangent Lorentz bundles.

\subsection{What has been achieved in mathematical \& theoretical physics with Finsler methods?}

The problem of formulating realistic Finsler MGTs and applications in classical and quantum physics and mechanics has many aspects. There is a great number of important ideas, results, and methods scattered in the literature. In this work, we have attempted to collect the most important and perspective mathematical and theoretical physics results and geometric methods and to present an overall perspective on classical and quantum theories with generic local anisotropy, MDRs and LIVs. We argue that such theories can be constructed on Lorentz manifolds with nonholonomic fibered structure and/or on (co) tangent Lorentz bundles enabled with nonlinear connection structure. The nonholonomic geometric scheme (with nonlinear quadratic elements and generalized (non) linear connections can be preserved for various noncommutative/supersymmetric/quantum/fractional etc. modifications and generalized symmetries when the  Lorentz symmetries are restricted, modified, or generalized.

\vskip5pt

The strategy of almost 35 years research activity and author's points of view were both personal but also adapted to a temporary state of art evolution of main directions in modern mathematical and theoretical physics. The choice of subjects was determined by own interests and skills in geometry and physics but also oriented to solve important problems in modern gravity, non-standard particle physics and modified cosmology. He apologizes to  colleagues in mathematics and physics for what is missing and omitted from considerations (see Appendix \ref{ashr} for a historical survey of achievements due to other authors, and references therein). So much is missing and the reason was author's intellectual and time limits and restrictions on lengths of this article: It is not a book and not a general or status review report but an attempt to formulate an axiomatic unification of main principles and results based on most important 20 directions of research in modern Finsler MGTs and applications.

\vskip5pt

Conventionally, main achievements on generalized Finsler geometry and applications in modern physics correlated to author's research activity, experience, and important publications can be structured into \textbf{Seven Strategic Directions}:

\begin{enumerate}
\item \textsf{MGTs with MDRs and LIVs } and Finsler-Lagrange-Hamilton geometry, with applications in modern cosmology and astrophysics

\item \textsf{Geometric methods for constructing exact and parametric off-diagonal solutions} in GR, MGTs and geometric flow theories

\item \textsf{Geometric and almost K\"{a}hler deformation quantization, perturmative and nonperturbative quantization methods} for models with nonlinear and/or nonholonomic dynamics and locally anisotropic field interactions

\item \textsf{Nonholonomic Clifford structures,} spinors and twistors, and Dirac operators for Lagrange--Hamilton and Riemann--Finsler spaces and analogous/ modified gravity

\item \textsf{(Non) commutative / associative Finsler geometry and gravity} related to locally anisotropic (super) string/ gravity theories

\item \textsf{Kinetic, stochastic, fractional, and geometric and statistical thermodynamical processes with local anisotropy}, classical and quantum gravity theories, quantum geometric informatics

\item \textsf{Nonholonomic geometric flow evolution, } modified Ricci solitons and field equations in supersymmetric, commutative and noncommutative, fractional (modified) gravity theories
\end{enumerate}

In a more general context, including author's collaborations with scholars and young researchers from Western Countries, Romania and R. Moldova, mentioned above seven strategic directions are related to \textbf{Twenty Main Research Directions} outlined and discussed in Appendix \ref{assrm} (in general, there are considered there more than 100 sub-directions of research in modern Finsler-Lagrange-Hamilton geometry and applications). We present also a synopsis of related works published by Western Editors but also by less known (former) Editors in Eastern Europe. A part of such works were usually reviewed in MathSciNet and Zentralblatt, partially dubbed in arXv.org,  but they are still un-known and (in many cases) not cited correctly in modern literature of particle physics, gravity, and cosmology and astrophysics.

\subsection{What is missing in the scheme and further perspectives}

\label{ssmissing} The main ambition of this work was to formulate a general axiomatic approach to Finsler-Lagrange-Hamilton MGTs. To  elaborate and study of such classical and quantum theories is inevitable if we consider generic nonlinear interactions and evolution models with MDRs. This results in Finsler like  generalizations of Einstein gravity, standard particle physics, astrophysics and cosmology theories. A number of papers involves research on locally anisotropic thermodynamics, statistics and kinetics, and quantum information theories with applications of a series of concepts and methods from quantum mechanics, QFT, and QG. The main aspects that are still missing or not sufficiently developed in our scheme are the following:
\begin{enumerate}
\item {\it What type of Finsler like MGTs should we elaborate?} MDRs arise naturally in various QG theories, for instance, non perturbative effects dominate at the Plank scale and yield to a discrete structure of spacetime. But we can also consider nonlinear interactions in classical physics with locally anisotropic stochastic and kinetic processes, nonholonomic dynamics and evolution models. For geometric modeling of such processes, MDRs and LIVs can be considered as certain effective ones, or as some fundamental relations imposed by experimental/ observational data. Another important question is if Finsler like theories could solve fundamental problems of accelerating cosmology and DE and DM physics? Having a well-defined axiomatic of a theory (we proved that this is possible for very general assumptions of Finsler-Lagrange-Hamilton theories) is not enough for knowing how to extract physics from it and provide logical and experimental tests. What is missing is a systematic analysis and experimental data which would allow us to decide if theories with generic local anisotropy are generic quantum ones, and/or for certain locally anisotropic media (classical and/or quantum), or a respective Finsler MGT provides a well defined approach in modern cosmology.

\item{\it Nonholonomic (co) tangent Lorentz bundles or other type modifications?} As we shown, self-consistent axiomatic approaches with well-defined causal structure can be naturally elaborated by extending the main principles of ST and GR. But certain interesting geometric and physical models can be elaborated for restricted/deformed/duble etc. linear and nonlinear symmetries. What physical implications may have Finsler-Lagrange-Hamilton theories with "very radical" modifications of the Einstein/string gravity theory, quantum mechanics etc. and how to treat such new concepts of locally anisotropic spacetime?

\item {\it Difference between Finsler-Lagrange-Hamilton and almost symplectic variables}. In general, such variables are described by different nonolonomic structures. For certain well defined conditions such classical theories are equivalent. But noncommutative/nonassociative/quantum/string/fractional/stochastic ... MGTs are positively elaborated differently for different classes of nonholonomic variables. Such theories are described by different Lagrangians/ Hamiltonians with corresponding generalized topological and geometric structures/symmetries. In principle, we have to elaborate rigorously on all 20 directions (and various sub-directions) outlined in appendix \ref{ass20directions} in order to distinguish the constructions on (co) tangent bundles, nonholonomic manifolds, with effective velocity/momentum variables etc.

\item {\it Physically important solutions on (co) tangent bundles} should be constructed explicitly and distinguished in velocity and momentum type variables, in particular, for Hamilton like configurations. Further partner works will be devoted to: 
\begin{enumerate}
\item Generalization of the AFDM for generating exact solutions with noncommutative/ nonassociative/ supersymmetric and stochastic/fractional Hamilton like and entropic variables

\item Exact and parametric black hole and cosmological solutions with modified dispersions in Finsler-Lagrange-Hamilton theories on (co) tangent Lorentz bundles

\item Modified dispersion relations for Einstein--Dirac systems and pseudo Finsler-Hamilton geometry on (co) tangent
Lorentz bundles

\item Exact and parametric solutions for Einstein-Yang-Mills-Higgs systems with modified dispersion and pseudo
Finsler-Hamilton geometry

\item  Cosmological solutions in Hamilton variables and quantum effects

\end{enumerate}

\item {\it Quantization of Finsler-Lagrange-Hamilton MGTs.} At present we have a quite well developed scheme of deformation quantization of such theories. For future research, it is very important to elaborate on

\begin{enumerate}
\item Quantum (non) commutative gauge like Lagrange-Hamilton MGTs
\item Geometric quantization of generalized Finsler gravity theorie
\item Loop Quantum Gravity, LQG, with MDRs.
\item Perturbative and semiclassical limits of quantum Finsler-Lagrange-Hamilton theories
\end{enumerate}

\item {\it Quantum geometric informatics and computations } and generalized entropic methods, entanglement etc. consist a series of recent developments in modern informatics, physics, and cosmology. Some series of works involve nonholonomic classical and quantum geometric methods,  entropic dynamics, geometric thermodynamics etc. How such constructions should be generalized for Finsler-Lagrange-Hamilton variables?
\end{enumerate}

Finally, we conclude that there are many problems that we have to understand and then solve in order to get credible, viable and physically important theories of locally anisotropic spacetime and gravitational and matter field interactions, nonholonomic evolution processes, quantization and cosmological applications. The author of this articel hopes that some readers that have followed the formulated axiomatic scheme for Finsler-Lagrange-Hamilton MGTs will develop the research and new directions to the 20 ones considered in appendix \ref{ass20directions}.

\vskip5pt \textbf{Acknowledgments:}\ During 2012-2017, author's research on modified Finsler gravity theories was supported by a project IDEI in Romania, PN-II-ID-PCE-2011-3-0256, and a series of senior research fellowships for CERN and DAAD. Recently, such results were presented at the IUCSS Workshop on "Finsler Geometry and Lorentz Violation", May 12-13, 2017, Indiana University, Bloomington, USA (organizer A. Kosteleck\'{y}) and a seminar at CSU Fresno, USA (host D. Singleton). The author of this work thanks all his co-authors (see bibliography at the end). He is grateful for support, important discussions and collaboration to C. Arcu\c{s}, H. E. Brandt (passed away), I. Bubuianu, L. Bubuianu, G. Calcagni, C. Castro Perelman, H. Dehnen, E. Elizalde, F. Etayo, I. Gottlieb (passed away), A. P. Kouretsis,  K. L\"{a}mmerzahl, N. Mavromatos, C. Mociu\c{t}chi (passed away), J. Moffat, S. Odintsov, V. Perlick, E. Peyghan, S. F. Ponomarenko, M. S\'{a}nchez, P. Stavrinos, E. V. Veliev, M. Vi\c{s}inescu and others.

\appendix

\setcounter{equation}{0} \renewcommand{\theequation}
{A.\arabic{equation}} \setcounter{subsection}{0}
\renewcommand{\thesubsection}
{A.\arabic{subsection}}

\section{ Two Important Corollaries for N-adapted Formulas}

\label{aspt} We provide some examples and technical results, formulas and proofs which are similar to those used for elaborating the AFDM and applications to MGTs in Refs. \cite{vexsol98,vjhep01,vijgmmp07,vijtp10,vijtp10a,vijgmmp11,vjpcs13,vepjc14,
vvyijgmmp14a,gvvepjc14,vmon98,vmon02,vmon06,vrev08}.

\subsection{Curvatures and torsions of d--connections on (co) tangent bundles}

By explicit computations for $\mathbf{X}=\mathbf{e}_{\alpha}, \mathbf{Y}=%
\mathbf{e}_{\beta }, \mathbf{D}=\{\mathbf{\Gamma }_{\ \alpha \beta}^{\gamma
}\}$ and $\ ^{\shortmid }\mathbf{X}=\ ^{\shortmid }\mathbf{e}_{\alpha },$ $\
^{\shortmid }\mathbf{Y}= \ ^{\shortmid}\mathbf{e}_{\beta }, \ ^{\shortmid }%
\mathbf{D}=\{\ ^{\shortmid }\mathbf{\Gamma }_{\ \alpha \beta }^{\gamma }\}$
in (\ref{dcurvabstr}), we prove

\begin{corollary}
\label{acorolcurv}For a d--connection $\mathbf{D}$ or $\ ^{\shortmid }%
\mathbf{D,}$ \ there are computed corresponding N--adapted coefficients:
\newline
d-curvature, $\mathcal{R}=\mathbf{\{R}_{\ \beta \gamma \delta }^{\alpha
}=(R_{\ hjk}^{i},R_{\ bjk}^{a},P_{\ hja}^{i},P_{\ bja}^{c},S_{\
hba}^{i},S_{\ bea}^{c})\},$ for
\begin{eqnarray}
R_{\ hjk}^{i} &=&\mathbf{e}_{k}L_{\ hj}^{i}-\mathbf{e}_{j}L_{\ hk}^{i}+L_{\
hj}^{m}L_{\ mk}^{i}-L_{\ hk}^{m}L_{\ mj}^{i}-C_{\ ha}^{i}\Omega _{\ kj}^{a},
\notag \\
R_{\ bjk}^{a} &=&\mathbf{e}_{k}\acute{L}_{\ bj}^{a}-\mathbf{e}_{j}\acute{L}%
_{\ bk}^{a}+\acute{L}_{\ bj}^{c}\acute{L}_{\ ck}^{a}-\acute{L}_{\ bk}^{c}%
\acute{L}_{\ cj}^{a}-C_{\ bc}^{a}\Omega _{\ kj}^{c},  \label{dcurv} \\
P_{\ jka}^{i} &=&e_{a}L_{\ jk}^{i}-D_{k}\acute{C}_{\ ja}^{i}+\acute{C}_{\
jb}^{i}T_{\ ka}^{b},\ P_{\ bka}^{c}=e_{a}\acute{L}_{\ bk}^{c}-D_{k}C_{\
ba}^{c}+C_{\ bd}^{c}T_{\ ka}^{c},  \notag \\
S_{\ jbc}^{i} &=&e_{c}\acute{C}_{\ jb}^{i}-e_{b}\acute{C}_{\ jc}^{i}+\acute{C%
}_{\ jb}^{h}\acute{C}_{\ hc}^{i}-\acute{C}_{\ jc}^{h}\acute{C}_{\ hb}^{i},%
\hspace{0in}\ S_{\ bcd}^{a}=e_{d}C_{\ bc}^{a}-e_{c}C_{\ bd}^{a}+C_{\
bc}^{e}C_{\ ed}^{a}-C_{\ bd}^{e}C_{\ ec}^{a},  \notag
\end{eqnarray}%
or $\ ^{\shortmid }\mathcal{R}=\mathbf{\{\ ^{\shortmid }R}_{\ \beta \gamma
\delta }^{\alpha }=(\ ^{\shortmid }R_{\ hjk}^{i},\ ^{\shortmid }R_{a\ jk}^{\
b},\ ^{\shortmid }P_{\ hj}^{i\ \ \ a},\ ^{\shortmid }P_{c\ j}^{\ b\ a},\
^{\shortmid }S_{\ hba}^{i},\ ^{\shortmid }S_{\ bea}^{c})\},$ for
\begin{eqnarray*}
\ ^{\shortmid }R_{\ hjk}^{i} &=&\ ^{\shortmid }\mathbf{e}_{k}\ ^{\shortmid
}L_{\ hj}^{i}-\ ^{\shortmid }\mathbf{e}_{j}\ ^{\shortmid }L_{\ hk}^{i}+\
^{\shortmid }L_{\ hj}^{m}\ ^{\shortmid }L_{\ mk}^{i}-\ ^{\shortmid }L_{\
hk}^{m}\ ^{\shortmid }L_{\ mj}^{i}-\ ^{\shortmid }C_{\ h}^{i\ a}\
^{\shortmid }\Omega _{akj}, \\
\ ^{\shortmid }R_{a\ jk}^{\ b} &=&\ ^{\shortmid }\mathbf{e}_{k}\ ^{\shortmid
}\acute{L}_{a\ j}^{\ b}-\ ^{\shortmid }\mathbf{e}_{j}\ ^{\shortmid }\acute{L}%
_{a\ k}^{\ b}+\ ^{\shortmid }\acute{L}_{c\ j}^{\ b}\ ^{\shortmid }\acute{L}%
_{a\ k}^{\ c}-\ ^{\shortmid }\acute{L}_{c\ k}^{\ b}\ ^{\shortmid }\acute{L}%
_{a\ j}^{\ c}-\ ^{\shortmid }C_{a\ }^{\ bc}\ ^{\shortmid }\Omega _{ckj}, \\
\ ^{\shortmid }P_{\ jk}^{i\ \ \ a} &=&\ ^{\shortmid }e^{a}\ ^{\shortmid
}L_{\ jk}^{i}-\ ^{\shortmid }D_{k}\ ^{\shortmid }\acute{C}_{\ j}^{i\ a}+\
^{\shortmid }\acute{C}_{\ j}^{i\ b}\ ^{\shortmid }T_{bk}^{\ \ \ a},\ \
^{\shortmid }P_{c\ k}^{\ b\ a}=\ ^{\shortmid }e^{a}\ ^{\shortmid }\acute{L}%
_{c\ k}^{\ b}-\ ^{\shortmid }D_{k}\ ^{\shortmid }C_{c\ }^{\ ba}+\
^{\shortmid }C_{\ bd}^{c}\ ^{\shortmid }T_{\ ka}^{c}, \\
\ ^{\shortmid }S_{\ j}^{i\ bc} &=&\ ^{\shortmid }e^{c}\ ^{\shortmid }\acute{C%
}_{\ j}^{i\ b}-\ ^{\shortmid }e^{b}\ ^{\shortmid }\acute{C}_{\ j}^{i\ c}+\
^{\shortmid }\acute{C}_{\ j}^{h\ b}\ ^{\shortmid }\acute{C}_{\ h}^{i\ c}-\
^{\shortmid }\acute{C}_{\ j}^{h\ c}\ ^{\shortmid }\acute{C}_{\ h}^{i\ b}, \\
\ ^{\shortmid }S_{a\ }^{\ bcd} &=&\ ^{\shortmid }e^{d}\ ^{\shortmid }C_{a\
}^{\ bc}-\ ^{\shortmid }e^{c}\ ^{\shortmid }C_{a\ }^{\ bd}+\ ^{\shortmid
}C_{a\ }^{\ bc}\ ^{\shortmid }C_{b\ }^{\ ed}-\ ^{\shortmid }C_{e}^{\ bd}\
^{\shortmid }C_{a}^{\ ec};
\end{eqnarray*}%
d-torsion, $\ \mathcal{T}=\{\mathbf{T}_{\ \alpha \beta }^{\gamma }=(T_{\
jk}^{i},T_{\ ja}^{i},T_{\ ji}^{a},T_{\ bi}^{a},T_{\ bc}^{a})\},$ for
\begin{equation}
T_{\ jk}^{i}=L_{jk}^{i}-L_{kj}^{i},T_{\ jb}^{i}=C_{jb}^{i},T_{\
ji}^{a}=-\Omega _{\ ji}^{a},\ T_{aj}^{c}=L_{aj}^{c}-e_{a}(N_{j}^{c}),T_{\
bc}^{a}=C_{bc}^{a}-C_{cb}^{a},  \label{dtors}
\end{equation}%
or $\ ^{\shortmid }\mathcal{T}=\{\ ^{\shortmid }\mathbf{T}_{\ \alpha \beta
}^{\gamma }=(\ ^{\shortmid }T_{\ jk}^{i},\ ^{\shortmid }T_{\ j}^{i\ a},\
^{\shortmid }T_{aji},\ ^{\shortmid }T_{a\ i}^{\ b},\ ^{\shortmid }T_{a\ }^{\
bc})\},$ for
\begin{equation*}
\ ^{\shortmid }T_{\ jk}^{i}=\ ^{\shortmid }L_{jk}^{i}-\ ^{\shortmid
}L_{kj}^{i},\ ^{\shortmid }T_{\ j}^{i\ a}=\ ^{\shortmid }C_{j}^{ia},\
^{\shortmid }T_{aji}=-\ ^{\shortmid }\Omega _{aji},\ \ ^{\shortmid }T_{c\
j}^{\ a}=\ ^{\shortmid }L_{c\ j}^{\ a}-\ ^{\shortmid }e^{a}(\ ^{\shortmid
}N_{cj}),\ ^{\shortmid }T_{a\ }^{\ bc}=\ ^{\shortmid }C_{a}^{\ bc}-\
^{\shortmid }C_{a}^{\ cb};
\end{equation*}%
d-nonmetricity, $\ \mathcal{Q}=\mathbf{\{Q}_{\gamma \alpha \beta }=\left(
Q_{kij},Q_{kab},Q_{cij},Q_{cab}\right) \},$ for
\begin{equation}
Q_{kij}=D_{k}g_{ij},Q_{kab}=D_{k}g_{ab},Q_{cij}=D_{c}g_{ij},Q_{cab}=D_{c}g_{ab}
\label{dnonm}
\end{equation}%
or $\ ^{\shortmid }\mathcal{Q}=\mathbf{\{\ ^{\shortmid }Q}_{\gamma \alpha
\beta }=\left( \ ^{\shortmid }Q_{kij},\ ^{\shortmid }Q_{kab},\ ^{\shortmid
}Q_{cij},\ ^{\shortmid }Q_{cab}\right) \},$ for
\begin{equation*}
\ ^{\shortmid }Q_{kij}=\ ^{\shortmid }D_{k}\ ^{\shortmid }g_{ij},\
^{\shortmid }Q_{k}^{\ ab}=\ ^{\shortmid }D_{k}\ ^{\shortmid }g^{ab},\
^{\shortmid }Q_{\ ij}^{c}=\ ^{\shortmid }D^{c}\ ^{\shortmid }g_{ij},\
^{\shortmid }Q^{cab}=\ ^{\shortmid }D^{c}\ ^{\shortmid }g^{ab}.
\end{equation*}
\end{corollary}

\subsection{The coefficients of canonical Lagrange and Hamilton d-connections%
}

Such d-connections are very important for elaborating MGTs on (co) tangent
bundles because they allow a very general decoupling and integration of
generalized Einstein and matter field equations. By explicit computations in
N-adapted frames, we can prove that necessary conditions for defining and
constructing, respectively, $\widehat{\mathbf{D}}$ \ (\ref{canondcl}) and $\
^{\shortmid }\widehat{\mathbf{D}}$ (\ref{canondch}), are satisfied following

\begin{corollary}
\label{acoroldand}The N-adapted coefficients of canonical Lagrange and
Hamilton d-connections are computed respectively:
\begin{eqnarray}
\mbox{ on }T\mathbf{TV},\ \widehat{\mathbf{D}} &=&\{\widehat{\mathbf{\Gamma }%
}_{\ \alpha \beta }^{\gamma }=(\widehat{L}_{jk}^{i},\widehat{L}_{bk}^{a},%
\widehat{C}_{jc}^{i},\widehat{C}_{bc}^{a})\},\mbox{ where },  \notag \\
&&\mbox{ for }\lbrack \mathbf{g}_{\alpha \beta }=(g_{jr},g_{ab})\mathbf{,N=\{%
}N_{i}^{a}\mathbf{\}]\simeq \lbrack }\widetilde{\mathbf{g}}_{\alpha \beta }=(%
\widetilde{g}_{jr},\widetilde{g}_{ab}),\widetilde{\mathbf{N}}_{i}^{a}=%
\widetilde{N}_{i}^{a}],  \notag \\
\widehat{L}_{jk}^{i} &=&\frac{1}{2}g^{ir}\left( \mathbf{e}_{k}g_{jr}+\mathbf{%
e}_{j}g_{kr}-\mathbf{e}_{r}g_{jk}\right) ,\ \widehat{L}%
_{bk}^{a}=e_{b}(N_{k}^{a})+\frac{1}{2}g^{ac}(e_{k}g_{bc}-g_{dc}\
e_{b}N_{k}^{d}-g_{db}\ e_{c}N_{k}^{d}),  \notag \\
\widehat{C}_{jc}^{i} &=&\frac{1}{2}g^{ik}e_{c}g_{jk},\ \widehat{C}_{bc}^{a}=%
\frac{1}{2}g^{ad}\left( e_{c}g_{bd}+e_{b}g_{cd}-e_{d}g_{bc}\right)
\label{canlc}
\end{eqnarray}%
\begin{eqnarray}
\mbox{ and, on }T\mathbf{T}^{\ast }\mathbf{V},\ \ ^{\shortmid }\widehat{%
\mathbf{D}} &=&\{\ ^{\shortmid }\widehat{\mathbf{\Gamma }}_{\ \alpha \beta
}^{\gamma }=(\ ^{\shortmid }\widehat{L}_{jk}^{i},\ ^{\shortmid }\widehat{L}%
_{a\ k}^{\ b},\ ^{\shortmid }\widehat{C}_{\ j}^{i\ c},\ ^{\shortmid }%
\widehat{C}_{\ j}^{i\ c})\},\mbox{ where, }  \notag \\
&&\mbox{ for }\lbrack \ ^{\shortmid }\mathbf{g}_{\alpha \beta }=(\
^{\shortmid }g_{jr},\ ^{\shortmid }g^{ab})\mathbf{,\ ^{\shortmid }N=\{}\
^{\shortmid }N_{ai}\mathbf{\}]\simeq \lbrack }\ ^{\shortmid }\widetilde{%
\mathbf{g}}_{\alpha \beta }=(\ ^{\shortmid }\widetilde{g}_{jr},\ ^{\shortmid
}\widetilde{g}^{ab}),\ ^{\shortmid }\widetilde{\mathbf{N}}_{ai}=\
^{\shortmid }\widetilde{N}_{ai}],  \notag \\
\ ^{\shortmid }\widehat{L}_{jk}^{i} &=&\frac{1}{2}\ ^{\shortmid }g^{ir}(\
^{\shortmid }\mathbf{e}_{k}\ ^{\shortmid }g_{jr}+\ ^{\shortmid }\mathbf{e}%
_{j}\ ^{\shortmid }g_{kr}-\ ^{\shortmid }\mathbf{e}_{r}\ ^{\shortmid
}g_{jk}),\   \notag \\
\ ^{\shortmid }\widehat{L}_{a\ k}^{\ b} &=&\ ^{\shortmid }e^{b}(\
^{\shortmid }N_{ak})+\frac{1}{2}\ ^{\shortmid }g_{ac}(\ ^{\shortmid }e_{k}\
^{\shortmid }g^{bc}-\ ^{\shortmid }g^{dc}\ \ ^{\shortmid }e^{b}\ ^{\shortmid
}N_{dk}-\ ^{\shortmid }g^{db}\ \ ^{\shortmid }e^{c}\ ^{\shortmid }N_{dk}),
\notag \\
\ ^{\shortmid }\widehat{C}_{\ j}^{i\ c} &=&\frac{1}{2}\ ^{\shortmid }g^{ik}\
^{\shortmid }e^{c}\ ^{\shortmid }g_{jk},\ \ ^{\shortmid }\widehat{C}_{\
a}^{b\ c}=\frac{1}{2}\ ^{\shortmid }g_{ad}(\ ^{\shortmid }e^{c}\ ^{\shortmid
}g^{bd}+\ ^{\shortmid }e^{b}\ ^{\shortmid }g^{cd}-\ ^{\shortmid }e^{d}\
^{\shortmid }g^{bc}).  \label{canhc}
\end{eqnarray}
\end{corollary}

In a similar form, we can prove that all N-adapted coefficient formulas necessary formulating and finding solutions of physically important field and evolution equations in theories with MDRs and LIVs.

\setcounter{equation}{0} \renewcommand{\theequation}
{B.\arabic{equation}} \setcounter{subsection}{0}
\renewcommand{\thesubsection}
{B.\arabic{subsection}}

\section{ A Historical Survey on Modified Finsler Gravity}

\label{ashr}

\subsection{Motivations and historical remarks}

The historical remarks in this work are related to subjective author's interests and research experience on geometric methods in physics and applications of Finsler geometry to (non) standard theories of physics, geometric mechanics etc. The bibliography is not exhaustive and not organized in a chronological order. Readers may consult monographs
\cite{cartan35,vranc57,rund59,yano73,matsumoto86,bao00,shen01,shen01a,youssef09}, for main references on works written by mathematicians, and monographs \cite{vlasov66,asanov85,asanov88,asanov89,bejancu90,bejancu03,bogoslovsky92,anton93,anton96,anton03},
 published more than 20 years ago and written by geometers and physicists oriented to Finsler applications in physics. Author's contributions to generalized Finsler geometry, standard theories of physics and MGTs are reviewed in
 \cite{vplb10,vap97,vnp97,vijgmmp12,vijmpd12,vmon98,vmon02,vmon06,vrev08,vjgp10,vjmp13,vch2416,ruchinepjc17}, see also references therein.\footnote{The historical and bibliographical remarks presented in this paper complete comments and references for Appendix A in review \cite{vrev08}. There are reviewed and listed in MathSciNet, Scopus, libraries bibliographic data etc. more than 5000 works on the keyword "Finsler". Such base data contain titles and abstracts of works published by mathematicians and physicists
before the year 2000 and certain information on such research elaborated in Eastern Europe, Japan, China etc. During last 25 years, a number of preprints on Finsler geometry and applications were put/ dubbed in arXiv.org and ispirehep.net . It is not possible to cite and discuss in this article all important contributions related to Finsler geometry and physics when authors developed "non-standard" physical paradigms or have been oriented to solution of "pure" mathematical tasks.}

\subsubsection{Politics and communist ideology, geometry and physics, and young researchers}

There are three important motivations for presenting in this Appendix a short review on Finsler geometry, methods, and applications in physics and geometric mechanics:

\begin{enumerate}
\item \textsf{A number of fundamental contributions to Finsler geometry and applications belong to authors who lived and worked for a long time in Countries with totalitarian regimes.} Their social and scientific destinies and conditions of research activity were different (and almost unknown) comparing to those in democratic Countries. For instance, L. Berwald
    \cite{berwald26,berwald41} who performed a new approach to Finsler geometry was a Jewish scientist killed in a Nazi concentration camp during the 2d World War. \newline
R. I. Pimenov \cite{pimenov87} tried to defend a habilitation thesis on an axiomatic formulation of the concept of Finsler spacetime in the former URSS. He was persecuted and tortured during many years for anti-communist activity (Stalin's GULAG system was also anti-human and with a huge network of concentration and extermination campuses). Dr. Pimenov was a close friend of the famous Soviet physicist, anti-soviet dissident and human right activist, academician A. D. Sakharov, a winner of Nobel Prize for Peace and the "main father" of the Soviet thermonuclear bomb. \newline
The research on Finsler geometry in Romania was taken partially under control of Elena Ceau\c{s}ecu (dictator's wife), secret service "securitate" (in English, security) supervised by the Soviet KGB, with a certain degree of independence allowed for the Romanian communist party. She had only a primary school education but supervised the Academy of Sciences of Romania and forced research institutions on chemistry to publish for her in Wester Countries prestigious works on chemistry. In all directions of
Romanian science, it was organized a system that only professors authorised by the communist elite and their secret service were allowed to travel, collaborate and publish in Western Countries. Those scientists and (in many cases) pseudo scientists had obligations to "bring" (i.e. still) fundamental ideas, results, and technological secrets and falsify scientific results in an ideological Marxist fashion. On differential and Finsler geometry, prof. Radu Miron (later he became a member of the Romanian
Academy) was assigned by the dictator as a "doctor docent" for political, ideological and administrative control. He got exclusive rights to travel abroad, participate at international scientific conferences, take any results of tenths of Romanian researchers and publish under his name with National and Western Editors, and even to supervise PhD theses for students from Japan.
\newline
The external financial support for science in the Republic of Moldova (for instance, provided by various NATO and UNESCO programs but also by private funds and organizations in Western Countries) was taken under control of the communist regime and Russian occupation army (existing up till present in the Eastern part of that Country). A quasi-communist / Soviet system acts up
till present in that small Country. A number of senior and young researchers on geometry and physics was persecuted politically and deported (or constrained by criminal methods to live) from R. Moldova during 1997-2001. A part of those researchers were prohibited to do science and get PhD theses after they published with Western Editors a series books and important works
on Finsler geometry and applications modern physics. It was not allowed by the communist regime to accept independently NATO and UNESCO fellowships for geometry and physics. For reviews of their results and original contributions, see
\cite{v83,v94,vog94,vgon95,vjmp96,vstoch96,vap97,vnp97,vexsol98,vhsp98,vmon98,vapl00,vapny01, vplbnc01,vncs01,vsbd01,vmon02}.

\item In Western Countries, Finsler geometry was considered for many years (in many cases, unfairly) as a non-canonical direction for elaborating particle models and classical and quantum field theories. Finsler gravity models were concluded (which was quite subjective and not motivated scientifically) to be in contradiction to the standard approach to physics. For instance, spacetime models with local anisotropy were classified as un-physical ones by a very influent author in GR (J. D. Bekenstein, 1993) \cite{bekenstein93}. In that paper, there were not studied possible important implications of the N--connection structure (a less known for physicists geometric formalism) which maid the analysis and conclusions to be incomplete. Before the end of previous Century, it was very difficult to publish Finsler papers in Phys. Rev. Lett., Phys. Rev., Phys. Lett. B and a number of top journals on astrophysics and cosmology.\footnote{Nevertheless, it should be noted that the author of this work was able to publish a series of original papers and reviews on generalized Finsler geometry and applications (and modeling such structures on nonholonomic Lorentz manifolds) in other important physical and math physical journals like Phys. Lett. A and B, EPJC, Nuclear Phys. B, JHEP, Annals Phys. NY, Class. Q. Grav., Gen. Rel. Gravitation, J. Geom. Phys., J. Math. Physics, IJGMMP, IJMPD etc.; see, for instance, \cite{vjmp96,vap97,vnp97,vhsp98,vapny01,vjhep01,vpcqg01,vtnpb02,vsjmp02} and
references therein.} During 1960-2005, the mainstream of works on Finsler geometry and physics/ geometric mechanics was formed by authors originating from Japan, URSS, Romania, R. Moldova, Hungary, Poland, Greece etc. They published their papers in a series of "less known" international and/or national journals. Those contributions were reviewed (in the bulk) in MathSciNet and Zentralblatt. Here we also cite a series of important works published by authors and editors in Japan, see \cite%
{kawaguchi37,kawaguchi52,matsumoto66,matsumoto86,ikeda77,ikeda78,ikeda79,ikeda81,ishikawa80,ishikawa81,ono90,ono93,takano68, yano73,takano74,takano83,bao07} and references therein. It should be noted that a number of important ideas on locally anisotropic supersymmetric models, noncommutative geometries, spinor variables were proposed and tested in those works. The main issue was that not all geometric constructions were N-adapted and that there were not elaborated self-consistent methods for constructing exact solutions and quantization of such theories. In many cases, those ideas, results and methods should be reformulated in order to consider and solve important problems for the advanced geometry and modern physics after 2000.

\item After discovering accelerating cosmology and publishing a series of important theoretical works on possible physical effects with MDRs and/or broken LIV in particle physics and classical and quantum gravity (we cite here
    \cite{amelino96,amelino98,amelino02,amelino02a,hooft96,mavromatos07,mavromatos10,mavromatos10a,mavromatos11}), a number of MGTs were elaborated. For reviews, one can be considered
    \cite{nojiri07,capozziello10,yang09,yang10,nojiri13jcap,magueijo04,mavromatos13a,mavromatos13b}. Certain directions in modified gravity and cosmology involve metrics with explicit dependence on velocity and/or momentum type variables. There were studied models with so-called "rainbow metrics", analogous metrics etc. containing phenomenologically, for instance, an energy type parameter, see Example \ref{exmdrenergy}. A series of works discussing possible Finsler applications in modern physics and cosmology appeared in top physical journals \cite{girelli06,gallego06,laemmerzahl08,laemmerzahl12,kostelecky11,
    kostelecky12,kouretsis10,kouretsis14, kouretsis14a,itin14,amelino14,bailey15}. Authors of such works were not able to motivate their constructions by exact solutions with nontrivial N--connection structures (it is a very difficult technical task) and to formulate a self-consistent axiomatic approach. Recently, a new generation of young researchers published
original papers on various types of Finsler spacetimes and locally anisotropic cosmology models (in some cases with static and/or spherical symmetries which is possible only for trivial N-connection structures). There were advanced very sophisticated schemes of "Finsler locally anisotropic causality" \cite{perlick87,meister90,pfeifer11,li14,li14a,lin14,li16,hohmann13,
hohmann16,minguzzi15,minguzzi15a, russell15,schreck15,schreck15a,schreck16,schreck16a,barcaroli15,lobo16}
and some results were re-discovered without a corresponding citation of previous original and important works.
\end{enumerate}

\subsubsection{Seven historical periods for Finsler MGTs}

\label{assshclass}Chronologically, the evolution of research on Finsler geometry, generalizations, and applications can be conventionally divided into seven periods\footnote{A similar historical separation is considered in Appendix B of the book \cite%
{rovelli03} for five periods evolution of Loop Quantum Gravity, LQG, which is very different from that for Modified Finsler gravity. A synthesis of LQG and deformation quantization methods with applications to certain generalized Finsler geometries and gravity models was studied in Ref. \cite{vlqg09}.} (with certain overlaps and parallel activities of different authors and schools on geometry and physics; we cite in this subsection some most important and typical works):

\begin{itemize}
\item \textit{The Prehistory (1854-1918) and/or beginning of History (till 1934).} The first "prehistorical" work was the B. Riemann's habilitation thesis from 1854 \cite{riem1854}. Here we note that authors of monograph \cite{bao00} consider that we shall use the term the term Riemann-Finsler geometry for spaces endowed with nonlinear quadratic elements. Following
their arguments, the real Finsler history began much early than the classical (P. Finsler, 1918) work \cite{finsler18} was published. Here we note that the first complete model of Finsler geometry determined by a triple of fundamental geometric objects consisting from a nonlinear homogeneous quadratic element, N-connection and d-connection was constructed in explicit coordinate form in E. Cartan's first monograph on Finsler geometry published in 1935 \cite{cartan35} (see that book for citations on much
early classical geometric works):

\item \textit{The Classical Geometric Ages (1935 - 1955).} There were published a series of fundamental geometric papers in that period due to L. Berwald \cite{berwald26,berwald41}, S. -S Chern \cite{chern48} and others. Here we mention certain very important works on the geometry of nonholonomic manifolds (G. Vranceanu and Z. Horak, 1926-1955)
    \cite{vranc31,vranc57,horac27} and global theory of N-connections and Finsler geometry (A. Kawaguchi, 1937-1952, and C. Ehresmann, 1955) \cite{kawaguchi37,kawaguchi52,ehresmann55}. Such methods of the geometry of nonholonomic spaces because later they became very important for generating exact solutions and quantization of Finsler MGTs. For summary of results and
a detailed bibliography on the classical period of Finsler geometry (being considered also certain applications), we cite
\cite{cartan35,vranc57,rund59,yano73,matsumoto86,bao00,shen01,shen01a,youssef09,vlasov66,asanov85,asanov88,asanov89,bejancu90,
bejancu03, bogoslovsky92,anton93,anton96,anton03}.

\item \textit{The Early Middle Ages} \textit{(1950-1974)} with the beginning of explicit applications in physics and mechanics and Finsler gravity modifications, see (J. I. Horvath, 1950) \cite{horvath50}. It was also proposed an alternative version of geometric mechanics, see (J. Kern, 1974) \cite{kern74}, where the concept of Lagrange geometry was elaborated as a
generalization of Finsler geometry. For a standard approach to classical geometric mechanics, we cite (M. de Leon and P. Rodrigues, 1985) \cite{deleon85} and references therein. Various schools on geometric and analytic mechanics followed in the bulk the standard approach to geometric mechanics and not Kern-Finsler modifications even certain research groups in Japan and
Romania elaborated on almost K\"{a}hler-Finsler/ -Lagrange and Hamilton-Finsler theories.

\item \textit{The Middle Ages (1959 - 1995)} are characterized by a number of works on non-standard locally anisotropic field theories and gravity. Conventionally, this period begins with the monograph (H. Rund, 1959) \cite{rund59}, with translation, in Russia by G. Asanov; and publication of M. Matsumoto's monograph \cite{matsumoto86}. There were published a number of
papers on modified Finsler gravity theories and applications which can not be related directly to accelerating cosmology, string gravity, noncommutative and nonassociative generalizations of Finsler geometry, see typical works
\cite{ikeda77,ikeda78,ikeda79,ikeda81,ishikawa80,ishikawa81,ono90,ono93,takano68,yano73,takano74,takano83,asanov85,asanov88,
asanov89, aringazin88,ingardentamassy93,roxburgh79,roxburgh86,tavakol86,bogoslovsky92,bogoslovsky04}.

\item \textit{The Dark Non-Standard Ages (1975-2011)} is a period of classical results in standard theories of physics and GR. At that period, it was concluded that Finsler geometry and applications had restrictions from experimental data \cite{bekenstein93} which restricted the possibilities to publish such papers in Phys. Rev. Lett and Phys. Rev. Nevertheless, a number of authors in the USA, Japan, Germany, UK, Canada, Greece, Russia, Poland, Romania, R. Moldova etc. published in other mathematical and theoretical physics journals a number of papers on generalized Finsler gravity theories (with nonholonomic supersymmetric and noncommutative variables and quantum models with locally anisotropic interactions). Here, we cite some series of such works elaborated by different research groups and authors \cite{beil93,beil03,brandt92,brandt98,
    brandt99,brandt00,brandt03,brandt04,mrugala90,mrugala92,vargas96,vrm2,vrm3,vap97,vnp97,vncs01,ingarden03,ingarden04,
    ingarden08,vjmp05,vijtp10,vcqg10,ingarden08a,stavr99, stavr04,stavr05,kouretsis10,laemmerzahl08}. There were published also a series of important monographs on generalized Finsler gravity and applications, see
    \cite{bejancu90,vmon98,vmon02,vmon06,anton93,anton96,anton03}.

\item \textit{The Renaissance via MDRs and LIVs (1996-2011).} This period is characterised by two classes of works on generalized Finsler gravity and applications. A series of important physical journals began to publish papers on physical models with MDRs and LIVs related to quantum gravity, string theory and certain modifications of gravity, see typical works \cite%
{amelino96,amelino98,amelino02,amelino02a,mavromatos07,magueijo04}. Such works usually had not involved Finsler geometric methods but contained nevertheless certain formulas with velocity/ momentum type variables on (co) tangent bundles. The second class of papers were based on various generalizations of Finsler geometry elaborating non-standard theories of physics, see publications \cite{vapny01,vijmpd03,vijmpd03a,vjmp05,vijtp10,vcqg10,girelli06,gallego06,mavromatos10,mavromatos10a,mavromatos11,castro97, castro07,castro08,castro08a,castro11,stavr99,stavr04,stavr05,kouretsis10}, and references therein.

\item \textit{Nowdays, MGTs, and Future (2012- ...).} After a number of publications in MGTs and accelerating cosmology, PRD and other important physical journals began to accept papers on modified Finsler gravity from some groups with young researchers in Germany, UK, Italy, China and other countries, see \cite{pfeifer11,li14,li14a,lin14,li16,hohmann13,hohmann16,minguzzi15,
    minguzzi15a,russell15,schreck15,schreck15a, schreck16,schreck16a,barcaroli15,lobo16}. A number of researchers working on modified gravity and applications in cosmology, non-standard particle physics, MDRs and LIVs begin to publish series of important works with Finsler like variables and involving geometric methods \cite{mavromatos13a,mavromatos13b,
laemmerzahl12,kostelecky11,kostelecky12,kouretsis14,kouretsis14a,itin14,amelino14,bailey15,castro14,castro16}. A number of
papers on Finsler MGTs and applications in modern cosmology were published in collaboration with authors originating from the R. Moldova, Romania, USA, Germany etc, see \cite{gvvepjc14,svvijmpd14,gvvepjc14,gvvcqg15,gheorghiuap16,bubuianucqg17,ruchinepjc17,
elizalde15,vport13,vijgmmp14,vepjc14,vepjc14a, vvyijgmmp14a} and references therein. We speculate on future and perspective in subsection \ref{ssmissing}.
\end{itemize}

It should be noted that a series of authors worked during some historical periods of Finsler geometry and related MGTs and, in some cases, are active and publishing important works up till present. Before the end of previous Century, it was also a conventional geographical separation of researchers publishing on Finsler geometry and applications.

\subsection{A survey of research on Finsler MGTs in Eastern and Western Europe and Japan}

A number of directions related to MGTs with MDRs and applications of Finsler-Lagrange-Hamilton geometries in physics, engineering and biophysics were elaborated after II World War. That research activity was negatively affected by the different conditions of work and collaboration corresponding to different political and ideological systems established in democratic countries and those being under control of the Soviet regime.

\subsubsection{Finsler geometry and physical theories in Central Europe and former USSR}

Perhaps, the first most important result for elaborating Finsler gravity theories after the publication of E. Cartan's monograph \cite{cartan35} was due to (J. I. Horvath, 1950) with a work \cite{horvath50} submitted from Debrecen, Hungary and published in Physical Review. The goal of that short paper was to consider a modification of GR using the Cartan d-connection in Finsler geometry. There were not analysed possible physical consequences and not provided solutions of such equations. Nevertheless, it was shown how such extensions could be geometrized on tangent bundles to (pseudo) Riemannian manifolds using respective analogues of the Cartan d-connection. In Debrecen, it was created an influent school on Finsler geometry and generalizations acting in Central Europe during almost 50 years, see works by L. Tam\'{a}ssy, J. Szilasi, L. Kozma \cite{tamassy00,kozma03,szilasi12} and references therein.

A series of important E. Cartan books and articles (see \cite{cartan35,cartan38,cartan63}) were known in the former USSR before II World War. Later, a number of such works were translated in Russian. Those works had a great influence on researchers in geometry and physics from Moscow, Leningrad (St. Petersburg), Kazan and other cities. That period may be a subject of a detailed study on the history of geometry and physics in the former USSR which is out of scopes of this work. Here we note that G. A. Asanov was the editor of the Russian translation of H. Rund's monograph on Finsler geometry \cite{rund59}. In that translation, there is a typical appendix with ideological Soviet type comments but also certain important developments on various models of Finsler geometry and gravity. Prof. Asanov was authorised by Soviet officials \footnote{having recommendations from professors A. Sokolov and D. Ivanenko from the department of theoretical physics of M. Lomonosov Moscow State University who wrote a series of works on Marxist ideology and physics and got the Stalin's prize for science; in 1954, it was renamed into the "State prize"} to supervise from the "M. Lomonosov" Moscow State University the research in USSR and Socialist Countries on developing Finsler geometry and
applications. Nevertheless, that quasi-political/ideological activity resulted also in a series of original articles and books on modified gravity theories when some important works were performed together with post-graduate students\footnote{for instance, A. Aringazin, originating and working at present in Kazakhstan, and S. F. Ponomarenko, originating from Ukraine with further
affiliations in the R. Moldova and Russia} and published beginning the end of 70th of previous century \cite{asanov85,asanov88,asanov89,aringazin88}. Asanov's group elaborated on various gauge and jet models of modified Finsler gravity with "compensating" terms (specific locally anisotropic gravitational and gauge fields) on tangent/vector/jet bundles. Certain original ideas and a post-Newton formalism for computing possible physical anisotropies of cosmic rays etc. were elaborated in a series of works in Russian and some further publications in English as a member of Editorial Body of Rep. Math. Phys. Publications of G. A. Asanov's school were not formulated in a rigorous language of the geometry of fiber bundles and nonlinear connections but performed in a stile more closed to theoretical physics. In that period of glory of GR, standard particle physics and canonical quantum field theory, Finsler like modifications in physics and mechanics were treated with suspicion and not accepted by the "Orthodox gravity" community. The constructions with "gauge symmetries" in Finsler models elaborated by G. Asanov's schools are based on different geometric and physical principles than those for particle fields and standard quantum
field theory.

A very original application of Finsler geometry to the theory of locally anisotropic kinetic processes was elaborated in monograph (A. A. Vlasov, 1966) \cite{vlasov66} as a generalization of the so-called Einstein-Vlasov theory. That direction was supported by a famous Soviet mathematical physicist, N. N. Bogolyubov, but not accepted by L. D. Landau's school in theoretical physics.\footnote{
Further applications and developments on locally anisotropic kinetic processes, diffusion, geometric flows on generalized (super) Finsler spaces were considered by the S. Vacaru and co-authors in \cite{vstoch96,vmon98,vapl00,vapny01,vsym13} and references therein.}

There were developed in the former USSR and Russia other alternative approaches to locally anisotropic spacetimes geometry and applications in physics. We emphasized above the original attempt to provide an axiomatic approach to Finsler gravity by R. I. Pimenov \cite{pimenov87}. Here we note also contributions by G. Yu. Bogoslovsky from Moscow and his collaboration with H. G. Goenner from G\"{o}ttingen (during almost 20 years) \cite{bogoslovsky92,bogoslovsky04}. More recently, research on Finsler geometry
and applications got certain financial support from some Russian business elite, see a typical work due to (G. Garas'ko and D. G. Pavlov, 2016) \cite{garasko06} and reference therein. As a final remark on research on Finsler geometry during the Soviet period, one should be emphasized that such geometric methods were considered potentially important (and had certain official State support, for instance, for nuclear centers in Siberia and the Ural region, Russia) because of certain tensor models of locally anisotropic
nuclear matter and kinetics of nuclear reactions. In the bulk, that was before formulating the Quantum Chromo Dynamics, QCD as a standard gauge theory of nuclear forces.

Research on Finsler geometry and applications in physics was developed in Poland (with certain collaborations with authors from Japan, Hungary, former USSR). Here we cite a series of works due to R. S. Ingarden, J. Lawrynowicz, and co-authors published beginning 1954, see \cite{ingarden54,ingardentamassy93,ingarden03,ingarden04,ingarden08,ingarden08a} and references therein. Such authors worked on long-term projects on Finsler geometry and physics including important issues on quantum Randers geometry and (non) associative/commutative gauge theories and Kaluza-Klein spaces. We also emphasize original constructions on statistical approaches and the geometric (including Finsler) structure of thermodynamics due to (R. Mrugala and co-authors, 1990-1992) \cite{mrugala90,mrugala92}. MDRs related to noncommutative and commutative Finsler geometry where studied in \cite{lukierski95} due to (J. Lukierski and co-authors, 1993).

\subsubsection{Finsler geometry and physics in Western Countries and Japan}

We outline in brief some most important achievements related to modified Finsler gravity theories beginning the (pre) Historic period till the Renaissance period stated in section \ref{assshclass}. In many cases, the geographic classification is quite conventional because of mobility of researchers and international collaborations\footnote{there were a number of international collaborations and examples when, for instance, certain students from a Country published works with their supervisors in another one, then obtained a PhD, changed their positions, and developed their research in a third Country}.

\begin{itemize}
\item \textbf{Germany, France, and Switzerland:} B. Riemann \cite{riem1854} and P. Finsler \cite{finsler18} can be considered respectively as the "grandfather" and "father" of Finsler geometry. They defended their famous theses in Germany, where the (pre)history of Finsler geometry and applications began. In France, E. Cartan elaborated the first complete model (with
metrics, N- and d-connections and adapted frames) of Finsler geometry in the book \cite{cartan35}, which stated the first "dynasty" for the "Finsler Kingdom". A number of geometers and physicists published in these Countries, and Switzerland, in various directions of Finserl geometry and applications in physics and mechanics. As activities which are typical both for the Late Middle and Dark Ages, we cite a series of papers due to H. G. Goenner, see \cite{bogoslovsky04} and references therein, and (V. Perlick, 1987) \cite{perlick87}. For elaborating axiomatic approaches to MGTs, it was important to characterize standard clocks by means of light rays and freely falling particles extended to the propagation of particles in Finsler gravity theories.
We cite here a series of works due to C. L\"{a}mmerzahl and co-authors (2009-2014), see \cite{laemmerzahl08,itin14} and references therein, on confronting Finsler spacetime with experiment and Finsler-type modification of the Coulomb law. Such works are important for elaborating experimental tests of various theories with MDRs and LIVs.
\vskip2pt

There are certain conceptual issues which may result both in theoretical and experimental ambiguities in the interpretation of locally anisotropic configurations with prescribed spherical/cylindric/toroid and other types higher symmetries in Finsler MGTs. For instance, a recent research is devoted to possible observable effects in a class of spherically symmetric static Finsler spacetimes, for geometric models of standard static spacetimes, Finslerian extensions of the Schwarzschild metric etc. \cite{laemmerzahl12,lin14,silagadze15,caponio16}. However, it is not clear how a generic anisotropic (generalized) Finsler gravity configuration could be with spherical symmetry and/or static because for a nontrivial N-connection structure any associated total space metric (\ref{offd}) is generic off-diagonal and depends at least on one spacetime (radial or time)
coordinate and one (co) fiber coordinate (velocity or momentum type). Such solutions must be at least stationary, and/or with certain generic anisotropic polarizations of constants, and dependencies of the metric coefficients on velocity/momentum coordinates. 
\vskip2pt

In a series of our works published during 2001-2017, see \cite{vapny01,vijmpd03,vijmpd03a,vjmp05,vijtp10,vcqg10,vijtp13,
vrev08,vijgmmp07,vijgmmp11,gvvepjc14,svvijmpd14,gvvepjc14,gvvcqg15,gheorghiuap16,bubuianucqg17,ruchinepjc17}, we studied black hole solutions in various classes of MGTs with Finsler modifications and nonholonomic manifolds for (non) commutative,
supersymmetric, fractional, stochastic, string and brane models etc. Our conclusion is that such exact or parametric solutions for nontrivial Finsler configurations could be with deformed horizons (for instance, black ellipsoids), ellipsoidal wormholes, parameterized by generic off-diagonal metrics. Such generalized locally anisotropic spacetime models cannot be static on the total (co) tangent Lorentz bundles, or on nonholonomic Lorentz manifolds, with nontrivial N-connection structure. Nevertheless, spherical symmetric configurations may exist and can be considered as self-consistently embedded into certain locally anisotropic background spaces with possible locally anisotropic sources. Such black ellipsoid and more general locally anisotropic spacetime and phase space configurations (with nontrivial topology, noncommutative parameters etc.) can be constructed as exact and/or parametric solutions for all classes of Finsler-Lagrange-Hamilton theories considered in this paper. It is possible to assign certain physical interpretation for such modified Einstein equations and their solutions if a well-defined axiomatic approach and physical principles are formulated (which is one of the main goals of this work). \vskip2pt

Two works \cite{dvgrg03,vdgrg03} on applications of the geometry of nonlinear connections and nearly autoparallel maps in GR and gauge gravity can be considered as an important application of the Finsler geometry methods. Such a collaboration between prof. H. Dehnen and author of this article was supported in 2001 by DAAD in Germany.\vskip2pt

We note that mathematicians in France and Switzerland elaborated on mathematical problems related to functional analysis, topology and Finsler geometry, see \cite{papadopoulos09,akbar88} and references therein. That was an activity characteristic for Classical Geometric Ages. The classical work on geometric mechanics published in Switzerland due to (J. Kern, 1974)
\cite{kern74} belongs conventionally to The (Early) Middle Ages and reflects a substantial influence of research o M. Matsumoto school in Japan (see next item). \vskip2pt

Works \cite{hohmann13,hohmann16,pfeifer11,pfeifer12} have been published by a new generation of authors working in Germany (or originating from and collaborating from Brazil, Estonia etc. with researchers in this Country) on Finsler gravity theories. Certain studies on causality structures related directly to nonlinear Finsler quadratic elements were elaborated. Nevertheless, those directions have not concluded yet in any geometric complete and self-consistent Finsler MGTs, with exact solutions and
perspectives for quantization. In our approach with (co) tangent Lorentz bundles and Lorentz manifolds of higher dimension, such programs have been partially realized by constructing physically important solutions, elaborating geometric quantization models, and offering various perspectives for quantization of MGTs.

\item \textbf{Japan:} There were established collaborations on Finsler geometry between Western Countries and Japan before the World War II. A. Kawaguchi (1937, 1952) \cite{kawaguchi37,kawaguchi52} published a series of important works on Finsler metrics and nonlinear connections, areal metrics and geometries. He was the editor of journal Tensor N. S. where a number of
mathematical papers on Finsler geometry and applications had been published during a period of more than 30 years. In certain cases, those works there were treated with scepticism by mathematicians and physicists working on canonical geometric models and standard theories of physics. In Japan, it was created one of the most influent schools on Finsler geometry in the World due to M. Matsumoto, see \cite{matsumoto66,matsumoto86} and references therein. The almost symplectic formulation of Finsler spaces, Lagrange and Hamilton geometry and higher order generalizations (with various ideas and discussions of perspectives to develop such models of modified and quantum gravity) can be found in M. Matsumoto and co-authors works. Here we cite a more recent collection of works with historical remarks \cite{bao07}.
\vskip2pt

Original unified (including supersymetric), modified gravity and gauge theories, with generalized conservation laws, were elaborated in a series of works due to (S. Ikeda, 1976,...,1981) \cite{ikeda76,ikeda77,ikeda78,ikeda79,ikeda81}. In an alternative way, lifts of geometric objects from base to Finslers spaces (a geometric formalism summarized in \cite{yano73}) were used for constructing models of Finslerian gravity with modified Einstein equations (H. Ishikawa, 1980, 1981)
\cite{ishikawa80,ishikawa81}. The author of this article was inspired by a series of works due to (Y. Takano and T. Ono, 1968-1993) \cite{takano68,takano74,takano83,ono90,ono93} devoted to the theory of fields (gravitational, gauge and spinors) in Finsler variables. In those papers, it was not given a definition of spinors in Finsler geometry and did not found solutions for a Finsler generalization of the Einstein-YM-Dirac equations. In a self-consistent mathematical form, geometric methods for constructing solutions of fundamental locally anisotropic spinor, gauge, and gravitational equations were elaborated in some series of works due to (S. Vacaru and co-authors), see
\cite{v87,v94,vog94,vgon95,vrm5,vrm1,vrm2,vrm3,vjmp96,vhsp98,vplbnc01,vncs01,vdgrg03,vvicol04,vmon98,vmon02,vmon06}
and references therein.

\item \textbf{The UK and Greece: }The problem of axiomatic formulation of unified field theories (including Finsler generalizations) was considered by D. A. Kilmister and G. Stephenson (1953, 1954) \cite{kilmister53,kilmister54}. A complete scheme for locally anisotropic gravity were proposed much later, see in (R. I. Pimenov, 1987) \cite{pimenov87}. Recently, a minimal set of principles for Einstein-Finsler gravity and locally anisotropic cosmology was analyzed in (S. Vacaru, 2012) \cite{vijmpd12} and in this work.
\vskip2pt

Physical gravitational effects in non-Riemann spaces and post Newtonian limits of Finsler theories for solar system were studied in a series of works due to (I. W. Roxburg and R. K. Tavakol, 1979, 1992) \cite{roxburgh79,roxburgh86}.\ Certain viability criteria for a class of Finsler spaces and gravity theories were studied in (R. K. Tavakol and N. Van der Bergh, 1986, 2009) \cite{tavakol86,tavakol09}.
\vskip2pt

Certain attempts to construct and study four dimensional spherically symmetric Finsler spaces for a generalization of vacuum Einstein were considered in (P. J. McCarthy and S. F. Rutz, 1993) \cite{mccarthy93,rutz93}. Those models do not provide formulas and solutions for self-consistent Finsler generalizations of the Einstein gravity with matter fields.
\vskip2pt

Studies of some special classes of Finsler spaces and a review of Landsberg spaces were elaborated by mathematicians, see (C. T. J. Dodson, 2006) \cite{dodson06} and references therein.
\vskip2pt

\textit{The Renaissance } period began with a series of works in collaboration with young-researches (at that time) originating and/or working at present in the UK, Italy, Greece, USA, etc. Those works were not written in and Finsler geometric language but studied MDRs and LIV, for instance, in the Liouville string approach to QG, doubly special relativity, spacetimes with short distance structure and (Planckian) length. Here we cite a series of works due to G. Amelino-Camelia, J. R. Ellis, N. E.
Mavromatos, D. V. Nanopoulos, S. Sarkar (1996-2002) \cite{amelino96,amelino98,amelino02,amelino02a,mavromatos07}. During last decade, N. E. Mavromatos and co-authors published a series of very important works on decoherence and CTP violation in string QG and stringy models of spacetime form; studied possible Lorentz-invariance violation and possible effects in gamma-ray astronomic, CPT-violations and leptogenesis induced by gravitational defects etc.
\cite{mavromatos10,mavromatos10a,mavromatos11,mavromatos13a,mavromatos13b}. We consider a common classification for research on Finsler MGTs in the UK and Greece because one of the most prolific authors in such directions originates from Greece but published a number of important papers with affiliations in the UK and CERN (N. E. Mavromatos). Such works contained a series of concepts and methods from Finsler geometry and emphasized the importance of locally anisotropic string and gravity configurations. There were cited mutually a series of former correlated results and methods elaborated by (S. Vacaru and co-authors, 1996-2007), see \cite{vrm2,vrm3,vap97,vnp97,vncs01,vmon98,vmon02,vmon06}. More recently, a series of papers due to (J. Magueijo and co-authors, 2004 - present) was submitted with the UK and international affiliations, see \cite{magueijo04,barcaroli15} and references therein. \vskip2pt

During last 40 years, an influent research team on Finsler geometry and applications in gravity and cosmology was distinguished due to P. Stavrinos from the Athens University and his co-authors (former students and international collaborations) \cite{stavr99,stavr04,stavr05,kouretsis10,kouretsis14,kouretsis14a}. Here, it should be mentioned an important collaborations between P. Stavrinos and S. Vacaru former groups in the R. Moldova and Romania which resulted in common publications of two monographs and a series of important papers with applications of Finsler geometry in physics, cosmology and astrophysics, see
\cite{vmon02,vmon06,svcqg13,svvijmpd14} and references therein. The Athens group on Finsler geometry and applications established certain collaborations on MGTs and cosmology with researchers working in the USA, China and other countries, see \cite{cai14}.
 \vskip2pt

Recently, a series of work with application of Finsler geometry to quantum information processing were published by authors (B. Russell and S. Stepney, 2013-2017) from York University in the UK and authors at Princeton University, USA, see \cite%
{russellb14,russellb14a} and references therein.

\item \textbf{The USA, Canada, Brazil and Mexico:} H. Rund's monograph \cite{rund59} published in 1959 began a conventional \textit{Middle Ages} period in developing Finsler geometry and applications. That book was translated into many languages and, for instance, inspired substantially G. Asanov to create a "Soviet" school of Finsler applications in physics
    \cite{asanov85,asanov88,asanov89,aringazin88}. Later, a renewed interest of S. -S. Chern \cite{chern48} resulted in the publication of a series of important mathematical works on Finsler geometry by his school (D. Bao and Z. Shen;
beginning 2000), see \cite{bao00,shen01,shen01a} and references therein. During last decade, Chern-Finsler geometry was applied for elaborating non-standard particle and gravity theories with applications in modern cosmology and astrophysics, see a series of work by Chinese authors \cite{chang08,chang12,li14,li14a,li16,lin14}. Nevertheless, such models are with metric noncompatible Finsler connections which results in a series of problems for the definition of spinors and self-consistent models for locally
anisotropic interactions, see critics in \cite{vplb10}.
 \vskip2pt

An original research on Finsler Kaluza-Klein gauge and relativistic field theories was performed by R. G. Beil (1993-2003) \cite{beil93,beil03}. Finslerian quantum fields and K\"{a}hler spacetime models on tangent bundles were studied in the USA by H. E. Brandt (1992-2012), see \cite{brandt92,brandt98,brandt99,brandt00,brandt03,brandt04} and references therein. His approach allowed to perform certain constructions for Lorentz-invariant quantum field locally anisotropic interactions on tangent
bundles. The main problems of the models elaborated by R. G. Beil and H. E. Brandt are those that they could not support their constructions with exact solutions and certain general quantum locally anisotropic theories involving the multi-metric/-connection character of such geometries. Here we note that J. Vargas and co-authors studied canonical connections of Finsler metrics and Finslerian connections on Riemannian metrics but those constructions have not resulted in complete physical theories, see \cite{vargas96} and other publications by this author. Nevertheless, such geometric models have not involved relativistic generalizations which would result in physically viable gravity and particle models with locally anisotropic interactions. %
\vskip2pt

A series of publications involving modifications of Finsler geometry and applications in modern physics belong to C. Castro (he also publishes under the name C. Castro Perelman), see his publications beginning 1995
\cite{castro97,castro07,castro08,castro08a,castro11,castro12,castro14,castro16}. He elaborated on various original constructions with Finsler and other types locally anisotropic structures in string gravity, noncommutative and nonassociative geometry, Born's deformed reciprocal complex gravity theories, Clifford group geometric unification etc. Even his first works were not based on N-adapted geometric formulas (which is important in Finsler geometry) the later ones were performed using advanced geometric methods with a number of mutual citations of S. Vacaru and co-authors works. It was established a strong collaboration and influence of ideas and geometric methods elaborated by these authors. Recent C. Castro papers are devoted to certain original modified gravity models on curved phase-spaces, Clifford space relativity, black hole entropy, modified dispersion relations etc. which is related to the paradigm of this article.
\vskip2pt

At present, a team of researchers in the USA, Brazil, Mexico, and other countries collaborating with and/or supervised by  professor V. A. Kosteleck\'{y} consists one of the most influent schools in the World on applications of Finsler geometry in modern physics, see \cite{bailey15,kostelecky11,kostelecky12,foster15,russell15} and references therein.  The most important publications of such authors are on phenomenological aspects of theories with MDR and LIVs which can be geometrized as modified Finsler spaces  \cite{silva15,silva16,reis16}.
\vskip2pt

In North America, a number of works were published on Finsler applications in biology and physics due to P. L. Antonelli and co-authors, see summaries of results in a series of monographs beginning 1993 \cite{anton93,anton96,anton03}. A number of collaborations with researchers from Canada, Japan, Romania, Brazil and other countries were established. Their two-volume handbook may be considered as a summary of a direction in Finsler geometry and applications related to the period of the Dark Non-Standard ages. Their theories had not involved relativistic Finsler constructions and modifications of standard gravity and particle theories. Nevertheless, it should be noted here that after P. Antonelli moved his research activity to Brazil a Finsler network for research on locally anisotropic MGTs was stated in South America. Here we mention papers by S. F. Rutz
\cite{mccarthy93,rutz93} and R. Gallego \cite{gallego06,gallego17} who published, for instance, on generalized Einstein equations and a Finslerian version of 't Hooft deterministic quantum models \cite{hooft96}. Lorentz-violation modifications of the standard model and GR are studied by M. Schreck and co-authors publishing at present from Brazil
\cite{schreck15,schreck15a,schreck16,schreck16a,reis16}. It was established a specific Finsler collaboration between physicists working in different periods in Canada, Brazil, Spain, the UK and Netherland. Finally, we note that and influent work on gravity's rainbow (involving Finsler like metrics) was published in 2004 by J. Magueijo and L. Smolin \cite{magueijo04} working
in Canada and the UK.
\vskip2pt

Finsler gravity and nonholonomic geometry methods were developed in the USA due to collaboration of D. Singleton with S. Vacaru and some  students from the R. Moldova, with support from California State University at Fresno in 2000-2002, and renewed in 2017, see typical publications \cite{vsbd01,vsjmp02} and references therein. There constructed in explicit form certain classes of locally anisotropic (ellipsoidal, cylindrical, bipolar and toroidal) wormhole and flux tube solutions in five dimensional gravity which can be redefined on (co) tangent Lorentz bundles. Recently, some support for such activities (oriented to constructing quasi periodic exact solutions in modified gravity and cosmology) was obtained from a business organization in Topanga, California \cite{vacaruepjc17}.

\item \textbf{Spain and Italy:} We cite here a series of works due to mathematicians (M. S\'{a}nchez, E. Caponio, M. A. Javaloyes and others) \cite{havaloyes13,caponio14,caponio16} collaborating on fundamental issues of causality of Finsler metrics, Zermelo' navigation and relativistic spacetimes, and attempts to define and construct static locally anisotropic configurations. Such configurations are at least stationary if they are derived as solutions of certain generalized Einstein-Finsler equations with nontrivial N-connections for d-metrics. Recent studies are devoted to the (non-)
uniqueness of linear connection structures in the Einstein-Hilbert-Palatiny approach which can be extended to modified Finsler gravity theories \cite{bernal16}. Locally anisotropic light cone physics, the Raychaudhury equation, and singularity theorems are studied for certain examples of Finsler spacetimes due to a series of recent works of E. Minguzzi
\cite{minguzzi15,minguzzi15a,minguzzi16}. Such constructions do not involve explicit examples of exact or parametric solutions in an Einstein-Finsler gravity model as it was considered in Refs. \cite{vjmp05,vijtp10,vcqg10}. %
\vskip2pt

In indirect and direct form, modified Finsler like gravity configurations can be considered in theories of relativity with short distance (Planckian) length scale, or in doubly special relativity and modified dispersion relations \cite{amelino02a,amelino14,girelli06} (see also citations therein on publication due to G. Amelino Camelia, F. Girelli and S. Liberati groups). Such theories can be related to MGTs with various f(R,T,....) and massive, bimetric, or biconnection modifications \cite{hassan12,nojiri07,nojiri13jcap,derham10,derham11,capozziello10}. A series of such works were published as collaborations of researchers from Spain, Romania and Germany
\cite{gvvepjc14,gvvcqg15,gheorghiuap16,elizalde15,vport13,vijgmmp14,vepjc14,vepjc14a,vvyijgmmp14a}.
\end{itemize}

\subsection{Developing Finsler geometry in Romania and R. Moldova}

\label{ssectromania} During a period of more than 80 years, the research on nonholonomic and Finsler geometry and applications developed by authors originating from, and/or working in, Romania and R. Moldova have an international impact and importance. As a consequence of the Hitler-Stalin pact which resulted in a long-term occupation of Bessarabia and North Bukovina by the Soviet Union, the conditions, method and style of research on geometry and physics (in partial, on Finsler MGTs) were very different in Romania and R. Moldova. After the anti-Ceau\c{s}escu revolution in Romania (1989) and crash of USSR (1991), one appeared new possibilities for collaboration and obtaining financial support for research from Western Countries. In Romania, there were
elaborated new directions in Finsler geometry with applications in geometric mechanics under the conditions that various bureaucratic and corruption constraints, and "higher order plagiarism", were imposed by the former "secret service guard for science", see footnote \ref{fn01a}. The bulk of new results were not published in top international mathematical and physics journals. Due to certain NATO, UNESCO, West governmental and private grants and fellowships, research on Finsler MGTs was performed during visits in Western Countries, which resulted in more than a hundred of scientific works published in high impact mathematical physics and applied mathematics journals.

\subsubsection{Research and activities before World War II}

Studies on nonholonomic manifolds and Finsler geometry began in Romania due to G. Vr\v{a}nceanu, see footnote \ref{fvranceanu}, and M. Haimovici fellowships in Western Countries. As young researchers, they worked under supervision, and/or collaboration, by famous mathematicians like D. Hilbert, E. Cartan, T. Levi-Civita etc. Later they were elected as members of the Romanian Academy of Sciences. After the Ward Wold II, their research was summarized, for instance, in monographs \cite{vranc57,haimovichi84}, censored by the communist regime. More than 300 research papers and tenths monographs published by G. Vr\v{a}nceanu and M. Haimovici were reviewed in MathSciNet and Zentralblatt MATH. It is important to consider reviews in such databases because they summarized a number of original ideas and results published in Eastern European journals which are almost not accessible to Western researchers.

It should be noted that the geometric schools supervised by G. Vr\v{a}nceanu and M. Haimovici developed a "pure" mathematical activity which was less influenced by communist ideology. In general, they followed Western standards for science. As members of the Academy and authorized leaders of such groups, they had access to international scientific literature and were allowed to participate in International Conferences and publish in Western Journals. Both authors contributed substantially in developing a number of directions in geometry and applications and selected and supervised a number of young researchers to work on nonholonomic and Finsler geometry. Later, the communist regime imposed a style of isolation of research even for internationally recognized mathematicians.

Review of mathematical works and certain applications elaborated in Romania during the Classical Period can be found in some monographs due to A. Bejancu \cite{bejancu90,bejancu03} published in 1990 and 2003. New directions of research and main references on Finsler methods in physics and biology are summarized in details in \cite{vmon98,vmon02,vmon06,anton93,anton96,anton03}.

\subsubsection{Ceau\c{s}escu's communist dictate and Finsler geometry}

Prof. Aurel Bejancu began his research activity in Ia\c{s}i before 1980 and held positions at the Polytechnic University. His first important for applications in physics publications were as preprints "Seminarul de Mecanica, University of Timi\c{s}oara". The main results of those papers were included and developed in monograph \cite{bejancu90}. A. Bejancu is the author of a series monographs published by prestigious Western Editors (as well some tenths of high impact papers) and can be considered as the most prominent Romanian mathematician who contributed to development, generalizations and applications of Finsler geometry after World War II. He
obtained fellowships in Western Countries and moved hist activity to Kuwait University. That allowed him to survive and later avoid conflicts  with Radu Miron's school.\footnote{From a formal point of view, R. Miron and co-authors published much more books and papers on Finsler geometry and applications than any international and/or Romanian scholar. Unfortunately, there are a number of moral and legal issues with those works directed and supervised by the communist secret service in a
 "half-plagiarism and half-slavery" style. It was a specific policy to hide and do not publish in Western Countries that activity. Contrary, A. Bejancu worked in Romania almost individually and later, in Kuwait, collaborated with internationally recognized scholars publishing his works with top Western Editors.} A. Bejancu elaborated on original Finsler geometry models with supersymmetric variables, Kaluza-Klein and gauge like theories of Finsler gravity. His results were not much known in the physical literature because the constructions were performed in a non-relativistic form, without solutions and quantization of locally anisotropic field equations.

The most ambiguous was the activity of R. Miron's school on Finsler gravity and applications as we mentioned in footnote \ref{fn01a}. For Western Scientists, it is known that the Stalin/Soviet regime imposed a style of ideological screening, falsification and stealing of scientific results from Western Countries and even from scientists not validated by communist censors. In 1968, it was an opening to international research collaboration after the dictator I. Ceau\c{s}escu refused to participate in the military invasion of Czechoslovakia by the Soviet block. The style of pseudo-research on Finsler geometry developed by R. Miron's school was that a number of ideas and preprint publications at universities in Japan, Romania and other countries were taken and modified without a corresponding citation. There were published more than 30 monographs and more than 250
papers under the name of R. Miron (much more than most important authors on geometry and physics in Western Countries), see
\cite{miron87,miron94,mironatanasiu94,mironatanasiu96,miron00,miron03} and references therein. Such monographs were not grounded on important authors' publications in top journals of the American / European mathematical and physical societies but contained a number of results elaborated and modified by unknown authors in Romania. It is a problem how to cite such works and distinguish original contributions from "higher order plagiarism", quasi-corruption and communist ideology screening. In the past, the author
of this article cited and highly appreciated and cited R. Miron and company works in Refs. \cite{vmiron98,vgrg05,vmon98,vmon02,vmon06} but have to reconsider his former scientific evaluation on Finsler geometry and
applications in Romania. Beginning 2008, one have been published some files and newspapers in Western Countries and Romania how Ceau\c{s}escu's secret service supervised scientific activity via their agents at the Academy of Science and Universities administration.

Finally, we note that in MathSciNet there are reviewed a number of works published in Romania and other countries by authors belonging/collaborating/related with the so-called R. Miron's "school" on Finser geometry and applications. Nevertheless, such mathematicians published individually, or with other co-authors, certain important works related to geometric methods in modern physics and mechanics. We cite such monographs and articles due to M. Anastasiei \cite{avjmp09,anton93,anton96}, C. Arcu\c{s} \cite{arcus14}, G. Atanasiu \cite{mironatanasiu94,mironatanasiu96},V. Balan and M. Neagu \cite{balan16}, I. Bucataru \cite{bucataru07}, G. Munteanu \cite{munteanu04}, V. Oproiu \cite{oproiu85}, C. Udri\c{s}te \cite{udriste02}, N. Voicu \cite{voicu17} etc.

\subsubsection{Support for researchers from the R. Moldova and geometry and physics}

\label{assrm}

During 1998-2001, the (neo) communist authorities tolerated an attempt to organize at the Institute of Applied Physics, the Academy of Sciences of the R. Moldova, a research group on geometric methods and applications in modern physics (see also comments on Directions 1, 2, 3, and 5; in respective appendices \ref{sssdir01},  \ref{sssdir02},  \ref{sssdir403}, and \ref{sssd05}, on related author's research activity during 1978-1997). More than ten young researchers were involved in that activity in the capital of that former Soviet republic, Chi\c{s}inau/Kishinev city. Main publications of that group were due S.\ Vacaru, S. Ostaf, Yu. Goncharenko, E. Gaburov, D. Gontsa, N. Vicol, I. Chiosa etc. \cite{v94,vog94,vgon95,vrm5,vrm1,vrm2,vrm3,vrm4,vrm5} (such works were published in the R. Moldova and Romania in almost not accessible journals but reviewed in MathSciNet and Zentrallblat), and \cite{vcb1,vgon95,vcb2,vjmp96,vstoch96,ve2,vap97,vnp97,vexsol98,vmon98,vapl00,vapny01,vplbnc01,vncs01,vjhep01,vpcqg01,
vsbd01,vmon02,vtnpb02,vsjmp02,vijmpd03,vijmpd03a,vdgrg03,vcb3,vvicol04}(series of important papers and high impact international journals and two monographs on theoretical and mathematical physics and gravity), see also references therein. In 2001, the new elected communist government of that Country persecuted directly on political grounds a number of researchers of Romanian ethnic origin and deported S. Vacaru's family out of R. Moldova. In
those actions supervised by the Russian secret service, there involved leaders and administrative workers of the Academy of Sciences of R. Moldova like A. M. Andrie\c{s}, V. A. Moscalenko, S. A. Moskalenko, M. K. Bologa, L. Kulyuk (Culiuc), V. Kantser (Can\c{t}er), D. Digor, and others.  During research activity on "Finsler geometry and applications", a series of fundamental human and intellectual rights of young researchers were violated because they accepted (without authorization by the communist regime) some NATO, DAAD and UNESCO fellowships with visits at ICTP Trieste, Italy, I. Newton Institute at Cambridge University, and Konstanz University in Germany.\footnote{The author of this article considers that it is important to inform researchers in Western Countries on the conditions of work and activity of scientists in the former Soviet Union and Eastern Europe. Even his opinions and evaluations could be very subjective, it is important to have evidence and understand how certain research under control of communist officials resulted in new results in Finsler geometry and applications.} A number of ideas, methods and results on elaborating physical theories on nonholonomic manifolds and (co) tangent Lorentz bundles, with various supersymmetric and noncommutative generalizations and using Finsler methods, were originally proposed and realized in the R. Moldova. Those publications were related to more than 20 long term research programs, see \ref{ass20directions} and fellowships (hosted by Universities and Research Institutes in USA, Germany, Canada, Spain, Portugal, Greece, Romania etc. which resulted in more than 150 publications and 3 monographs).

Here we emphasize that a number of professors in Western Countries (P. Stavrinos, Greece; D. Sigleton and S. Rajpoot, USA; N. Mavromatos, UK; H. Dehnen, Germany; J. P. Lemos, Portugal; M. de Leon, Spain; J. Moffat, Canada) provided important recommendations and supported collaborations which were crucial for activity of research on Finsler geometry and applications in R. Moldova and obtaining new positions in Western Countries and Romania. It was also certain financial support from R. M. Santilli and K. Irwin (some private business organizations in the USA with certain interests in mathematics and physics) but unfortunately, those collaborations imposed certain pseudo-scientific flavour contrary to S. Vacaru's research programs.

It should be noted that certain research on Finsler geometry and applications was performed in the R. Moldova by authors not belonging to S. Vacaru's research group. Here we cite a monograph \cite{asanov88} published in Russian (Kishinev /Chisinau, 1988) by S. F. Ponomarenko together with his Ph. D supervisor, G. Asanov, from Moscow. Some works on Finsler methods and systems of nonlinear partial differential equations were published by V. Dryuma and M. Matsumoto, see \cite{dryuma98,dryuma00} and references therein.

The first collaborations between researchers on geometry and physics in the R. Moldova and Romania were established at a Romanian National Conference on Physics (Cluj-Napoca, October 1990), which was possible after the anti-Ceau\c{s}escu revolution. Author's talk was on nonholonomic geometric methods for elaborating a diagram techniques for minisuperspace twistor quantum cosmology was published in Studia Universitatis Babes-Bolyai, Cluj-Napoca, Romania, \cite{v90} (reviewed in MathSciNet). During 1991-1992, the research on gravity and Finsler geometry at the Academy of Sciences of R. Moldova was re-oriented to a more than 10 years collaboration with Romanian professors and senior researchers at Ia\c{s}i and Bucharest (I. Gottlieb, C. Mociu\c{t}chi, M.Vi\c{s}inescu, G. Zet, M. Piso, H. Alexandrescu, R.\ Miron, M. Anastasiei and others). That collaboration was possible due to substantial support by Cleopatra Mociu\c{t}chi at Al. I. Cuza Ia\c{s}i University (she originated from Bessarabia, at present R. Moldova), V. Manu, and M. Vi\c{s}inescu at the Institute of Atomic Physics, Bucharest-Magurele. Having support from Western Countries and Romania, the author of this review was able to elaborate on a series of new directions on generalized Finsler geometry and applications (see next subsection). We cite here a series of works elaborated in the R. Moldova during 1990-2001: \cite{v94t,vgon95,vrm1,vrm2,vrm3,vrm4,vrm5,vjmp96,vstoch96,vcb2,ve2,vcp1,vap97,vnp97,vmiron98,vhsp98,vmon98,anton96,vgon95,
vncs01,vjhep01,vpcqg01,vsbd01,vtnpb02}. Unfortunately, that collaboration had a number of difficulties and stopped because of control of corruption and control of former "security" administrators in Romanian science like Luci Popa and D. Ha\c{s}egan (from the Institute of Space Sciences, Bucharest-Magurele) and Radu Miron (Ia\c{s}i).

\subsection{Twenty "Moldavian" research directions in generalized Finsler geometry and applications}
\label{ass20directions}

We present a synopsis of \textbf{Twenty Main Research Directions, and respective sub-directions,} of author's research activity, collaborations, and main publications. The goal of this subsection is to describe in brief the conditions for scientific activity in the post-Soviet R. Moldova and review twenty main directions of author's research related to Finsler geometry and applications in modern physics. There will be also discussed further important developments, support from Western Countries, and collaborations.

\subsubsection{Nonlinear gravitational - electromagnetic optics, twistors and gauge gravity}
\label{sssdir01}
For Direction 1, there are distinguished such sub-directions and important results (based on publications \cite{v83,v87,v90,v94t}):
\begin{itemize}
\item[a)] Tensor models of locally anisotropic interactions of nuclear matter

\item[b)] Nonholonomic constraints in nonlinear dispersive media and nonlinear optics

\item[c)] Parametric multi-photon graviton processes

\item[d)] Nonlinear interactions in locally anisotropic media

\item[e)] Gauge models of gravity and twistors
\end{itemize}

\textit{Comments:} In 1978, the author of this article studied optionally a series of textbooks and monographs written/ translated in Russian on differential and Finsler geometry, when he began his research activity as an undergraduate student on nuclear physics at Tomsk Polytechnic University in former USSR/ Russia. The goal of that program was to model geometrically some examples of locally anisotropic interactions of nuclear matter and radiation processes in media with dispersion. There were not considered gauge models of strong interactions (quantum chromodynamics, QCD) but only nonlinear locally anisotropic generalizations for scalar, spinor and electromagnetic interactions. Later, certain methods of nonholonomic geometry were applied in his university diploma thesis (an equivalent of a master thesis in Western Countries) performed at the Laboratory of Nuclear Problems, Joint Institute of Nuclear Research in Dubna, Russia. The aim of that graduation work and a further program for a young researcher activity (supervised by habilitation doctor A. F. Pisarev during 1980-1984) was to study multiphoton-graviton resonance processes and parametric amplification in gaze like and locally anisotropic media and with modified dispersion relations. There were also elaborated models of "grasers" (gravitational analogues of lasers) using methods of nonlinear optics, geometry, and quantum theory in condensed matter physics, see \cite{v83} and references in that collection of works (such works are listed in inspirehep.net). Certain models of locally anisotropic cosmic radiation (considered by G. A. Asanov and co-authors at Moscow State University, and published a series of works on such applications of Finsler geometry, see references in \cite{asanov85,asanov88,aringazin88}) were studied by S. Vacaru in 1982 when he began his job as an adjunct engineer at the Astrophysical Observatory (in Ciuciuleni-Lozova, district Nisporeni) of the State University of the R. Moldova. \vskip2pt

Another research program involving nonholonomic geometric methods was on twistor-gauge methods. It was authorised by the faculty administration for S. Vacaru when he became a post graduate student at the Department of Theoretical Physics at
"M. Lomonosov" Moscow State University during 1984-1987. Geometric models of gravity and gauge theories on complex manifolds and spinor and twistor spaces were elaborated using the formalism on nonholonomic distributions in the sense of G. Vr\v{a}nceanu \cite{vranc57} and other works on nonholonomic mechanics and classical and quantum field interactions. A typical such publication was \cite{v87} on a twistor-gauge interpretation of the Einstein-Hilbert equations generalized in order to include bi-metric theories. It was not possible for S. Vacaru to defend his Ph. D thesis at Moscow (supervisor prof. D. Ivanenko and co-supervisor dr. Yu. N. Obukhov) because of various political, human right activity and historical issues related to "perestroika" and "crash" of the Soviet Union. He had to translate his thesis from Russian into Romanian and complete with new results on gauge models of spin glasses, high temperature superconductivity, and minisuperspace twistor quantum cosmology \cite{v90}. Those results were obtained during his work as a young researcher at the Institute of Applied Physics, Academy of Sciences of R. Moldova (during 1988-1992); master student studies (with second university specialization) at the Department of Radiophysics, "T. Shevchenko" Kiev State University (during 1988-1989); and postgraduate student studies at the Physics Department, University "Al. I. Cuza" at Ia\c{s}i, Romania (supervisor prof. Ioan Gottlieb). S. Vacaru's Ph. D thesis \cite{v94t} was defended and validated by the Romanian Ministry of Education in 1994. \vskip2pt

It should be emphasized here that it was organized a specific policy for selecting young researchers in the former USSR who could  be allowed to perform individually certain research activities related to new directions of mathematics and modern physics. Supervision of such activities by senior scientists was quite formal. For S. Vacaru, the priority was given to gravitational waves, gauge theories of gravity, spinor and twistor geometry). In parallel, he became familiar with international advances of Finsler geometry and applications in physics after discussions at Moscow with A. K. Aringazin
\cite{aringazin88} and S. F. Ponomarenko \cite{asanov88} - at that time, they were postgraduate students of G. Asanov. Those young researchers provided copies of a series of monographs on Finsler geometry due to M. Matsumoto \cite{matsumoto86} and A. Asanov \cite{asanov85,asanov88}. In 1992, S. Vacaru studied the Romanian version (published in 1987) of the R. Miron and M. Anastasiei monograph (translated and up-dated for Kluwer in 1994, \cite{miron94}) which was a conventional review (with a communist style political and ideological screening of authors) of a number of works elaborated by authors in Japan, USSR, Romania, Hungary and other countries.%
\vskip2pt

So, the beginning and young research activity of the author of this article was on geometric methods in physics, when he was undergraduate/graduate/ master/postgraduate student in the former USSR and Romania during 1977-1994. The main conclusion for that period with publications \cite{v83,v87,v90,v94t} on the mentioned sub-directions a)- e) was that relativistic models of locally anisotropic continuous media, gauge gravity theories, definitions of twistors for curved spacetime etc. could be
formulated in certain forms closed to GR and standard particle physics but extending the geometric constructions on nonholonomic manifolds and Finsler models on (co) tangent Lorentz bundles.

\subsubsection{Locally anisotropic (non)commutative gauge gravity and perturbative quantization}
\label{sssdir02}

Direction 2 was structured into four sub-directions following a series of publications
\cite{v87,vmon98,vmon02,vmon06,vrm2,vrm3,vgon95,vplbnc01,vjmp05,vcqg10,dvgrg03,vdgrg03,vijgmmp10,vijgmmp10a,vch2416}:
\begin{itemize}
\item[ a)] Affine and de Sitter models of gauge (super) gravity

\item[ b)] Gauge like models of Lagrange--Finsler and Hamilton-Cartan gravity

\item[ c)] Locally anisotropic gauge theories and perturbative quantization

\item[ d)] Noncommutative gauge gravity on (co) tangent Lorentz bundles and supermanifolds.
\end{itemize}

\textit{Comments:} This direction can be considered as a development of research outlined in some chapters and sections in author's Ph. D thesis \cite{v94t}, where affine and de Sitter gauge like models were studied in the framework of the so-called twistor-gauge formulation of gravity. There were elaborated models with one metric structure but with nonholonomic
(non-integrable) twistor configurations and (alternatively) bi-metric structures when a (background) metric allows a definition of integrable twistor structures. Such constructions were elaborated originally in \cite{v87,v90}. The first publications on anisotropic gauge gravity theories were in R. Moldova \cite{vrm2,vrm3} (1994-1996), see also a paper in a peer reviewed journal, together with a graduate student, Yu. Goncharenko (1995) \cite{vgon95}. Authors were allowed by the (neo) communist officials in R. Moldova to present their results for publication in  Western Journals only beginning 1994. The geometric constructions in those works were based on the idea that the Einstein equations can be equivalently reformulated as some Yang-Mills equations for the affine (we used the Popov--Dikhin approach, 1976) \cite{popov75,popov76} and/or de Sitter frame bundles (A. Tseytlin generalization, 1982) \cite{v87,v90}. In those works, there were considered nonlinear realizations of corresponding gauge groups and well-defined projections on base spacetime (non) commutative manifolds, or supermanifolds.
\vskip2pt

To formulate (non) commutative and/or supersymmetric gauge theories of Lagrange-Finsler gravity, we used the Cartan connection in corresponding generalizations of the affine and/or de Sitter bundles. The Cartan connection in such bundle spaces should be not confused with the Cartan d-connection in Finsler geometry. Here we note that the constructions with
the Cartain d-connection on a nonholonomic manifold and/or (co) tangent bundles can be lifted on total spaces of (in general, complex) vector bundle and/or associate frame affine/ de Sitter bundle spaces. In result, we can construct induced affine/ de Sitter d-connectoin structures in total space. As base spaces, there were considered Finsler (super) manifolds and various
anisotropic generalizations including constructions for nonholonomic higher order tangent/vector superbundles. Such results were summarized in a series of chapters of monographs \cite{vmon98,vmon02,vmon06} and presented at a NATO workshop in 2001, see \cite{vncs01}. \vskip2pt

Re-writing the Einstein gravity and generalizations as gauge like models allowed this author to perform one of the most cited his works \cite{vplbnc01}. That paper was devoted to nonholonomic Seiberg--Witten transforms and noncommutative generalizations of Einstein and gauge gravity. The gauge gravitational equations with noncommutative deformations can be
integrated in very general off-diagonal forms \cite{vjmp05}; see also Ref. \cite{vcqg10} on noncommutative Finsler black hole solutions. \vskip2pt

It should be mentioned here an important collaboration with Prof. H. Dehnen (Konstanz University, Germany, 2000-2003) supported by DAAD. The research was on higher order Finsler-gauge theories, generalized conformal and nearly autoparallel maps and conservation laws, see \cite{dvgrg03,vdgrg03}. A recent approach to two-connection perturbative quantization of gauge gravity models \cite{vijgmmp10,vijgmmp10a} had been elaborated. All results outlined in above sub-directions can be considered on base (super/noncommutative) spaces endowed with generalized Finsler-Lagrange-Hamilton structures. It is possible to integrate such generalized gauge gravitational models for Lagrange-Hamilton commutative and noncommutative/nonassociative variables, see a recent work \cite{vch2416}. The two connection renormalization can be compared with certain non-perturbative constructions in loop gravity, asymptotic freedom of GR and MGTs.

\subsubsection{Nonholonomic Clifford structures and spinors for Finsler-Lagrange-Hamilton spaces}
\label{sssdir403}
There were distinguished such sub-directions of Direction 3 (see publications
\cite{vrm5,vcb3,ve2,vjmp96,vhsp98,vvicol04,vmon98,vmon02,vmon06,vjmp09,vpcqg01,vtnpb02,vhthes12,vaaca15}):

\begin{itemize}
\item[a) ] Spinors and Dirac operators on generalized Finsler spaces

\item[b) ] Nonholonomic Clifford structures with nonlinear connections

\item[c) ] Spinors and field interactions in higher order anisotropic spaces.

\item[d) ] Solutions for nonholonomic Einstein--Dirac systems and extra dimension gravity
\end{itemize}

\textit{Comments:} A nonholonomic Lorentz manifold and/or it (co) tangent bundle space are by definition enabled with a nonholonomic (equivalently, anholonomic, or non-integrable) distribution defining a N-connection structure. The problem of definition of spinors and Dirac operators on nonholonomic manifolds and/or Finsler-Lagrange spaces was not solved during
60 years after first E. Cartan's monographs on spinors in curved spaces and Finsler geometry (during 1932--1935). A series of works (beginning 1995, due to S. Vacaru and co-authors) \cite{vrm5,vcb3,ve2,vjmp96,vhsp98,vvicol04} were devoted to the geometric definition of locally anisotropic spinors and Dirac operators and elaboration of physical models with
Finsler-Lagrange-Hamilton spaces. The monographs \cite{vmon98,vmon02,vmon06} and recent papers \cite{vjmp09,vpcqg01,vtnpb02,vhthes12,vaaca15} contain a series of chapters and sections devoted to nonholonomic spinors and twistors and developments for supergravity and noncommutative geometry with applications in modern physics of the so-called nonholonomic Clifford geometry.
\vskip2pt

The first attempts to introduce spinors in Finsler geometry were made in a series of works due to Takano and Ono, in Japan, and Stavrinos, in Greece; see main references and historical remarks in \cite{ono90,ono93,takano83}. The idea was to consider two-dimensional spinor bundles on Finsler spaces defined by generating functions depending on spinor and dual spinor
variables. Nevertheless, those works were not complete because they had not answered the question if and when spinors of arbitrary dimensions could be defined for Finsler spaces. For instance, there were not studied possible relations between Clifford generations (generalized Dirac matrices) and Finsler metrics, or related Hessians. There were not provided self--consistent definitions of Dirac operators for Finsler spaces and there were not studied definitions of N-adapted Clifford algebra structures, and spin d-operators, to Finsler metrics and connections. \vskip2pt

In 1994, the author of this article became interested in elaborating theories of gravitational and matter field interactions on generalized Finsler spaces. The problem was considered in a more general context related to G. Vr\v{a}nceanu's definition of nonholonomic manifolds. The paper \cite{vrm5} contains the first self-consistent definition of Clifford structures
and spinors for Finsler spaces and generalizations. The geometric and physical theories were formulated in a more rigorous form (with developments for complex and real spinor Lagrange--Finsler structures and Dirac operators adapted to N-connections) were published also in 1996, see \cite{vjmp96}. In that year, it was approved a special research grant Romania-R. Moldova, supported by the Ministry of Research and Education of Romania. The (neo) communist government of R. Moldova prohibited that collaboration and the Romanian grant was affiliated with Al. I. Cuza University at Ia\c{s}i. Main
research was performed by S. Vacaru in Chi\c{s}inau, R. Moldova, with an administrative supervision by R. Miron, from Ia\c{s}i, Romania. It was bought a computer for the Moldavian team but 90\% of financial resources were taken by the R. Miron team. In result, any further collaboration was stopped by S. Vacaru in 1998 (who decided to not involve his group in any
corruption academic activities in Romania). \vskip2pt

Having elaborated the concept of nonholonomic Clifford bundles and Clifford-Lie algebroids \cite{vrm5,vcb3,ve2,vjmp96}, it was possible to construct geometric models of gravitational and field interactions on (super) spaces with higher order anisotropy \cite{vhsp98,vvicol04,vmon02}. There were obtained important results on formulating differential spinor
geometry and supergeometry studied possible applications in high energy physics. Two professors from Greece (P. Stavrinos and G. Tsagas) and some young researchers from R. Moldova\footnote{at that time a under-graduate and a postgraduate students, N. Vicol and I. Chiosa}, see \cite{vncs01,vmon98,vvicol04}, and Romania\footnote{postgraduates F. C. Popa and O. \c{T}in\c{t}\v{a}reanu-Mircea} supervised by M. Vi\c{s}inescu, see \cite{vpcqg01,vtnpb02}. Those works were summarized in monographs \cite{vmon02,vmon06} containing a series of new results on locally anisotropic Dirac spinor waves and solitons, spinning particles, in Taub NUT anisotropic spaces, solutions for Einstein--Dirac systems in nonholonomic higher dimension gravity and Finsler modifications of (super) gravity. \vskip2pt

Using N-adapted Dirac operators on generalized Finsler spaces, it was possible to elaborate models of noncommutative Finser geometry and gravity and generalized G. Perelman functionals to describe evolution of noncommutative geometries defined in the approach defined by A. Connes \cite{vjmp09}.

\subsubsection{Nonholonomic gerbes and algebroid structures in MGTs}

For Direction 4, the main results were on (see publications \cite{vmathsc07,vjmp06,vbrasov05,vindia05,valgebroid05}):

\begin{itemize}
\item[ a)] Gerbes and Finsler-Lagrange-Hamilton spaces

\item[ b)] Exact solutions with nonholonomic algebroid structure in gravity
and geometric flow theories

\item[ c)] Lie algebroids and gauge gravity

\item[ d)] Clifford algebroids and noncommutative geometry
\end{itemize}

\textit{Comments:} There were elaborated four sub-directions related to topological nontrivial constructions for nonholonomic and locally anisotropic (Finsler) spinors, exact solutions and (non) commutative Finsler-Lagrange-Hamilton geometry. Papers \cite{vindia05,vbrasov05} (the first one was with participation of J. F. Gonzalez--Hernandez; in 2005, he was a master a student from Madrid, Spain) were devoted to the geometry and applications in physics of nonholonomic gerbes, Clifford-Finsler structures and index theorems. Article \cite{vjmp06} formulated the theory of nonholonomic Clifford and Finsler-Clifford algebroids. It also contain proofs of main theorems on properties of indices of connections in such spaces. There were defined the Dirac operator on Lie and Clifford algebroid spaces. Such constructions can be generalized on (co) tangent Lorentz bundles and/or nonholonomic manifolds. For instance, a perspective sub-direction is that on generalized gauge gravity models with generalized gerbes and/or algebroid symmetries. \vskip2pt

A series of preprints (see references in \cite{valgebroid05}) are devoted to examples of exact solutions with algebroid symmetries in MGTs. That sub-direction has perspectives in constructing locally anisotropic black hole and cosmological solutions with the pattern forming structure. Mechanical models with Lie algebroid symmetries on Lagrange-Hamilton spaces are studied in \cite{vmathsc07}. \vskip2pt

Author's research program on nonholonomic gerbes, Lie algebroids and applications began in 2005 during his visiting professor fellowship at IMAF, CSIC, Madrid, and in a result of certain seminars and discussions with professors. M. de Leon (Madrid) and R. Picken (Lisbon). In 2006, during a visit at Fields Institute, Toronto, Canada, the author of this work established a collaboration and elaborated a further research program for prof. M. Anastasiei and his post-graduate student C. Arcu\c{s} at Al. I. Cuza University at Iasi, Romania. That program on mathematical structures got certain financial support from the Ministry of Educations and Science of Romania (2006-2009). Later, certain research on nonholonomic Lie algebroid
configurations and Finsler geometry was performed partially by (at that time post-graduate student) dr. Constantin Arcu\c{s}. He also collaborated and published a series of works with Dr. E. Peyghan from Iran, see \cite{arcus14,arcusjmp14,arcus15} and references therein. Nevertheless, that research has not been concluded yet with (perspectives) applications in modern
quantum field theory and gravity. This is a perspective for further investigations. \vskip2pt

Finally, in this direction, it should be noted that the constructions for the nonholonomic Diract operators were applied for definition of noncommutative Finsler spaces and Ricci flows in A. Connes sense, see details in Ref. \cite{vjmp09} and in Part III of monograph \cite{vmon06}. Recently, geometric flows of of almost K\"{a}hler and Lie algebrod structures, encoding Finsler-Lagrange-Hamilton theories, were studied in \cite{vmjm15} and references therein.

\subsubsection{ Nearly autoparallel maps, twistors, and conservation laws in Finsler-Lagrange spaces}
\label{sssd05}

For Direction 5, there were developed four sub-directions (see publications \cite{ve2,vcb2,vcb1,v94,v94t,vdgrg03,dvgrg03,vmon98}):

\begin{itemize}
\item[ a)] Generalized geodesic maps and classification of Finsler-Lagrange-Hamilton spaces

\item[ b)] Nearly autoparallel maps, na-maps, and conservation laws in GR and MGTs

\item[ c)] Nearly autoparallel flat spaces and non conformal flat twistor spaces

\item[ d)] Supersymmetric na-maps
\end{itemize}

\textit{Comments:} The geometry of geodesic and conformal maps of Riemannian spaces and some examples of generalizations for Einstein and metric-affine spaces were studied by \ H. Weyl and A. Z. Petrov \cite{petrov71}. In an almost complete form it was formulated in a monograph published in Russian by N. Sinyukov in 1979 \cite{sinyukov79}. The author of this was was
interested in physical applications of generalized geodesic and conformal maps beginning 1992, see \cite{ve2,vcb2,vcb1,v94,v94t,vdgrg03,dvgrg03,vmon98}. Such models can be characterised by certain classes of generalized map
equations and corresponding Thomas and Weyl invariant values. In principle, any two (pseudo) Riemannian spaces (in particular, one of them being an Euclidean / Minkowski space) can be related to a set of nearly autoparallel maps (na-maps). Some authors call such maps also as nearly geodesic, or almost geodesic, maps \cite{mikes08}. \vskip2pt

Some chapters of author's Ph. D thesis \cite{v94t} (see supersymmetric developments in \cite{vmon98}) were devoted to na-transforms and definitions of corresponding invariants and conservation laws for spaces with nontrivial torsion, endowed with spinor/twistor/supersymmetric/noncommutative structures etc. Supersymmetric / complex manifolds and maps can be encoded
into certain map geometries of nonholonomic manifolds. This is an open direction for further geometric research and applications in physics. For instance, the geometry of spinors and twistors for curved spaces was studied in Ref. \cite{ve2} using nearly autoparallel and generalized conformal maps. Local twistors were defined on conformal flat spaces and mapped via
generalized transforms to more general (pseudo) Riemannian and Einstein spaces. The key result was that even the twistor equations are not integrable on general curved spaces such couples of spinors structures can be defined via nonholonomic deformations and generalize nearly autoparallel maps. Following this approach, we can consider analogs of Thomas invariants
and Weyl tensors (in certain generalized forms, with corresponding symmetries and conservation laws). The constructions were generalized for Lagrange and Finsler spaces \cite{vcb2} - it was a collaboration with a former author's student in R. Moldova, S. Ostaf. \vskip2pt

There are relevant certain results published in \cite{vcb1,vcb2} in collaboration with the Ph. D superviser in Romania, Prof. I. Gottlieb, devoted to the A. Moor's tensor integrals \cite{vcb2}. Here we note that an article on tensor integration and conservation laws on nonholonomic spaces was published by S. Vacaru in R. Moldova, see \cite{vrm4}. \vskip2pt

For explicit applications of na-maps in particle physics and field theories, there are cited the articles \cite{vdgrg03,dvgrg03} (a collaboration of S. Vacaru with Prof. H. Dehnen, supported by DAAD). In those works, generalized
geodesic and conformal maps were considered in (higher order) models of Finsler gravity and gauge and Einstein gravity. \vskip2pt

The theory of n-maps was formulated for supermanifolds enabled with N-connection structure. The results on supersymmetric na-maps and applications in supergravity and superstring theories were summarized in two chapters of monograph \cite{vmon98}.

In 2013, S. Vacaru met at a geometric conference prof. J. Mike\v{s} who gave a copy of monograph published in 2008
 due to  J. Mike\v{s}, V. Kiosak and A. Van\v{z}urov\'{a}) on  geodesic mappings of manifolds with affine connection, see \cite{mikes08}. That research began more than 35 years ago under the supervision of N. Sinyukov in Odessa, USSR/ Ukraine and concluded at Faculty of Science, Palack\'{y} University, Olomouc, Czechia. Only the results on na-maps and applications performed by S. Vacaru's group were dubbed in arXiv.org and peer-reviewed Western Journals. Unfortunately, it was not organized any collaboration of the mentioned Czech and Moldo-Romanian teams and they worked and published independently without mutual citations.

\subsubsection{Locally anisotropic (non) commutative gravity in low energy limits of string theories}

There were elaborated such sub-directions of Direction 6 (see publications \cite{vap97,vnp97,vmon98,vencyc03,vcb3,vncs01,
vpcqg01,vtnpb02,vcb3,vijgmmp08,gvvcqg15,bubuianucqg17,vacaruepjc17, vijgmmp09,vcqg11,vgrg12,rajpoot17ap}):

\begin{itemize}
\item[ a)] Nonholonomic background methods and locally anisotropic string configurations

\item[ b)] Supersymmetric generalizations of Finsler-Lagrange-Hamilton spaces

\item[ c)] Finsler branes and Ho\v{r}ava - Witten gravity

\item[ d)] Noncommutative brane black ellipsoid and cosmological solutions
\end{itemize}

\textit{Comments:} Relativistic Finsler-Lagrange-Hamilton geometry models encode naturally theories with MDRs and LIVs. If the goal is to elaborate sub-directions related to standard directions of physics, such theories have to be derived in some low energy limits of (super) string theory and/or as (non) commutative brane configurations. Two papers \cite{vap97,vnp97} published in 1997 in top particle physics journals were devoted to supersymmetric generalizations of theories with nonholonomic structures and local anisotropies. Part I of monograph \cite{vmon98}, see also references
therein, is devoted to the geometry of nonholonomic supermanifolds and possible applications in physics. Here we note that the concept of superspace/superbundle involves a special class of nonholonomic distributions allowing to model curved complex spaces, see author's Supersymmetry Encyclopedia term and explanations for Finsler superspaces in \cite{vencyc03}. \vskip2pt

Perhaps, the first attempt to elaborate a "supergeometric" Finsler like formulation of unified field theories is due to (S. Ikeda, 1978) \cite{ikeda78}. But the term "super" in that work was not related to modern concepts of supersymmetry, supergravity, and superstrings. Then, it should be noted that A. Bejancu introduced N-connections with super-fiber indices,
and super algebra symmetries, in some his preprints published by at Vest University of Timi\c{s}oara. Those results are summarized in his monograph on Finsler geometry and applications, from 1990, \cite{bejancu90}. In those works, the total spaces were not constructed as supermanifolds even the idea of a base B. DeWitt superspace \cite{dewitt84} was discussed. In 1996, there were elaborated nonholonomic models of supergravity and superstring theories (S. Vacaru, 1996-1997) \cite{vap97,vnp97,vmon98} following the N-connection supergeometry determined by nonholonomic distributions in total and base
superspaces modelled a class of supermanifolds. Here we cite some further constructions on supersprays and N-connections by authors from Iran (M. M. Rezaii and E. Azizpour, beginning 2005) \cite{rezaii05}. It was proven that certain models of (super) Finsler-Lagrange-Hamilton spaces, nonholonomic supergauge like theories, and respective superparticle models can be
constructed as low energy limits of superstring theories if corresponding nonholonomic constraints are imposed. \vskip2pt

One of the main problems in elaborating Finsler like generalizations of models of supergeometry and supergravity theories is that there is not a generally accepted definition of supermanifolds and superspaces (the existing ones differ for global constructions). Using nonholonomic super-distributions, N-connections can be defined for all considered concepts of superspace. This allows us to elaborate different types of models of supersymmetric Lagrange-Finsler geometry, or dual Hamilton and almost symplectic theories. Following the background field method with supersymmetric and N--adapted derivatives, and a correspondingly adapted variational principle, locally anisotropic configurations can be derived in low
energy limits of string theory. \vskip2pt

A series of works on supersymmetric models of nonholonomic superspaces and supergravity was elaborated by S. Vacaru's group in R. Moldova, and Romania (after he was deported in 2001 by the pro-Russian communist regime), and communicated at International Conferences \cite{vncs01,vpcqg01,vtnpb02,vcb3}. It was also a collaboration with former S. Vacaru's students, N. Vicol and I. Chiosa, and with young researchers from Bucharest-Magurele, Romania, F. C. Popa and O. T\^{\i}n\c{t}\^{a}reanu-Mircea. Part I of monograph \cite{vmon98} is devoted to the geometry of nonholonomic supermanifolds and possible applications in physics. It should be noted that a series of papers on string, brane and quantum gravity and MDRs were published during last 15 years, in certain alternative ways (not involving N-connections but with various types of Finsler like metrics), by Prof. N. Mavromatos and co-authors from King's College of London, see
\cite{mavromatos10,mavromatos10a} and references therein. \vskip2pt

Further constructions on Finsler-Lagrange-Hamilton modifications in Einstein and string gravity, models of modified dynamical supergravity breaking, and exact solutions in heterotic supergravity, were elaborated in \cite{vijgmmp08,gvvcqg15,bubuianucqg17,vacaruepjc17}. There is a correlation with research on Finsler branes and various (non) commutative, quantum gravity phenomenology and supersymmetric geometric flows \cite{vijgmmp09,vcqg11,vgrg12,rajpoot17ap}.

\subsubsection{Anisotropic Taub--NUT spaces, solitons, pp- and Dirac spin waves, and Ricci flows}
\label{sssdir07}
There were stated conventionally such sub-directions of Direction 7 (see
publications \cite{vpcqg01,vtnpb02,vijmpa06,vvijmpa07,vejtp09}):

\begin{itemize}
\item[ a)] Off-diagonal deformations of Taub-NUT (super) spaces
configurations

\item[ b)] Locally anisotropic Dirac spin waves

\item[ c)] Locally anisotropic solitionic solutions

\item[ d)] Ricci flows, solitonic waves and pp-waves with modified dispersion
\end{itemize}

\textit{Comments:} The research in this direction began in 2001 as a collaboration of S. Vacaru\footnote{%
he had a 6 months senior research position at the Institute of Spaces Sciences, ISS, Bucharest-Magurele, Romania, just after deportation by the communist regime in R. Moldova in March 2001} with M. Vi\c{s}inescu at the Institute of Nuclear Physics and Engineering at Bucharest-Magurele, Romania, and his PhD students also employed at the Insititute of Space Sciences, ISS, F. C. Popa and O. \c{T}\^{\i}n\c{t}\v{a}reanu-Mircea. The director of  ISS, D. Ha\c{s}egan, and M. Vi\c{s}inescu imposed the condition that S. Vacaru would co-supervise the research of two young researchers from Romania and 'help them with publications necessary for Ph.D. There were published two important papers locally anisotropic Taub NUT (super) spaces in Classical and Quantum Gravity and Nuclear Physics B) \cite{vpcqg01,vtnpb02} and some preprints at a University in Timi\c{s}oara. A sub-direction of that collaboration with M. Vi\c{s}inescu was extended in 2006 to Ricci flow solutions and pp-waves which can be related to Taub NUT \cite{vijmpa06,vvijmpa07}, see also results on geometric flows of exact solutions \cite{vejtp09}. \vskip2pt

Taub-NUT (super) space configurations consist an important subclass of solutions in (super) gravity theories. In our works, such prime spaces where used for generating new classes (targets) of generic off-diagonal solutions. We proved that corresponding systems of nonlinear PDEs can be decoupled and integrated in very general forms for nonholonomic deformations of spaces with high symmetries. Such constructions can be performed in MGTs for generalized Einstein-Dirac systems, solitonic waves and pp-waves. The works cited in sub-directions (a) - (d) contain a number of examples of exact solutions generated by applying the N--connection formalism and nonholonomic frame deformations which originated from Finsler geometry and nonholonomic mechanics. \vskip2pt

New sub-directions of research could be correlated to Ricci flows, solitonic waves and pp-waves with MDRs. Such solutions can be constructed in different classes of MGTs. They consist explicit examples of nonlinear evolution and nonholonomic dynamics on Finsler-Lagrange-Hamilton spaces with prescribed Taub-NUT symmetries and deformed into solutions for modified Einstein-Dirac and gauge and scalar field interactions, see \cite{vijgmmp07} and \cite{vpcqg01,vtnpb02}.

\subsubsection{Locally anisotropic diffusion, kinetic and thermodynamical processes, and MGTs}

We emphasize such sub-directions of Direction 8 (see publications \cite{vrm1,vcp1,vstoch96,vmon98,vapl00,vapny01,vrmp09,vcsf12,vsym13,ruchinepjc17}):

\begin{itemize}
\item[ a)] Stochastic processes, diffusion and thermodynamics in nonholonomic curved spaces (super) bundles

\item[ b)] Locally anisotropic kinetic processes and non-equilibrium thermodynamics in curved spaces

\item[ c)] Fractional diffusion and kinetic processes

\item[ d)] Stochastic nonholonomic geometric flows and structure formation

\item[ e)] Relativistic thermodynamic processes and geometric flows
\end{itemize}

\textit{Comments:} The idea to consider Finsler metrics in the theory of stochastic/ diffusion processes and applications in biophysics was exploited in a series of works (1995-2004) due to P. Antonelli, T. Zastawniak and D. Hrimiuc in Alberta, Canada, see references in \cite{anton93,anton96}. It was a collaboration between P. Antonelli's group in Canada and R. Miron's group at Ia\c{s}i University in Romania oriented to applications of Finsler geometry methods in biophysics.  Those teams were not familiar with existing at that time results in relativistic diffusion processes and kinetics, relativistic
thermodynamics. In order to study diffusion processes on locally anisotropic spaces, it was important to define a generalization of Laplace (diffusion) operators for Finsler spaces. That was an ambiguous problem for general nonlinear and
distinguished connections with nonmetricity, for instance, of Chern-Finsler type. The problem was considered by M. Anastasiei during 1992--1994 who explained it (at that time a young researcher) S. Vacaru.

The author of this article was familiar with results due to a famous Soviet (Russian) physicist, A. A. Vlasov, who published in 1966 in a book on statistical distribution functions see \cite{vlasov66}. That book was perhaps the first one where the theory of kinetics in curved spaces was extended to certain classes of Finsler metrics and connections. Most important results in this direction were based on the fact that main equations and conservation laws for anisotropic processes with additional nonholonomic constraints can be adapted to nonholonomic distributions using metric compatible distinguished connections like in Finsler geometry. In general, such constructions can be performed for velocity/momentum type variables
in Finsler-Lagrange-Hamilton spaces. \vskip2pt

S. Vacaru proposed his definitions of Finsler-Laplace metric compatible operators using the canonical distinguished connection and the Cartan distinguished connection and corresponding \^{I}to and Stratonovich types (those books existed in R. Moldova in Russian) of anisotropic stochastic calculus on generalized Finsler space during Ia\c{s}i Academic days in
October 1994. Those results with a study of N-connections, Finsler measures, stochastic and diffusion processes on Finsler--Lagrange spaces and vector bundles enabled with nonlinear connection structure were published later (1995-1996) in the R. Moldova and Proceedings of a Conference in Greece, see Refs. \cite{vrm1,vcp1,vstoch96}. \vskip2pt

During next five years, some sub-directions of S. Vacaru's research activity were oriented to the exploration of locally anisotropic diffusion processes, generalizations of stochastic theories for nonholonomic superspaces, and elaborating applications in modern physics and cosmology. Those results were published in two chapters of monograph \cite{vmon98} (see also references therein) in 1998. In 2001, he was able to publish two important papers in Annals of Physics (Leipzig) and Annals of Physics (New York) on stochastic processes and anisotropic thermodynamics in GR, Einstein-Finsler theories,
and nonholonomic spacetime, respectively, on locally anisotropic kinetic processes and non-equilibrium thermodynamics see Refs. \cite{vapl00,vapny01}. Those works contained also solutions for locally anisotropic black holes and black ellipsoids and related non-equilibrium thermodynamics related to certain applications in MGTs. \vskip2pt

A chapter in monograph \cite{vmon98} was devoted to author's results on generalizations of stochastic and kinetic theories for Lagrange and Hamilton geometries and their higher order anisotropic (including supersymmetric, in general, noncommutative and various quantum deformations) extensions. It was shown that geometric methods with adapting the constructions to nonholonomic distributions, the N-- connection formalism and adapted frames play a substantial role in the definition of anisotropic nonholonomic stochastic and diffusion processes. Similar constructions were developed in kinetics and
geometric thermodynamics of constrained physical systems, locally anisotropic black holes etc. \vskip2pt

Mentioned above sub-directions a) - c) consist a perspective for future investigations related to above-mentioned diffusion geometry and analogous thermodynamics. In papers \cite{vrmp09,vcsf12,vsym13}, there were studied generalizations of Grisha Perelman's entropy and thermodynamical functionals for nonholonomic Ricci flows and Lagrange-Finsler evolution models. The
constructions can be re-defined in dual Hamilton and/or almost symplectic variables, for fractional derivatives, stochastic calculus, noncommutative and/or supersymmetric generalization. In equilibrium, such processes can be described as certain generalized Ricci solitonic systems or effective Einstein spaces with nonholonomic constraints. Various classes of solutions
of such evolution and effective field equations can be modelled by stochastic generating functions. To relate the thermodynamical values for Ricci flows to some analogous diffusion processes and standard kinetic and thermodynamic theory, or to black hole thermodynamic processes, is a difficult mathematical physics problem with less known implications in
modern physics. \vskip2pt

One of the most prospective for further research sub-direction d) is on the theory of relativistic locally anisotropic stochastic and kinetic processes, thermodynamics and generalized Ricci flows. Such mathematical and physical problems are difficult to be elaborated because generalized Ricci operators in certain relativistic cases are related to nonlinear wave (and not to a nonlinear diffusion) like operators. Generalized Perelman's functionals are not of entropy type but determine certain dynamical entropy with nonholonomic mixed dynamical and evolution types, see discussions and references in a recent paper due to (V. Ruchin, O. Vacaru and S. Vacaru, 2017) \cite{ruchinepjc17}. Finally, we note that above sub-directions can be extended to classical and quantum informatics, entropic dynamics, and entanglement which are of interest in modern applied mathematics, informatics and quantum physics.

\subsubsection{Differential fractional derivative geometry, MGTs, and deformation quantization}

We consider such sub-directions of Direction 9 (see \cite{vcsf12,vijtp12,bvnd11,bvcejp11,bvcibf12}):

\begin{itemize}
\item[ a)] Models of differential fractional geometry

\item[ b)] Differential fractional geometric and quantum mechanics

\item[ c)] Modified gravity with fractional derivatives

\item[ d)] Fractional geometric deformation quantization
\end{itemize}

\textit{Comments:} Fractional derivatives $\mathcal{D}$ were introduced in mathematics with the aim to satisfy conditions of type $\mathcal{D} ^{k/s}e^{ax}/dx^{^{k/s}}=a^{^{k/s}}e^{ax},$ where the real constant $a\neq 0;\ k/s$ define a fraction; and $x$ is a real variable. This is not a trivial task, for instance, if $k/s=1/2$. In general, we can change $k/s$ into some general operators, functionals etc. and consider various classes of functions etc. The problem of constructing a fractional derivative and integral calculus was studied in a number of classical works due to Leibnitz, Riemann and other famous mathematicians. We note that fractional derivatives should be not confused with "fractals", fractional dimensions, fractional stochastic processes, fractional spins etc. The main problems in elaborating geometric and gravitational models with fractional derivatives can be related to certain very cumbersome integro-differential relations used, for instance, for the definition of Riemann-Liouville integral operators. In general, such fractional derivatives acting on scalars do not
result in zero. This implies in a series of difficulties for constructing models of fractional differential/ integral geometry and elaborating self-consistent physical theories with fractional derivative interactions. At present, there is an increasing number of publications with applications in modern engineering, economics etc. For instance, there were published a series of works with fractional derivative diffusion by F. Mainardi (last 30 years) \cite{carpinteri97} and a reformulation of physical theories on flat spaces to fractional derivatives was proposed by V. E. Tarasov (beginning 2005),
see \cite{tarasov06,vhthes12,vcsf12,vijtp12} and references therein.
\vskip2pt

The author of this article elaborated his first works on fractional derivative geometric flows (certain versions of nonholonomic Ricci flows), modified gravity and field theories, exact solutions with fractional derivatives, and deformation quantization in 2010, see papers \cite{vcsf12,vijtp12}. The main goals of those works were to formulate a self-consistent model of fractional differential geometry and study possible implications in modern gravity and cosmology. It was shown how exact solutions generalizing black holes and cosmological models can be constructed in explicit form in fractional MGTs by generalizing the AFDM. Those models of Ricci flows and gravity theories were developed for the so-called Caputo's fractional
derivative transforming scalar values in zero. Such constructions can be re-defined for the Riemann-Liouville, or other types, fractional derivatives via corresponding nonholonomic integro-differential transforms. %
\vskip2pt

During 2010-2011, S. Vacaru had a collaboration with D. Baleanu (from Ankara, Turkey, and Bucharest-Magurele, Romania) - one of the most prolific authors in the World on fractional derivative calculus and various applications in engineering, condensed matter physics etc. Together with co-authors, D. Baleanu published more than 300 papers in various directions of fractional calculus, control theory, economy and physics. Their activity is typical for a number of applied mathematics/ physics and engineering journals when, usually, 10 pages papers are signed by many co-authors considering some solutions of PDEs with fractional derivatives and certain applications are considered.  In a series of works \cite{bvnd11,bvcejp11}, see also references therein, there were elaborated fractional models of almost K\"{a}hler - Lagrange geometry, constructed and studied physical implications of exact and parametric solutions in gravity and geometric mechanics, with solitonic hierarchies and deformation quantization of such theories. The results were published in Proceedings of two International Conferences \cite{bvcibf12} and a seminar in Italy hosted by F. Mainardi at Bologna University, Italy, in 2012. In S. Vacaru's approach, the geometric formalism and related fractional partial derivatives depend on certain assumptions on nonholonomic structures modelling respective types of nonlocal and "memory" nonlinear effects we try to study, for instance, in theories of condensed matter or quantum gravity. \vskip2pt

It should be emphasized that there are not standard and unique ways for constructing geometric and physical models with fractional derivatives. For instance, a series of papers elaborated by G. Calcagni (beginning 2011; for review of his results, see \cite{calcagni16}) is based on a quite different approach with the aim to unify fractional dimensions, fractional derivatives, noncommutative and diffusion processes. In a recent paper \cite{calcagni17}, a version of black hole solution for certain multi-fractional theories (with fractional derivative radial coordinate) was studied in details and compared to those constructed in \cite{vijtp12}. Let us discuss in brief this issue. For fractional derivatives (with holonomic or nonholonomic structures), it is not clear how to formulate certain black hole uniqueness theorems. In general, it was elaborated till present only one complete model of fractional differential geometry elaborated in \cite{vcsf12,vijtp12}. For such theories, we can formulate both pure geometric and N-adapted variational models of fractional gravity and matter fields interaction taking an explicit Ricci tensor for a fractional d-connection. The constructions are quite similar to those for GR but with nonholonomically modified rules of fractional derivation and integration.
The exact solutions studied in \cite{calcagni17} are limited to certain fractional modifications of nonlinear systems of ordinary differential equations (for black holes) and do not involve the anholonomic deformation method like in S. Vacaru and D. Baleanu works.
\vskip2pt

For elaborating classical and quantum theories with MDRs and LIVs, the nonholonomic structure can be chosen to involve various types of fractional configurations. In such cases, the systems of fundamental nonlinear PDEs (for nonlinear evolution and/or dynamical interactions of generalized thermodynamical systems and classical/ quantum fields) can be decoupled and integrated in general forms for certain classes of fractional derivatives (for instance, of Caputo type). The new classes of solutions can describe fractional soliton interactions, black ellipsoid/ toroid / wormhole and other type configurations, various types of locally anisotropic cosmological scenarios with generic off-diagonal metrics. In particular cases of nonholonomic configurations with Caputo derivative, the method elaborated in
\cite{vcsf12,vijtp12,bvnd11,bvcejp11,bvcibf12} can reproduce the Calcagni solutions in \cite{calcagni17}. Nevertheless, such constructions can not be equivalent because the geometries and elaborated theories are, in general, different.

\subsubsection{Geometric methods of constructing off-diagonal solutions for (non) commutative / supersymmetric Ricci solitons, GR and MGTs}
\label{sssdir10}

Direction 10 is the most important and general one for author's activity during last 20 years. It splits in a number of mutually related sub-directions (see publications \cite%
{vexsol98,vmon98,vapny01,vjhep01,vpcqg01,vsbd01,vbebt01,vmon02,vsjmp02,vtnpb02,vijmpd03,vijmpd03a, vmurcia04,vjmp05,valgebroid05,vjmp05,vmon06,vijmpa06,vijgmmp07,vijgmmp08,avjgp09,vijtp10,vijtp10a,vcqg10, vijgmmp11,vepl11,vcqg11,vhthes12,vcsf12,vijtp12,vgrg12,vsym13,vjpcs13,vijtp13,vport13,svcqg13,svvijmpd14, vvyijgmmp14a,gvvepjc14,vepjc14a,gvvcqg15,rajpoot15,vmjm15,gheorghiuap16,vacaruplb16,rajpoot17ap, ruchinepjc17,rajpoot17ijgmmp,vacaruepjc17,bubuianucqg17,vbubuianu17,vcosmbc,rajpoot17ap}):

\begin{itemize}
\item[ a)] Decoupling property of (generalized) Einstein equations and
integrability for (modified) theories with commutative and noncommutative
variables

\item[ b)] Generic off--diagonal (off-diagonal) Einstein--Yang--Mills--Higgs
configurations with MDRs

\item[ c)] Exact solutions for nonholonomic (modified) Einstein-Dirac systems

\item[ d)] Nonholonomic and generalized Finsler configurations in (super)
string and brane gravity

\item[ e)] Off-diagonal solutions in noncommutative and/or nonassociative
gravity, and/or nonsymmetric metrics

\item[ f)] Exact solutions with generalized metrics and connections in (non)
commutative and/or supersymmetric Finsler-Lagrange-Hamilton theories

\item[ i)] Off-diagonal solutions and metric-affine configurations in
multi-metrics/ -connections and massive MGTs

\item[ j)] Almost Kaehler configurations, Horava-Witten structure and
off-diagonal black hole and cosmological configurations in (super) string
and (non) commutative brane models

\item[ k)] Exact and parametric solutions for nonholonomic geometric flows
and \ generalized Ricci solitions modeling evolution and/stationary
configurations with:

\item[ l)] Noncommutative A. Connes type and/or almost symplectic structures
associated to exact solutions in MGTs

\item[ m)] Gerbes and algebroid structures

\item[ n)] Fractional derivatives and nonholonomic geometries

\item[ o)] Stochastic / diffusion and nonlinear kinetic and geometric
thermodynamic processes

\item[ p)] Exact solutions with nonlinear waves, curve flows, solitons and
pp--waves;

\item[ q)] Solutions with MDRs and LIVs for generalized
Einstein-Yang-Mills-Higgs-Dirac systems and pseudo (almost symplectic)
Finsler-Lagrange-Hamilton structures

\item[ r)] Locally anisotropic black holes and black ellipsoids in (non)
commutative MGTs and nonholonomic stability

\item[ s)] Generalized Finsler black holes/ ellipsoids / toroids /wormholes
/flux tubes and locally anisotropic collapse

\item[ t)] Exact solutions in nonholonomic jet extended gravity

\item[ u)] Generic off-diagonal and locally anisotropic cosmological
solutions in supersymmetric / noncommutative / massive MGTs

\item[ v)] Diffusion and self-organized criticality in Ricci flow evolution
of Einstein and Finsler spaces

\item[ w)] Deforming black hole and cosmological solutions by quasiperiodic
and/or pattern forming structures in modified and Einstein gravity
\end{itemize}

\textit{Comments:} The gravitational field equations in Einstein gravity and modifications and/or geometric flows evolution equations consist very sophisticated systems of nonlinear partial differential equations, PDEs. For theories with MDRs and quantum corrections such systems are generic off-diagonal and with nonholonomic constraints. Former methods of finding
exact solutions involved procedures of transforming systems of nonlinear PDEs into systems of nonlinear ordinary differential equations, ODEs, which  can be integrated in explicit, or parametric, forms. The most important examples are those of black hole and homogeneous cosmological solutions. Such constructions are possible for some special ansatz, for instance, with
diagonal metrics depending on 1-2 variables and imposed higher symmetries (for spherical, cylindrical, wormhole configurations) and/or asymptotic conditions resulting in physically important solutions. Nevertheless, all methods with simple ansatz and transforms of PDEs into ODEs exclude from the very beginning the possibility to construct (for instance, in GR) exact solutions with generic off-diagonal metrics depending, for instance, on 3-4 variables. Such generalized classes of solutions also have physical importance because they describe, for instance, nonlinear wave and solitonic interactions, off-diagonal nonhomogeneous and locally anisotropic cosmological scenarios, gravitational vacuum and matter structure formation etc. \vskip2pt

The anholonomic deformation method, AFDM, of constructing exact and/or parametric solutions in GR, MGTs and geometric evolution theories was elaborated as a geometric and analytic one involving generic off-diagonal metrics and generalized (non) linear connections. The original purpose was to construct exact solutions with nontrivial N-connection structures in
Finsler like gravity theories. But from the very beginning, it became obvious that the same geometric methods can be applied generating off-diagonal solutions in GR and various MGTs. Such solutions are determined by generating and integration functions depending on various parameters and, in principle, all spacetime/ phase / supersymmetric/ noncommutative/
stochastic / fractional and other type variables in (non)commutative. The AFDM is reviewed and completed with new additional methods and examples in a series of our articles and monographs, see \cite{vmon02,vjmp05,vijgmmp07,vijgmmp08,vijtp10,
vijtp10a,vjpcs13,svvijmpd14,gvvcqg15,vvyijgmmp14a,vacaruepjc17, bubuianucqg17,vbubuianu17}. Parts I and II in the collection of works \cite{vmon06} contain both geometric details and examples for solutions in generalized metric-affine and Lagrange-Finsler-affine gravity theories, noncommutative gravity, extra dimension models etc. \vskip2pt

The idea of existence and first examples of solutions with general decoupling (for off-diagonal metrics and generalized connections depending on all spacetime variables) of gravitational field equations in Einstein, string and Finsler gravity was communicated in 1998 at a conference in Poland \cite{vexsol98}. Solutions for EYMH systems with off-diagonal deformations (for black hole and cosmological configurations) were presented in 2001 in Annals of Physics (NY) and JHEP see \cite{vapny01,vjhep01}.\footnote{The PRD referees considered those papers to be very mathematical in nature and not related to physics because that Finsler like gravity theories were subjected to critics in \cite{bekenstein93} (in that work, there were not analysed important physical effects due to the N-connection structure). Beginning 2005, a series of Finsler gravity and cosmology papers were published in PRD. Nevertheless, none such a paper (up to the present) has been related to exact solutions with nontrivial N-connection structure. Such constructions request a more advanced geometric technique which perhaps is less compatible with the style and purposes of journals like PRL-PRD. Fortunately, a series of other important  theoretical and mathematical physics journals (EPJC, JHEP, NPB,PLB, AP NY, CQG,\ GRG, J. Math. Phys., IJGMMP, etc.) accepted and published tenths of such works  cited for this and other directions.} It should be emphasized here that a substantial moral and technical support this author got from his former students and co-authors (S. Ostaf, E. Gaburov, D. Gontsa and Nadejda Vicol, from R. Moldova). With the  participation of those young researchers, there were elaborated Maple and Mathematica programs and verified by direct analytic calculus all theorems and formulas necessary for elaborating of the AFDM, see references in \cite{vmon98,vmon02,vmon06,vjmp05,vijgmmp07,vijgmmp08,vijtp10,vijtp10a}. That new geometric method of constructing off-diagonal exact solutions in MGTs was completed by 2001. Further developments involved more general technical work and possible physical implications related to supersymmetric/noncommutative/nonassociative/stochastic/fractional etc. distributions (see above sub-directions a)-w)).\vskip2pt

Exact solutions for nonholonomic and Finsler like generalizations for the Einstein-Dirac systems were studied in a series of works \cite{vpcqg01,vtnpb02} elaborated at ISS Bucharest-Magurele in 2001, see also the monograph \cite{vmon02} on (higher) order nonholonomic and Finsler-Lagrange-Hamilton spinors published in Athens in 2002. A collaboration with D. Singleton from the USA allowed to elaborate on new applications of the AFDM for constructing exact solutions for locally anisotropic black holes, ellipsoidal and toroid black holes in MGTs with brane and warped configurations, during a visit in the USA in 2001, see \cite{vsbd01,vbebt01,vsjmp02}. Important proofs on the stability of such locally anisotropic configurations (in explicit form, for black ellipsoids) in GR and MGTs were provided in \cite{vijmpd03,vijmpd03a}. \vskip2pt

A more rigorous mathematical formulation of the AFDM was presented at a conference in Murcia, Spain (2004) in \cite{vmurcia04} - it contained all geometric formalism necessary for constructing exact solutions for theories on (co) tangent bundles. The method was extended for generating new classes of solutions with noncommutative symmetries in GR and gauge gravity theories and disk type solutions with algebroid symmetries \cite{vjmp05,valgebroid05}. New developments and applications of the AFDM for Finsler-Lagrange and string gravity theories, nonholonomic metric-affine generalizations of
gravity were published and reviewed in \cite{vijgmmp07,vijgmmp08,vijtp10,vijtp10a,vijgmmp11}, see also respective
chapters in \cite{vmon06}. \vskip2pt

A series of works were devoted to constructing exact solutions in (non) commutative Finsler--Lagrange theories
\cite{vcqg10,vcqg11,avjgp09,vijtp13,vgrg12} defining Finsler black holes and ellipsoid/toroid configurations, bi-Hamiltonian structures and solitons, locally anisotropic Finsler banes and generic off-diagonal cosmological configurations in GR. There were analyzed possible physical implications of MDRs and locally anisotropic wormhole and toroid Finsler like configurations
\cite{vport13,rajpoot15}. Another class of locally anisotropic black hole solutions can be generated on jet extensions of MGTs \cite{rajpoot17ijgmmp}. Here we note that a number of new classes and possible physically important off-diagonal solutions with ellipsoid/ toroid symmetries and/or wormhole, solitons, Dirac waves and nontrivial Einstein-Yang-Mills-Higgs
configurations, cosmological solutions etc. (allowing formulation of renormalizable theories and Finsler like theories on Lorentz bundles) were studied in Refs. \cite{vepl11,vhthes12,vjpcs13,vvyijgmmp14a,svcqg13,svvijmpd14}.\vskip2pt

All solutions of generalized Einstein equations in MGTs can be redefined to generate certain classes of (relativistic) modified Ricci solitons. Applying the AFDM, it is possible to extend such solutions with dependencies on geometric flow parameter and construct exact solutions for generalized Ricci flow theories
\cite{vijmpa06,vvijmpa07,vcsf12,vijtp12,vsym13,vmjm15,gheorghiuap16,rajpoot17ap,ruchinepjc17}. Such a geometric techniques allows us to compute off-diagonal deformations of Kerr black hole and solitonic configurations in various (super)
symmetric/ string / brane configurations \cite{gvvepjc14,gvvcqg15,vacaruepjc17,bubuianucqg17,vbubuianu17}. For
cosmological applications, a new class of cosmological solutions with nonlinear wave, locally anisotropic pattern structures allows elaborating realistic scenarios of inflationary and accelerating cosmology and dark energy and dark matter models (a series of important works published beginning 2014 in PLB, EPJC, CQG, AP NY, etc.) \cite{vcosmbc,vepjc14a,gvvcqg15,vacaruplb16,rajpoot17ap}.

\subsubsection{Off-diagonal wormhole/ellipsoid/toroid configurations in (non)commutative MGTs}
\label{sssdir11}

We emphasize such sub-directions of Direction 11 (see papers \cite%
{vscqg02,vscqg02a,vsjmp02,vsbd01,vbebt01,vmon06,vport13,vepjc14,rajpoot15}):

\begin{itemize}
\item[ a)] Off-diagonal deformation of black holes with toroid, cylindrical
and planar horizons

\item[ b)] Locally anisotropic wormhole and flux-tube configurations

\item[ c)] Ellipsoid - toroid off-diagonal solutions and propagating black
holes in extra dimensions

\item[d)] Wormhole and black ellipsoid in Finsler-Lagrange-Hamilton geometry

\item[ e)] Noncommutative locally anisotropic wormholes
\end{itemize}

\textit{Comments:} This direction was inspired by J. P. Lemos works on solutions for black holes with toroid, cylindrical and planar horizons in GR and modifications, see a review in \cite{lemos01}, who hosted in 2001-2002 a NATO senior fellowship for the author of this work. That research program was correlated to a fellowship for supporting scientists in the R. Moldova
(hosted by D. Singleton from California State University at Fresno, USA). In result of those collaborations, there were published a series of works related to articles \cite{vscqg02,vscqg02a,vsjmp02,vsbd01} providing explicit applications of the AFDM to generating new classes of exact solutions in brane and extra dimension gravity. \vskip2pt

The locally anisotropic wormhole solutions were defined by generic off-diagonal metrics with a polarization of constants and deformed symmetries, for instance, ellipsoidal and toroid ones and additional solitonic waves. There were constructed in explicit nonholonomic ellipsoid - toroid configurations \cite{vbebt01}. The results of that preprint were developed
in Parts I and II of monograph \cite{vmon06}. Those solutions (generated using ellipsoidal and toroid coordinates, nonholonomic frames and generalized connections) are more general that, for instance, constructions for the so-called black Saturn metrics \cite{elvang07,yazadjiev08}. There new classes of locally anisotropic black holes and black rings/ tori
constructed in our works are different than those in \cite{emparan06}. For extra dimensions, such alternative constructions are not prohibited by black hole uniqueness theorems if extra dimensions are involved \cite{gibbons02}. %
\vskip2pt

New sub-directions of research can be related to (ringed) wormhole and black ellipsoid solutions with MDRs, which exist in Finsler-Lagrange-Hamilton and massive MGTs, in general, with supersymmetric and/or noncommutative variables \cite{vport13,vepjc14,rajpoot15}. Wormhole configurations exist as graphen structures, or various type of pattern forming networks with random processes and geometric flows.

\subsubsection{Solitonic gravitational hierarchies in Einstein and
generalized Finsler gravity}
\label{sssdir12}

We elaborated on such sub-directions of Direction 12 (see papers \cite{vjhep01,vpcqg01,vscqg02a,vjmp05,vmon06,vijmpa06,
vijgmmp07,vijgmmp08,vacap10,avjgp09,vejtp09,bvcejp11,bubuianucqg17,vbubuianu17,gvvepjc14,rajpoot15}):

\begin{itemize}
\item[ a)] Solitonic gravitational hierarchies as solutions of the Einstein
equations in general relativity

\item[ b)] Nonlinear waves with MDRs in Finsler-Lagrange-Hamilton spaces

\item[ c)] Black hole and solitonic backgrounds in MGTs

\item[ d)] Off-diagonal cosmological solutions with solitonic gravitational
and matter field waves

\item[ e)] (Non) commutative Ricci solitons, (super) symmetric nonlinear
wave and soliton solutions in string gravity
\end{itemize}

\textit{Comments:} A various class of solitonic and nonlinear wave solutions have been constructed in GR and MGTs, see reviews of results in monographs \cite{belinski01,kramer03}. Those gravitational solitonic solutions were generated using, for instance, the inverse scattering method and can be parameterized by diagonal metrics depending on two (in some exceptional cases, on three) spacetime coordinates. In certain cases, the Belinski-Verdaguer type solutions can be related to some classes of solitonic coordinates. \vskip2pt

In 1998, the author of this work was interested to develop the AFDM in some forms which would allow finding soltionic type solutions of generalized Einstein type equations generic off-diagonal metrics depending on 3 and more spacetime (and/or phase space like velocity/momentum) coordinates. In general, such solutions would describe nonholonomic solitonic configurations with nontrivial N-connection structure in certain modified Finsler like theories. For elaborating relativistic models, such constructions have to be elaborated (co) tangent Lorentz bundles. The first examples of such anholonomic soliton-dilaton and black hole solution in GR were found in 1999-2000 and published in a prestigious journal on high energy physics \cite{vjhep01}. Other classes of generic off-diagonal solutions were constructed for Dirac spinor waves and solitons in anisotropic Taub-NUT spaces \cite{vpcqg01}.\vskip2pt

Warped solitonic deformations and propagation of black holes in 5D vacuum gravity were studied in a collaboration of S.\ Vacaru and D. Singleton during a fellowship in California in 2001 \cite{vscqg02a} and those classes of off-diagonal soltionic solutions were generalized for noncommutative nonholonomic / Finsler gravity and Finsler-Lagrange-Hamilton metric-affine
generalizations, see original results and summaries of works in \cite{vjmp05,vmon06,vijgmmp07,vijgmmp08}. The geometric methods of constructing off-diagonal solitonic and pp-wave solutions were applied to a study of Ricci type geometric evolution of certain physical models of particle physics and gravity \cite{vijmpa06}.\vskip2pt

For a research collaboration during a visiting professor fellowship at Brock University, Ontario, Canada, in 2005-2005, this author studied a series of work on the geometry of curve flows, bi-Hamilton structures and solitons. A new development of the AFDM was the constructing of such nonholonomic frame transforms and deformations of (non) linear connection structures when the curvature and Ricci tensors are characterized by constant coefficients (with respect to certain classes of N-adapted frames), see \cite{vacap10}. In such cases, all geometric and physical information, for instance, of Einstein metrics can be encoded into solitonic hierarchies characterized by respective invariants, nonlinear symmetries and basic solitonic equations. Such solutions depend, in principle, on all spacetime coordinates via various types of generating and integration solitonic waves and (non) commutative parameters. Imposing respective nonholonomic constraints, we can extract Levi-Civita configurations or elaborate on certain metric-affine generalizations. Those constructions were extended for curve flows in
Lagrange-Finsler geometry and related to corresponding models with bi-Hamiltonian structures and solitons \cite{avjgp09} (a collaboration with S. Anco). \vskip2pt

A conventional N--connection splitting can be considered on (pseudo) Riemannian (Einstein) spaces of arbitrary dimensions and various supersymmetric/ noncommutative generalizations. This allows us to redefine equivalently the geometric/physical data in terms of necessary type auxiliary connections which can be restricted to LC-configurations if it is necessary. Solitonic hierarchies can be derived similarly as in Lagrange-Finsler geometry but mimicking such structures on nonholonomic (pseudo) Riemann and effective Einstein--Cartan manifolds completely determined by the metric structure. This approach provides us with a new scheme of solitonic classification of very general classes of exact solutions in Einstein, Einstein--Finsler, modified Lagrange-Hamilton, and nonholonomic Ricci flow equations. Such sub-directions were elaborated in
series of works with solitonic configurations in pp-wave spacetimes, solitonic propagation of black holes in extra dimensions and in modified theories, solitonic wormholes and metric-affine and/or noncommutative models of solitons in gravity and string/brane models, fractional solitonic hierarchies etc, see a number of examples in Refs. \cite{vjmp05,vejtp09,bvcejp11}, Parts I and II in monograph \cite{vmon06} and references therein.\vskip2pt

Recently, the AFDM was applied for constructing solutions with black ring/ ellipsoid - solitonic configurations in modified Finselr gravity \cite{vijtp13,rajpoot15}. Off-diagonal solitonic deformations of classical black hole and cosmological solutions were studied in some new sub-directions of research related to study of quasiperodic and/or pattern-forming structures in string gravity, Einstein and MGTs, massive gravity \cite{bubuianucqg17,vbubuianu17,gvvepjc14}. Considering extra-dimension coordinates as nonholonomic phase spaces ones (velocity or momentum type), such solutions can be defined for supersymmetric/noncommutative/nonassociative Finsler-Lagrange-Hamilton theories and various MGTs with MDRs.

\subsubsection{Axiomatic formulation of MGTs with MDRs, generalized EDYMH \ systems, and (non) commutative / associative
Finsler-Lagrange-Hamilton geometries}

Direction 13 splits conventionally into such sub-directions (see papers \cite{vijmpd12,vijgmmp12,vjpcs11,vexsol98,
vapny01,vjhep01,vmon98,vmon02,vmon06,vjmp05,vijgmmp07,vijgmmp08,vijtp10,vijtp10a,vrm2,vrm3,vgon95,vplbnc01,vncs01,dvgrg03,
vdgrg03,vijgmmp10,vijgmmp10a}):

\begin{itemize}
\item[ a)] Modeling MGTs with MDRs on (co) tangent Lorentz bundles

\item[ b)] Critical remarks on Finsler like gravity theories

\item[ c)] Classification of generalized Finsler-Lagrange-Hamilton and
nonholonomic metric-affine spaces

\item[ d)] Principles of modeling modified EDYMH with MDR interactions

\item[ e)] Methods for constructing and criteria for selecting physical
solutions in generalized Finsler MGTs and nonholonomic geometric evolution
models

\item[ f)] Axiomatic formulation of theories with MDRs and LIVs

\item[ i)] Axiomatic formulation of gravitational and geometric flow
Finsler-Lagrange-Hamilton theories

\item[ j)] Physical principles and axiomatic for (relativistic) geometric
flow evolution theories and modified Ricci solitons and Finsler like theories

\item[ k)] Modeling nonholonomic and Finsler like geometries for (super)
string/brane theories, noncommutative / nonassociative / stochastic
generalizations, statistical and/or irreversible/ nonequilibrium
thermodynamics and locally anisotropic processes, classical and quantum
information theories etc.

\item[ l)] Nonholonomic and generalized Finsler geometry methods in
quasi-classical and quantum gravity and matter fields theories
\end{itemize}

\textit{Comments:} The sub-directions a) - d) and f) are related to the goals of this work. Other sub-directions outline a series of important results and perspectives. \vskip2pt

Classical gravitational and matter field interactions are modelled geometrically on certain classes of manifolds and (co) bundle spaces enabled with the frame, metric and connection structures. They involve partial derivatives at least of second order which mean that geometrical constructions are elaborated on (co) tangent bundles of second and higher orders. If the quadratic line elements are not quadratic on (co) fiber variables, the elaborated physical models are with MDRs and LIVs even the base manifolds are chosen to be Lorentz ones as in GR. For such MGTs, it is necessary to elaborate on geometric and physical principles and axiomatic formulation. This is very important for investigating physical consequences of generalized Finsler gravity theories and cosmology, see (S. Vacaru, 2012) \cite{vijmpd12} -  this is the main purpose of our work. \vskip2pt

The problem of axiomatic formulation of unified field theories was analysed by D. A. Kilmister and G. Stephenson (1953, 1954) \cite{kilmister53,kilmister54} but there were was not proposed any solutions for models with MDRs. A self-consistent approach was considered for Finsler like gravity theories much later - a complete scheme for locally anisotropic gravity due to (R. I. Pimenov, 1987) \cite{pimenov87}. Those axiomatic concepts were not supported by explicit examples of exact off-diagonal
solutions (for instance, for locally anisotropic black holes and cosmological spaces) till 1998, see our works
\cite{vexsol98,vapny01,vjhep01,vmon98,vmon02,vmon06,vjmp05,vijgmmp07,vijgmmp08,vijtp10,vijtp10a}. To select respective classes of physically viable solutions is important to formulate certain fundamental physical principles generalized for Finsler like theories. \vskip2pt

At present, there are not complete solutions for three very important problems for formulating theories with Finsler like metrics and N-connection and d-connection structures for generalized Hamilton spaces:\  1) to construct in explicit Hamilton and almost symplectic variables exact solutions (for instance, some analogs of Einstein-Finsler-Hamilton black holes and locally anisotropic cosmological solutions); 2) to define Lagrange-Hamilton spinors and construct exact solutions of modified
Einstein-Dirac systems in dual variables; and 3) to elaborate on commutative and noncommutative models of generalized Lagrange-Hamilton gravity related to (super) string/brane and noncommutative geometry/gravity theories. In above sub-directions, there are outlined the main steps how solutions of such problems were performed for MGTs with nonholonomic and Finsler-Lagrange variables , see also references
\cite{vrm2,vrm3,vgon95,vplbnc01,vncs01,dvgrg03,vdgrg03,vijgmmp10,vijgmmp10a} on (non) commutative Finsler - gauge formulations of gravity, with analogous Yang-Mills equations for gravity. It is an important task for future works
to elaborate an axiomatic for theories outlined in points 1) - 3). \vskip2pt

In Part I, of monograph \cite{vmon06}, an important classification of Finsler-Lagrange spaces, and certain nonholonomic generalizations, depending on compatibility of fundamental geometric structures was elaborated. There are summarized the results on various classes of geometric and physical theories with metric compatible and noncompatible Finsler d--connections
with general non-vanishing torsion structure. Using the AFDM (see Direction 10), there were constructed explicit examples of exact solutions in generalized Lagrange-Finsler-affine gravity, certain Hamilton and almost K\"{a}hler configurations, and analyzed possible physical implications. \vskip2pt

The geometric constructions for noncommutative Finsler generalizations of gravity and geometric flow theories \cite{vjmp05,vejtp09,vjmp09,bvcejp11} (see also Parts I and II in monograph \cite{vmon06} and \cite{vcqg10,vrev08}) can be extended in order to study evolution and stationary configurations of Lagrange-Hamilton type. Principles of Finsler gravity and cosmology \cite{vjpcs11,vijmpd12} admit straightforward generalizations to supersymmetric models of nonholonomic gravity \cite{bubuianucqg17,vbubuianu17,gvvepjc14} and various classes of geometric evolution theories
\cite{vijtp12,vcsf12,vjmp09,gheorghiuap16}. Considering the critics on Finsler like theories with noncompatible d-connections \cite{vplb10,vijgmmp12,vjpcs11}, we conclude that most closed to standard theories of physics are the Finsler-Lagrange-Hamilton models with metric compatible d-connections (for instance, the Cartan, or canonical,
d-connection) constructed on (co) tangent bundle to Lorentz manifolds. Such theories allow us to define spinor and fermions in a form similar to GR but on nonholonomic manifolds/bundles \cite{vmon02,vmon06}. Finsler-Ricci evolution models can be introduced via nonholonomic deformations of the (pseudo) Riemannian ones. \vskip2pt

The main conclusion of this direction is that using the canonical and/or Cartan's d-connection it is possible to construct
Einstein - Finsler and various Lagrange-Hamilton like theories of gravity on tangent/vector bundles, or on nonholonomic manifolds, following the same principles as in general relativity and generalizing the Ehlers-Pirani-Schild (EPS) axiomatic \cite{ehlers72}, as we performed in \cite{vjmp09,vjpcs11,vijmpd12,vijgmmp12,vcsf12,gheorghiuap16} and the main part of this article.

\subsubsection{Stability of nonholonomic MGTs and geometric flow evolution models with nonsymmetric metrics, generalized connection structures, and jet extensions}

We emphasize such sub-directions of Direction 14 (see \cite{vijtp09nsg,vsigma08,vijtp09nsf}):

\begin{itemize}
\item[ a)] Nonholonomic geometry with nonsymmetric metrics

\item[ b)] Gravity and matter field theories with nonsymmetric metrics and
MDRs

\item[c)] Stability of gravity theories and solutions with nonsymmetric
metrics

\item[ d)] Nonsymmetric and nonholonomic geometric flows

\item[ e)] Nonholonomic jet extensions in MGTs and exact solutions

\item[ f)] Nonsymmetric metrics and noncommutative / nonassociative gravity
\end{itemize}

\textit{Comments:} The author of this work extended his research to geometric and physical models with nonsymmetric metrics after visits to Perimeter Institute, Canada (2007-2008), hosted by prof. J. W. Moffat as an expert in MGTs and cosmology. 'Nonsymmetric' gravity theories, NSGT, were originally elaborated by A. Einstein and L. P. Eisenhardt (1925-1945 and 1951-1952) \cite{einstein25,einstein45,eisenhardt51,eisenhardt52}. \vskip2pt

NSGTs were intensively studied in details in a series of works due to J. W. Moffat and his co-authors J. L\'{e}gar\'{e} and M. A. Clayton beginning 1979, see \cite{moffat79,moffat91,moffat95,moffat96,clayton96,clayton96a,moffat05} and references therein. The bulk of those works were published before the paradigm of 'dark energy and dark matter' was accepted for accelerating cosmology. Those works and further developments on (non) commutative MGTs presented a series of original ideas on classical and quantum gravity and various applications in modern cosmology, for instance, attempts to explain galaxy rotation curves and cosmology without dark matter.\vskip2pt

There were published in physical literature some papers with critical remarks on NSGT because of un-physical modes and un-stability existing in some models, see (T. Damour, S. Deser and J. McCarthy, 1993) \cite{damour93} and (T. Prokopec and W. Valkenburg, 2006) \cite{janssen06,janssen07,janssen07a}. It should be noted that in 1995 an improved model with nonintegrable constants and additional terms in Lagrangians was elaborated by J. L\'{e}gar\'{e} and J. W. Moffat by analogy with introducing integrable constants as in Lagrange mechanics with singular Hessians. It was proven that certain configurations can be stabilized perturbatively with re-defined geometric variables. Nevertheless, the problem of existence of NSGTs which are stable with respect to non-perturbative nonholonomic deformations was not solved in a general geometric form. \vskip2pt

In \cite{vijtp09nsg}, there were constructed explicit classes of exact solutions for "nonsymmetric" ellipsoids which are stable as deformations of black hole solutions. The spherical symmetric configurations which were found to be unstable in \cite{janssen06,janssen07,janssen07a} could be stabilized by introducing special nonholonomic constraints on nonsymmetric degrees of freedom. Such constructions could be stabilized by adapting the constructions to certain classes of generic off-diagonal metrics and nonholonomic frames with N-connection structure. In a more general context, such theories can be
re-written in almost K\"{a}hler and/or Lagrange-Finsler variables \cite{vsigma08}. Those proofs on the stability of such nonsymmetric metrics can be technically extended for generalized Finsler-Lagrange-Hamilton geometries studied in this work. It should be noted here two works on nonsymmetric generalizations of Finsler geometry \cite{atanasiu85,miron83} published due
to a collaboration of Japanese and Romanian geometers (Gh. Atanasiu, M. Hashiguchi and R. Miron, 1983-1985). Publication of those works was possible because R. Miron was authorized by Ceau\c{s}escu's secret service as a co-author. There were considered only generalizations of the geometries in \cite{eisenhardt51,eisenhardt52} as Riemann-Finsler configurations. The
goals of our works \cite{vijtp09nsg,vsigma08} was to elaborate on generalized NSGTs extended relativistically to Finsler-Lagrange spaces, see also noncommutative models for deformation quantization of Lagrange-Hamilton geometries, and to generalize the AFDM in a form which would allow constructing exact solutions of generalized Einstein-Eisenhardt-Finsler
equations. It was also proven that such solutions can be stabilized by imposing additional nonholonomic constraints. \vskip2pt

The author of this work addressed the problem of stability of nonsymmetric metrics in gravity in the framework of the theory of nonholonomic geometric flows characterized by nonsymmetric Ricci tensors \cite{vijtp09nsf}. For such models, nonsymmetric components of metrics appear naturally during geometric flow evolution, which results also in nonsymmetric Ricci soliton
configurations as certain equilibrium states. Such methods and results have a number of perspectives in studies of noncommutative/ nonassociative metrics in string and brane gravity theories, deformation quantization etc.
\cite{vncs01,vjmp05,vjmp09,avjmp09,vjmp13,vch2416}. Geometric constructions and exact solutions on higher order (co) tangent bundles and nonholonomic jet deformations were studied in \cite{vmon02,rajpoot17ijgmmp}.

\subsubsection{Deformation, A-brane, and loop like quantization of almost K\"{a}hler models of gravity}

We discuss a series of results in such sub-directions of Direction 15 (see
\cite{etayo05,vjmp07,vpla08,vjgp10,avjmp09,vjmp13,vijgmmp09,bvnd11,vch2416,bubuianucqg17}
and references therein):

\begin{itemize}
\item[a)] Almost K\"{a}hler and Lagrange-Finsler variables in geometric mechanics and gravity theories

\item[b)] Fedosov quantization of Einstein gravity and modifications

\item[c)] Deformation quantization of generalized Lagrange-Finsler and Hamilton-Cartan theories

\item[d)] A-brane quantization of gravity

\item[e)] Loop gravity and deformation quantization of MGTs with MDRs and/or noncommutative/fractional variables theories.
\end{itemize}

\textit{Comments:} One of the most important problems in modern physics is to formulate a viable model of quantum gravity, QG. There were elaborated various ideas, approaches and techniques but up till present, we could not overcome the fundamental task to connect the quantum world theories with astrophysical and cosmological scales physics and even solved technique problems arising in each quantization scheme. For any model of classical gravity theory, the fundamental field equations consist nonlinear systems of PDEs. There have been not formulated methods of a nonlinear functional analysis which would provide a rigorous scheme of quantization. %
\vskip2pt

Beginning 2004, a series of our papers were devoted to deformation quantization, DQ, (called also Fedosov quantization) of the Einstein and Finsler-Lagrange-Hamilton MGTs with nonholonomic variables \cite{vjmp07,vpla08,vjgp10,avjmp09,vjmp13}. The original constructions were elaborated for K\"{a}ehler and Poisson structures \cite{fedosov96,kontsevich03} which could not be applied directly to the Einstein and MGTs which are not described by "pure" symplectic forms. Our main idea was to use some methods (due to A. Karabegov and M. Schlichenmeier, 2001) \cite{karabegov01} on DQ of almost K\"{a}hler geometries. The point was to reformulate the Einstein gravity and/or Finsler-Lagrange-Hamilton MGTs in almost symplectic/ K\"{a}hler variables. For such geometric models, the scheme of DQ can be naturally applied. We elaborated on extensions with
noncommutative, supersymmetric, fractional and other type variables for geometric and physical theories admitting formal such parameterizations \cite{vijgmmp09,bvnd11,vch2416,bubuianucqg17}. \vskip2pt

The first constructions on DQ of Finsler-Lagrange spaces deformed into Fedosov type nonholonomic manifolds were obtained in collaboration with F. Etayo and R. Santamaria (University of Cantabria, Santander, Spain; 2005) \cite{etayo05} with an extension for almost K\"{a}hler-Lagrange geometries in \cite{vjmp07}. The main task was to understand how in such a scheme the Einstein gravity could be included \cite{vpla08}. During a visit at Fields Institute at Toronto (Canada), the author of this work got also associate and permanent research professor positions at the University Alexandru Ioan Cuza, UAIC, in Romania. It was developed a research program (2006-2017) on DQ of Einstein-Finsler/ -Lagrange gravity \cite{vjgp10,vjmp13} and Hamilton spaces \cite{avjmp09}. Here we note that for Hamilton configurations on co-tangent bundles the geometry of phase space posses additional symplectic symmetries. This results in a very sophisticated structure of induced N-connections and linear connections. The DQ scheme has to be applied in a quite different form for Lagrange spaces, i.e. on tangent bundles, and for Hamilton spaces, or any other geometries on co-tangent bundles. The GR and MGTs of Einstein gravity can be modelled by constructions on (co) tangent Lorentz manifolds. In the last case, a more advanced geometric technique adapted to the  Legendre transforms, symplectomorphisms, and almost symplectic structure had to be elaborated. Recently, the DQ formalism was generalized fractional derivative geometries and fractional mechanics and gravity \cite{bvnd11}, see also a geometry of almost K\"{a}hler structures for nonsymmetric metrics \cite{vijtp09nsg}.\vskip2pt

The DQ scheme is not considered as a generally accepted quantization formalism with perturbative limits for operators acting on Hilbert spaces, renormalization formalism etc. For instance, S. Gukov and E. Witten (2008) \cite{gukov08} elaborated an alternative brane quantization with A-model complexification (using almost K\"{a}hler geometries). In \cite{vijgmmp09},
it was proved that the Einstein gravity and various modifications re-written in almost K\"{a}hler variables can be quantized following the A-model method. Possible connections to other approaches were analyzed in \cite{vlqg09} (for loop gravity\ with Ashtekar-Barbero variables determined by Finsler-Lagrange-Hamilton like connections, see \cite{rovelli03} for review
and references on loop QG) and in \cite{vijgmmp10a} for the so-called bi-connection formalism and perturbative quantization of gauge gravity models.\vskip2pt

The DQ and geometric quantization formalism have perspectives in elaborating noncommutative and relativistic geometric flow theories and almost K\"{a}hler and Lie algebroid geometries in superstring theory \cite{vmjm15,vch2416,bubuianucqg17}.

\subsubsection{Covariant renormalizable anisotropic theories and two connection gauge like gravity}

We elaborated in such sub-directions of Direction 16 (see \cite{vepl11,vgrg12,vijgmmp14,vdgrg03,vijgmmp10,vijgmmp10a}):

\begin{itemize}
\item[ a)] Generalized pseudo-Finsler and Ho\v{r}ava-Lifshitz theories on (co) tangent bundles

\item[ b)] Covariant renormalizable models for generic off-diagonal spacetimes and anisotropically modified gravity

\item[ c)] Two connection and nonholonomic gauge like quantization of gravity theories

\item[ d)] Renormalizable Finsler-Lagrange-Hamilton structures

\item[ e)] Noncommutative and/or supersymmetric gauge gravity models and QG
\end{itemize}

\textit{Comments:} The Newton gravitational constant for four dimensional interactions results in a generic non-renormalizability of GR theory. In Direction 15, we considered various schemes of geometric, deformation and
non--perturbative quantization. Such constructions do not solve the problem of constructing a realistic model of QG with a perturbative scheme without divergences from the ultraviolet region in momentum space. The so-called Ho\v{r}ava-Lifshitz model (2009) was developed for nonhomogeneous anisotropic scaling of spacetime and Finsler like variables which allow to develop covariant renormalization schemes for Finsler-Lagrange-Hamilton MGTs, see \cite{vepl11} and references therein.\vskip2pt

Various models of QG involve locally anisotropic configurations and result in MDRs. Indicators of such dispersion can be associated with certain classes of Finsler fundamental generating functions. In \cite{vgrg12}, we developed a scheme of perturbative quantization of such Ho\v{r}ava-Finsler models. In both cases (for constructions from above two mentioned papers), a crucial role in the quantization procedure is played by the type of nonholonomic constraints, generating functions and parameters. Such values determine various classes of generic off-diagonal solutions of Einstein equations and generalizations, see Direction 10. In \cite{vijgmmp14} and \cite{vepl11}, we proved that the nonlinear gravitational dynamics, geometric evolution, and corresponding nonholonomic constraints can be parameterized in such ways that certain "renormalizable" configurations are generated in locally anisotropic form. For such configurations, a covariant Ho\v{r}ava-Lifshitz quantization formalism can be elaborated. Considering respective nonholonomic variables, this perturbative quantization method can be extended to massive (non) commutative gravity theories on (co) tangent
Lorentz bundles formulated in terms of Finsler-Lagrange-Hamilton variables. %
\vskip2pt

Perturbative quantization schemes can be elaborated for certain classes of models of quantum gauge gravity. In explicit form, such computations with nonholonomic distributions are provided for the method of two-connection renormalization for gauge models of Einstein gravity and MGRs \cite{vijgmmp10,vijgmmp10a}. This sub-direction is related to a series of works
on gauge Einstein and Finsler like theories published this author and co-authors during 1994-1997 in R. Moldova, see \cite{v94t,vgon95,vrm2,vrm3}. The approach was developed for nonholonomic spinor, twistor, and supersymmetric variables, see monographs \cite{vmon98,vmon02} and references therein. Noncommutative gauge models on base spaces enabled with N-connection
structure were studied in \cite{vplbnc01,vncs01,vjmp05,vdgrg03,vmon06} (a collaboration supported by DAAD hosted by prof. H. Dehnen at Konstanz University with further developments and collaborations in Spain, Portugal, Romania, and Canada, see Direction 2).

\subsubsection{Nonholonomic Ricci flows and thermodynamic characteristics in geometric mechanics and gravity theories}

We emphasize such sub-directions of Direction 17 (see papers \cite{vrmp09,vjmp08,vijtp09nsf,vijmpa06,vvijmpa07,vejtp09,vcsf12,vsym13,vijgmmp12,vmjm15,vjmp09,vhthes12,
rajpoot15,gheorghiuap16,rajpoot17ap,ruchinepjc17}):

\begin{itemize}
\item[a)] Generalized Perelman's functionals and Hamilton's equations for
nonholonomic Ricci flows

\item[ b)] Analogous statistical and thermodynamic values for flow evolution
of Finsler-Lagrange--Hamilton geometries and analogous gravity theories

\item[c)] Nonholonomic Ricci solitons and exact solutions in GR and MGTs

\item[ d)] Nonholonomic geometric evolution and nonsymmetric metrics

\item[ e)] Geometric evolution of pp--wave and Taub NUT spaces

\item[ f)] Nonholonomic Dirac operators, distinguished spectral triples and
evolution of models of noncommutative geometry and gravity theories

\item[ i)] Stoachastic geometric flow evolution

\item[ j)] Fractional geometric flows and gravity

\item[ k)] Relativisic geometric flows as nonholonomic evolution

\item[ l)] Supersymmetric geometric flows and MGTs

\item[ m)] Almost K\"{a}hler Ricci flows and Einstein and Lagrange-Finsler
structures on Lie algebroids

\item[ n)] Modified Ricci flows and curve flows in \ GR and
Finsler-Lagrange-Hamilton geometry, bi-Hamiltonian structures and solitons
\end{itemize}

\textit{Comments:} One of the most remarkable result in modern mathematics with fundamental implications in physics is the proof of the Poincar\'{e} conjecture by Grisha Perelman (2002-2003) following methods of the theory of Ricci flows (1982) \cite{perelman02,perelman03,perelman03a}. Those constructions were originally elaborated for evolution of Riemannian and/or
K\"{a}hler metrics using the Levi--Civita connection, see monographs \cite{cao06,morgan07,kleiner08} for reviews of mathematical results. Here, we note that physical models with renorm group flows and geometric evolution were well known in quantum field theory since the D. Friedan (1980-1985) works on gauge and nonlinear sigma theories
\cite{friedan80,friedan80a,friedan82} inspired by A. M. Polyakov's works (in 75) \cite{polyakov75}. Those works on nonlinear models in $2+\varepsilon $ dimensions contained renorm group equations which are equivalent in certain approximations to R. S. Hamilton's (1982-1995) geometric evolution equations for Ricci flows \cite{hamilton82,hamilton88,hamilton95}.
\vskip2pt

The author of this work became interested in research on geometric flows and applications in physics in 2005 when he had a sabbatical professor fellowship at CSIC, Madrid, in Spain. His research program was on nonholonomic geometric methods in mechanics and physics. A sub-direction in that program was devoted to a study of flow evolution of Lagrange-Hamilton
systems geometrized on (co) tangent bundles, which resulted in a series of works on the nonholonomic geometric evolution of Finsler-Lagrange spaces. Such spaces can be endowed with general nonsymmetric metric and nonlinear and distinguished connection structures, see \cite{vrmp09,vjmp08,vijtp09nsf} (in a dual form, such constructions are valid for almost Kaehler and/or Hamilton spaces etc). Those works provided mathematical proofs and explicit examples when different types of commutative and noncommutative, (non) Riemannian, geometries transform from one to another if corresponding nonholonomic constraints are imposed on geometric flow evolution. Such results can be generalized for geometric and physical models with nonholonomic constraints on Ricci flows of (pseudo) Riemannian and/or vector bundles.
\vskip2pt

One of our goals was to study Ricci flow evolution of systems encoding rich geometric structure and analyze possible applications in modern physics. A series of our works (it was also a collaboration with M. Vi\c{s}inescu, in 2005-2006) were devoted to nonholonomic Ricci flows/ solitons to off-diagonal solitonic pp-wave and parametric deformations of Taub-NUT and
Schwarzschild metrics and running cosmological constants \cite{vijmpa06,vvijmpa07,vejtp09}. It should be noted here that the concept of Ricci soliton is different from that of a solitonic wave, see above references and a research on curve flows and solitonic hierarchies generated by Einstein metrics, Lagrange-Finsler geometry and bi-Hamilton structures (a collaboration with S. Anco, 2006-2007) \cite{vacap10,avjgp09}. \vskip2pt

Two sub-directions of research on non-Riemannian geometric flows are related to fractional nonholonomic Ricci flows \cite{vcsf12} and diffusion and self-organized criticality in the geometric evolution of Einstein and Finsler spaces \cite{vsym13}. Such theories of nonholonomic Ricci flows and Finsler-Lagrange-Hamilton geometries can be formulated in metric compatible and/or noncompatible form \cite{vijgmmp12}. In a more general context, geometric flow evolution theories encode models of almost K\"{a}hler Ricci flows and Einstein and Lagrange-Finsler structures on Lie algebroids \cite{vmjm15}.\vskip2pt

Perhaps, one of the most important developments of the Hamilton-Perelman Ricci flow theory was for noncommutative geometric flows. Here we note the results on nonholonomic Clifford and Finsler structures related to spectral functionals, nonholonomic Dirac operators and noncommutative Ricci flows summarized in \cite{vjmp09,vhthes12}. Such constructions were generalized
recently (2014-2017, collaborations with S. Rajpoot, V. Ruchin and some young researchers) to include supergeometric flows and solitonic configurations in supergravity and modified (super) gravity and cosmology \cite{rajpoot15,gheorghiuap16,rajpoot17ap}. \vskip2pt

Further research is outlined in sub-directions related to geometric flows and renormalizations of QG models, noncommutative and supersymmetric models of geometric evolution with almost Kaehler structure for MDRs, and constructing exact solutions for stationary Ricci soliton configurations and modified gravity theories, with applications in modern cosmology and astrophysics etc. An original contribution due to (V. Ruchin, O. Vacaru and S. Vacaru, 2013-2017) was on statistical and relativistic thermodynamic description of gravitational fields using Perelman's W-entropy \cite{ruchinepjc17}.\vskip2pt

Finally, we note that many constructions elaborated upon by physicists have not been validated by mathematicians - this who request proofs of theorems on hundred pages using sophisticate geometric analysis methods and elaborating new branches of nonlinear functional analysis. In our works, we neither concentrated on pure quantum field theories issues nor considered
how certain geometric models may provide a natural connection between nonlinear sigma models, geometric analysis and topological implications of (modified) Ricci flow theories. Our goals, methods, and results are very different from the ones employed in geometric analysis with relation to quantum theories. Instead, we elaborated a series of new results based on to
exact solutions for nonlinear off-diagonal systems in MGTs.

\subsubsection{Stationary locally anisotropic and/or quasiperiodic solutions in MGTs}
\label{sssdir18}

There were considered such sub-directions of Direction 18 (see papers \cite{vexsol98,vapny01,vjhep01,vbebt01,vmon02,vtnpb02,
vpcqg01,vijmpd03,vijmpd03a,vrev08,vijgmmp08,vmon06,vjmp05,vijtp09nsg,vsigma08,vijtp09nsf,vijgmmp07,vijgmmp11,vijtp10,vijtp13,
vcqg11,vgrg12,svvijmpd14,vvyijgmmp14a,gvvepjc14,rajpoot15,rajpoot17ijgmmp,rajpoot17ijgmmp,vbubuianu17}):
\begin{itemize}
\item[ a)] Exact solutions in locally anisotropic gravity and strings

\item[ b)] Locally anisotropic (2+1)-dimensional black holes

\item[ c)] Thermodynamic geometry and locally anisotropic black holes

\item[ d)] Anholonomic soliton-dilaton and black hole solutions in GR, brane
and gauge gravity

\item[ e)] Dirac spinor waves and solitons in locally anisotropic Taub NUT
spinning spaces

\item[ f)] Off-diagonal 5-dimensional metrics and mass hierarchies with
anisotropies and running constants

\item[ i)] Ellipsoidal black hole - black tori systems in four dimensional
gravity

\item[ j)] Deformed black holes with soliton, quasiperiodic and/or
pattern-forming structures in heterotic supergravity, GR, and MGTs

\item[ k)] Perturbations and stability of black ellipsoids; horizons and
geodesics of black ellipsoids with anholonomic conformal symmetries

\item[ l)] Stationary (non) commutative Finsler like solutions in Einstein,
(super) gauge and string, and metric-affine gravity

\item[ m)] Clifford-Finsler / Einstein--Cartan algebroid and stationary
Einstein-Dirac structures

\item[ n)] Nonholonomic deformations of disk solutions and algebroid
symmetries in Einstein and extra dimension gravity

\item[ o)] Black holes, ellipsoids, and nonlinear waves in pseudo-Finsler
spaces and Einstein gravity

\item[ p)] Nonholonomic black ring and solitonic solutions in Finsler and
extra dimension gravity theories

\item[ q)] Finsler black holes induced by noncommutative anholonomic
distributions in Einstein gravity

\item[ r)] Stochastic and/or fractional dynamics from Einstein-Finsler
gravity, general solutions, and black holes

\item[ s)] Hidden symmetries of ellipsoid-solitonic deformations of Kerr-Sen
black holes and quantum anomalies

\item[ t)] Off-diagonal deformations of Kerr black holes in Einstein and
modified massive gravity in higher dimensions

\item[ u)] Nonholonomic jet deformations and exact solutions for modified
Ricci soliton and Einstein equations

\item[ v)] A geometric method of constructing stationary solutions in
modified f(R,T)-gravity with Yang-Mills and Higgs interactions

\item[ w)] Black ring and Kerr ellipsoid - solitonic configurations in
modified Finsler gravity
\end{itemize}

\textit{Comments:} Sub-directions a) - w) outline more than 20 years research activity on physical implications
of stationary solutions in MGTs constructed following the AFDM, see Direction 10, and (additionally, for wormhole/cylindric/ toroid systems) Direction 11. We emphasize here that such solutions with nontrivial N-connection structure in Finsler-Lagrange-Hamilton theories are generic off-diagonal and cannot be described by static total space configurations.
Stationary solutions are physically important because they describe, for instance, generalizations of black hole solutions for nonholonomic and locally anisotropic configurations. One of the goals was to construct analogues of Finsler black holes which can be modeled by nonholonomic configurations on Lorentz manifolds, in a higher dimension (pseudo) Riemannian
and non-Riemannian spacetimes, on in various Finsler like MGTs. The first examples on nonholomic three dimensional solutions in locally anisotropic gravity and strings were communicated and published in Conference proceedings (1998, Poland) \cite{vexsol98}. Solutions for four dimensional and higher dimension locally anisotropic black holes/ ellipsoid /torus/ \ soliton - dilaton configurations where studied in a series of works published in 2001 \cite{vapny01,vjhep01,vbebt01}. Certain classes of three dimensional nonholonomic black hole and solitonic configurations were analyzed in \cite{vmon02}. \vskip2pt

There were two projects for further researches on nonholonomic stationary configurations in (super) symmetric and/or brane type MGTs. The first one was a collaboration with postgraduate students at the Institute of Space Sciences, Bucharest-Magurele, Romania (S. Vacaru, F. C. Popa and O. \c{T}in\c{t}areanu-Mircea, 2001) on anisotropic Taub NUT spinning spaces and Dirac spinor waves and solitons on such spaces \cite{vtnpb02,vpcqg01}. The second one was based on a series of articles \cite{vsjmp02,vscqg02a,vsbd01} due to (S. Vacaru, D. Singleton and students from R. Moldova, 2001-2002) performed during a visit at California State University at Fresno, USA, and a NATO fellowship at IST, Lisbon, Portugal. Those classes of generic
off-diagonal exact solutions describe locally anisotropic (ellipsoid, cylindrical, bipolar, toroid) wormholes and flux tubes, and propagating black holes, in five dimensional gravity. Such solutions can be re-defined in respective Finsler-Lagrange-Hamilton variables, for instance, on (co) tangent Lorentz manifolds.\vskip2pt

In theories with MDRs and LIVs, there is a subclass of solutions generalizing the Kerr metric in GR to certain (black) ellipsoid configurations. For small off-diagonal parametric deforms, such solutions can be also constructed in GR. Such solutions are stable with respect linear perturbations and, for respective nonholonomic constraints, can be stabilized for respective non-perturbative deformations, see \cite{vijmpd03,vijmpd03a}. Similar classes of solutions were found for  noncommutative/nonsymmetric metric MGTs \cite{vjmp05,vijtp09nsg,vsigma08,vijtp09nsf}, in general with Lie algebroid symmetries \cite{valgebroid05} which allows to provide explicit examples of noncommutative Finsler metric-affine stationary spaces with generalized symmetries. In details, applications of the AFDM for such configurations are reviewed in \cite{vrev08,vijgmmp08,vmon06}.\vskip2pt

Important results for the AFDM in GR and nonholonomic and Finsler like MGTs were formulated and proven in a series of works
\cite{vijgmmp07,vijgmmp11,vijtp10,vijtp13,vvyijgmmp14a,gvvepjc14,rajpoot17ijgmmp,vbubuianu17} and chapters of monographs \cite{vmon02,vmon06} published by S. Vacaru and co-authors beginning 2001. Stationary configurations were generated as
particular vacuum and nonvacuum configurations for which generalized Einstein equations admit a general decoupling and integrability when the coefficients of generic off-diagonal metrics and (non) linear connections possess a Killing symmetry on the time like coordinate. The AFDM and respective formulas for stationary configurations can be extended for nontrivial sources of generalized Einstein-Yang-Mills-Higgs and Finsler-Einstein-Dirac systems, in commutative and noncommutative forms,
with fractional and stochastic variables, supersymmetric modifications etc., see \cite{vijtp12,vsym13}, and for noncommutative Finsler black holes \cite{vcqg10}. The geometric methods for constructing exact solutions for modified Einstein and Finsler theories on tangent Lorentz bundles are summarized in \cite{svvijmpd14}, see also examples of solutions for black holes/ellipsoids and nonlinear waves \cite{vijtp13} and for Finsler branes, QG phenomenology with LIVs violations, MDRs in Horava-Lifshitz gravity etc. \cite{vcqg11,vgrg12}.\vskip2pt

A collaboration (S. Rajpoot and S. Vacaru, 2014-2014) was devoted to constructing exact solutions for modified Ricci solitons and Einstein equations in generalized Finsler gravity which describe nonholonomic jet deformations of black hole solutions and black ring and Kerr ellipsoid-solitonic configurations \cite{rajpoot15,rajpoot17ijgmmp}. Certain projects supported by CERN, DAAD and (recently, by a private organization, QGR-Topanga) were devoted to a study of (stationary)off-diagonal deformations of Kerr metrics and black ellipsoids, instantons and almost-K\"{a}hler internal spaces, in heterotic supergravity \cite{vacaruepjc17,bubuianucqg17}. The most important physical property of such solutions is that they can be generated with nontrivial soliton, quasiperiodic and/or pattern-forming structures which is important in modern astrophysics and cosmology. Here we note that following the AFDM the off-diagonal stationary solutions are dual with respect to locally anisotropic cosmological solutions all generated by quasiperiodic and/or pattern forming structures in modified and Einstein gravity (a recent collaboration by L. Bubuianu and S. Vacaru) \cite{vbubuianu17}.\vskip2pt

Finally (for this direction) we note that a series of works were elaborated recently using methods which are alternative to the AFDM. They also claim to construct some classes of "vacuum field equation in Finsler spacetime", in certain cases with static spherical symmetry and as extensions of the Schwarzschild metric, or for Finslerian wormhole models
\cite{li14,lin14,silagadze15,rahman16,caponio16}. It is not clear how such solutions could be static (for d-metrics on total
Finsler-Lagrange-Hamilton spacetime) and with nontrivial N-connection structure and if such configurations can be extended for nontrivial sources of certain self-consistent Finsler generalizations of the Einstein equations. Following the AFDM, there are possible only stationary configurations describing deformed spherical symmetries in order to have anisotropy and/or
to embed certain high symmetries into nontrivial (non) vacuum local anisotropic backgrounds, for instance, with quasi-periodic structure.

\subsubsection{Locally anisotropic (quasiperiodic) cosmological structures and dark matter / energy}
\label{sssdir19}

We elaborated recently such sub-directions of Direction 19 (see papers \cite{v90,vijtp10a,vijmpd12,vcosmbc,svcqg13,
elizalde15,vacaruplb16,rajpoot17ap,bubuianucqg17,vacaruepjc17,amaral17,vbubuianu17}):

\begin{itemize}
\item[ a)] Off-diagonal metrics and anisotropic brane inflation

\item[ b)] Locally anisotropic cosmological metrics with noncommutative
symmetries in Einstein and gauge gravity

\item[ c)] Noncommutative Finsler cosmology from (super) string / M--theory

\item[ d)] Cosmological Finsler like solutions in Einstein, string and
metric-affine gravity

\item[e)] Cosmological spacetimes with Lie algebroid symmetries

\item[ f)] New classes of off-diagonal cosmological solutions in
Einstein-Finsler gravity

\item[ i)] Cyclic and ekpyrotic universes in modified Finsler osculating
gravity on tangent Lorentz bundles

\item[ j)] Off-diagonal ekpyrotic scenarious and equivalence of modified,
massive and/or Einstein gravity

\item[ k)] Cosmological solutions in biconnection and bimetric gravity
theories

\item[ l)] Effective Einstein cosmological spaces for non-minimal modified
gravity

\item[ m)] Ghost-free massive f(R) theories modelled as effective Einstein
spaces and cosmic acceleration

\item[ n)] Equivalent off-diagonal cosmological models and ekpyrotic
scenarious in f(R)-modified, massive and Einstein gravity

\item[ o)] Supergeometric geometric flows and inflations in modified
supergravity

\item[ p)] Cosmological attractors and anistropies in two measure theories,
effective EYMH systems, and off-diagonal inflation models

\item[ q)] Starobinsky inflation and dark energy and dark matter effects
from quasicrystal like spacetime structures

\item[ r)] Anamorphic quasiperiodic universes in modified and Einstein
gravity with loop quantum gravity corrections

\item[ s)] Deforming cosmological solutions by quasiperiodic and/or pattern
forming structures in modified and Einstein gravity
\end{itemize}

\textit{Comments:} We applied nonholonomic geometric methods for elaborating a twistor diagramm techniques in a work on minisuperspace twistor quantum cosmology \cite{v90}. The first off-diagonal locally anisotropic five dimensional brane cosmological and anisotropic black hole solutions with nontrivial N--connection structure where studied in a series of electronic preprints (arXiv: hep-th/0108065, hep-ph/0106268, hep-th/0109114 in collaboration of S. Vacaru and two students from R. Moldova, E. Gaburov and D. Gontsa) published in hard form as respective sections of monograph \cite{vmon06}. That research was partially supported by a 2000-2001 Legislative Award at California State University at Fresno, USA, hosted by D. Singleton. The goal was to study anisotropic reheating and entropy production on three dimensional branes from a decaying bulk scalar field in the brane-world picture of the Universe by considering off-diagonal metrics and anholonomic
frames.\vskip2pt

Locally anisotropic cosmological spacetimes in GR and MGTs are constructed as exact and/or parametric non-stationary solutions of generalized Einstein equations. Such solutions can be generated using nonholonomic variables and
(generalized) connections, or in off-diagonal form for metrics following the AFDM applied in extra dimension and Finsler MGTs, see a series of works \cite{vcp1,vap97,vnp97,vncs01,vjhep01,vtnpb02,vjmp06,vstoch96,vcb2,ve2} on
various applications of locally anisotropic kinetic/stochastic/wormhole/black ellipsoid solutions in modern astrophysics and cosmology. Geometric methods for constructing exact off-diagonal solutions was reviewed in \cite{vmon06,vijgmmp11}. \vskip2pt

The AFDM was applied for constructing in explicit form new classes of off-diagonal cosmological solutions in Einstein gravity (S. Vacaru, 2010) \cite{vijtp10a}. Those classes of solutions can be re-written and extended for nonholonomic variables on (co) tangent Lorentz bundles for various classes of Finsler like MGTs. For instance, the work \cite{vijmpd12} is devoted to principles of Einstein-Finsler gravity and perspectives in modern cosmology. That work was a "metric compatible" reply to Finsler like cosmological models with metric non-compatible connections, see \cite{vplb10} on critical remarks on Finsler modifications of gravity and cosmology by Zhe Chang and Xin Li. Generic off-diagonal cosmological solutions became
important for elaborating locally anisotropic inflation and cosmological scenarios in biconnection and bimetric gravity theories \cite{vcosmbc} (there were presented a series of author's talks at the Thirteenth Marcel Grossmann Meeting on General Relativity, Stockholm University, Sweden, 1-7 July 2012). Two collaborations (2012-2014) of S. Vacaru with P. Stavrinos \cite{svcqg13}\ and E. Elizalde \cite{elizalde15} were devoted respectively to a study of cyclic and ekpyrotic universes in modified Finsler osculating gravity on tangent Lorentz bundles and locally anisotropic cosmological
configurations in MGTs with nonminimal coupling.\vskip2pt

A new direction of research in accelerating cosmology and locally anisotropic dark energy and dark matter effects was stated in 2013 (see paper due to S. Vacaru published in 2016 \cite{vacaruplb16}). It was devoted to elaborating off-diagonal ekpyrotic scenarios and study of equivalence of modified, massive and/or Einstein gravity. It resulted in a collaboration
(S. Rajpoot and S. Vacaru, 2015-2016, see \cite{rajpoot17ap} and references therein) on inflation models determined by  supersymmetric geometric flows and modified scale invariant supergravity. \vskip2pt

We discuss some recent works on so-called quasiperiodic cosmological models, see \cite{amaral17,vbubuianu17} and references therein. Researchers from Eastern Europe do not have many possibilities to obtain financial support for their work. In certain cases, business companies in Western Countries are interested to provide some assistance but they impose
certain non-professional (even dilettante) conditions for such support like to "discover could fusion with gold rates and consciousness quasicrystals". Such pseudo-scientific directions are promoted for instance, by Quantum Gravity Research, QGR, Topanga, California, USA, founded by Klee Irwin in 2009, see links in internet \cite{irwin01,irwin02} in order to understand how experts in science, and K. Irwin himself, evaluate such activities. That team is supervised by  self-declared "research director" without a university education and/or individual publications on mathematics and physics (this is possible in the USA for certain tax facilities etc).\footnote{One of the goals of this Appendix is to provide objective information on conditions of research and collaboration of scientists from less favored countries. Even such an information may be very subjective it is important a correct evaluation and further development of science.} Nevertheless, even under such constrained conditions, it was possible a one year consulting agreement with QGR when this author was able to perform a research on QG and cosmology and publish in 2016-2017 for that team a series of important publications top peer reviewed journals \cite{bubuianucqg17,vacaruepjc17,amaral17}. Those papers contained new results on anamorphic quasiperiodic universes in modified and Einstein gravity, for Starobinsky inflation and dark energy and dark matter effects from quasicrystal like spacetime structures and loop quantum gravity corrections, solitons and/or pattern-forming structures in heterotic supergravity etc. Such constructions could be performed for cosmological models in MGTs with MDRs on (co) tangent Lorentz bundles. Finally, we cite certain  work on duality and deformations of black hole/ cosmological solutions by quasiperiodic and/or pattern forming structures in modified and Einstein gravity (a recent collaboration due to L. Bubuianu and S. Vacaru), see \cite{vbubuianu17} and references therein.

\subsubsection{Locally anisotropic structure in supergravity and (super) string/ branes theories}

We publish recently in such sub-directions of Direction 20 (see
\cite{vhsp98,vmon98,vmon02,vap97,vnp97,vencyc03,vncs01,vcb3,vtnpb02,vjmp05,vijgmmp08,vcqg10,vijgmmp09,gvvepjc14,gvvcqg15,
rajpoot17ap,vacaruepjc17,bubuianucqg17}):
\begin{itemize}
\item[ a)] Spinors in higher dimensional and locally anisotropic spaces

\item[ b)] Nonlinear connections in superbundles and (higher order) locally anisotropic supergravity

\item[ c)] Superstrings in higher order extensions of Finsler spaces

\item[ d)] Locally anisotropic supergravity and gauge gravity on noncommutative spaces

\item[ e)] Noncommutative Finsler gravity from (super) string / M-theory

\item[ f)] Modified dynamical supergravity breaking and off-diagonal
super-Higgs effects

\item[ i)] Supersymmetric nonholonomic configurations in scale invariant
supergravity

\item[ j)] Off-diagonal deformations of Kerr metrics and black ellipsoids in heterotic supergravity

\item[ k)] Heterotic supergravity with internal almost-K\"{a}hler spaces and deformed black holes with nonlinear waves, soliton, quasiperiodic and/or pattern-forming structures
\end{itemize}

\textit{Comments:} Sub-directions a) - k) are also related to Directions 3 and 6. We outline chronologically such research discussing some connections to a series of recent works on exact solutions with black hole, locally anisotropic and soliton/quasiperiodic/pattern-forming structures in modern supergravity and brane theories. \vskip2pt

Nonholonomic models of supergravity encoding Finsler like gravity theories have been elaborated after there were formulated theories of locally anisotropic spinors and Dirac operators, see Direction 3 and Refs. \cite{vhsp98,vmon98,vmon02}, on higher order anisotropic and Finsler-Lagrange-Hamilton superspaces. Such nonholonomic Clifford and (super) symmetric geometric models were constructed with the goal to extend the theory of locally anistoropic gravity and strings \cite{vap97}  to the
theory of superstrings in higher order extensions of Finsler superspaces \cite{vnp97}. The term "Finsler Superspace" in Concise Encyclopedia of Supersymmetry \cite{vencyc03} was grounded by a research on generalized Finsler superspaces and locally anisotropic supergravity and gauge gravity on noncommutative spaces (a collaboration with two students from R. Moldova, I. A. Chiosa and Nadejda A. Vicol, communicated at a NATO, Kiev, Ukraine, workshop and conferences at some universities in Greece) \cite{vncs01,vcb3}. Those research programs were performed during 1996-2003.\vskip2pt

Supersymmetric anholonomic frames, generalized Killing equations, and anisotropic Taub NUT spinning spaces were studied in \cite{vtnpb02}. The AFDM was generalized for supermanifolds and superstring and noncommutative gravity theories \cite{vjmp05,vijgmmp08,vcqg10}. Further developments concerned Finsler branes and quantum gravity phenomenology with Lorentz
symmetry violations \cite{vcqg11} and nonholonomic brane quantization for an A-model complexification of Einstein gravity in almost K\"{a}hler variables \cite{vijgmmp09}. All constructions were performed on various models of commutative and noncommutative (superbundles) with cosmological and black hole solutions in massive gravity and for Lagrange-Hamilton variables of various dimensions \cite{gvvepjc14} (those sub-directions of research were elaborated during 2004-2014).\vskip2pt

Nonholonomic supersymmetric geometric methods were applied in modern supergravity and supergravity theories (a collaboration during 2015-2017 between L. Bubuianu, responsible for computer support; T. Gheorghiu and O. Vacaru, responsible for computer graphics and editing; K. Irwin, providing some financial support, and S. Rajpoot, an expert on supergravity, and S.
Vacaru, an expert on nonholonomic geometry and author of the AFDM). There were elaborated such sub-directions of nonholonomic supergeometry and flow evolutions:\ a) modified dynamical supergravity breaking and off-diagonal super-Higgs effects \cite{gvvcqg15}; supersymmetric geometric flows and R2 inflation from scale invariant supergravity \cite{rajpoot17ap}. The AFDM was developed for investigating off-diagonal deformations of Kerr metrics and black ellipsoids in heterotic supergravity \cite{vacaruepjc17} and study internal almost-Kaehler spaces; instantons for SO(32), or E8 x E8, gauge groups; and deformed black holes with soliton, quasiperiodic and/or pattern-forming structures \cite{bubuianucqg17}. Such constructions, together with solutions for deforming black hole and cosmological metrics by quasiperiodic and/or pattern forming structures in modified and Einstein gravity \cite{vbubuianu17}, can be redefined for modified supergravity theories constructed on (co) tangent Lorentz superbundles (one can be considered various supersymmeteric generalizations). This provides perspectives for defining Finsler-Lagrane-Hamilton supersymmetric geometries and search for possible applications in modern gravity, cosmology and astrophysics.

\subsection{On some recent results on Finsler MGTs}

New discoveries on accelerating cosmology and related problems on dark energy and dark matter physics, and research on quantum information, requested applications of new advanced geometric methods and elaborating MGTs and QG models with non-Riemannian generalized Finsler spacetimes. During last decade, the number of papers on Finsler geometry and generalized
MGTs published in top physical journals on gravity and particle physics and cosmology, for instance, Phys. Rev. Lett., Phys. Rev. D, Phys. Lett. B etc. increased substantially.  In this subsection, we analyze some recent results published by teams of scholars and young researchers working in some directions related to our axiomatic approach to relativistic
Finsler-Lagrange-Hamilton geometries and applications in physics.

\subsubsection{International teams of scholars and young researchers}

Recently, some new teams of researchers were formed in Western Europe and begin to publish intensively on Finsler geometry and physics:

Let us analyse firstly a series of works published by authors  originating from/ working  in Germany and Estonia. Finsler geometric extensions of Einstein gravity and the possibility to define a causal structure and electrodynamics on Finsler spacetimes were considered by C. Pfeifer and M. N. R. Wohlarth (beginning 2011) \cite{pfeifer11,pfeifer12}. There were
studied Finsler-Cartan extensions of Lorentz geometry, generalized symmetries and certain cosmological models in Refs. \cite{hohmann13,hohmann16}, see references therein. The geometric and physical constructions in those work are for special Finsler connections and metrics when models of "exotic" causal structures are elaborated for a class of theories. That approach is very different from Direction 13 with an axiomatic for nonholonomic deformations of Lorentz manifolds and bundles. Following C. Pfeifer, M. Hohmann and co-authors approaches, it was not possible to elaborate on geometric methods for constructing exact solutions and general methods of quantization like for Directions 10, 15, and 16. Certain examples of Finsler like solutions for "Finslerian extensions of the Schwarzschild metric" \cite{silagadze15} and Finsler pp-waves \cite{fuster16} do not involve the N-connection structure (which is crucial for nontrivial Finsler configurations) and well-defined generalizations of the Einstein and matter field equations in Finsler geometry. Such examples of exact solutions
and approximate methods do not compete with the AFDM and there are not cited former results published in tenths of works on Finsler black holes/ellipsoids/cosmological models Directions 7, 11,12, 18, 19, 20.%
\vskip2pt

Secondly, another series of works is devoted to phase space geometry and theories with MDRs \cite{amelino14,barcaroli15,lobo16,lobo17} (G. Amelino-Camelia, L. Barcaroli, L. K. Brunkhorst, G. Gubitosi, I. P. Lobo, N.
Loret, F. Nettel and co-authors, 2014-2017). Authors investigated Finsler generalizations of Plank-scale-deformed relativity, rainbows metrics, and realization of doubly special relativistic symmetries for Finsler-Hamilton configurations. There are two main issues (problems) of such works: The first one is that the constructions are not based on self-consistent exact or
parametric solutions of correspondingly generalized Einstein equations like in Directions 10, 11, 18, and 19. The second one is that there were not elaborated in such works any complete scheme of quantization of Finsler-Lagrange-Hamilton theories like we attempted in Directions 15 and/or 16. For instance, the dependencies on the energy type variables in the works
with rainbow metrics are pure parametrical and phenomenological and not explained theoretically in a self-consistent form by certain classical or quantum solutions.\vskip2pt

Finally (in this subsection), we note that a series of works \cite{aazami14,minguzzi15a,caponio14,caponio16,caponio17} due to authors from Spain and Italy, and collaborations, (B. Aazami, E. Caponio, M. A. Javaloyes, E. Minguzzi, M. S\'{a}nchez, and G. Stancarone, 2014-2017). The results on Zermelo's navigation and causality of Finsler like spacetimes are connected and important for developing our Direction 13. The results on singularity theorems in certain models of Finsler spacetime are of special interest and related to Directions 18 and 19. Nevertheless, it should be noted that such constructions are very dependent on the type of (non) linear connection and metric structures and that mentioned authors have not supported their results with examples of exact solutions for well-defined Finsler generalizations of the Einstein equations. For instance, it is not clear how ideas like "standard static Finsler structures" could be realized physically for nontrivial N-connections. Any exact or parametric solution in a model of Einstein-Finsler gravity results in generic off-diagonal terms and generic
locally anisotropic stationary and non-stationary configurations as we proved in Directions 10, 18, and 19. The Raychaudhuri equations can be generalized/ lifted on (co) tangent bundles using geometric methods but the constructions, and results on singularity theorems, are strongly dependent on the type of Finsler nonlinear and distinguished connection structures are
used.

\subsubsection{Finsler geometry and applications in China, Egypt, India, Iran, Pakistan etc.}

In this subsection, we restrict our attention to research on Finsler geometry and MGTs in some countries outside North and South America, Europe and Japan. 

We begin with a series of important works published during 2004-2009 by researchers with affiliations in Iran (E. Azizpour, D. Latify, A. Razavi, M. M. Rezaii and others) when some authors  collaborated with researchers in the UK, see \cite{latifi06,esfrafilian04,rezaii05,tavakol86,tavakol09} and references therein. Such papers were devoted to a study of the spacetime geometry and Finsler geometry, on homogeneous Finsler spaces. The works on nonlinear connections and supersprays in Finsler and Lagrange superspaces develops and cite our results outlined in Direction 20. Hidden connections between general relativity and Finsler geometry were studied by M. Panahi \cite{panahi03} and a series of works were devoted to constructing K\"{a}hler and Riemannian metrics arising from Finsler structures (A. Tayebi and E. Peyghan) \cite{tayebi10,tayebi11}. The results of such works are alternative (it is not clear how they can be related to almost K\"{a}hler-Finsler structures and relativistic configurations) to those obtained in Directions 15 and 10. A recent paper \cite{nekouee17} studies the relation between non-commutative and Finsler geometry in Horava-Lifshitz black holes (there are not constructed exact solutions with nontrivial N-connection structure) and dubs in certain sense (but not cite) former similar results in Directions \ 2, 10, 13. \vskip2pt

An influent school on global Finsler geometry and non-Riemannian geometries was created in Egypt (and collaborations with French mathematicians), see works \cite{youssef09,youssef12} due to N. L. Youssef and co-authors. \vskip2pt

One of the most prolific team of young researchers (Zhe Chang, Sai Wang, Xin Li and others, beginning 2008) publishing important works on Finsler MGTs and cosmology was created in China, see \cite{chang08,chang12,li14,li14a,li16} and references therein. Their research covers a number of directions related to Finsler modifications of Newton's gravity, locally anisotropic variations of fundamental constants and interactions, Finslerian perturbations of the $\Lambda $CDM model and alternatives to the dark energy and dark matter physics. Such research is supported by a collaboration with mathematicians in the USA (Z. Shen and co-authors) \cite{bao00,shen01,shen01a}. The bulk of such applications were based on geometric constructions with the Chern-Finsler connections which is metric noncompatible, see critical remarks in \cite{vplb10,vijgmmp12,vijmpd12}. Here we also note that the attempts to derive exact solutions of vacuum field equations and spherical symmetrical spacetime solutions in Finsler gravity seem to be not complete and not self-consistent. It is not clear what type of anisotropy would describe spherical configurations with the highest degree of symmetry and with trivial N-connection structures. \vskip2pt

Recently, Finsler geometric methods were applied by a team researchers in India in order to elaborate on locally anisotropic star and wormhole models \cite{rahman15,rahman16}. Finally, we note a very original research involving Finsler geometry and particle entanglement was performed in Pakistan \cite{wasay17}.

\end{document}